\documentclass[prd,aps,nofootinbib]{revtex4}
\usepackage{amssymb,amsmath,mathtools,dsfont}
\usepackage{epsfig,wrapfig}
\usepackage{subfig}
\usepackage[bookmarks]{hyperref}
\usepackage[usenames,dvipsnames]{color}

\newcommand{\Natural}{\mathbb{N}}
\newcommand{\Real}{\mathbb{R}}

\newcommand{\diag}{\mbox{diag}}
\newcommand{\divrg}{\mbox{div}\,}
\newcommand{\dvol}{\mbox{dvol}}
\newcommand{\pz}{\mbox{\em \r{p}\hspace{0.3mm}}}
\newcommand{\piz}{\mbox{\em $\mathring{\pi}$\hspace{0.3mm}}} 
\newcommand{\ve}[1]{\underline{#1}}
\newcommand{\proof}{\noindent {\bf Proof. }}

\newcommand{\qed}{\hfill \fbox{} \vspace{.3cm}}
\newcommand{\Ker}{\mathrm{Ker}\,}

\newtheorem{definition}{Definition}
\newtheorem{lemma}{Lemma}
\newtheorem{proposition}{Proposition}

\newtheorem{theorem}{Theorem}

\begin{document}

\title{An introduction to the relativistic kinetic theory on curved spacetimes}

\author{Rub\'en O. Acu\~na-C\'ardenas, Carlos Gabarrete and Olivier Sarbach}
\affiliation{Instituto de F\'\i sica y Matem\'aticas,
Universidad Michoacana de San Nicol\'as de Hidalgo,\\
Edificio C-3, Ciudad Universitaria, 58040 Morelia, Michoac\'an, M\'exico.}

\begin{abstract}
This article provides a self-contained pedagogical introduction to the relativistic kinetic theory of a dilute gas propagating on a curved spacetime manifold $(M,g)$ of arbitrary dimension. Special emphasis is made on geometric aspects of the theory in order to achieve a formulation  which is manifestly covariant on the relativistic phase space. Whereas most previous work has focused on the tangent bundle formulation, here we work on the cotangent bundle associated with $(M,g)$ which is more naturally adapted to the Hamiltonian framework of the theory.

In the first part of this work we discuss the relevant geometric structures of the cotangent bundle $T^* M$, starting with the natural symplectic form on $T^* M$, the one-particle Hamiltonian and the Liouville vector field, defined as the corresponding Hamiltonian vector field. Next, we discuss the Sasaki metric on $T^* M$ and its most important properties, including the role it plays for the physical interpretation of the one-particle distribution function.

In the second part of this work we describe the general relativistic theory of a collisionless gas, starting with the derivation of the collisionless Boltzmann equation for a neutral simple gas. Subsequently, the description is generalized to a charged gas consisting of several species of particles and the general relativistic Vlasov-Maxwell equations are derived for this system.

The last part of this work is devoted to a transparent derivation of the collision term, leading to the general relativistic Boltzmann equation on $(M,g)$. To this end, we introduce the collision manifold, describing the set of all possible binary elastic collisions and discuss its most important geometric properties, including the metric and volume form it is equipped with and its symmetries. We show how imposing full Lorentz symmetry leads to microscopic reversibility and the relativistic $H$ theorem. The meaning of global and local equilibrium and the stringent restrictions for the existence of the former on a curved spacetime are discussed. We close this article with an application of our formalism to the expansion of a homogeneous and isotropic universe filled with a collisional simple gas and its behavior in the early and late epochs.
\end{abstract}

\date{\today}


\maketitle
\tableofcontents

\section{Introduction}
\label{Sec:Intro}

The kinetic theory of dilute gases is an elegant and profound theory which allows one to understand a wide range of interesting macroscopic phenomena from the microscopic laws of particle mechanics. Three prominent examples are: (i) the deduction of the equations underlying fluid dynamics for a gas in a near-equilibrium state including the equation of state and the terms describing viscosity and heat transport~\cite{Huang-Book,DavidTong-Book}, (ii) Landau damping (see~\cite{cMcV11} and references therein), a relaxation process that takes place in a collisionless charged gas and plays a central role in plasma physics, and (iii) the modeling of a galaxy as a self-gravitating kinetic gas in which the stars play the role of the gas particles~\cite{BinneyTremaine-Book}. A key concept in kinetic theory is the characterization of the state of the gas through the {\em one-particle distribution function}, a time-dependent function defined on the one-particle phase space of the theory, whose time evolution is determined by Boltzmann's equation. The relevant macroscopic quantities (such as particle density, energy density, pressure, entropy, etc.) are obtained {\em a posteriori} from the one-particle distribution function from suitable averages over the momentum space.

In recent years, there has been an increasing interest in applications of {\em relativistic} kinetic theory in which a significant fraction of the gas particles have relativistic speeds or are subject to strong gravitational fields. Typical situations involving relativistic speeds are encountered in hot plasmas (i.e. a gas at such extreme temperatures that the electrons are separated from the nuclei) when the electron temperature reaches about $10^{10} K$. On the other hand, gas particles which are subject to a strong gravitational fields are found in the early stages of our Universe (and thus are relevant for cosmology, see e.g.~\cite{cMeB95}) or in accretion processes around compact objects such as black holes. In particular, the study of a kinetic gas in the vicinity of black holes is likely necessary to gain a more thorough understanding of the properties of accretion disks and the jet launching mechanism which can now be observed at event-horizon-scales, as has recently been revealed by the spectacular results from the Event Horizon Telescope Collaboration who was able to reconstruct the image of the supermassive black hole in the center of the galaxy M87~\cite{kAetal19}. In such situations, a (special or general) relativistic formulation of the theory must be put forward which generalizes the well-known notions from the non-relativistic theory to a manifestly Poincar\'e-covariant or diffeomorphism-invariant setting.

The special relativistic formulation has a long history and, in fact, started just a few years after Einstein introduced the theory of special relativity with J\"uttner's work in 1911~\cite{fJ11a,fJ11b}. In his work, J\"uttner generalized the Maxwell-Boltzmann distribution for an equilibrium gas to the relativistic case and also derived the equation of state describing such a gas. A next important step was performed by Synge in 1934~\cite{jS34} who based the description of the gas on the world lines of the individual gas particles (as opposed to their position and linear momentum in phase space at a given time) and in this way made it possible to provide a covariant description of the theory. For further work and textbooks on the special relativistic kinetic gas theory, see for example Refs.~\cite{Synge-Book1,Synge-Book2,Groot-Book,CercignaniKremer-Book}.

The general relativistic formulation of kinetic gases began shortly after Synge's work in 1934, and was elaborated by Tauber and Weinberg~\cite{gTjW61}, Israel~\cite{wI63}, Lindquist~\cite{rL66} among others, and put on a more geometric basis by Ehlers~\cite{jE71,jE73}. These works led to the development of transient relativistic thermodynamics~\cite{Stewart-Book,wI76,jS77,wIjS76,wIjS79a,wIjS79b,wHlL83,wHlL85}, which allows one to formulate irreversible thermodynamics for fluid or fluid mixtures which is applicable to nonstationary processes without violating causality. Further important achievements that have been at least partially motivated by kinetic theory include extended thermodynamics~\cite{MullerRuggeri-Book} and divergence-type fluids~\cite{iLiMtR86,sP87,rGlL90,gNoR95,oRgN97} which provide theories for dissipative relativistic fluids with hyperbolic and causal propagation. For a recent summary on  the progress in the field of relativistic nonequilibrium thermodynamics and open issues, see~\cite{jStZ20}. For a textbook on the general relativistic kinetic theory, see the one by Cercignani and Kremer~\cite{CercignaniKremer-Book} and also Refs.~\cite{oStZ13,oStZ14b,oStZ14a} for recent reviews on the formulation of relativistic kinetic theory with special emphasis on the geometric structures of the tangent bundle $TM$ associated with the spacetime manifold $(M,g)$ in a modern differential geometric language.

There has been considerable progress on mathematical aspects of relativistic kinetic theory, including the (local in time) well-posedness for the Cauchy problem of the Einstein-Maxwell-Boltzmann system~\cite{dByC73} and, more recently, results on the nonlinear stability of Minkowski spacetime~\cite{mT17,hLmT20,dFjJjS21,lBdFjJjSmT20} and the future stability of the Universe~\cite{Ringstrom-Book,hAhR16,dF16,lAdF20,hBdF20,jJmTjV20} for the Einstein-Vlasov or Einstein-Maxwell-Vlasov systems. Other mathematical work on the Einstein-Vlasov system has analyzed the complete gravitational collapse of a spherical cloud~\cite{aRjV10,hA12,hA14} or has established the existence of static, spherically symmetric solutions~\cite{gR93,hAgR06,hAdFmT15} or axisymmetric solutions with are either static or stationary~\cite{hAmKgR11,hAmKgR14}; see also Refs.~\cite{hAmEgR09,hAgR06b,aAmC14,eAhAaL16,eAhAaL19} for related numerical works and~\cite{hA11} for a review on these results. For recent work on the relativistic Boltzmann equation including the collision term, see for instance~\cite{hLaR13,hL13,hLeN17,hLeN18,hLeNpT19}. Recently, the properties of the solutions to the Vlasov equation on a fixed black hole background have also been studied in the context of the accretion problem~\cite{pRoS17,pRoS17b,aCpM20,pMaO21a,pMaO21b,aGcGpDdNoS21}, in the context of mixing of a gas of massive particles~\cite{pRoS18,pRoS20}, and in the context of decay in the case of massless particles~\cite{lApBjS18,lB20}.

The present work aims at providing a self-contained, pedagogical introduction to the general relativistic formulation of kinetic gases (which contains the special relativistic formulation as a particular case). Unlike most previous work, which defines the relativistic phase space as an appropriate submanifold of the tangent bundle $TM$, we base our formulation on the {\em cotangent} bundle $T^*M$ associated with the spacetime manifold $(M,g)$. Of course, the two formulations are equivalent, since the spacetime metric $g$ provides a natural isomorphism between $TM$ and $T^*M$; however, we here choose to focus on the cotangent bundle because it is more naturally adapted to the Hamiltonian formulation used in this review. For recent work based on the cotangent bundle formulation, see Refs.~\cite{pR19-thesis,pRoS17,pRoS18,jJmTjV20,aCpM20,pMaO21a,pMaO21b}.

This article is organized as follows. In section~\ref{Sec:Cotangent} we introduce the cotangent bundle $T^*M$ associated with $(M,g)$ and recall a few geometric concepts that are key to the formulation of the relativistic kinetic gas theory, including the natural symplectic form it admits, the mass shell condition and the Sasaki metric on $T^*M$ which is naturally induced from the spacetime metric $g$. Next, in section~\ref{Sec:SimpleGas} we provide the description of a relativistic, collisionless and uncharged simple gas, that is, a gas configuration consisting of classical, identical, massive, neutral and non-spinning particles which follow future-directed timelike geodesic trajectories in $(M,g)$. In particular, we introduce the one-particle distribution function describing such a gas configuration, derive the collisionless Boltzmann equation and discuss the expressions and main properties of the relevant observables (including the particle current density and energy-momentum-stress tensor) describing the macroscopic properties of the gas. In section~\ref{Sec:Charged} we generalize the results of section~\ref{Sec:SimpleGas} to the case of a gas configuration consisting of several species of (possibly charged) particles and derive the relativistic Vlasov-Maxwell system on a curved spacetime $(M,g)$ which is relevant to the description of a relativistic plasma propagating in a (possibly strong) gravitational field. Section~\ref{Sec:Collision} is devoted to the derivation of the collision term, where for simplicity and definiteness we limit ourselves to the case of purely elastic binary collisions between gas particles of a single species, and to the derivation of the general relativistic Boltzmann equation. Next, in section~\ref{Sec:Equilibrium} we discuss the relativistic H-theorem which gives rise to a relativistic formulation of the second law of thermodynamics. The {\em global equilibrium configurations} are defined as those solutions of the Boltzmann equation for which the entropy production is zero everywhere on $M$, and as we will see, for a gas configuration on a curved spacetime $(M,g)$ in an electromagnetic field $F$ this yields rather strong conditions on $g$ and $F$. The more general concept of {\em local equilibrium configurations} which does not require such strong conditions on $g$ and $F$ is also discussed and motivated in section~\ref{Sec:Equilibrium}. In section~\ref{Sec:MomentMethod} we provide a brief outline of the method of moments, which is used to convert the integro-differential Boltzmann equation into a system of conservation laws on the spacetime manifold, and serves as a starting point for several approximation techniques. Next, in section~\ref{Sec:Application} we apply our formalism to the evolution of a Friedmann-Lema\^itre-Robertson-Walker spacetime filled with a collisional, isotropic and homogeneous gas in the expanding direction and assuming a particular ansatz for the differential cross section, we show that the limits of early and late epochs can be treated analytically. Conclusions are drawn in section~\ref{Sec:Conclusions}. Technical details, including the analysis of the manifold structure of the cotangent bundle, the use of an orthonormal set of basis covectors, an alternative derivation of the volume element on the mass hyperboloid, the mass shell, and the collision manifold, as well as a compilation of several definitions and useful results regarding symmetries are discussed in appendices~\ref{App:Cotangent}--\ref{App:Symmetries}.

Throughout this work, $(M,g)$ denotes a $n$-dimensional, $C^\infty$-differentiable, connected and time-oriented Lorentzian manifold. Greek Letters refer to spacetime indices running from $0,1,\ldots,d=n-1$ while Latin indices run from $1,2,\ldots,d$. We use the signature convention $(-,+,+,\ldots,+)$ for the spacetime metric and geometrized units in which Newton's constant and the speed of light are one, i.e. $G_N = c = 1$. We shall make use of the modern differential geometry language, which has the conceptual advantage of exhibiting the general covariance of the theory. For example, ${\cal F}(N)$ and ${\cal X}(N)$ refer to the class of smooth (i.e. $C^\infty$-differentiable) functions and vector fields, respectively, on a $C^\infty$-differentiable manifold $N$. If $X$ denotes a vector field and $\omega$ a $p$-form on $N$, then $\pounds_X\omega$, $i_X\omega$ and $d\omega$ refer to the Lie, the interior and the exterior derivatives of $\omega$, respectively.\footnote{See, for example, chapter~14 in Ref.~\cite{Straumann-Book} for a definition of these derivatives and a discussion of Cartan's calculus of differential forms.} Nevertheless, we shall also give the most relevant equations in their coordinate form for the reader who is unfamiliar with the differential geometry language or is mostly interested in applications.

\section{Geometric properties of the cotangent bundle}
\label{Sec:Cotangent}

In this section we introduce the cotangent bundle $T^* M$ associated with the spacetime manifold $(M,g)$ and discuss its most important geometric properties that are relevant for the formulation of relativistic kinetic theory. In subsection~\ref{SubSec:T*M} we start by recalling the definition of $T^* M$ and its basic properties. Next, in subsection~\ref{SubSec:Ham} we introduce a natural symplectic form on $T^* M$ used for the Hamiltonian formulation of the theory. In subsection~\ref{SubSec:MassShell} we define the future mass shell (for positive masses) $\Gamma_m^+$ and show that it describes a submanifold of $T^* M$. Physically, it represents the relativistic phase space for a simple kinetic gas. Next, in subsection~\ref{SubSec:Sasaki} we show that the spacetime metric $g$ and the associated Levi-Civita connection induces a natural metric on $T^* M$, called the Sasaki metric. This metric, in turn, induces a natural metric and volume form on $\Gamma_m^+$, which will turn out to be important for the physical interpretation of the one-particle distribution function. Finally, in subsection~\ref{SubSec:Liouville}, we discuss Liouville's theorem.

\subsection{Definition and basic properties of $T^* M$}
\label{SubSec:T*M}

In the following, for any spacetime event $x\in M$, we denote by $T_x M$ the tangent space of $M$ at $x$; and $T_x^* M$ is the cotangent space at $x$, that is, the linear space consisting of all covectors $p$ at $x$. We shall say that an element $p\in T_x^* M$ is future-directed timelike, if and only if its associated vector $\tilde{p} := g^{-1}(p,\cdot)$ lies in the future light cone at $x$. Physically, such elements describe the canonical momentum of a gas particle. We define:

\begin{definition}
\label{Def:Cotangent}
The cotangent bundle $T^*M$ is the set
\begin{equation}
\boxed{ T^* M := \{ (x,p) : x\in M, p\in T^*_x M \}. }
\end{equation}
\end{definition}
Associated with it is the natural projection map
\begin{equation}
\begin{array}{cccc}
\pi: & T^* M & \to & M,\\
 & \left(x,p\right) & \mapsto & x,
\end{array}
\label{Eq:ProjectionMap}
\end{equation}
which projects an element $(x,p)\in T^* M$ onto its base point $x\in M$. The fibre at $x$ is the inverse set
\begin{equation}
\pi^{-1}(x) = \left(x,T_{x}^* M \right) \simeq T_x^* M,
\end{equation}
see Fig.~\ref{Fig:T*M}. The spacetime metric induces a natural metric (the fibre metric) on $\pi^{-1}(x)$:
\begin{equation}
h_x\left( (x,p),(x,q) \right) := g_x^{-1}(p,q),\qquad (x,p),(x,q)\in \pi^{-1}(x).
\end{equation}

The following lemma shows that the cotangent bundle inherits the manifold property of $M$, and that it is orientable regardless of whether or not $M$ is orientable.

\begin{lemma}
\label{Lem:CotangentSpace2n}
$T^* M$ is a $2n$-dimensional, orientable manifold.
\end{lemma}

\begin{proof}
Given local coordinates $x^\mu$ in a neighborhood $U$ of $M$, one can assign to each point $(x,p)\in \pi^{-1}(U)$ the local coordinates $(x^\mu,p_\mu)$, where $x^\mu$ are the local coordinates associated with $x$ in $U$ and $p_\mu$ are the components of $p$ with respect to the basis covectors $(dx^\mu_x)$ of $T_x^* M$, that is,
\begin{equation}
p = p_\mu dx^\mu_x.
\end{equation}
This defines a coordinate chart on $\pi^{-1}(U)$. By taking a differentiable atlas of $M$ one obtains a corresponding differential atlas of $T^* M$ which is oriented, see appendix~\ref{App:Cotangent} for the details.
\qed
\end{proof}

{\bf Remark}: We call the coordinates $(x^\mu,p_\mu)$ adapted local coordinates on $T^* M$. These coordinates provide at each point $(x,p)\in T^* M$ a basis of vector fields
\begin{equation}
\left\{ \left. \frac{\partial}{\partial x^\mu} \right|_{(x,p)}, \left. \frac{\partial}{\partial p_\mu} \right|_{(x,p)} \right\}
\end{equation}
and the associated dual basis of covector fields
\begin{equation}
\left\{ dx^\mu_{(x,p)}, (dp_\mu)_{(x,p)} \right\}.
\end{equation}
Therefore, a vector field $X$ on $T^* M$ can be expanded, locally, as
\begin{equation}
X = X^\mu\frac{\partial}{\partial x^\mu} + Y_\mu\frac{\partial}{\partial p_\mu},\qquad
X^\mu = dx^\mu(X),\quad
Y_\mu = dp_\mu(X).
\label{Eq:VectorFieldX}
\end{equation}

\begin{figure}[h]
\begin{centering}
\includegraphics[scale=0.50]{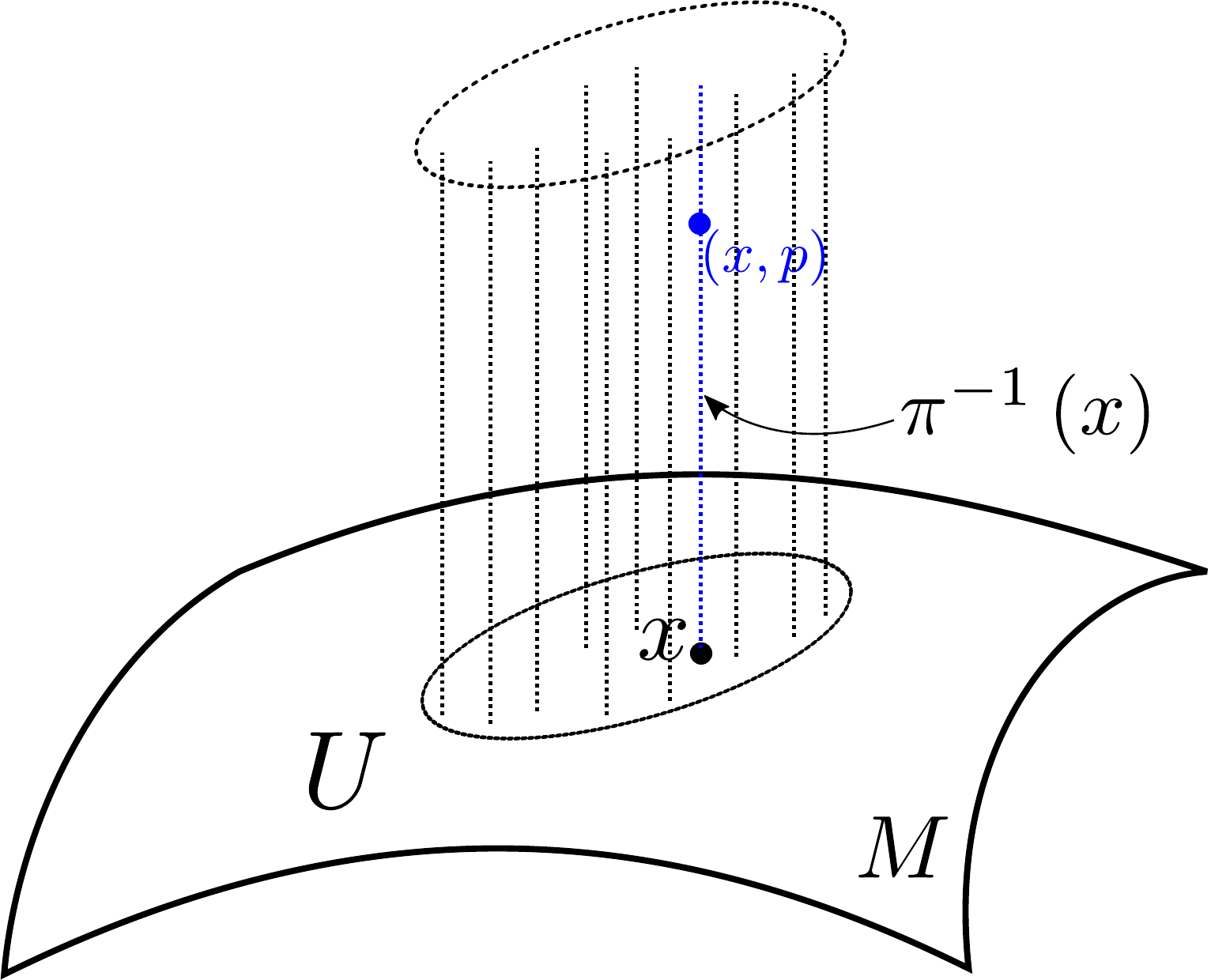}
\par\end{centering}
\caption{An illustration of the spacetime manifold $M$ and some fibres $\pi^{-1}(x)$ over the open neighborhood $U\subset M$.}
\label{Fig:T*M}
\end{figure}

\subsection{Symplectic form and Hamiltonian formulation}
\label{SubSec:Ham}

After having defined the cotangent bundle $T^* M$, we introduce on it a symplectic form, that is an antisymmetric bilinear form $\Omega_s: {\cal X}(T^* M)\times {\cal X}(T^* M)\rightarrow {\cal F}(T^* M)$ which is
\begin{enumerate}
\item[(i)] non-degenerate: $\displaystyle \Omega_{s}(\:\cdot\:,X) = 0 \Longrightarrow X=0$,
\item[(ii)] closed: $\displaystyle d\Omega_{s}=0$.
\end{enumerate}
The symplectic form arises naturally as the differential of the Poincar\'e (or canonical) one-form which is given in the following definition.

\begin{definition}
\label{Def:simplectic_form}
The Poincar\'e one-form over $T^{*}M$ is the one-form $\Theta: {\cal X}(T^{*}M)\rightarrow {\cal F}(T^{*}M)$ defined as
\begin{equation}
\Theta_{(x,p)}(X_{(x,p)}):=p(d\pi_{(x,p)}(X_{(x,p)})),\qquad
X_{(x,p)}\in T_{(x,p)}(T^{*}M),
\end{equation}
where $d\pi_{(x,p)}: T_{(x,p)}(T^{*}M)\to T_x M$ is the differential of the projection map $\pi$ defined in Eq.~(\ref{Eq:ProjectionMap}) at $(x,p)$.
\end{definition}

In terms of adapted local coordinates $(x^\mu,p_\mu)$ and the corresponding expansion in Eq.~(\ref{Eq:VectorFieldX}) it is not difficult to show that
$$
d\pi_{(x,p)}\left(X_{(x,p)}\right) = X^\mu(x,p)\left.\frac{\partial }{\partial x^\mu}\right|_x.
$$
Consequently,
$$
\Theta_{(x,p)}(X_{(x,p)}) = p\left( X^\mu(x,p)\left.\frac{\partial }{\partial x^\mu}\right|_x \right) = p_\mu X^\mu(x,p)
 = p_\mu dx^\mu_{(x,p)} (X_{(x,p)}),
$$
from which one obtains the following expression for the Poincar\'e one-form in terms of adapted local coordinates:
\begin{equation}
\boxed{ \Theta = p_{\mu}dx^{\mu}. }
\end{equation}

The exterior differential of the Poincar\'e one-form defines the symplectic form on the cotangent bundle $T^{*}M$:
\begin{equation}
\boxed{\Omega_s := d\Theta=dp_{\mu}\wedge dx^{\mu}.}
\label{Eq:Omega}
\end{equation}
By construction, $\Omega_s$ is a closed two-form. In order to show that $\Omega_s$ is non-degenerate, we take an arbitrary vector field $X$ on $T^* M$ which we decompose according to Eq.~(\ref{Eq:VectorFieldX}). Then, we have
\begin{equation}
\Omega_s(\:\cdot\:,X) = X^\mu dp_\mu - Y_\mu dx^\mu.
\label{Eq:OmegaDotX}
\end{equation}
Obviously, the right-hand side is zero if and only if both $X^\mu$ and $Y_\mu$ vanish, that is, if and only if $X = 0$. This proves that $\Omega_s$ is non-degenerate.

The existence of the symplectic form allows one to introduce the Hamiltonian vector field associated with a given function $\mathcal{H}$ on $T^*M$:

\begin{definition}
\label{Def:HamVF}
Given a smooth function $\mathcal{H}\in\mathcal{F}(T^{*}M)$ on the cotangent bundle, the associated Hamiltonian vector field $X_{\mathcal{H}}\in {\cal X}(T^{*}M)$ is defined by
\begin{equation}
\boxed{d\mathcal{H}=\Omega_{s}(\:\cdot\:,X_{\mathcal{H}})=-i_{X_{\mathcal{H}}}\Omega_{s}.}
\end{equation}
\end{definition}

Explicitly, in terms of adapted local coordinates, if $X = X_{\mathcal{H}}$ is parametrized as in Eq.~(\ref{Eq:VectorFieldX}), then it follows from Eq.~(\ref{Eq:OmegaDotX}) that
\begin{equation}
\boxed{ X_{\mathcal{H}} = \frac{\partial\mathcal{H}}{\partial p_\mu} \frac{\partial}{\partial x^\mu}
 -  \frac{\partial\mathcal{H}}{\partial x^\mu} \frac{\partial}{\partial p_\mu}. }
\label{Eq:XH}
\end{equation}
The integral curves of $X_{\mathcal{H}}$ are determined by Hamilton's equations of motion
\begin{eqnarray}
\frac{dx^\mu}{d\lambda} &=& \frac{\partial\mathcal{H}}{\partial p_\mu},\\
\frac{dp_\mu}{d\lambda} &=& -\frac{\partial\mathcal{H}}{\partial x^\mu},
\end{eqnarray}
with $\lambda$ the parameter along the curve.

A particular important example for the purpose of relativistic kinetic theory is the free one-particle Hamiltonian, given by
\begin{equation}
\boxed{ \mathcal{H}(x,p)=\frac{1}{2}g^{-1}_{x}(p,p)=\frac{1}{2}g^{\mu\nu}(x)p_{\mu}p_{\nu}.}
\label{Eq:FreeParticleH}
\end{equation}
The associated Hamiltonian vector field, also called the Liouville vector field, is explicitly given by
\begin{equation}
\boxed{ L := X_{\mathcal{H}} = g^{\mu\nu}p_{\nu}\frac{\partial}{\partial x^{\mu}}-\frac{1}{2}\frac{\partial g^{\alpha\beta}}{\partial x^{\mu}}p_{\alpha}p_{\beta}\frac{\partial}{\partial p_{\mu}}. }
\label{Eq:LiouvilleVF}
\end{equation}
In this case, the integral curves of $L$, when projected onto the spacetime manifold by means of $\pi$, describe (affinely parametrized) geodesics of $(M,g)$, see Fig.~\ref{Fig:Liouville}.

\begin{figure}
\center
\includegraphics[scale=0.32]{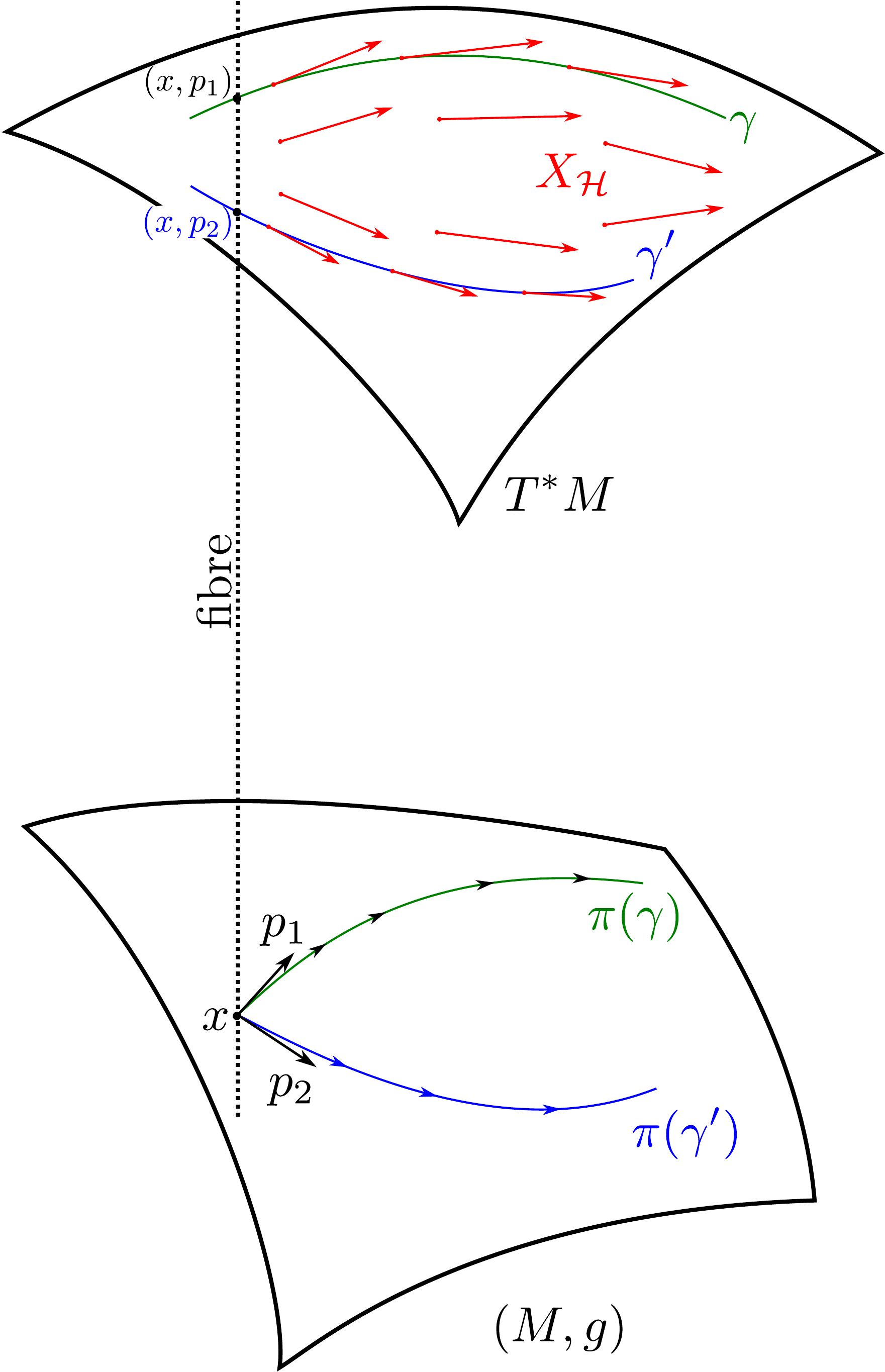}
\caption{An illustration of the Hamiltonian vector field $X_{\mathcal{H}}$ on the cotangent bundle $T^{*}M$ and two of its integral curves $\gamma$ and $\gamma'$ which are projected onto the spacetime manifold $M$. Note that although $\pi(\gamma)$ and $\pi(\gamma')$ cross each other at $x$, the corresponding curves in $T^* M$ do not cross since they have different momenta at $x$.}
\label{Fig:Liouville}
\end{figure}

The Hamiltonian vector field associated with any function $\mathcal{H}$ leaves the symplectic form invariant. This is a direct consequence of Cartan's formula, the fact that $d\Omega_s = 0$, and Definition~\ref{Def:HamVF}:
\begin{equation}
\pounds_{X_{\mathcal{H}}}\Omega_s 
 = (d i_ {X_{\mathcal{H}}} + i_{X_{\mathcal{H}}} d)\Omega_s = -d^2\mathcal{H} = 0.
\label{Eq:LOmega}
\end{equation}
Therefore, $X_{\mathcal{H}}$ generates canonical transformations. Next, we introduce:

\begin{definition}
\label{Def:PoissonBracket}
Given two smooth functions $\mathcal{H},\mathcal{G}\in {\cal F}(T^* M)$, their Poisson bracket is defined by
\begin{equation}
\boxed{ \left\{ \mathcal{H},\mathcal{G}\right\} := \Omega_{s}(X_{\mathcal{H}},X_{\mathcal{G}}). }
\end{equation}
\end{definition}

According to Definition~\ref{Def:HamVF}, it describes the change of $\mathcal{G}$ along the Hamiltonian flow associated with $\mathcal{H}$:
\begin{equation}
X_{\mathcal{H}}[\mathcal{G}] = d\mathcal{G}(X_{\mathcal{H}}) =  \left\{ \mathcal{H},\mathcal{G}\right\}.
\label{Eq:dG}
\end{equation}

\begin{lemma}
The Poisson bracket satisfies the following well-known properties for all $\mathcal{F},\mathcal{G},\mathcal{H}\in {\cal F}(T^* M)$ and all $\lambda\in \Real$:
\begin{enumerate}
\item[(i)] $\displaystyle \left\{ \mathcal{H},\mathcal{G}\right\}=-\left\{\mathcal{G}, \mathcal{H}\right\}$.
\item[(ii)] $\displaystyle \left\{ \mathcal{H},\mathcal{F} + \lambda\mathcal{G}\right\}=\left\{\mathcal{H}, \mathcal{F}\right\}+\lambda\left\{\mathcal{H}, \mathcal{G}\right\}$.
\item[(iii)] $\displaystyle \left\{ \mathcal{H},\left\{ \mathcal{G},\mathcal{F}\right\}\right\}+\left\{ \mathcal{G},\left\{ \mathcal{F},\mathcal{H}\right\}\right\}+\left\{ \mathcal{F},\left\{ \mathcal{H},\mathcal{G}\right\}\right\}=0$.
\item[(iv)] In adapted local coordinates one has
\begin{equation}
\boxed{\left\{ \mathcal{H},\mathcal{G}\right\}
= \frac{\partial \mathcal{H}}{\partial p_{\mu}}\frac{\partial\mathcal{G}}{\partial x^{\mu}}
 -\frac{\partial \mathcal{H}}{\partial x^{\mu}}\frac{\partial\mathcal{G}}{\partial p_{\mu}}. }
\label{Eq:PoissonBracketCoordinate}
\end{equation} 
\end{enumerate}
\end{lemma}

\begin{proof}
(i)--(iii) follow easily from the definitions above. (iv) follows immediately from Eq.~(\ref{Eq:XH}) and the fact that $\left\{ \mathcal{H},\mathcal{G}\right\} = d\mathcal{G}(X_{\mathcal{H}})$.
\qed
\end{proof}

Property (i) of the Lemma together with the identity~(\ref{Eq:dG}) implies the Hamiltonian version of Noether's theorem:

\begin{theorem}
$\mathcal{G}$ is invariant under the Hamiltonian flow associated with $\mathcal{H}$ if and only if $\mathcal{H}$ is invariant under the flow of $\mathcal{G}$. (That is, $\mathcal{G}$ is an integral of motion if and only if $X_{\mathcal{G}}$ generates a continuous symmetry of $\mathcal{H}$.)
\end{theorem}

\subsection{Mass shell}
\label{SubSec:MassShell}

In this section we introduce the mass shell, defined as the following subset of the cotangent bundle:

\begin{definition}
Let $m\geq 0$. The mass shell is defined by
\begin{equation}
\boxed{ \Gamma_m := \left\{ (x,p)\in T^{*}M : g^{-1}_{x}(p,p)=-m^{2} \right\}. }
\end{equation}   
\end{definition}

\begin{lemma}
Let $m > 0$ be positive. Then $\Gamma_m$ it is a differentiable submanifold of $T^{*}M$ of dimension $2n-1$.
\end{lemma}

\begin{proof}
Note that $\Gamma_m$ is the level set of the free-particle Hamiltonian function~(\ref{Eq:FreeParticleH}) with constant $-m^2/2$. Since
$$
d\mathcal{H} = -i_{X_\mathcal{H}}\Omega_s,
$$
and since $\left. X_\mathcal{H} \right|_{(x,p)}$ is different from zero as long as $p\neq 0$, it follows that this set describes a smooth hypersurface of codimension $1$.
\qed
\end{proof}

{\bf Remark}: When $m = 0$ the set $\Gamma_m$ fails to be differentiable at points where $p$ vanishes.\\

Since $(M,g)$ is assumed to be connected and time-oriented, the mass shell (with $m > 0$) consists of precisely two connected components (see, for instance, Proposition~3 in~\cite{oStZ13}). Hence, $\Gamma_m$ is the disjoint union
\begin{equation}
\Gamma_m = \Gamma_m^+\, \dot{\cup}\, \Gamma_m^-,
\end{equation}
with $\Gamma_m^\pm$ denoting the future (past) mass shell, respectively. Explicitly,
\begin{equation}
\Gamma^{+}_{m} = \left\{ (x,p)\in T^* M  : g_x^{-1}(p,p) = -m^2,\; \hbox{$p$ is future-directed}  \right\},
\end{equation}
and similarly for $\Gamma^{-}_m$. $\Gamma^+_m$ is again a fibre bundle over $M$ whose fibres consist of the future mass hyperboloids
\begin{equation}
\boxed{P_x^+(m) := \left\{ p\in T_x^* M : g_x^{-1}(p,p) = -m^2,\; \hbox{$p$ is future-directed}  \right\}}
\label{Eq:MassHypo}
\end{equation}
at each $x\in M$. The Liouville vector field $L$ defined in Eq.~(\ref{Eq:LiouvilleVF}) is tangent to the future mass shell (see Fig.~\ref{Fig:LGammam}):

\begin{figure}
\center
\includegraphics[scale=1]{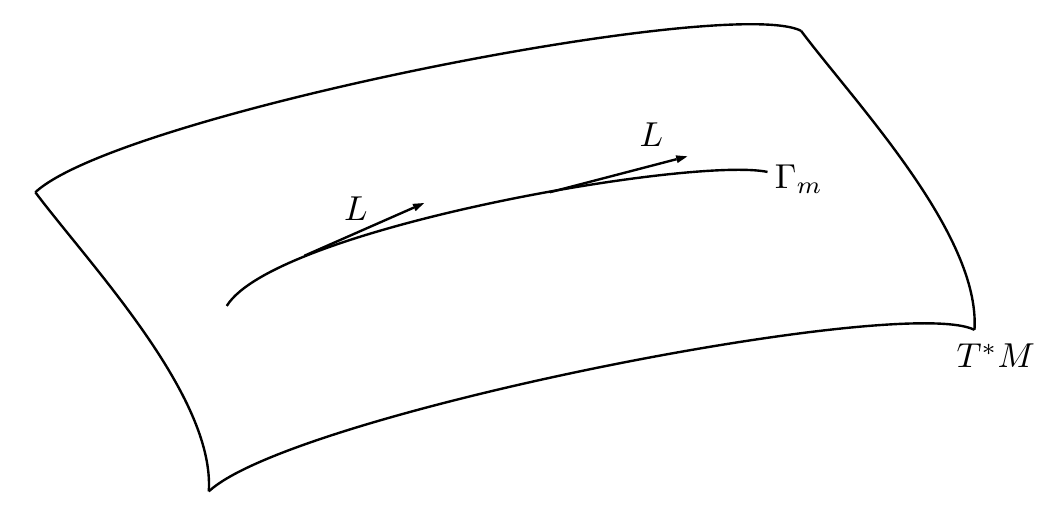}
\caption{Illustration of $T^* M$, the submanifold $\Gamma_m$ and the Liouville vector field tangent to it.}
\label{Fig:LGammam}
\end{figure}

\begin{lemma}
\label{Lem:LTangentMassShell}
At each point $(x,p)\in \Gamma_m$, $L$ is tangent to $\Gamma_m$.
\end{lemma}

\begin{proof}
Using the fact that $L$ is the Hamiltonian vector field associated with the one-particle Hamiltonian $\mathcal{H}$, one obtains immediately
$$
L[\mathcal{H}] = d\mathcal{H}(L) = \Omega_{s}(L,X_{\mathcal{H}})
 = \Omega_{s}(X_{\mathcal{H}},X_{\mathcal{H}}) = 0,
$$
and hence, $L$ leaves the level sets of $\mathcal{H}$ invariant.
\qed
\end{proof}

Therefore, we can also consider $L$ as a vector field on the future mass shell $\Gamma_m^+$. As we will see later, the Boltzmann equation for a simple collisionless gas is simply $L[f] = 0$, with $f\in {\cal F}(\Gamma_m^+)$ the one-particle distribution function.

\subsection{Sasaki metric}
\label{SubSec:Sasaki}

In previous subsections we showed that the cotangent bundle $T^* M$ associated with the spacetime manifold $M$ is a $2n$-dimensional differentiable manifold which admits a natural symplectic form. This in turn, gives rise to the necessary structure for a Hamiltonian formalism on $T^* M$, as we have discussed. In this subsection, we show that the spacetime metric $g$ gives rise to a natural metric on $T^* M$, called the Sasaki metric.\footnote{This metric was introduced by Shigeo Sasaki  in 1958 in the context of the tangent bundle associated with a Riemannian manifold, see~\cite{sS58} and Refs.~\cite{sS62,pD62,sGeK02,rA19} for further related work on the geometry of the tangent bundle relevant to this work.} This metric also induces a metric on the future mass shells $\Gamma_m^+$, and as we will show towards the end of this subsection this allows one to interpret the mass shells with $m > 0$ as Lorentzian submanifold of $T^*M$. Most of the material presented in this subsection is based on~\cite{pR19-thesis}.

The definition of the Sasaki metric makes use of a vector space decomposition of the tangent space, $T_{(x,p)}\left(T^{*}M\right)$, into horizontal $H_{(x,p)}$ and vertical $V_{(x,p)}$ subspaces:
\begin{equation}
T_{(x,p)}\left(T^{*}M\right) = H_{(x,p)} \oplus V_{(x,p)},
\label{Eq:tanS_of_cotanM}
\end{equation}
where each of these subspaces is isomorphic to the tangent space $T_x M$ at $x$. This allows one to define the Sasaki metric $\hat{g}$ on $T^* M$ by requiring that the decomposition~(\ref{Eq:tanS_of_cotanM}) is orthogonal with respect to $\hat{g}$ and that its restrictions to $H_{(x,p)}$ and $V_{(x,p)}$ are compatible with the action of the spacetime metric on $T_x M$.

\begin{definition}
\label{Def:HorVer}
The horizontal and vertical subspaces are defined as
\begin{equation}
\label{Eq:horizontal_vertical}
H_{(x,p)} := \Ker K_{(x,p)},\qquad
V_{(x,p)} := \Ker d\pi_{(x,p)},
\end{equation}
where $K_{(x,p)}: T_{(x,p)}\left(T^{*}M\right) \to T_x^* M$ is the connection map (defined below) and $d\pi_{(x,p)}: T_{(x,p)}(T^{*}M)\to T_x M$ is the differential of the projection map $\pi$ at $(x,p)$ (cf. Definition~\ref{Def:simplectic_form}).
\end{definition}

To define the connection map, let $Z\in T_{(x,p)}\left(T^{*}M\right)$ be a vector at $(x,p)$ which is tangent to a given, smooth curve $\gamma(\lambda)$ in $T^* M$ through $(x,p)$, see Fig.~\ref{Fig:ConnectionMap}. The curve $\gamma(\lambda)$ consists of points $\left( x(\lambda),p(\lambda)\right)$, where $x(\lambda)\in M$ and $p(\lambda)\in T^*_{x(\lambda)}M$. Therefore, $\gamma(\lambda)$ gives rise to a curve $x(\lambda)$ in $M$ through the point $x(0) = x$ and a covector field $p(\lambda)$ along it, such that $p(0) = p$. By means of the parallel transport in $M$ (defined by the Levi-Civita connection $\nabla$ belonging to the spacetime metric $g$), one can transport each of these covectors $p(\lambda)$ along $x(\lambda)$ to the point $x$, giving rise to a family $\hat{p}(\lambda)$ of covectors at $x$. The connection map is defined as the first variation of $\hat{p}(\lambda)$, that is,
\begin{equation}
K_{(x,p)}(Z) := \left. \frac{d}{d\lambda} \hat{p}(\lambda) \right|_{\lambda = 0}
 = \left.\nabla_{\dot{x}(0)}p(\lambda)\right|_{\lambda=0}.
\end{equation}

\begin{figure}[h]
\center
\includegraphics[scale=0.48]{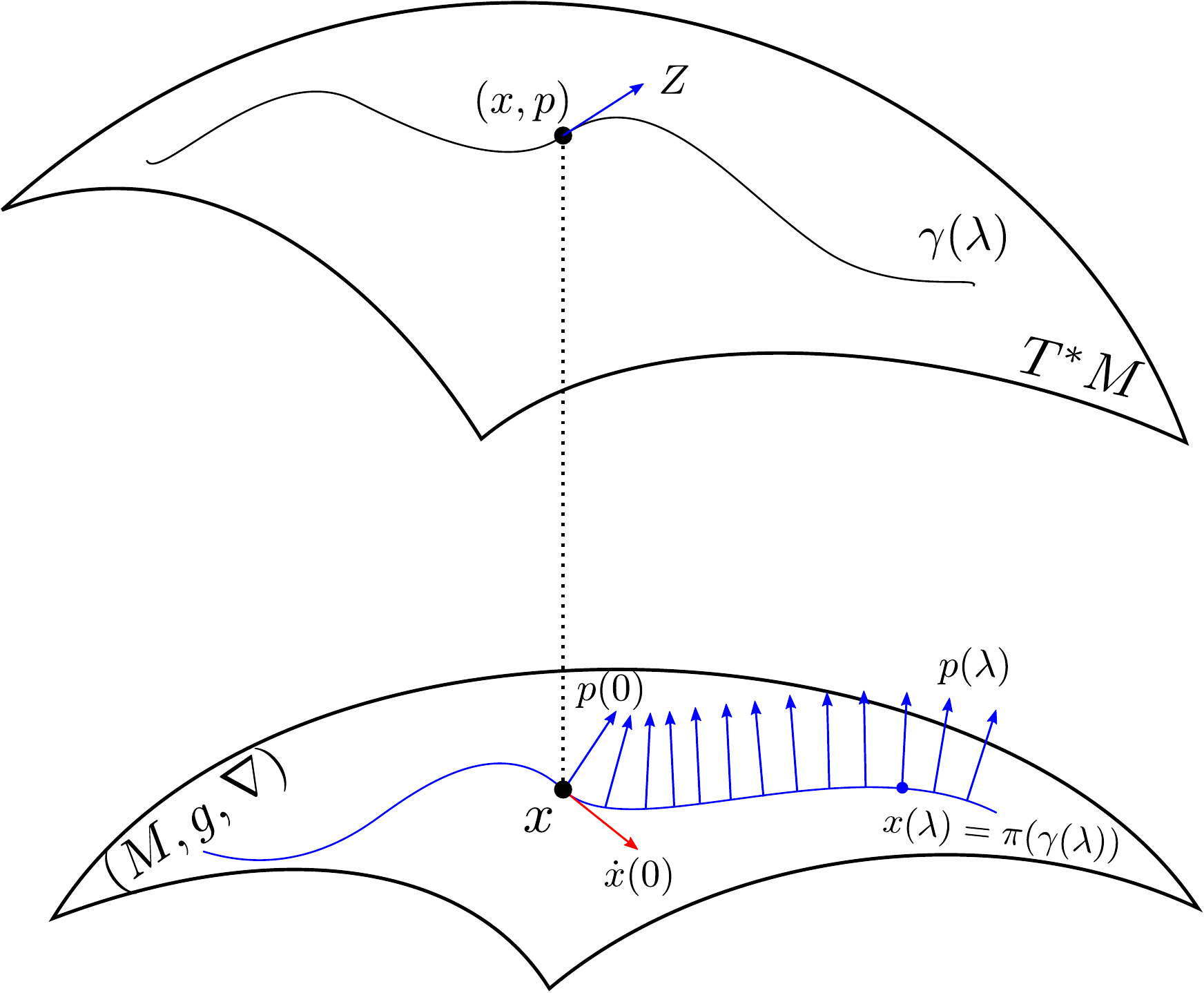}
\caption{Illustration for the definition of the action of the connection map on a tangent vector $Z\in T_{(x,p)}\left(T^{*}M\right)$.}
\label{Fig:ConnectionMap}
\end{figure}

In terms of adapted local coordinates $(x^{\mu},p_{\mu})$ we have 

\begin{equation}
Z=\left.X^{\mu}\frac{\partial}{\partial x^{\mu}}\right|_{(x,p)}+\left.Y_{\mu}\frac{\partial}{\partial p_{\mu}}\right|_{(x,p)},\qquad
p(\lambda) = p_\nu(\lambda)\left. dx^{\nu}\right|_{x(\lambda)},
\end{equation}
with
\begin{equation}
\dot{x}(0) := \frac{dx}{d\lambda}(0) = d\pi_{(x,p)}(Z)
 = \left.X^{\mu}\frac{\partial}{\partial x^{\mu}}\right|_{x},\qquad
\frac{dp_\nu}{d\lambda}(0) = Y_\nu.
\end{equation}
Therefore,
\begin{eqnarray*}
\nabla_{\dot{x}(0)}p(\lambda) &=& \left.\nabla_{\dot{x}(0)}\left(\left. p_{\nu}(\lambda)dx^{\nu}\right|_{x(\lambda)} \right)\right|_{\lambda=0}\\
 &=&\left.\left.\dot{x}(0)\left[p_{\nu}(\lambda) \right]\right|_{\lambda=0}dx^{\nu}\right|_{x}+\left. p_{\nu}(0)\nabla_{\dot{x}(0)}dx^{\nu}\right|_{x} \\
&=& \left.\left.\frac{d}{d\lambda} p_{\nu}(\lambda)\right|_{\lambda=0} dx^{\nu} \right|_{x}+\left. p_{\nu}(0)X^{\mu}\nabla_{\frac{\partial}{\partial x^{\mu}}} dx^{\nu}\right|_{x} \\
&=& \left.\left.\frac{d}{d\lambda} p_{\alpha}(\lambda)\right|_{\lambda=0} dx^{\alpha} \right|_{x} - \left. p_{\nu}(0)X^{\mu}\Gamma^{\nu}{}_{\mu\alpha} dx^{\alpha}\right|_{x} \nonumber \\
&=& \left. \left[ Y_\alpha - \Gamma^{\nu}{}_{\mu\alpha} p_{\nu} X^{\mu}\right] dx^{\alpha}\right|_{x},
\end{eqnarray*}
where $\Gamma^\nu{}_{\mu\alpha}$ denote the Christoffel symbols associated with $\nabla$. Hence, one obtains the following explicit representation for the connection map in terms of adapted local coordinates:
\begin{equation}
\boxed{ K_{(x,p)}(Z)
 = \left[Y_{\alpha} - \Gamma^{\nu}{}_{\mu\alpha}(x) p_{\nu}X^{\mu}\right] dx^{\alpha}_x,
}
\label{Eq:ConnectionMap}
\end{equation}
where $(X^\mu,Y_\alpha)$ are the components of $Z$. Based on the representation~(\ref{Eq:ConnectionMap}) we can make the following observations:
\begin{enumerate}
\item $K_{(x,p)}(Z)$ only depends on $Z$, it does not depend on the choice of the curve $\gamma$ tangent to it.
\item The map $K_{(x,p)}: T_{(x,p)}(T^{*}M)\rightarrow T^{*}_{x}M$ is linear.

\item $\displaystyle Z\in \Ker K_{(x,p)}\Leftrightarrow Y_{\alpha}= \Gamma^{\nu}{_{\mu\alpha}} p_{\nu} X^{\mu} \Leftrightarrow Z=X^{\mu}\left[\frac{\partial}{\partial x^{\mu}} +\Gamma^{\nu}{_{\mu\alpha}} p_{\nu} \frac{\partial}{\partial p_{\alpha}}\right]$. Therefore, a suitable basis for the horizontal subspace $H_{(x,p)} = \Ker K_{(x,p)}$ consists of the $n$ linearly independent tangent vectors
\begin{equation}
\label{Eq:Ddx}
\boxed{ \left. \frac{D}{dx^{\mu}} \right|_{(x,p)} := \left. \frac{\partial}{\partial x^{\mu}} \right|_{(x,p)} + \Gamma^{\nu}{}_{\mu\alpha}(x) p_{\nu} \left. \frac{\partial}{\partial p_{\alpha}} \right|_{(x,p)}.}
\end{equation}

\item From the considerations below Definition~\ref{Def:simplectic_form}, a suitable basis for the vertical subspace $V_{(x,p)} = \Ker d\pi_{(x,p)}$ consists of the $n$ linearly independent tangent vectors
\begin{equation}
\left. \frac{\partial}{\partial p_{\mu}} \right|_{(x,p)}.
\label{Eq:ddp}
\end{equation}
Clearly, the $2n$ vector $\displaystyle \left. \frac{D}{dx^{\mu}} \right|_{(x,p)}$ and $\displaystyle \left. \frac{\partial}{\partial p_{\mu}} \right|_{(x,p)}$ are linearly independent, and hence they provide a basis for the tangent space $T_{(x,p)} (T^* M)$ which is adapted to the decomposition~(\ref{Eq:tanS_of_cotanM}).

\item Accordingly, any tangent vector $Z\in T_{(x,p)}(T^{*}(M))$ can be decomposed uniquely in the form
\begin{equation}
Z = Z^{H} + Z^{V},\qquad Z^{H}\in H_{(x,p)},\quad Z^{V}\in V_{(x,p)}.
\end{equation}

\item For later use, we will also need the basis $\{ dx^\mu, Dp_\mu \}$ of $T_{(x,p)}^* (T^* M)$ dual to $\displaystyle \left\{ \frac{D}{dx^\mu} ,\frac{\partial}{\partial p_{\mu}} \right\}$, with $Dp_\mu$ given by
\begin{equation}
\boxed{ \left. Dp_{\mu} \right|_{(x,p)} := \left. dp_{\mu}\right|_{(x,p)}
 - \Gamma^{\alpha}{}_{\beta\mu} (x) p_{\alpha} \left. dx^{\beta} \right|_{(x,p)},}
\label{Eq:dualbasis}
\end{equation} 
such that
\begin{equation}
dx^{\mu}\left( \frac{D}{dx^\nu} \right)  = \delta^\mu{}_\nu, \qquad
dx^{\mu}\left( \frac{\partial}{\partial p_{\nu}}\right) = 0, \qquad 
Dp_\mu\left( \frac{D}{dx^\nu} \right)  = 0,\qquad
Dp_\mu\left( \frac{\partial}{\partial p_{\nu}}\right) = \delta_\mu{}^\nu.
\end{equation}
The following commutator identities will also turn out to be useful:
\begin{equation}
\left[ \frac{D}{dx^\mu}, \frac{D}{dx^\nu} \right] 
 = R^{\alpha}{}_{\beta\mu\nu} p_{\alpha} \frac{\partial}{\partial p_{\beta}}, \qquad \left[ \frac{D}{dx^\mu}, \frac{\partial}{\partial p_{\nu}} \right]
  = -\Gamma^{\nu}{}_{\mu\beta} \frac{\partial}{\partial p_{\beta}}, \qquad \left[ \frac{\partial}{\partial p_{\mu}}, \frac{\partial}{\partial p_{\nu}} \right] = 0,
\label{Eq:Commutators}
\end{equation}
where $\displaystyle R^{\alpha}{}_{\beta\mu\nu}$ denotes the Riemann tensor.\footnote{In particular, the first commutator implies that $H_{(x,p)}$ does not always define an integrable distribution if the curvature is non-zero.}

\item In addition to the decomposition~(\ref{Eq:tanS_of_cotanM}) into horizontal and vertical subspaces, the connection map $K_{(x,p)}$ and $d\pi_{(x,p)}$ also provide a means to identify horizontal and vertical vectors with tangent vectors on $M$. This is realized through the following maps (cf. Fig.~\ref{Fig:IVH})
\begin{eqnarray}
I^{H}_{(x,p)} & : & H_{(x,p)}\rightarrow T_{x}M\: :\: Z^{H}\mapsto d\pi_{(x,p)}(Z^{H}),
\label{Eq:IH_map}\\
I^{V}_{(x,p)} & : & V_{(x,p)}\rightarrow T_{x}M\: :\: Z^{V}\mapsto g_x^{-1}\left( K_{(x,p)}(Z^{V}),\,\cdot\, \right),
\label{Eq:IV_map}
\end{eqnarray}
where here the inverse metric $g_x^{-1}$ provides the isomorphism $T_x^*M\to T_x M$, $\omega\mapsto g_x^{-1}(\omega,\cdot)$ between covectors and vectors at $x$. Due to the aforementioned properties, these maps are linear and invertible, and hence they allow one to identify the horizontal subspace as well as the vertical subspace with $T_x M$. In terms of the basis vectors~(\ref{Eq:Ddx},\ref{Eq:ddp}), one has
\begin{eqnarray}
I^{H}_{(x,p)}(Z^{H}) &=& Z^{\mu}\left. \frac{\partial}{\partial x^\mu} \right|_{x},\qquad
Z^{H} = Z^{\mu} \left. \frac{D}{dx^\mu} \right|_{(x,p)},\\
I^{V}_{(x,p)}(Z^{V}) &=& 
 Y_{\mu} g^{\mu\nu}(x)\left. \frac{\partial}{\partial x^\nu} \right|_{x},
\qquad
Z^{V} = Y_{\mu}\left. \frac{\partial}{\partial p_{\mu}} \right|_{(x,p)}.
\end{eqnarray}

\item Based on the previous observation, one can introduce an ``almost complex structure" (analogous to a rotation in the complex plane by the angle $\pi/2$), that is a linear map $J_{(x,p)} : T_{(x,p)}(T^*M)\to T_{(x,p)}(T^*M)$ satisfying $J_{(x,p)}^{2}= -\mathds{1}$, in the following way:
\begin{equation}
J_{(x,p)}(Z^{H} + Z^{V}) := -Q_{(x,p)}^{-1}\left( Z^{V} \right) + Q_{(x,p)}\left( Z^{H} \right),
\qquad Z\in T_{(x,p)}(T^* M),
\label{Eq:AlmostComplexStruc}
\end{equation}
where the map $Q_{(x,p)}: H_{(x,p)}\to V_{(x,p)}$ is defined as $Q_{(x,p)}:= \left( I^{V}_{(x,p)} \right)^{-1}\circ I^{H}_{(x,p)}$.

\end{enumerate}

\begin{figure}[h]
\includegraphics[scale=0.75]{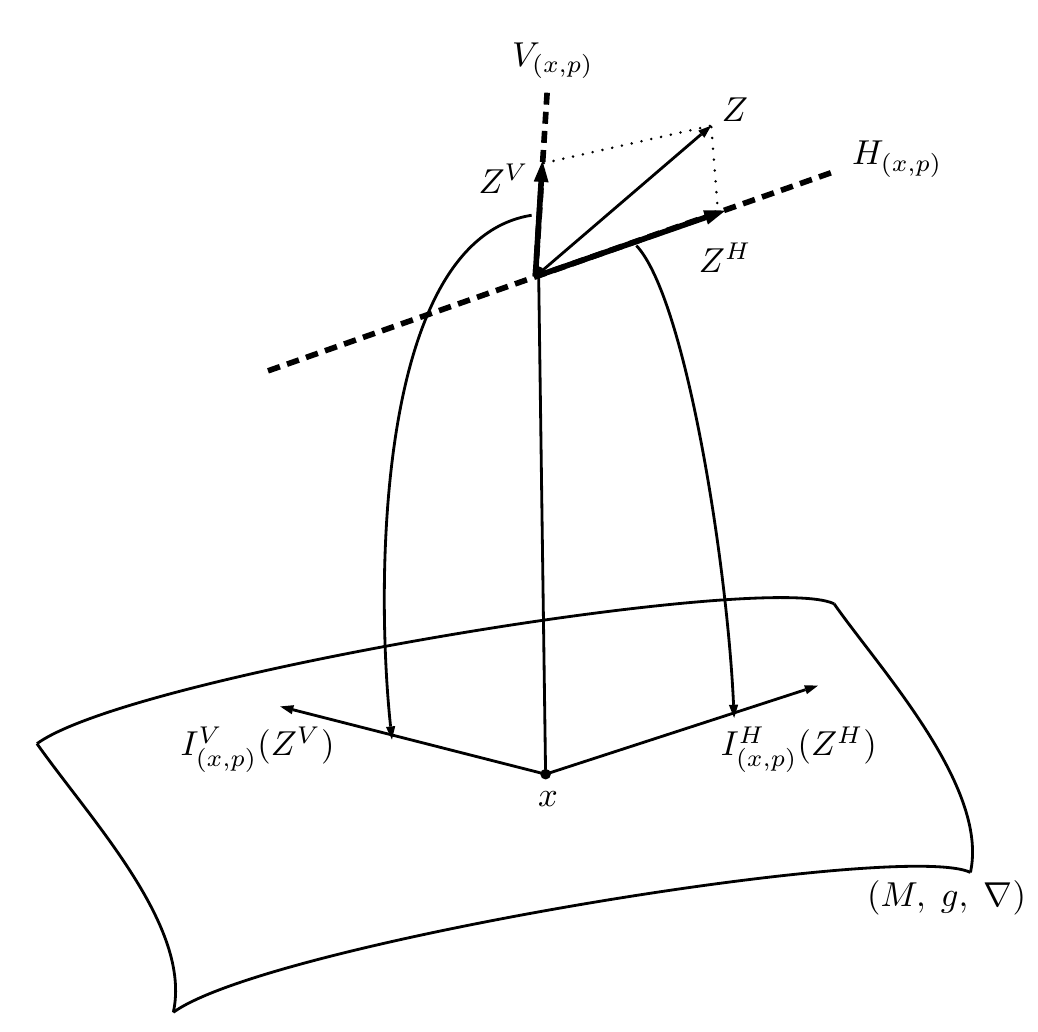}
\caption{Illustration of the maps $I^H_{(x,p)}: H_{(x,p)}\to T_x M$ and $I^V_{(x,p)}: V_{(x,p)}\to T_x M$ which allows one to identify the horizontal and the vertical subspaces with the tangent space $T_x M$.}
\label{Fig:IVH}
\end{figure}

After these observations, it is straightforward to equip $T^* M$ with a metric $\hat{g}$. Since any vector $Z\in T_{(x,p)}(T^*M)$ can be decomposed uniquely into horizontal and vertical components, $Z = Z^H + Z^V$, and since both components can be identified with a tangent vector on $M$, we define $\hat{g}$ as follows:

\begin{definition}
\label{Def:Sasaki}
Given $\displaystyle Z,W\in T_{(x,p)}(T^{*}M)$, we define
\begin{eqnarray}
\hat{g}_{(x,p)}(Z,W) &=& \hat{g}_{(x,p)}\left( Z^{H}+Z^{V},W^{H}+W^{V}\right)  
\nonumber\\
 &:=& \hat{g}_{(x,p)}\left( Z^{H},W^{H}\right)+\hat{g}_{(x,p)}\left( Z^{V},W^{V}\right) \nonumber \\
&:=& g_{x}\left( I^{H}_{(x,p)}(Z^{H}),I^{H}_{(x,p)}(W^{H})\right)
 + g_{x}\left( I^{V}_{(x,p)}(Z^{V}),I^{V}_{(x,p)}(Z^{V})\right) \nonumber\\
&=& g_{x}\left(d\pi_{(x,p)}(Z),d\pi_{(x,p)}(W) \right)+g_{x}^{-1}\left(K_{(x,p)}(Z),K_{(x,p)}(W) \right).
\label{Eq:Sasaki}
\end{eqnarray}
\end{definition}

Thus $\hat{g}$ is the unique metric on $T^* M$ which is orthogonal with respect to the decomposition~(\ref{Eq:tanS_of_cotanM}) and turns the maps $I^{H}_{(x,p)}$ and $I^{V}_{(x,p)}$ defined in Eqs.~(\ref{Eq:IH_map},\ref{Eq:IV_map}) into isometries. Note also that  $\hat{g}$ is symmetric and invariant with respect to the almost complex structure $J_{(x,p)}$. In order to verify that $\hat{g}$ is non-degenerate, we compute its components with respect to the basis vectors $\displaystyle  \frac{D}{dx^\mu} ,\frac{\partial}{\partial p_{\mu}} $:
\begin{equation}
\hat{g}\left( \frac{D}{dx^\mu},\frac{D}{dx^\nu} \right) = g_{\mu\nu}, \qquad 
\hat{g}\left(\frac{\partial}{\partial p_{\mu}},\frac{\partial}{\partial p_{\nu}} \right) = g^{\mu\nu}, \qquad 
\hat{g}\left( \frac{D}{dx^\mu},\frac{\partial}{\partial p_{\nu}} \right) = 0,
\label{Eq:ghatComponents}
\end{equation}
such that\footnote{Since we use geometrized units, the coordinates $x^\mu$ and $p_\mu$ have the same units of length. When working in natural units one should introduce a factor $G_N^2/c^6$ in front of the second term on the right-hand side of Eq.~(\ref{Eq:SasakiExplicit}).}
\begin{equation}
\boxed{\hat{g} = g_{\mu\nu} dx^{\mu}\otimes dx^{\nu}
 + g^{\mu\nu} Dp_{\mu}\otimes Dp_{\nu}, \qquad 
 Dp_{\mu} = dp_{\mu} - \Gamma^{\alpha}{}_{\beta\mu}p_{\alpha}dx^{\beta}.}
\label{Eq:SasakiExplicit}
\end{equation}
Therefore, the Sasaki metric~(\ref{Eq:Sasaki}) is a pseudo-Riemannian metric with signature $(-2,2n-2)$. Note also that in terms of the basis used above, the determinant of the metric components is one.

In the remainder of this subsection, we list some results which will be useful later. The first result concerns the following relation between the symplectic form, the Sasaki metric and the almost complex structure:

\begin{proposition}
\label{Prop:OmegaJgRelation}
Let $\Omega_s$, $J$ and $\hat{g}$ be defined as in Eqs.~(\ref{Eq:Omega},\ref{Eq:AlmostComplexStruc},\ref{Eq:Sasaki}), then for all $W,Z\in {\cal X}(T^* M)$,
\begin{equation}
\Omega_s\left(W,Z\right)=\hat{g}\left(W,J\left(Z\right)\right).
\label{Eq:OmegaJgRelation}
\end{equation}
\end{proposition}

\begin{proof} First we note that due the fact that the Christoffel symbols are symmetric in the lower indices, the symplectic form $\Omega_s$ defined in Eq.~(\ref{Eq:Omega}) can also be written as
\begin{equation}
\boxed{\Omega_s = Dp_\mu\wedge dx^\mu,}
\label{Eq:OmegaBis}
\end{equation}
where $Dp_\mu$ is defined in Eq.~(\ref{Eq:dualbasis}).

Next, we set for each $W,Z\in {\cal X}(T^* M)$
$$
\tilde{\Omega}(W,Z) := \hat{g}\left(W,J\left(Z\right)\right).
$$
This defines a bilinear form which is antisymmetric under the exchange $W$ and $Z$:
$$
\hat{g}\left(W,J\left(Z\right)\right) = \hat{g}\left(J\left(W\right),J^2\left(Z\right)\right)
  = -\hat{g}\left(Z,J\left(W\right)\right),
$$
and which vanishes if both $W$ and $Z$ are horizontal or both are vertical. Furthermore, with respect to the basis vectors ${\displaystyle \frac{D}{dx^{\mu}},\frac{\partial}{\partial p_{\mu}}}$ defined in Eqs.~(\ref{Eq:Ddx},\ref{Eq:ddp}) we find
\begin{eqnarray*}
\tilde{\Omega}\left( \frac{D}{dx^{\mu}}, \frac{\partial}{\partial p_{\nu}} \right)
 &=& \hat{g}\left(\frac{D}{dx^{\mu}},J\left(\frac{\partial}{\partial p_{\nu}}\right)\right) \\
 &=&-\hat{g}\left(\frac{D}{dx^{\mu}},g^{\nu\alpha}\frac{D}{dx^{\alpha}}\right)\\
&=&-g^{\nu\alpha}g\left(\frac{\partial}{\partial x^\mu},\frac{\partial}{\partial x^\alpha}\right)\\
&=&-g^{\nu\alpha}g_{\alpha\mu}=-\delta^{\nu}{_{\mu}}.
\end{eqnarray*}
where in the second step we have used
\begin{equation}
J\left(\frac{\partial}{\partial p_{\nu}}\right)=-Q^{-1}\left(\frac{\partial}{\partial p_{\nu}}\right)
 = \left(-I^{H}\right)^{-1}\left[g^{\mu\nu}\frac{\partial}{\partial x^{\mu}}\right]=-g^{\nu\mu}\frac{D}{dx^{\mu}}.
\label{Eq:Jp}
\end{equation}
Therefore, since $\tilde{\Omega}$ is antisymmetric,
$$
\tilde{\Omega} = Dp_\mu \wedge dx^\mu,
$$
which agrees with~(\ref{Eq:OmegaBis}).
\qed
\end{proof}

The next result shows that the Liouville vector field $L$ generates geodesics in the pseudo-Riemmannian manifold $(T^* M,\hat{g},\hat{\nabla})$, where $\hat{\nabla}$ denotes the Levi-Civita connection associated with the Sasaki metric $\hat{g}$:

\begin{proposition}
\label{Prop:GeodesicL}
The Liouville vector field defined in Eq.~(\ref{Eq:LiouvilleVF}) is geodesic:
\begin{equation}
\hat{\nabla}_L L = 0.
\end{equation}
\end{proposition}

\begin{proof}
This proposition was originally demonstrated by Sasaki in 1958~\cite{sS58} in the context of the tangent bundle associated with a Riemannian manifold; see also~\cite{oStZ14b} for an alternative proof in the Lorentzian case which avoids explicitly computing the connection $\hat{\nabla}$. The following proof is a straightforward generalization of the one presented in~\cite{oStZ14b} to the cotangent bundle. For the proof, we require the following lemma.

\begin{lemma}
\label{Lem:LH}
\begin{enumerate}
\item[(a)] The Liouville vector field $L$ is horizontal and
\begin{equation}
\boxed{L = g^{\mu\nu} p_\nu\frac{D}{dx^\mu}.}
\label{Eq:LiouvilleHor}
\end{equation}
\item[(b)] The Poincar\'e one-form (see Definition~\ref{Def:simplectic_form}) is dual to $L$, that is:
\begin{equation}
\boxed{\Theta = \hat{g}(L,\;\cdot\;).}
\end{equation}
\end{enumerate}
\end{lemma}

\begin{proof} Using the explicit expressions~(\ref{Eq:LiouvilleVF}) and (\ref{Eq:ConnectionMap}) for the Liouville vector field and connection map, it is not difficult to verify that $K(L) = 0$, which shows that $L$ is horizontal. Since $\displaystyle I^{H}_{(x,p)}(L) = d\pi_{(x,p)}(L) = g^{\mu\nu} p_\mu\frac{\partial}{\partial x^\nu}$, the formula~(\ref{Eq:LiouvilleHor}) follows. Furthermore, taking into account the definition of the Sasaki metric $\hat{g}$ and the Poincar\'e one-form $\Theta$ it follows for any $Z\in {\cal X}(T^* M)$ that
$$
\hat{g}(L,Z) = \hat{g}(L^H,Z^H) = g( d\pi(L), d\pi(Z) ) = p_\mu\left( d\pi(Z) \right)^\mu 
 = \Theta(Z),
$$
which also proves statement (b).
\qed
\end{proof}

After these preliminary results, we return to the proof of Proposition~\ref{Prop:GeodesicL}. First, note that $\Theta(L) = \hat{g}(L,L) = g^{\mu\nu} p_\mu p_\nu = 2\mathcal{H}$. Next, using Cartan's formula,
\begin{equation}
\pounds_L \Theta = d i_L\Theta + i_L d\Theta = 2d\mathcal{H} + i_L\Omega_s
 = d\mathcal{H},
\label{Eq:LTheta1}
\end{equation}
since $i_L\Omega_s = -d\mathcal{H}$. An alternative way of computing the Lie derivative of $\Theta$ is based on the fact that $\hat{\nabla}$ is metric-compatible and torsion-free, such that $\pounds_L Z=\left[L,Z\right]=\hat{\nabla}_L Z-\hat{\nabla}_Z L$:
\begin{eqnarray}
(\pounds_L \Theta)(Z) &=& L\left[\Theta(Z)\right] - \Theta(\pounds_{L} Z) 
\nonumber\\
&=& L\left[\Theta(Z)\right] - \Theta(\hat{\nabla}_L Z-\hat{\nabla}_Z L)
\nonumber\\
&=& L\left[\Theta(Z)\right] - \Theta(\hat{\nabla}_L Z) + \hat{g}(L,\hat{\nabla}_Z L)
\nonumber\\
&=& L\left[\Theta(Z)\right] - \Theta(\hat{\nabla}_L Z)+d\mathcal{H}(Z)
\nonumber\\
&=& (\hat{\nabla}_L \Theta)(Z)+d\mathcal{H}(Z),
\label{Eq:LTheta2}
\end{eqnarray}
where in the fourth step we have used the Ricci identity to conclude that $2\hat{g}(L,\hat{\nabla}_Z L) = Z[\hat{g}(L,L)] = Z[2\mathcal{H}] = 2d\mathcal{H}(Z)$. Comparing Eq.~(\ref{Eq:LTheta1}) with Eq.~(\ref{Eq:LTheta2}) we conclude that $\hat{\nabla}_L\Theta = 0$, which in view of Lemma~\ref{Lem:LH}(b) implies the desired result.
\qed
\end{proof}

Finally, we collect some useful statements regarding the geometric properties of the future mass shells $\Gamma_m^+$ of positive masses $m > 0$ which are induced from the Sasaki metric. For this, we recall that these surfaces are level sets of the free-particle Hamiltonian $\mathcal{H}$. Consequently, any unit normal vector $N$ on $\Gamma_m^+$ satisfies
\begin{equation}
\hat{g}(N,\;\cdot\;) = \alpha d\mathcal{H},
\label{Eq:N-dH}
\end{equation}
for some appropriate normalization constant $\alpha$. Combining the fact that $d\mathcal{H} = \Omega_s(\;\cdot\;,L)$ with Proposition~\ref{Prop:OmegaJgRelation} we find
\begin{equation}
N = \alpha J(L).
\end{equation}
Furthermore, for points on $\Gamma_m^+$ one has
$$
\hat{g}(N,N) = \alpha^2\hat{g}(L,L) = \alpha^2 g(p,p) = -\alpha^2 m^2,
$$
which shows that $N$ is timelike and fixes the normalization constant $\alpha = \pm 1/m$ up to a sign. Summarizing, we obtain the following result:

\begin{proposition}
\label{Prop:MassShell}
Denote by $\hat{h}$ the induced metric on $\Gamma_m^+$. For any $m > 0$, $(\Gamma_m^+,\hat{h})$ is a $(2n-1)$-dimensional Lorentzian manifold with unit normal vector field\footnote{We choose the orientation of $N$ is such that $\mathcal{H}$ decreases along the flow lines of $N$.}
\begin{equation}
N = \frac{1}{m} J(L) = \frac{1}{m} p_\mu\frac{\partial}{\partial p_\mu}.
\label{Eq:NormalVector}
\end{equation}
\end{proposition}

\begin{proof}
Although the statement already follows from the arguments preceding the proposition, an alternative proof which does not use the map $J$ explicitly is based on the observation that the differential of $\mathcal{H}$ may be written as
\begin{equation}
\boxed{ d\mathcal{H} = g^{\mu\nu}p_\mu Dp_\nu,}
\end{equation}
with $Dp_\nu$ defined in Eq.~(\ref{Eq:dualbasis}). Comparing this expression with
\begin{equation}
\hat{g}(N,\,\cdot\,) = g_{\mu\nu}dx^{\mu}(N)dx^{\nu}+g^{\mu\nu}Dp_{\mu}(N)Dp_{\nu}
\end{equation}
obtained from the representation~(\ref{Eq:SasakiExplicit}) of the Sasaki metric, one realizes that $dx^\mu(N) = 0$ and $dp_\mu(N) = \alpha p_\mu$, and then the normalization follows again from the observation that $\hat{g}(N,N) = \alpha^2 g^{\mu\nu} p_\mu p_\nu = -\alpha^2 m^2$ on $\Gamma_m^+$.
\qed
\end{proof}

Since $L$ is tangent to $\Gamma_m^+$, it can be regarded as a vector field on the mass shell (which, by a slight abuse of notation we denote again by $L$). Since $\hat{h}(L,L) = \hat{g}(L,L) = 2\mathcal{H} = -m^2 < 0$ it is timelike and thus defines a time-orientation in $(\Gamma_m^+,\hat{h})$. Further, Proposition~\ref{Prop:GeodesicL} implies:

\begin{proposition}
\label{Prop:GeodesicL_on_MS}
Denoting by $\hat{D}$ the Levi-Civita connection on $(\Gamma_m^+,\hat{h})$, one has:
\begin{equation}
\hat{D}_L L = 0.
\end{equation}
\end{proposition}

\begin{proof}
This follows directly from $\hat{\nabla}_L L = 0$ and the fact that $\hat{D}$ coincides with the induced connection on $\Gamma_m^+$, such that
$$
\hat{g}(\hat{\nabla}_X Y ,Z) = \hat{h}( \hat{D}_X Y,Z)
$$
for all $X,Y,Z\in {\cal X}(T^*M)$ tangent to $\Gamma_m^+$.
\qed
\end{proof}


\subsection{Volume forms and Liouville's theorem}
\label{SubSec:Liouville}

In this subsection, we show that the Sasaki metric allows one to introduce volume forms on $T^* M$ and the future mass shells $\Gamma_m^+$ with positive mass $m > 0$, which will be relevant for the physical interpretation of the one-particle distribution function.

The volume form $\eta_{T^* M}$ on $T^* M$ induced by the Sasaki metric is
\begin{equation}
\eta_{T^* M} = \pm\sqrt{\det(\hat{g}_{AB})} dy^1\wedge dy^2\wedge \cdots \wedge dy^{2n},
\end{equation}
for any coordinate system $(y^A)$, $A=1,2,\ldots,2n$ on $T^* M$, where one still has the freedom to choose the sign corresponding to a particular orientation of $T^* M$. We choose the orientation such that with respect to adapted local coordinates $(x^\mu,p_\mu)$ (in which case $\det(\hat{g}_{AB}) = 1$ according to Eq.~(\ref{Eq:SasakiExplicit})) one has
\begin{equation}
\eta_{T^* M} = -dp_0 \wedge dp_1 \wedge \cdots \wedge dp_d
\wedge dx^0 \wedge dx^1 \wedge \cdots \wedge dx^d.
\label{Eq:EtaOnT*M}
\end{equation}

\begin{proposition}[Liouville's theorem on $T^* M$]
\label{Prop:LiouvilleT*M}
The volume form is invariant with respect to the flow generated by the Liouville vector field $L$, that is,
\begin{equation}
\pounds_L \eta_{T^* M} = 0.
\end{equation}
\end{proposition}

\begin{proof}
First, we note that
\begin{eqnarray}
\eta_{T^* M} &=& - dp_{0}\wedge\cdots\wedge dp_{d}\wedge
dx^{0}\wedge\cdots\wedge dx^{d}\\
 &=& -\frac{1}{n!}  dp_{\mu_{1}}\wedge\cdots\wedge dp_{\mu_{n}}\wedge
 dx^{\mu_{1}} \wedge\cdots\wedge dx^{\mu_{n}}  \\
&=& -\frac{(-1)^k}{n!} (dp_{\mu_1} \wedge dx^{\mu_1}) \wedge (dp_{\mu_2} \wedge dx^{\mu_2}) \wedge \cdots \wedge (dp_{\mu_n} \wedge dx^{\mu_n})\\
 &=& -\frac{(-1)^k}{n!}\Omega_s\wedge\Omega_s\wedge\cdots\wedge\Omega_s,
\label{Eq:volumeformonT*M}
\end{eqnarray}
with $k = d(d+1)/2$. Now the proposition is a direct consequence of the fact that $\pounds_L\Omega_s = 0$, see Eq.~(\ref{Eq:LOmega}).
\qed
\end{proof}

Thanks to the presence of the normal vector field $N$ to the mass shell (see Proposition~\ref{Prop:MassShell}), the volume form $\eta_{T^* M}$ induces a volume form on $\Gamma_m^+$:

\begin{definition}
\label{Def:InducedVolumeForm}
The volume form on the future mass shell can be defined as
\begin{equation}
\eta_{\Gamma_m^+}(X_2,X_3,\dots,X_{2n}):=\eta_{T^* M}(N,X_2,X_3,\dots,X_{2n})
 = (i_N\eta_{T^* M})(X_2,X_3,\dots,X_{2n})
\label{Eq:volumeformonGammaM+}
\end{equation}
for $X_2,X_3,\ldots,X_{2n}\in {\cal X}(\Gamma_m^+)$ and where the normal vector field $N$ was defined in Eq.~(\ref{Eq:NormalVector}).
\end{definition}

\begin{theorem}[Liouville's theorem on $\Gamma_m^+$]
\label{Thm:Liouville}
Let $m > 0$. The volume form $\eta_{\Gamma_m^+}$ is invariant with respect to the flow generated by the Liouville vector field $L$ on $\Gamma_m^+$, that is,
\begin{equation}
\pounds_L \eta_{\Gamma_m^+} = 0.
\end{equation}
\end{theorem}

\begin{proof} Let $X_2,X_3,\ldots,X_{2n}\in {\cal X}(\Gamma_m^+)$. Using the definitions of the Lie derivative and $\eta_{\Gamma_m^+}$, we obtain
\begin{eqnarray}
\pounds_{L}\eta_{\Gamma_m^+}(X_2,X_3,\dots,X_{2n}) &=& L\left[\eta_{\Gamma_m^+}(X_2,X_3,\dots,X_{2n})\right]-\eta_{\Gamma_m^+}(\pounds_{L} X_2,X_3,\dots,X_{2n})-\cdots-\eta_{\Gamma_m^+}(X_2,X_3,\dots,\pounds_{L} X_{2n})
\nonumber\\
&=& L\left[\eta_{T^* M}(N,X_2,\dots,X_{2n}) \right] - \eta_{T^* M}(N,\pounds_{L}X_2,\dots,X_{2n}) - \cdots - \eta_{T^* M}(N,X_2,\dots,\pounds_{L}X_{2n})
\nonumber\\
&&-\eta_{T^* M}(\pounds_{L} N,X_2,\dots,X_{2n}) + \eta_{T^* M}(\pounds_{L} N,X_2,\dots,X_{2n})
 \nonumber\\
&=& (\pounds_{L}\eta_{T^* M})(N,X_2,\dots,X_{2n})
 + \eta_{T^* M}(\pounds_{L} N,X_2,\dots,X_{2n}),
\label{Eq:Leta}
\end{eqnarray}
where the first term on the right-hand side of Eq.~(\ref{Eq:Leta}) vanishes as a consequence of Proposition~\ref{Prop:LiouvilleT*M}. It remains to show that the second term vanishes as well. For this, we use Eq.~(\ref{Eq:NormalVector}) and obtain
$$
\pounds_{L}N = \frac{1}{m}\pounds_{L}J(L) = \frac{1}{m}\left[L,J(L) \right].
$$
Next, using Eqs.~(\ref{Eq:LiouvilleHor}) and~(\ref{Eq:Jp}) we find (setting $p^\mu := g^{\mu\nu} p_\nu$)
\begin{eqnarray}
\left[L,J(L) \right] 
 &=& \left[p^\mu\frac{D}{dx^\mu} , p_\alpha\frac{\partial}{\partial p_\alpha} \right]
\nonumber\\
 &=& p^\mu\underbrace{\frac{D p_\alpha}{dx^\mu}}_{ = \Gamma^{\beta}{}_{\mu\alpha}p_\beta}\frac{\partial}{\partial p_\alpha} - p_\alpha\underbrace{\left(\frac{\partial p^\mu}{\partial p_\alpha}\right)}_{ = g^{\mu\alpha}} \frac{D}{dx^\mu}
  + p^\mu p_\alpha\underbrace{\left[ \frac{D}{dx^\mu},\frac{\partial}{\partial p_\alpha} \right]}_{= -\Gamma^{\alpha}{}_{\mu\beta}\frac{\partial}{\partial p_\beta}}\nonumber\\
&=& -L,
\end{eqnarray}
where we have used the definition of $\displaystyle\frac{D}{dx^\mu}$ in Eq.~(\ref{Eq:Ddx}) and the commutation relation~(\ref{Eq:Commutators}) in the last step.

Since the $2n$ vectors $L, X_2, X_3, \dots, X_{2n}$ are tangent to $\Gamma_m^+$ they must be linearly dependent; hence it follows that $\eta_{T^* M}(\pounds_L N,X_2,\dots,X_{2n}) = 0$, which concludes the proof of the theorem.
\qed
\end{proof}

To conclude this section, we provide explicit expressions for the induced metric and volume form on the future mass shell $\Gamma_m^+$. These expressions acquire a rather elegant and simple form when working in terms of an orthonormal basis of covector fields $\left\lbrace\theta^{\hat{\alpha}}\right\rbrace$ on $M$ (as opposed to the coordinate basis $dx^\alpha$), where here and in the following, hatted Greek indices $\hat{\alpha},\hat{\beta},\hat{\gamma},\ldots$ running over $\hat{0},\hat{1},\dots,\hat{d}$ refer to such an orthonormal frame. If $\left\lbrace e_{\hat{\alpha}} \right\rbrace$ denotes the dual basis of vector fields, we assume that $e_{\hat{0}}$ is future-directed timelike. Before considering the mass shell, let us start with the future mass hyperboloid $P_x^+(m)$ at a given (fixed) point $x\in M$, see Eq.~(\ref{Eq:MassHypo}). This set is characterized by those $p = p_{\hat{\alpha}} \theta^{\hat{\alpha}}\in T_x^* M$, for which
\begin{equation}
p_{\hat{0}} = -\sqrt{m^2 + \delta^{\hat{a}\hat{b}} p_{\hat{a}} p_{\hat{b}}},
\label{Eq:p0}
\end{equation}
where $\hat{a},\hat{b}=1,2,\ldots,d$. The induced metric on $P_x^+(m)$ is
\begin{equation}
H_x = g^{\mu\nu}(x) dp_\mu\otimes dp_\nu 
 = -dp_{\hat{0}}\otimes dp_{\hat{0}} + \delta^{\hat{a}\hat{b}} dp_{\hat{a}} \otimes dp_{\hat{b}}
 = H^{\hat{a}\hat{b}} dp_{\hat{a}} \otimes dp_{\hat{b}},
\label{Eq:Induced_metric_gx-1}
\end{equation}
with components
\begin{equation}
H^{\hat{a}\hat{b}} = \delta^{\hat{a}\hat{b}} - \frac{1}{p_{\hat{0}}^2} p^{\hat{a}} p^{\hat{b}}.
\label{Eq:Hab}
\end{equation}
The associated volume form on $P_x^+(m)$ is
\begin{equation}
\eta_{P_x^+(m)} = \sqrt{ \det(H^{\hat{a}\hat{b}})} 
dp_{\hat{1}}\wedge dp_{\hat{2}}\wedge\cdots \wedge dp_{\hat{d}} 
 = m\dvol_x(p),
\label{Eq:etaPxm}
\end{equation}
where here and in the following,
\begin{equation}
\boxed{
\dvol_x(p) := \frac{1}{m}\eta_{P_x^+(m)} 
 = \frac{dp_{\hat{1}} \wedge dp_{\hat{2}} \wedge \cdots \wedge dp_{\hat{d}}}{
\sqrt{m^2 + p_{\hat{1}}^2 + p_{\hat{2}}^2 + \cdots + p_{\hat{d}}^2}}
}
\label{Eq:dvolx}
\end{equation}
denotes the the well-known Lorentz-invariant volume element on $P_x^+(m)$ (see, for instance, Ref.~\cite{Weinberg-QMBook}).\footnote{Note also that our definition of $\eta_{P_x^+(m)}$ differs from the corresponding volume form $\pi_x$ defined in~\cite{oStZ14b} by a factor of $m$.}

After these remarks, we return to the future mass shell $\Gamma_m^+$ which, we recall, can be thought of as the fibre bundle over $M$ with fibre $P_x^+(m)$ at $x\in M$. As we have just seen, it is convenient to expand the momentum $p\in T_x^* M$ in terms of an orthonormal basis. For this reason we shall use (instead of the adapted local coordinates $(x^\mu,p_\mu)$ used so far) the new coordinates $(x^\mu,p_{\hat{\alpha}})$ with $p_{\hat{\alpha}} = p\left(e_{\hat{\alpha}}\right)$ the orthonormal components of the momentum covector, such that $p = p_{\hat{\alpha}} \theta^{\hat{\alpha}}$. In terms of these new coordinates, the Sasaki metric~\eqref{Eq:SasakiExplicit} reads
\begin{equation}
\hat{g} = g_{\mu\nu} dx^{\mu}\otimes dx^{\nu} 
+ \eta^{\hat{\alpha}\hat{\beta}} Dp_{\hat{\alpha}}\otimes Dp_{\hat{\beta}},
\label{Eq:Sasaki_ONbasis}
\end{equation}
where $(\eta^{\hat{\alpha}\hat{\beta}}) = \diag(-1,+1,\ldots,+1)$ and
\begin{equation}
Dp_{\hat{\alpha}} := dp_{\hat{\alpha}} - \hat{\Gamma}^{\hat{\beta}}{_{\mu\hat{\alpha}}}p_{\hat{\beta}}dx^{\mu},\qquad
\hat{\Gamma}^{\hat{\beta}}{_{\mu\hat{\alpha}}} 
 := \theta^{\hat{\beta}}\left[\nabla_{\frac{\partial}{\partial x^{\mu}}}e_{\hat{\alpha}}\right],
\label{Eq:Dphat}
\end{equation}
see Appendix~\ref{App:ON_basis} for details. Using the fact that $\hat{\Gamma}_{\hat{\beta}\mu\hat{\alpha}} := \eta_{\hat{\beta}\hat{\gamma}}\hat{\Gamma}^{\hat{\gamma}}{_{\mu\hat{\alpha}}}$ is antisymmetric in $\hat{\alpha}\hat{\beta}$, it is not difficult to verify that Eq.~(\ref{Eq:p0}) implies that
\begin{equation}
Dp_{\hat{0}} = \frac{p^{\hat{b}}}{p_{\hat{0}}} dp_{\hat{b}} - p^{\hat{b}}\hat{\Gamma}_{\hat{b}\mu\hat{0}} dx^\mu = \frac{p^{\hat{b}}}{p_{\hat{0}}} Dp_{\hat{b}},
\end{equation}
such that the induced metric $\hat{h}$ on the future mass shell can be written as
\begin{equation}
\boxed{ \hat{h} = g_{\mu\nu} dx^{\mu}\otimes dx^{\nu} 
+ H^{\hat{a}\hat{b}} Dp_{\hat{a}}\otimes Dp_{\hat{b}},}
\end{equation}
with $H^{\hat{a}\hat{b}}$ given by Eq.~(\ref{Eq:Hab}). The corresponding volume form is
\begin{equation}
\eta_{\Gamma_m^+} = \sqrt{-\det(g_{\mu\nu})} dx^0\wedge dx^1\wedge\cdots\wedge dx^d
\wedge \frac{m}{|p_{\hat{0}}|} dp_{\hat{1}}\wedge dp_{\hat{2}}\wedge\cdots \wedge dp_{\hat{d}},
\label{Eq:etaGammam+}
\end{equation}
and agrees with the one obtained from the definition in Eq.~(\ref{Eq:volumeformonGammaM+}), see Appendix~\ref{App:VolumeFormOnTheHyperboloidOfMass}. With a slight abuse of notation, we can write Eq.~(\ref{Eq:etaGammam+}) as
\begin{equation}
\eta_{\Gamma_m^+} = \eta_M\wedge \eta_{P_x^+(m)},
\label{Eq:etaGammam+Short}
\end{equation}
with $\eta_M$ the volume form on $M$ and $\eta_{P_x^+(m)}$ the volume form on $P_x^+(m)$, see Eq.~(\ref{Eq:etaPxm}). This provides the Fubini-type integration formula
\begin{equation}
\boxed{ \int\limits_{\Gamma_m^+} f(x,p) \eta_{\Gamma_m^+} 
= \int\limits_M \left( \int\limits_{P_x^+(m)} f(x,p) \eta_{P_x^+(m)} \right) \eta_M,}
\label{Eq:Fubini}
\end{equation}
for any Lebesgue-integrable function $f: \Gamma_m^+\to \Real$ on the future mass shell. This (and other, similar) formula will play an important role when discussing the physical interpretation of the one-particle distribution function in the next section.

\section{One-particle distribution function and description of a collisionless simple gas}
\label{Sec:SimpleGas}

After having discussed the relevant geometric properties of the cotangent bundle, we are ready to provide the description of a relativistic, collisionless and uncharged simple gas propagating on a curved spacetime $(M,g)$. For this, we consider a Gibbs ensemble of identical, classical (i.e. not quantum), massive, free-falling test particles in $(M,g)$ which are further assumed to be electrically neutral and without spin. (Generalizations to gas configurations involving different species of electrically charged particles will be discussed in the next section.) At the macroscopic level, we describe such a gas configuration by a one-particle distribution function, that is a non-negative function $f$ on the future mass shell $\Gamma_m^+$, with $m > 0$ the mass of the gas particles.

In the next subsection, we provide a fully covariant description of the one-particle distribution function $f$ and explain its physical meaning through a flux integral interpretation. The connection with the usual Newtonian-type interpretation will also be discussed. Next, in subsection~\ref{SubSec:Observables} we discuss the physical observables associated with the one-particle distribution function. The main observables that will be used in this article consist of the particle current density, the energy-momentum-stress tensor and the entropy flux density of the gas, which are constructed from appropriate fibre integrals over $f$. The important example of the J\"uttner distribution function (representing a gas configuration in thermodynamic equilibrium) is discussed in subsection~\ref{SubSec:Juttner}. Finally, in subsection~\ref{SubSec:CollisionlessBoltzmann} we derive the collisionless Boltzmann equation as a simple application of our interpretation of the distribution function $f$.

\subsection{Physical interpretation of the one-particle distribution function}
\label{SubSec:OneParticleDistributionFunction}

In this subsection we describe the physical meaning of the general relativistic one-particle distribution function $f: \Gamma_m^+\to \Real$. To this end, we first introduce the following vector field on $\Gamma_m^+$:

\begin{definition}
Let $L$ be the Liouville vector field defined in Eq.~(\ref{Eq:LiouvilleVF}), 
which is tangent to $\Gamma_m^+$ (see Lemma~\ref{Lem:LTangentMassShell}). Then, we define the following vector field on $\Gamma_m^+$:
\begin{equation}
\boxed{ \mathcal{U} := \frac{1}{m} L.}
\label{Eq:UDef}
\end{equation}
\end{definition}
Recall that $\hat{h}(L,L) = -m^2$, which implies that $\mathcal{U}$ is a unit timelike vector field on $(\Gamma_m^+,\hat{h})$ which defines a time-orientation. Furthermore, as follows from Proposition~\ref{Prop:GeodesicL_on_MS} and Theorem~\ref{Thm:Liouville}, $\mathcal{U}$ is geodesic and expansion-free. Therefore, $\mathcal{U}$ can be interpreted as a $(2d+1)$-velocity field in relativistic phase space $\Gamma_m^+$ which generates a future-directed timelike geodesic and incompressible flow.

Before discussing the role of the $(2d+1)$-velocity vector field $\mathcal{U}$, let us first recall the analogous role played by the four-velocity vector field $u$ of a fluid flow on the four-dimensional spacetime manifold $(M,g)$. Let $n: M\to \Real$ be the particle density describing this flow, that is, the number of particles per unit volume measured by an observer which is co-moving with the flow. If $S$ is a three-dimensional spacelike compact hypersurface in $M$ (representing a volume at some given time) with future-directed unit normal $s$, then the total number $N[S]$ of particles contained in $S$ is given by the flux integral
\begin{equation}
N[S] = -\int\limits_S g(J,s) \eta_S,\qquad J := n u,
\label{Eq:FluxN}
\end{equation}
with $\eta_S = i_s\eta$ the induced volume element on $S$, and where the minus sign arises from the requirement that $N[S]$ should be positive if $J$ is future-directed timelike. In the particular case that $u$ is hypersurface orthogonal and $S$ is chosen orthogonal to $u$, Eq.~(\ref{Eq:FluxN}) reduces to the familiar expression $N[S] = \int_S n \eta_S$ in which $N[S]$ is just the volume integral over the particle density; however, in the general case $N[S]$ should be interpreted as the flux of the particle current density $J = n u$ through the surface $S$. If $(M,g)$ is an asymptotically flat, globally hyperbolic spacetime in which the particle density $n$ falls off sufficiently fast in spacelike directions, the  continuity equation $\nabla_\mu J^\mu = 0$ implies (through Gauss' theorem) that $N[S_1] = N[S_2]$ for two Cauchy surfaces $S_1$ and $S_2$. This expresses the conservation of the total particle number.

In analogy to the particle density $n$ and its associated current density $J = n u$ in the fluid case, we now introduce the one-particle distribution function $f:\Gamma_m^+\to \Real$ and the associated $(2d+1)$-current density
\begin{equation}
\boxed{ \mathcal{J} :=f \mathcal{U} = \frac{1}{m} f L }
\label{Eq:JDef}
\end{equation}
in phase space. Given a compact (either spacelike or timelike) $2d$-dimensional surface $\Sigma$ with unit normal vector field $\nu$, then\footnote{In order for $\mathcal{N}[\Sigma]$ to be dimensionless, the distribution function $f$ should have units of (length)$^{-2d}$. As is known from statistical mechanics, a dimensionless distribution function can be constructed by dividing the volume form by the factor $h^{2d}$ with Planck's constant $h$.}
\begin{equation}
\boxed{ \mathcal{N}\left[\Sigma\right] 
 := -\int\limits_{\Sigma} \hat{h}(\mathcal{J},\nu)\eta_{\Sigma},
\qquad
\eta_{\Sigma} := i_\nu\eta_{\Gamma_m^+},}
\label{Eq:NFlux}
\end{equation}
is interpreted as the averaged (over the Gibbs ensemble) number of occupied trajectories crossing $\Sigma$, see Fig.~\ref{Fig:PS_trajectories}. Here, we recall that $\hat{h}$ refers to the induced metric on $\Gamma_m^+$, and as before, the minus sign in the definition~(\ref{Eq:NFlux}) arises from the requirement that $\mathcal{N}\left[\Sigma\right]$ should be positive if $f > 0$ and $\nu$ is future-directed timelike. Further, $i_\nu\eta_{\Gamma_m^+}$ is the induced volume form on $\Sigma$. The flux integral~(\ref{Eq:NFlux}) is the analogue of the flux integral~(\ref{Eq:FluxN}) in the fluid case, and provides the physical interpretation of the distribution function.\footnote{Note that if $\mathcal{U}$ were orthogonal to $\Sigma$, such that $\nu = \mathcal{U}$, Eq.~(\ref{Eq:NFlux}) would reduce to $\mathcal{N}\left[\Sigma\right] = \int_\Sigma f \eta_\Sigma$. However, this interpretation cannot be given since the Liouville vector field $L$ is not hypersurface-orthogonal. Indeed, using Lemma~\ref{Lem:LH}(b), it follows that the associated one-form $\Theta$ satisfies
$$
\Theta\wedge d\Theta = \Theta\wedge \Omega_s = -p_\mu dp_\nu\wedge dx^\mu\wedge dx^\nu,
$$
and hence it does not satisfy the Frobenius condition for $\mathcal{U}$ to be hypersurface-orthogonal.
}

In order to make contact with the usual (Newtonian) definition of the distribution function we consider in $(M,g)$ a $d$-dimensional spacelike hypersurface $S$ (representing a certain volume at some given time) with future-directed unit normal vector $s$. Let
\begin{equation}
\Sigma := \{ (x,p) : x\in S, p\in P_x^+(m) \}
\label{Eq:SigmaFromS}
\end{equation}
be the corresponding surface in the future mass shell $\Gamma_m^+$. By definition, $\mathcal{N}\left[\Sigma\right]$ is the averaged number of occupied trajectories in $\Gamma_m^+$ whose projection on $M$ intersect the spacelike hypersurface $S$. Therefore, $\mathcal{N}\left[\Sigma\right]$ represents the averaged number of particles contained in the volume $S$. To compute $\mathcal{N}[\Sigma]$ for the particular surface $\Sigma$ given in Eq.~(\ref{Eq:SigmaFromS}), we first claim that $\Sigma$ is a spacelike hypersurface in $(\Gamma_m^+,\hat{h})$ with future-directed unit normal given by
\begin{equation}
\nu = s^\mu \frac{D}{d x^\mu}.
\label{Eq:nu}
\end{equation}
To prove this, consider an arbitrary curve $\gamma(\lambda)$ in $\Sigma$ which has the tangent vector $Z$ at the point $(x,p)$. We decompose $\displaystyle Z = X^\mu \frac{D}{d x^\mu} + Y_\mu\frac{\partial}{\partial p_\mu}$ and note that the tangent vector of the projected curve $\pi\circ\gamma$ is
\begin{equation}
\left. \frac{d}{dt} \pi\circ\gamma(t) \right|_{t=0} = d\pi_{(x,p)}(Z) = X^\mu \left. \frac{\partial}{\partial x^\mu} \right|_x .
\end{equation}
By definition, this vector is tangent to $S$, which implies that the horizontal components $X^\mu$ of $Z$ satisfy $g_{\mu\nu} X^\mu s^\nu = 0$. Since the normal vector to $\Sigma$ at $(x,p)$ must be orthogonal to all tangent vectors $Z$ of this form, it follows that it must be proportional to $\nu$ defined in Eq.~(\ref{Eq:nu}). Finally, we note that $\hat{h}(\nu,\nu) = g(s,s) = -1$ and $\hat{h}(\nu,L) = g(s,p) < 0$ which shows that $\nu$ is future-directed timelike and has unit norm. This proves the claim regarding the normal vector field. Next, we use Eqs.~(\ref{Eq:LiouvilleHor},\ref{Eq:nu}) and compute
\begin{equation}
\hat{h}(\mathcal{J},\nu) = \frac{1}{m} f\hat{g}(L,\nu) = \frac{1}{m} f p(s).
\end{equation}
Finally, using the representation~(\ref{Eq:etaGammam+}) for the volume form on $\Gamma_m^+$, combined with the observation that
\begin{equation}
i_\nu dx^\mu = s^\mu,\qquad i_\nu dp_{\hat{\alpha}} = 0,
\end{equation}
we find the following expression for the volume form on $\Sigma$:
\begin{equation}
\eta_\Sigma = i_\nu\eta_{\Gamma_m^+} = i_s\eta_M \wedge \eta_{P_x^+(m)}.
\end{equation}

Therefore, we conclude that the number $\mathcal{N}[\Sigma]$ for the particular surface $\Sigma$ given in Eq.~(\ref{Eq:SigmaFromS}), representing the averaged number of particles contained inside the volume $S$, is equal to
\begin{equation}
\mathcal{N}[\Sigma]  = -\frac{1}{m}\int\limits_S\left( \int\limits_{P_x^+(m)} f(x,p) p(s) \eta_{P_x^+(m)} \right) \eta_S
 = -\int\limits_S\left( \int\limits_{P_x^+(m)} f(x,p) p_{\hat{\alpha}} s^{\hat{\alpha}} \dvol_x(p)
\right) \eta_S,
\label{Eq:NSigmaFromS}
\end{equation}
with $\eta_S = i_s\eta_M$ the induced volume form on $S$ and $\dvol_x(p)$ the Lorentz-invariant volume element on $P_x^+(m)$ defined in Eq.~(\ref{Eq:dvolx}). Orienting the orthonormal basis $\left\lbrace e_{\hat{\alpha}} \right\rbrace$ such that at each point of $S$ its timelike leg $e_{\hat{0}} = s$ coincides with the normal vector $s$ to $S$, we obtain $p_{\hat{\alpha}} s^{\hat{\alpha}} = p_{\hat{0}}$ and Eq.~(\ref{Eq:NSigmaFromS}) further simplifies to
\begin{equation}
\mathcal{N}[\Sigma] = \int\limits_S \left( \int\limits_{\Real^d} f(x,p_{\hat{\alpha}}\theta^{\hat{\alpha}}) dp_{\hat{1}} \wedge dp_{\hat{2}} \wedge \cdots \wedge dp_{\hat{d}} \right) \eta_S.
\end{equation}
Assuming the existence of a local chart $(U,\phi)$ of $M$ in which $S$ is characterized by $x^0 = const$ such that $dp_{\hat{1}}\wedge \cdots \wedge dp_{\hat{d}} \wedge \eta_S = dp_1\wedge \cdots \wedge dp_d\wedge dx^1\wedge\cdots\wedge dx^d$ one obtains, assuming that $f$ is zero outside $U$,
\begin{equation}
\mathcal{N}[\Sigma] = \int\limits_{\Real^d}\int\limits_{\Real^d} f(\phi^{-1}(x^\mu),p_\mu dx^\mu) d^d p d^d x,
\end{equation}
where here $p_0$ is determined by $p_1,\ldots,p_d$ in such a way to ensure that $p_\mu dx^\mu\in P_x^+(m)$. This closely resembles the usual (non-relativistic) interpretation of the distribution function as a density over the phase space $(x^1,\ldots,x^d,p_1,\ldots,p_d)$ with volume element $d^d x d^d p$ and also provides the connection with the approach in Ref.~\cite{fDwL09a}.
\begin{figure}
\begin{centering}
\includegraphics[scale=0.30]{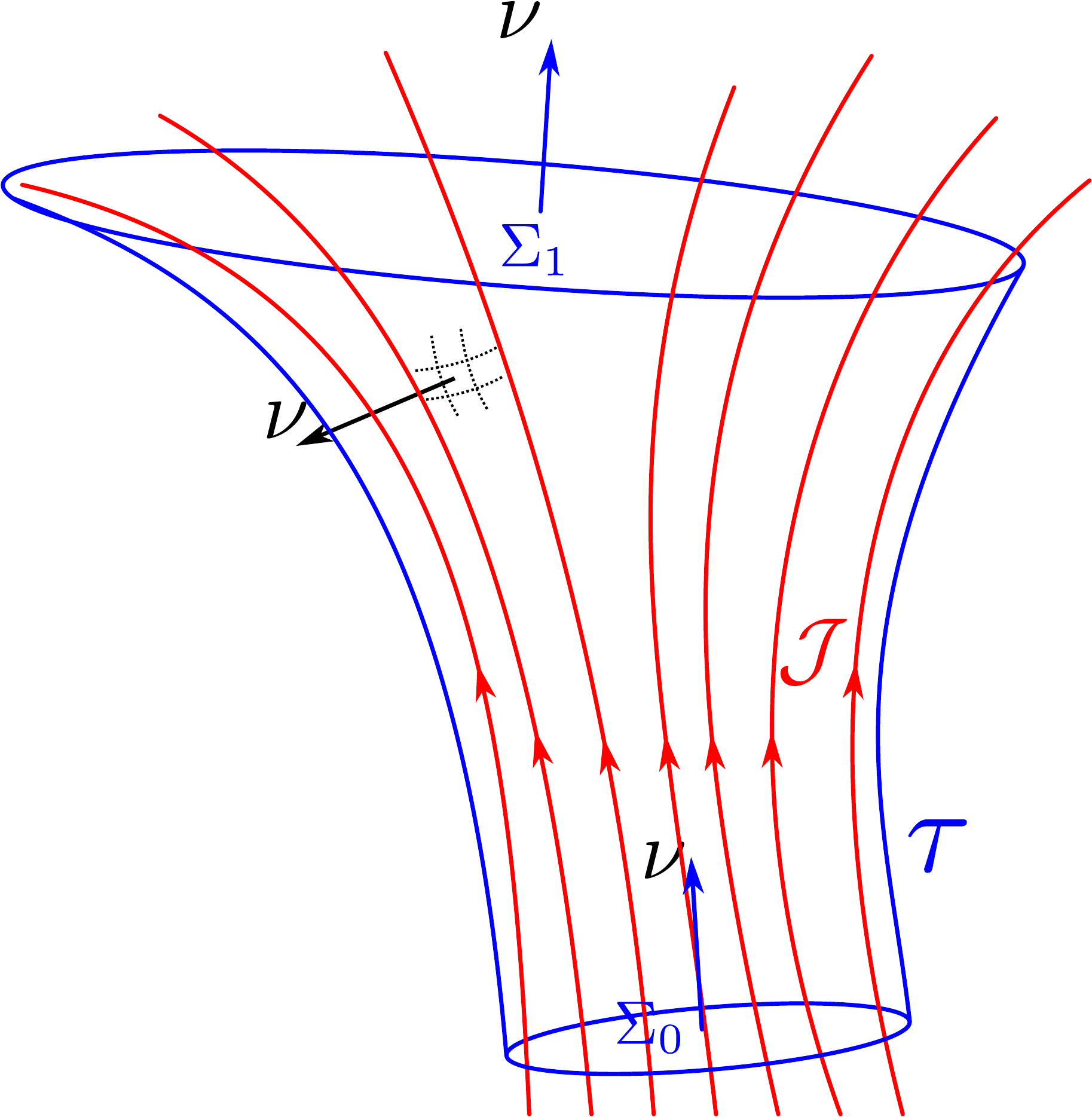}
\par\end{centering}
\caption{An illustration of the flow of the $(2d+1)$-current $\mathcal{J}$ crossing the $2d$-dimensional surface $\Sigma_0$ with normal vector $\nu$. The surface $\Sigma_1$ is obtained by transporting $\Sigma_0$ along this flow.}
\label{Fig:PS_trajectories}
\end{figure}

\subsection{Physical observables constructed from the one-particle distribution function}
\label{SubSec:Observables}

In this subsection, we use the one-particle distribution function for a simple gas $f$ and the results from the previous subsection to construct important tensor fields on the spacetime manifold $M$ by computing suitable fibre integrals of $f$. These quantities are the particle current density vector field $J$ and the energy-momentum-stress tensor field $T$, which, as we will show, are divergence-free:
\begin{equation}
\nabla_\mu J^\mu = 0, \qquad\hbox{and}\qquad \nabla_\mu T^{\mu\nu} = 0.
\label{Eq:LawOfConservation}
\end{equation}

We start with the derivation of $J$ and the associated conservation law $\nabla_\mu J^\mu = 0$. For this, we consider as in the previous subsection a $d$-dimensional spacelike hypersurface $S$ with future-directed unit normal vector $s$ in $(M,g)$ and the associated surface $\Sigma$ in $\Gamma_m^+$ given by Eq.~(\ref{Eq:SigmaFromS}). As we have seen, in this case $\mathcal{N}[\Sigma]$ represents the averaged number of particles contained inside the volume $S$. On the other hand, from Eq.~(\ref{Eq:NSigmaFromS}) we can also write this number as a flux integral over $S$, namely
\begin{equation}
\mathcal{N}[\Sigma] = -\int\limits_S J(s) \eta_S,
\label{Eq:NSigmaFromSBis}
\end{equation}
with the particle current density covector field $J\in {\cal X}^*(M)$ given by
\begin{equation}
\boxed{ J_x(X) := \int\limits_{P_x^+(m)} f(x,p) p(X) \dvol_x(p),\qquad X\in T_x M.}
\label{Eq:ParticleDensity}
\end{equation}
Comparing Eq.~(\ref{Eq:NSigmaFromSBis}) with Eq.~(\ref{Eq:FluxN}) we are led to the conclusion that this covector field (or its associated vector field) $J$ represents the averaged particle current density of the gas. Decomposing $J = n u$ with $n := \sqrt{-g^{-1}(J,J)}$, one obtains from this the mean particle velocity $u$ and the mean particle density $n$ of the gas configuration.

Introducing $\displaystyle X = \frac{\partial}{\partial x^\mu}$ (or $X = e_{\hat{\alpha}}$) in Eq.~(\ref{Eq:ParticleDensity}) we obtain the corresponding expressions for the components
\begin{equation}
J_\mu(x) = \int\limits_{P_x^+(m)}f(x,p)p_\mu \dvol_x(p).
\label{Eq:ParticleDensityCoord}
\end{equation}

The definition~(\ref{Eq:ParticleDensity}) of the particle current density as a fibre integral corresponding to the first moment of the distribution function can be generalized to higher moments in a straightforward way, yielding for each $s\in \Natural$ the totally symmetric tensor field
\begin{equation}
\boxed{T_x^{(s)}(X_1,X_2,\ldots,X_s)  
 := \int\limits_{P_x^+(m)} f(x,p) p(X_1) p(X_2)\ldots p(X_s) \dvol_x(p),
\qquad X_1,X_2,\ldots,X_s\in T_x M,}
\label{Eq:Ts}
\end{equation}
with associated components
\begin{equation}
T^{(s)}_{\mu_1\mu_2\ldots\mu_s}(x)
 = \int\limits_{P_x^+(m)} f(x,p) p_{\mu_1} p_{\mu_2}\ldots p_{\mu_s} \dvol_x(p).
\label{Eq:TsCoord}
\end{equation}
Of particular relevance are the first moment (i.e. the particle current density $J$) and the second one, $T_{\mu\nu} := T^{(2)}_{\mu\nu}$ which represents the energy-momentum-stress tensor of the gas. 

We conclude this subsection with the following identity which will result fundamental for the further development of the theory:

\begin{theorem}
\label{Thm:sMoments}
The $s$-moments $T^{(s)}$ defined in Eq.~(\ref{Eq:Ts}) satisfy the following identities:
\begin{equation}
\boxed{ \nabla^{\mu_1} T^{(s)}_{\mu_1\mu_2\ldots\mu_s}(x)
 = \int\limits_{P_x^+(m)} L[f](x,p) p_{\mu_2}\ldots p_{\mu_s} \dvol_x(p), \qquad s = 1,2,\ldots,}
\label{Eq:DivergenceIdentity}
\end{equation}
with $L$ the Liouville vector field defined in Eq.~(\ref{Eq:LiouvilleVF}).
\end{theorem}

\begin{proof}
We provide two different proofs for this important result. The first one, which provides more geometric insight, is based on Gauss' theorem and the identity (see the last subsection)
\begin{equation}
\int\limits_{\Sigma} \hat{h}(\mathcal{J},\nu) \eta_\Sigma
 = \int\limits_S\left( \int\limits_{P_x^+(m)} f(x,p) p(s) \dvol_x(p) \right) \eta_S
 = \int\limits_S J(s) \eta_S,
\label{Eq:KeyIdentity}
\end{equation}
which holds for any $d$-dimensional hypersurface $S$ in $M$ with normal vector field $s$, and the corresponding $2d$-dimensional hypersurface $\Sigma$ in $\Gamma_m^+$ defined in Eq.~(\ref{Eq:SigmaFromS}) with associated normal field $\nu$ given by Eq.~(\ref{Eq:nu}). Note that this identity holds regardless of whether $S$ is timelike or spacelike. We apply this identity to the case in which $S = \partial K$ arises as the boundary of a compact subset $K\subset M$ of spacetime, in which case $\Sigma = \partial V$ is the boundary of the corresponding volume
\begin{equation}
V = \{ (x,p) : x\in K, p\in P_x^+(m) \}
\end{equation}
in the future mass shell $\Gamma_m^+$. Applying Gauss' theorem on both sides of Eq.~(\ref{Eq:KeyIdentity}), we obtain
\begin{equation}
\int\limits_V \divrg \mathcal{J} \eta_{\Gamma_m^+} = \int\limits_K \divrg J \eta_M,
\end{equation}
with $\divrg$ denoting the divergence operator (in $\Gamma_m^+$ or $M$). Due to Liouville's theorem, see Theorem~\ref{Thm:Liouville}, we have
\begin{equation}
\divrg\mathcal{J} = \divrg\left( \frac{f}{m} L \right) = \frac{1}{m} L[f].
\end{equation}
This shows that
\begin{equation}
\int\limits_K \divrg J \eta_M = \frac{1}{m}\int\limits_V L[f] \eta_{\Gamma_m^+}
 = \int\limits_K \left( \int\limits_{P_x^+(m)} L[f] \dvol_x(p) \right) \eta_M,
\label{Eq:IdJ}
\end{equation}
where we have used the Fubini-type formula~(\ref{Eq:Fubini}) in the last step. Since this holds for any $K$, the statement of the theorem for $s=1$ follows. To generalize the proof to arbitrary $s$, we take arbitrary vector fields $X_2,\ldots,X_s\in {\cal X}(M)$ and consider instead of $J$ the covector field
\begin{equation}
\hat{J} := T^{(s)}(\;\cdot\;,X_2,\ldots,X_s),
\end{equation}
which is equivalent to replacing the distribution function $f(x,p)$ with $\hat{f}(x,p) := f(x,p) p(X_2)\cdots p(X_s)$ in the identity~(\ref{Eq:KeyIdentity}). Using
\begin{eqnarray*}
\divrg\hat{J} &=& \nabla^{\mu_1}\left( T^{(s)}_{\mu_1\mu_2\ldots\mu_s}
X_2^{\mu_2}\ldots X_s^{\mu_s} \right)
\nonumber\\
 &=& \left(\nabla^{\mu_1} T^{(s)}_{\mu_1\mu_2\ldots\mu_s}\right)
 X_2^{\mu_2}\ldots X_s^{\mu_s}
 + T^{(s)}_{\mu_1\mu_2\ldots\mu_s}(\nabla^{\mu_1} X_2^{\mu_2})
 X_3^{\mu_3}\ldots X_s^{\mu_s} + \cdots +
 T^{(s)}_{\mu_1\mu_2\ldots\mu_s} X_2^{\mu_2}
 \ldots X_{s-1}^{\mu_{s-1}}\nabla^{\mu_1} X_s^{\mu_s},
\end{eqnarray*}
and
\begin{eqnarray*}
L[\hat{f}] &=& L\left[ f p(X_2)\ldots p(X_s) \right]\\
 &=& L[f] p(X_2)\ldots p(X_s) + f L[p(X_2)] p(X_3) \ldots p(X_s)
  + \cdots + f p(X_2)\ldots p(X_{s-1}) L[p(X_s)],
\end{eqnarray*}
and the fact that $\displaystyle \frac{D}{dx^\nu}(p_\mu X_j^\mu) = p_\mu\nabla_\nu X_j^\mu$ which implies $\displaystyle L\left[ p(X_j) \right] =  p_\mu p_\nu \nabla^\mu X_j^\nu$ for all $j=2,3,\ldots,s$, the statement of the theorem follows from Eq.~(\ref{Eq:IdJ}) and the definition of $T^{(s)}$.

The alternative proof, which is technically simpler but offers less geometric insight, makes use of a Gaussian coordinate system at $x$ (i.e. a local inertial frame at $x$), such that $\displaystyle \left. e_{\hat{\alpha}}^\mu \right|_x = \delta^\mu{}_{\hat{\alpha}}$ and $\displaystyle \left. \frac{\partial}{\partial x^\nu} e_{\hat{\alpha}}^\mu \right|_x = 0 $. Then, using Eqs.~(\ref{Eq:TsCoord}) and~(\ref{Eq:dvolx}),
\begin{eqnarray*}
 \nabla^{\mu_1} T^{(s)}_{\mu_1\mu_2\ldots\mu_s}(x)
 &=& \frac{\partial}{\partial x^{\mu_1}}  
 \int\limits_{\Real^d} f(x,p) p^{\mu_1} p_{\mu_2}\ldots p_{\mu_s} \frac{d^d p}{\sqrt{m^2 + p_1^2 + p_2^2 + \cdots + p_d^2}}\\
 &=& \int\limits_{\Real^d} p^{\mu_1} \frac{\partial f}{\partial x^{\mu_1}}(x,p)  p_{\mu_2}\ldots p_{\mu_s} \frac{d^d p}{\sqrt{m^2 + p_1^2 + p_2^2 + \cdots + p_d^2}}\\
 &=& \int\limits_{P_x^+(m)} L[f] p_{\mu_2}\ldots p_{\mu_s} \dvol_x(p),
\end{eqnarray*}
where in the last step we have used the explicit expression~(\ref{Eq:LiouvilleVF}) for the Liouville vector field and the vanishing of the first derivatives of the metric components at $x$ due to the choice of the Gaussian system. This concludes the second proof of the identity.
\qed
\end{proof}

\subsection{An example: The J\"uttner distribution function}
\label{SubSec:Juttner}

As a typical an important example we consider the distribution function~\cite{fJ11a,fJ11b}
\begin{equation}
\boxed{ f(x,p) = \alpha e^{p(\beta)} = \alpha e^{\beta^\mu p_\mu}, }
\label{Eq:Juttner}
\end{equation}
with $\alpha\in {\cal F}(M)$ a positive function and $\beta\in {\cal X}(M)$ a future-directed timelike vector field on $M$. Introducing the generating function~\cite{wI63}
\begin{equation}
Z(\alpha,\beta) := \alpha\int\limits_{P_x^+(m)} e^{p(\beta)} \dvol_x(p),
\label{Eq:GeneratingFunction}
\end{equation}
the $s$-moment tensor can be computed according to
\begin{equation}
T^{(s)}_{\mu_1\mu_2\ldots\mu_s} = \frac{\partial^s Z(\alpha,\beta)}{\partial\beta^{\mu_1}\partial\beta^{\mu_2}\ldots\partial\beta^{\mu_s}}.
\end{equation}
Exploiting the Lorentz-invariance of the volume element $\dvol_x(p)$, it is sufficient to compute $Z(\alpha,\beta)$ in an orthonormal frame $\left\lbrace e_{\hat{\alpha}} \right\rbrace$ whose timelike leg $e_{\hat{0}}$ is aligned with $\beta$. Parametrizing the momentum in the form
$$
p = -m\cosh\chi e_{\hat{0}} + m\sinh\chi n^{\hat{a}} e_{\hat{a}} ,
$$
with $\chi\geq 0$ a hyperbolic angle and $\ve{n} := (n^{\hat{1}},n^{\hat{2}},\ldots,n^{\hat{d}})\in S^{d-1}$ a unit vector, we find
$$
\dvol_x(p) = m^{d-1}\sinh^{d-1}\chi d\chi\wedge d\Omega(\ve{n}),
$$
with $d\Omega(\ve{n})$ the standard volume element on $S^{d-1}$. Introducing this into Eq.~(\ref{Eq:GeneratingFunction}) one obtains
\begin{equation}
Z(\alpha,\beta) = 2\alpha\left( \frac{2\pi m^2}{z} \right)^{\frac{d-1}{2}} K_{\frac{d-1}{2}}(z),
\qquad z = m\sqrt{-g(\beta,\beta)} = m\sqrt{-\beta^\mu\beta_\mu},
\end{equation}
where here
\begin{equation}
K_\nu(z) := \int\limits_0^\infty e^{-z\cosh(\chi)} \cosh(\nu\chi) d\chi
 = \frac{\sqrt{\pi}\left( \frac{z}{2} \right)^\nu}{\Gamma\left( \nu + \frac{1}{2} \right)}
 \int\limits_0^\infty e^{-z\cosh(\chi)} \sinh^{2\nu}(\chi) d\chi,\qquad
\nu > -\frac{1}{2},\quad z > 0,
\end{equation}
denote the modified Bessel functions of the second kind, see for instance~\cite{DLMF}. Based on these expressions and the relation
\begin{equation}
z^\nu\frac{d}{dz}\left[ z^{-\nu} K_\nu(z) \right] = -K_{\nu+1}(z),
\end{equation}
one easily finds the following expressions
\begin{equation}
J^\mu = n(z) u^\mu,\qquad
T_{\mu\nu} = n(z) h(z) u_\mu u_\nu + P(z) g_{\mu\nu},
\label{Eq:PerfectFluid}
\end{equation}
for the particle current density and the energy-momentum-stress tensor, where here
\begin{eqnarray}
&& n(z) = 2\alpha m\left( \frac{2\pi m^2}{z} \right)^{\frac{d-1}{2}} K_{\frac{d+1}{2}}(z),
\qquad
u^\mu = \frac{\beta^\mu}{\sqrt{-\beta^\mu\beta_\mu}},
\label{Eq:PerfectFluidnu}\\
&& h(z) = m\frac{K_{\frac{d+3}{2}}(z)}{K_{\frac{d+1}{2}}(z)},
\qquad
P(z) = \frac{\alpha}{\pi}\left( \frac{2\pi m^2}{z} \right)^{\frac{d+1}{2}} K_{\frac{d+1}{2}}(z).
\label{Eq:PerfectFluidhp}
\end{eqnarray}
These formulae, which generalize the corresponding ones in~\cite{fJ11a,fJ11b,wI63} to arbitrary space dimensions,\footnote{See also Refs.~\cite{gChetal2010,lAgCh2018,gL2016} for similar distribution functions in $d$ dimensions with additional parameters.} describe a perfect fluid configuration with velocity $u$, particle density $n$, enthalpy per particle $h$ and pressure $P$ satisfying (formally) the ideal gas equation
\begin{equation}
\frac{P}{n} = \frac{m}{z} = k_B T,\qquad T := \frac{1}{k_B}(-\beta^\mu\beta_\mu)^{-1/2}.
\label{Eq:IdealGas}
\end{equation}
As will be shown in section~\ref{Sec:Equilibrium}, when collisions are present, the distribution function~(\ref{Eq:Juttner}) describes a local equilibrium configuration of temperature $T$.

\subsection{The collisionless Boltzmann equation}
\label{SubSec:CollisionlessBoltzmann}

An evolution equation for the distribution function describing a collisionless and uncharged simple gas follows immediately from the definition~(\ref{Eq:NFlux}) of the averaged number of occupied trajectories crossing a given compact spacelike $2d$-dimensional surface $\Sigma_0$ in $\Gamma_m^+$ and Gauss' theorem. For this, consider for each $\Delta t > 0$ the tubular region
\begin{equation}
V := \bigcup_{0\leq t \leq \Delta t} \psi^t(\Sigma_0),
\label{Eq:VDef}
\end{equation}
which is obtained by transporting the surface $\Sigma_0$ along the flow $\psi^t$ of $\mathcal{J}$. The boundary of $V$ consists of the ``initial" and ``final" hypersurfaces $\Sigma_0$ and $\Sigma_1 := \psi^{\Delta t}(\Sigma_0)$ and the cylindrical piece, $\displaystyle {\cal T} := \bigcup_{0\leq t\leq \Delta t}\partial(\psi^t(\Sigma_0))$, see Fig.~\ref{Fig:PS_trajectories}. Integrating the divergence of $\mathcal{J}$ over $V$, using Gauss' theorem and noting that the flux integral over ${\cal T}$ vanishes since $\mathcal{J}$ is tangent to it, one obtains the balance law
\begin{equation}
\mathcal{N}[\Sigma_1] - \mathcal{N}[\Sigma_0] 
 = \int\limits_V (\divrg\mathcal{J}) \eta_{\Gamma_m^+}
 = \frac{1}{m}\int\limits_V L[f] \eta_{\Gamma_m^+}. 
\label{Eq:BalanceLaw}
\end{equation}

In the absence of collisions between the gas particles, one has $\mathcal{N}[\Sigma_1] = \mathcal{N}[\Sigma_0]$, and hence the one-particle distribution function $f$ must satisfy the Liouville (or Vlasov or collisionless Boltzmann) equation
\begin{equation}
\boxed{L[f] = g^{\mu\nu}p_{\nu}\frac{\partial f}{\partial x^{\mu}}
 - \frac{1}{2}\frac{\partial g^{\alpha\beta}}{\partial x^{\mu}}p_{\alpha}p_{\beta}
 \frac{\partial f}{\partial p_{\mu}} = 0.}
\label{Eq:LiouvilleEq}
\end{equation}
In terms of the local coordinates $(x^\mu,p_{\hat{a}})$ parametrizing the future mass shell $\Gamma_m^+$ (see Section~\ref{SubSec:Liouville}), one can also write the Liouville equation more explicitly as
\begin{equation}
\boxed{ L[f] = p^{\hat{\beta}} e_{\hat{\beta}}^\mu\left( \frac{\partial f}{\partial x^{\mu}} 
 + p_{\hat{\alpha}} \hat{\Gamma}^{\hat{\alpha}}{}_{\mu\hat{b}} 
 \frac{\partial f}{\partial p_{\hat{b}}} \right) = 0,}
\label{Eq:LiouvilleEqMassShell}
\end{equation}
where it is understood that $p_{\hat{0}}$ is determined in terms of $p_{\hat{a}}$ according to Eq.~(\ref{Eq:p0}) and $\hat{\Gamma}^{\hat{\alpha}}{}_{\mu\hat{b}}$ was defined in Eq.~(\ref{Eq:Dphat}).\footnote{In order to prove that Eq.~(\ref{Eq:LiouvilleEqMassShell}) is equivalent to Eq.~(\ref{Eq:LiouvilleEq}) one can start from the alternative representation~(\ref{Eq:LiouvilleHor}) for the Liouville vector field, along with the observation that in terms of the coordinates $(x^\mu,p_{\hat{\alpha}})$ on $T^* M$ one has
$$
\frac{D}{dx^\mu} = \frac{\partial}{\partial x^\mu} + \hat{\Gamma}^{\hat{\alpha}}{}_{\mu\hat{\beta}} p_{\hat{\alpha}}\frac{\partial}{\partial p_{\hat{\beta}}}.
$$
Using the fact that $\hat{\Gamma}_{\hat{\alpha}\mu\hat{\beta}}$ is antisymmetric in $\hat{\alpha}\hat{\beta}$ and that when the coordinates $p_\mu$ are considered to be functions of $p_{\hat{a}}$, $\hat{a} = 1,2,\ldots,d$, on the mass shell, one has
$$
\frac{\partial}{\partial p_{\hat{a}}} = \left( \frac{p^{\hat{a}}}{p_{\hat{0}}} \theta^{\hat{0}}{}_\mu
 + \theta^{\hat{a}}{}_\mu \right) \frac{\partial}{\partial p_\mu},
$$
the desired equivalence follows immediately.}
Recently, the properties of the solutions of the Liouville equation for a fixed (Schwarzschild or Kerr) black hole background have been analyzed in detail. In this case, it is possible through the use of action-angle-type variables to provide a formal solution for the distribution function~\cite{oStZ14b}. For applications to the accretion problem onto a Schwarzschild black hole, see Refs.~\cite{pRoS17,pRoS17b,aCpM20,pMaO21a,pMaO21b,aGcGpDdNoS21}. Other recent applications~\cite{pRoS18,pRoS20} show that a collisionless gas which is trapped in the gravitational potential of a black hole relaxes in time to a stationary state (even though the motion of the individual particles is quasi-periodic!), an effect which is due to phase space mixing. For the decay properties of the solutions to the Liouville equation for a massless gas propagating on a Schwarzschild or Kerr background, see Refs.~\cite{lApBjS18,lB20}.

When the self-gravity of the gas becomes important, one cannot assume that the spacetime metric $g$ is fixed, and in this case one should consider instead the coupled Einstein-Liouville (or Einstein-Vlasov) system
\begin{eqnarray}
&& L[f] = 0,
\label{Eq:LiouvilleEqBis}\\
&& G_{\mu\nu} =\kappa T_{\mu\nu},
\label{Eq:EinsteinEq}
\end{eqnarray}
with $G_{\mu\nu}$ the components of the Einstein tensor, $\kappa$ the gravitational coupling constant (which in natural units reads $\kappa = 8\pi G_N/c^4$ with Newton's constant $G_N$), and $T_{\mu\nu}$ the energy-momentum-stress tensor defined by Eq.~(\ref{Eq:TsCoord}) with $s = 2$. Note that the integrability condition $\nabla^\mu G_{\mu\nu} = 0$ for the Einstein field equations~(\ref{Eq:EinsteinEq}) is a consequence of the Liouville equation~(\ref{Eq:LiouvilleEqBis}) and the identity~(\ref{Eq:DivergenceIdentity}). In the presence of a Killing vector field $k\in {\cal X}(M)$, the latter provides an infinite family of conservation laws
\begin{equation}
\nabla_\mu J_{(s)}^\mu = 0,\qquad
J_{(s)}^\mu := g^{\mu\nu} T^{(s)}_{\nu\mu_2\ldots\mu_s} k^{\mu_2} \ldots k^{\mu_s},
\qquad s = 1,2,\ldots
\end{equation}
For work on the properties of solutions to the Einstein-Vlasov system we refer the reader to the review article~\cite{hA11}, see also Refs.~\cite{aAmC14,eAhAaL16,eAhAaL19} for recent numerical work. Recently, the nonlinear stability of Minkowski spacetime as a solution of the massless~\cite{mT17,lBdFjJjSmT20,jJmTjV20} and massive~\cite{hLmT20,dFjJjS21} Einstein-Vlasov systems has also been established.

For a collisional gas, the left-hand side of Eq.~(\ref{Eq:BalanceLaw}) represents the ensemble average of the net change in number of occupied trajectories between $\Sigma_0$ and $\Sigma_1$, due to collisions. This change is described by the transition probability, as will be discussed in detail in section~\ref{Sec:Collision}.

\subsection{A collisionless simple gas with a continuous mass distribution}
\label{SubSec:MassDistribution}

We would like to emphasize that one might encounter physical situations in which the particles constituting the gas do not have a fixed mass, but are subject to a continuous mass distribution. A prominent example arises in galactic dynamics (see, for instance~\cite{BinneyTremaine-Book}) where the stars of the galaxy represent the individual particles whose mass distribution might be supported on an open interval. In the case, the one-particle distribution function $f: T^*M \to \Real$ should in principle be regarded as a function on the cotangent bundle (instead of a particular mass shell) whose support lies inside the set
\begin{equation}
\Gamma^+ := \bigcup_{m > 0} \Gamma_m^+
\end{equation}
in order to guarantee that each particle follows a future-directed timelike trajectory. In this case the flux integral~(\ref{Eq:NFlux}) should be replaced with\footnote{Here the factor $(-1)^d$ is introduced such that $(-1)^d  i_\nu\eta_{T^* M} = i_s\eta\wedge \eta_{T_x^*M}$ for the particular case in which $\Sigma$ stems from a $d$-dimensional spatial hypersurface $S$ with  unit normal $s$.}
\begin{equation}
\mathcal{N}\left[\Sigma\right] 
 = -(-1)^d\int\limits_{\Sigma} f(x,p)\hat{g}(L,\nu) i_\nu\eta_{T^* M},\qquad
\label{Eq:NFluxMass}
\end{equation}
with $\Sigma$ a $(2d+1)$-dimensional spacelike hypersurface in $(\Gamma^+,\hat{g})$ with unit normal $\nu$. Gauss' theorem, Liouville's theorem on $\Gamma^+$ (Proposition~\ref{Prop:LiouvilleT*M}) and the hypothesis of absence of collisions lead to the Liouville equation $L[f] = 0$ on $\Gamma^+$, which in adapted local coordinates can be written in the form~(\ref{Eq:LiouvilleEq}). The particle current density $J$ and the higher moments $T^{(s)}$ tensor fields are given by the same expressions as in Eqs.~(\ref{Eq:ParticleDensity},\ref{Eq:Ts}) where one replaces $\dvol_x(p)$ with the following expression:
\begin{equation}
\frac{d^{d+1} p}{\sqrt{|\det(g_{\mu\nu})|}},
\end{equation}
and integrates over the whole fibre $T_x^* M$. With these replacements, the divergence identity~(\ref{Eq:DivergenceIdentity}) still holds, and the Einstein-Liouville system of equations~(\ref{Eq:LiouvilleEqBis},\ref{Eq:EinsteinEq}) leads again to a consistent set of equations.

However, due to the weak equivalence principle, the presence of a mass dispersion does not affect the dynamics of the Einstein-Liouville system. In order to see this explicitly it is sufficient to notice that the moments $T^{(s)}$ described above can be rewritten as fibre integrals over the unit mass hyperboloid $P_x^+(1)$ as follows:
\begin{equation}
T^{(s)}_{\mu_1\mu_2\ldots\mu_s}(x)
 = \int\limits_{P_x^+(1)} f_s(x,\tilde{p}) 
 \tilde{p}_{\mu_1}\tilde{p}_{\mu_2}\ldots \tilde{p}_{\mu_s}\dvol_x(\tilde{p}),
\qquad
f_s(x,\tilde{p}) := \int\limits_0^\infty f(x,m\tilde{p}) m^{s+d} dm,
\end{equation}
with $f_s$ describing the $s$'th mass moment of the distribution function. Since the Liouville vector field $L$ leaves the future mass shells $\Gamma_m^+$ invariant, these mass moments $f_s$ satisfy the Liouville equation as well and hence it is sufficient to consider the Einstein-Liouville system~(\ref{Eq:LiouvilleEqBis},\ref{Eq:EinsteinEq}) with $f$ replaced by $f_2$.

\section{Generalization to a charged gas with several species of particles}
\label{Sec:Charged}

In the previous section we discussed the relativistic kinetic theory for a simple, collisionless uncharged gas. In this section we generalize the theory to the case of a kinetic gas consisting of several species of charged particles (still neglecting collisions between the gas particles for the moment). This theory is particularly relevant for the description of hot plasmas.

In contrast to the electrically neutral case, in general, the individual particles of a charged gas configuration do not follow geodesic trajectories in spacetime, due to the presence of the electromagnetic field $F$ generated by the electric current in the configuration. Therefore, in this case, one ends up with a coupled system of transport equations describing the evolution of the distribution functions $f^{(a)}$ associated with each particle species (which depend on the electromagnetic field $F$) and Maxwell's equations for $F$ with the particle current density vector depending on the first moments of $f^{(a)}$. This gives rise to the Vlasov-Maxwell system of equations, and in this section we shall provide the details for its derivation, generalizing the description of the previous sections to the charged case. For the corresponding description on the tangent bundle $T M$, see Refs.~\cite{jE71,oStZ14a} and references therein.

\subsection{Description for a charged simple gas}

For simplicity, we start with the case of a single species of charged particles propagating on a curved spacetime $(M,g)$ with an external electromagnetic field which, we assume, is described by a closed two-form $F$ on $M$ ($dF = 0$). As explained above, the gas particles do not follow geodesic trajectories anymore, since they are accelerated due to the presence of the Lorentz force. As is well-known from classical electrodynamics (see for instance Ref.~\cite{Jackson-Book} Section~12.1) the transition from the free particle Hamiltonian $\mathcal{H}$ defined in Eq.~(\ref{Eq:FreeParticleH}) to the Hamiltonian describing the motion of a charged particle can be obtained by expressing the physical momentum $p$ of the particle as
\begin{equation}
p = \Pi - q A,
\label{Eq:Pi}
\end{equation}
where $A$ is the electromagnetic potential one-form (such that $F = dA$), $q$ the charge of the particle, and $\Pi$ the canonical conjugate momentum. The equations of motion are thus obtained from the new Hamiltonian function $\mathcal{H} = \mathcal{H}(x,\Pi - q A)$ on $T^* M$ and since $\Pi$ is the canonical conjugate momentum, the symplectic form is
\begin{equation}
d\Pi_\mu\wedge dx^\mu = dp_\mu\wedge dx^\mu 
 + \frac{q}{2} F_{\mu\nu} dx^\mu\wedge dx^\nu.
\end{equation}

These considerations lead to the following definition, which generalize the definition of the symplectic form in Eq.~(\ref{Eq:Omega}) to the charged case:
\begin{equation}
\boxed{ \Omega_{s,F} := \Omega_s + q\pi^* F = dp_\mu\wedge dx^\mu + \frac{q}{2} F_{\mu\nu} dx^\mu\wedge dx^\nu,}
\label{Eq:OmegaF}
\end{equation}
where $\pi^*$ denotes the pull-back of the projection map $\pi: T^*M \to M$. Since $dF = 0$ it follows that $\Omega_{s,F}$ is closed. Furthermore, it is non-degenerate as can easily be deduced from the identity
\begin{equation}
\Omega_{s,F}(\:\cdot\:,X) = X^\mu dp_\mu - (Y_\mu - q F_{\mu\nu} X^\nu) dx^\mu,
\qquad X = X^\mu\frac{\partial}{\partial x^\mu} + Y_\mu\frac{\partial}{\partial p_\mu}\in T_{(x,p)}(T^* M),
\label{Eq:OmegaFDotX}
\end{equation}
which implies that $\Omega_{s,F}(\:\cdot\:,X) = 0$ if and only if $X = 0$. Note that here and in the following, we choose to formulate the theory in terms of the physical momentum $p$ instead of the canonical conjugate momentum $\Pi$, which is the reason for the explicit appearance of the electromagnetic field $F$ in the expression for the symplectic form $\Omega_{s,F}$. This choice has the advantage that it avoids the need of introducing the potential one-form $A$, which might exist only locally and leads to a gauge-dependent definition of the momentum $\Pi$. Thus, in the approach adopted here (see also Refs.~\cite{eK00,oI03}), the Hamiltonian describing the motion of the particle is unchanged with respect to the uncharged case:
\begin{equation}
\boxed{ \mathcal{H}(x,p)=\frac{1}{2}g^{-1}_{x}(p,p)=\frac{1}{2}g^{\mu\nu}(x)p_{\mu}p_{\nu},}
\label{Eq:FreeParticleHBis}
\end{equation}
and the electromagnetic field manifests itself in the definition of the symplectic form~(\ref{Eq:OmegaF}).\footnote{Note that locally, one might assume $F = dA$ and the new local coordinates $(x^\mu,\Pi_\mu) := (x^\mu,p_\mu + q A_\mu)$ are Darboux coordinates, such that $\Omega_{s,F} = d\Pi_\mu\wedge dx^\mu$ and $\mathcal{H} = \frac{1}{2} g^{\mu\nu}(x)(\Pi_\mu - q A_\mu)(\Pi_\nu - q A_\nu)$ as in the usual approach.} Comparing $d\mathcal{H}$ with the right-hand side of Eq.~(\ref{Eq:OmegaFDotX}) one obtains the Liouville vector field for the charged case, defined as the Hamiltonian vector field associated with $\mathcal{H}$ on the symplectic manifold $(T^* M,\Omega_{s,F})$:
\begin{equation}
\boxed{ L_F =  g^{\mu\nu}p_{\nu}\frac{\partial}{\partial x^{\mu}}-\frac{1}{2}\frac{\partial g^{\alpha\beta}}{\partial x^{\mu}}p_{\alpha}p_{\beta}\frac{\partial}{\partial p_{\mu}} 
 + q F_\mu{}^\nu p_\nu \frac{\partial}{\partial p_\mu}. }
\label{Eq:LiouvilleVFCharged}
\end{equation}
As can be easily verified, the corresponding integral curves $(x(\lambda),p(\lambda))$, when projected onto the spacetime manifold by means of the projection map $\pi$, satisfy
\begin{equation}
\ddot{x}^\mu + \Gamma^\mu{}_{\alpha\beta}\dot{x}^\alpha\dot{x}^\beta 
 = q F^\mu{}_\nu\dot{x}^\nu,
\end{equation}
which are the equations of motion for a charged particle in an exterior electromagnetic field. Comparing with the Liouville vector field $L$ in the uncharged case (see Eqs.~(\ref{Eq:LiouvilleVF},\ref{Eq:LiouvilleHor})), one can write 
\begin{equation}
L_F = L + V,\qquad
L = g^{\mu\nu} p_\nu\frac{D}{dx^\mu},\qquad
V = q F_\alpha{}^\beta p_\beta \frac{\partial}{\partial p_\alpha},
\label{Eq:LiouvilleVFChargedBis}
\end{equation}
and thus, in the charged case the Liouville vector field acquires a vertical component $V$.\footnote{Recall that the decomposition~(\ref{Eq:tanS_of_cotanM}) of tangent vectors into horizontal and vertical components depends on the choice of the connection map $K_{(x,p)}$ introduced below Definition~\ref{Def:HorVer}, which was defined naturally in terms of the Levi-Civita connection belonging to the spacetime metric. In the charged case, it would be tempting to introduce a modified version of the map $K_{(x,p)}$, defined as
$$
K_{(x,p)}^{(F)}(Z) := K_{(x,p)}(Z) + q F_x(d\pi(Z),\;\cdot\;),\qquad
Z\in T_{(x,p)}(T^* M), 
$$
such that $L_F$ is horizontal with respect to the induced ``tilted" horizontal space $\ker K_{(x,p)}^{(F)}$. The definitions of the almost complex structure $J$ and bundle metric $\hat{g}$ would be modified accordingly and in fact, many of the nice properties listed in Section~\ref{SubSec:Sasaki} could be generalized. On the other hand, it is not clear (to us) what the geometric interpretation of the map $K_{(x,p)}^{(F)}$ would be. Moreover, since the bundle metric $\hat{g}$ would be modified, the corresponding geometry of $T^*M$ and the future mass shell $\Gamma_m^+$ would change. Most notably, one does not obtain the direct generalization of the relation~(\ref{Eq:OmegaJgRelation}) between the symplectic form and the bundle metric and consequently, not all tilted horizontal vector fields would be tangent to $\Gamma_m^+$, leading to an inequivalent definition of the normal vector field $N$. For these reasons, we prefer to adopt the same definitions of the horizontal space and bundle metric as in the  uncharged case.} The volume forms $\eta_{T^* M}$ and $\eta_{\Gamma_m^+}$ on the cotangent bundle and future mass shell are defined exactly in the same way as in the uncharged case, see Section~\ref{SubSec:Liouville}, and Liouville's theorem on $T^* M$ and $\Gamma_m^+$ (Proposition~\ref{Prop:LiouvilleT*M} and Theorem~\ref{Thm:Liouville}) are still valid if $L$ is replaced with $L_F$. To see this, it is sufficient to notice that one can replace $\Omega_s$ with $\Omega_{s,F}$ in the identity~(\ref{Eq:volumeformonT*M}), such that $\pounds_{L_F}\eta_{T^* M} = 0$ is again a direct consequence of $\pounds_{L_F}\Omega_{s,F} = (d i_{L_F} + i_{L_F} d)\Omega_{s,F} = -d^2\mathcal{H} = 0$. Moreover, when generalizing the proof of Theorem~\ref{Thm:Liouville} one uses the identity $N = J(L)/m$ and notes that
\begin{equation}
m\pounds_{L_F} N = [ L_F,J(L)]  = [L, J(L)] + [V,J(L)] = -L,
\end{equation}
where in the last step we have used the fact that
\begin{equation}
[V,J(L)] = q F_\alpha{}^\beta
\left[ p_\beta\frac{\partial}{\partial p_\alpha}, p_\mu\frac{\partial}{\partial p_\mu} \right]
 = 0.
\end{equation}

The description of a collisionless, simple charged gas proceeds exactly as in Section~\ref{Sec:SimpleGas}, replacing $L$ with $L_F$ in the definition of the vector field $\mathcal{U}$ in Eq.~(\ref{Eq:UDef}), which leads to the current density
\begin{equation}
\mathcal{J}_F = f\mathcal{U}_F,\qquad
\mathcal{U}_F := \frac{1}{m} L_F = \frac{1}{m}(L + V),
\end{equation}
with $f: \Gamma_m^+ \to \Real$ the one-particle distribution function. Note that by the generalization of Liouville's theorem to the charged case mentioned above the vector field $\mathcal{U}_F$ is still divergence-free. However, due to its vertical component, its squared norm is given by
\begin{equation}
\hat{g}(\mathcal{U}_F,\mathcal{U}_F) = \frac{1}{m^2}\left[ \hat{g}(L,L) + \hat{g}(V,V) \right]
  = -\left( 1 - \frac{q^2}{m^2} F^{\mu\alpha} F_\mu{}^\beta p_\alpha p_\beta \right),
\end{equation}
and thus $\mathcal{U}_F$ is only timelike as long as $qF$ is small enough. The physical interpretation of the one-particle distribution is again given by the flux integral~(\ref{Eq:NFlux}) where one replaces $\mathcal{J}$ with $\mathcal{J}_F$:
\begin{equation}
\boxed{ \mathcal{N}_F\left[\Sigma\right] 
 := -\int\limits_{\Sigma} \hat{h}(\mathcal{J}_F,\nu)\eta_{\Sigma},
\qquad
\eta_{\Sigma} := i_\nu\eta_{\Gamma_m^+},}
\label{Eq:NFluxF}
\end{equation}
for any $2d$-dimensional surface $\Sigma$ in $\Gamma_m^+$ with unit normal vector field $\nu$. Since $\mathcal{J}_F$ distinguishes itself from its uncharged counterpart $\mathcal{J}$ only by a vertical vector field, it follows that $\mathcal{N}_F[\Sigma] = \mathcal{N}[\Sigma]$ for any surface $\Sigma$ whose normal vector field $\nu$ is purely horizontal. This applies, in particular, to all surfaces $\Sigma$ of the form described in Eq.~(\ref{Eq:SigmaFromS}), which are induced by a $d$-dimensional oriented hypersurface $S$ in $M$, such that the integral identity~(\ref{Eq:KeyIdentity}) holds also with $\mathcal{J}$ replaced with $\mathcal{J}_F$. This leads to the following generalization of Theorem~\ref{Thm:sMoments}:

\begin{theorem}
\label{Thm:sMomentsF}
The $s$-moments $T^{(s)}$ defined in Eq.~(\ref{Eq:Ts}) satisfy the following identities
\begin{equation}
\boxed{ \nabla^{\mu_1} T^{(s)}_{\mu_1\mu_2\ldots\mu_s}(x)
 = \int\limits_{P_x^+(m)} L_F[f] p_{\mu_2}\ldots p_{\mu_s} \dvol_x(p)
  + q(s-1) F_{(\mu_2}{}^\beta T^{(s-1)}_{\mu_3\ldots \mu_s)\beta}, \qquad s=1,2,\ldots,}
\label{Eq:DivergenceIdentityF}
\end{equation}
with $L_F$ the Liouville vector field defined in Eq.~(\ref{Eq:LiouvilleVFCharged}).
\end{theorem}

\begin{proof}
With the above observations, the generalization of the proof of Theorem~\ref{Thm:sMoments} based on Gauss' theorem is straightforward. The case $s = 1$ follows immediately, while for $s > 1$ it is sufficient to notice that
\begin{equation}
L_F[ p(\beta) ] = p_\mu p_\nu\nabla^\mu \beta^\nu + q F_\mu{}^\nu p_\nu \beta^\mu,
\label{Eq:LFpbeta}
\end{equation}
for all $\beta\in {\cal X}(M)$.

Alternatively, one may insert $L = L_F - V$ in the identity~(\ref{Eq:DivergenceIdentity}) and use the local coordinates $(x^\mu,p_{\hat{a}})$ on $\Gamma_m^+$ to write the vertical part as
\begin{equation}
V = q F_{\hat{b}}{}^{\hat{\alpha}} p_{\hat{\alpha}} \frac{\partial}{\partial p_{\hat{b}}},
\end{equation}
where we recall that $\hat{b} = 1,2,\ldots,d$ and $p_{\hat{0}}$ is obtained from Eq.~(\ref{Eq:p0}). Using this representation and integration by parts, one obtains
\begin{eqnarray*}
&& \int\limits_{P_x^+(m)} V[f] p_{\hat{\alpha}_2} \ldots p_{\hat{\alpha}_s} \dvol_x(p)
 = q F_{\hat{b}}{}^{\hat{\alpha}}\int\limits_{\Real^d} \frac{\partial f}{\partial p_{\hat{b}}}  
p_{\hat{\alpha}} p_{\hat{\alpha}_2} \ldots p_{\hat{\alpha}_s}
\frac{d^d p}{(-p_{\hat{0}})}\\
&&\qquad\qquad\qquad = qF_{\hat{b}}{}^{\hat{\alpha}}\int\limits_{\Real^d} f
\left[ \frac{\partial}{\partial p_{\hat{b}}} \left( \frac{p_{\hat{\alpha}}}{p_{\hat{0}}} \right)
p_{\hat{\alpha}_2} \ldots p_{\hat{\alpha}_s}
 + \frac{p_{\hat{\alpha}}}{p_{\hat{0}}}  \frac{\partial p_{\hat{\alpha}_2}}{\partial p_{\hat{b}}}
 p_{\hat{\alpha}_3}\ldots p_{\hat{\alpha}_s} + \cdots +
 \frac{p_{\hat{\alpha}}}{p_{\hat{0}}} 
p_{\hat{\alpha}_2}\ldots p_{\hat{\alpha}_{s-1}} 
 \frac{\partial p_{\hat{\alpha}_s}}{\partial p_{\hat{b}}}  \right] d^d p.
\end{eqnarray*}
Owing to the fact that $F_{\hat{\alpha}\hat{\beta}}$ is antisymmetric and using Eq.~(\ref{Eq:p0}) it is easy to verify that the first term on the right-hand side vanishes. Finally, the statement follows after applying the identity
\begin{equation}
F_{\hat{b}}{}^{\hat{\alpha}} p_{\hat{\alpha}}
\frac{\partial p_{\hat{\alpha}_j}}{\partial p_{\hat{b}}}
 = F_{\hat{\alpha}_j\hat{\beta}} p^{\hat{\beta}},
\qquad j = 2,3,\ldots,s,
\end{equation}
to the other terms.
\qed
\end{proof}

A collisionless charged gas propagating in a curved spacetime $(M,g)$ with an exterior electromagnetic field $F$ is described by a one-particle distribution function $f: \Gamma_m^+\to \Real$ which satisfies the Liouville (or Vlasov) equation
\begin{equation}
\boxed{ L_F[f] =  p^\mu\frac{\partial f}{\partial x^\mu}
 + p_{\hat{\alpha}}\left( p^\mu\hat{\Gamma}^{\hat{\alpha}}{}_{\mu\hat{b}}
 - q F^{\hat{\alpha}}{}_{\hat{b}} \right) \frac{\partial f}{\partial p_{\hat{b}}} = 0,\qquad
p^\mu = p^{\hat{\beta}} e_{\hat{\beta}}^\mu.}
\label{Eq:VlasovEq}
\end{equation}
As a consequence of the previous theorem, any solution $f$ of Eq.~(\ref{Eq:VlasovEq}) has a particle current density $J_\mu$ and energy-momentum-stress tensor $T_{\mu\nu}$ satisfying
\begin{equation}
\boxed{ \nabla^\mu J_\mu = 0,\qquad
\nabla^\mu T_{\mu\nu} = -q  J_\mu F^\mu{}_\nu. }
\label{Eq:JTConservation}
\end{equation}
Thus, while the particle current density is divergence-free, in the charged case, the kinetic energy-momentum-stress tensor is not necessarily divergence-free. As discussed in the next subsection, this is due to the fact that the gas may interchange energy-momentum with the electromagnetic field.

\subsection{The Vlasov-Maxwell system of equations}

After these remarks, it is a simple task to generalize the description to a system of charged gas consisting of several species $(a)$ of particles propagating on the curved spacetime $(M,g)$ and in the electromagnetic field $F$ they generate. Each species $(a)$ is described by its one-particle distribution function $f^{(a)}: \Gamma_{m_a}^+ \to \Real$ defined on its future mass shell $\Gamma_{m_a}^+$ of mass $m_a$ and satisfies its own Vlasov equation
\begin{equation}
\boxed{ L_F^{(a)}[f^{(a)}] = p^\mu\frac{\partial f^{(a)}}{\partial x^\mu}
 + p_{\hat{\alpha}}\left( p^\mu\hat{\Gamma}^{\hat{\alpha}}{}_{\mu\hat{b}}
 - q_a F^{\hat{\alpha}}{}_{\hat{b}} \right) \frac{\partial f^{(a)}}{\partial p_{\hat{b}}} = 0,}
\label{Eq:VlasovEqSpecies}
\end{equation}
where $q_a$ denotes the charge of the specie $(a)$. The associated particle current covector $J_\mu^{(a)}$ and energy-momentum-stress tensor $T_{\mu\nu}^{(a)}$ satisfy
\begin{equation}
\nabla^\mu J_\mu^{(a)} = 0,\qquad
\nabla^\mu T_{\mu\nu}^{(a)} = -q_a  J_\mu^{(a)} F^\mu{}_\nu.
\label{Eq:DivTCharged}
\end{equation}
Furthermore, the electromagnetic field tensor $F$ should satisfy Maxwell's field equations
\begin{equation}
dF = 0,\qquad
\divrg F = -j,
\label{Eq:Maxwell}
\end{equation}
or, in terms of local coordinates,
\begin{equation}
\boxed{ \nabla_{[\mu} F_{\alpha\beta]} = 0,\qquad
\nabla^\mu F_{\mu\nu} = -j_\nu,}
\label{Eq:MaxwellCoords}
\end{equation}
where $j$ refers to the electric current covector field generated by the charged gas, given by
\begin{equation}
\boxed{ j_\mu = \sum\limits_a q_a J_\mu^{(a)}.}
\label{Eq:ElectricCurrent}
\end{equation}

The energy-momentum-stress tensor associated with the electromagnetic field is given by
\begin{equation}
T_{\mu\nu}^{(em)} 
 = F_\mu{}^\alpha F_{\nu\alpha} - \frac{1}{4} g_{\mu\nu} F^{\alpha\beta} F_{\alpha\beta},
\end{equation}
and Maxwell's equations~(\ref{Eq:Maxwell}) imply that
\begin{equation}
\nabla^\mu T_{\mu\nu}^{(em)} = j_\mu F^\mu{}_\nu.
\label{Eq:DivTem}
\end{equation}
As a consequence of Eqs.~(\ref{Eq:DivTCharged},\ref{Eq:ElectricCurrent},\ref{Eq:DivTem}) it follows that the total energy-momentum-stress tensor
\begin{equation}
T_{\mu\nu}^{(tot)} := \sum\limits_a T_{\mu\nu}^{(a)} + T_{\mu\nu}^{(em)}
\end{equation}
is divergence-free (leading to a conservation law for each Killing vector field). The relativistic Vlasov-Maxwell system of equations~(\ref{Eq:VlasovEqSpecies},\ref{Eq:MaxwellCoords},\ref{Eq:ElectricCurrent}) describes a whole range of interesting physical scenarios. For stimulating theoretical work regarding the justification that such a system approximately describes a relativistic system of a large number $N$ charged point particles, see~\cite{mKeY19}. For rigorous results on nonlinear Landau damping, see~\cite{cMcV11} which treats the nonrelativistic limit and~\cite{bY15,bY16} for a treatment of the ``relativistic Vlasov-Poisson" system which is a special case of the relativistic Vlasov-Maxwell system in which the magnetic field is identical zero in a given inertial frame. Another interesting limit of the Vlasov-Maxwell system with applications to the physics of pulsars and active black holes is the force-free approximation, in which the electromagnetic field energy-momentum is assumed to dominate that of the plasma (see~\cite{sGtJ14} and references therein). This leads to the force-free condition $F_{\mu\nu} j^\nu = 0$ which yields a nonlinear system for the electromagnetic field (see, e.g.~\cite{fCoR16} and references therein for recent results on hyperbolic formulations of these systems). If the self-gravity of the gas is important one should couple the relativistic Vlasov-Maxwell system of equations~(\ref{Eq:VlasovEqSpecies},\ref{Eq:Maxwell},\ref{Eq:ElectricCurrent}) to Einstein's field equations~(\ref{Eq:EinsteinEq}) with $T_{\mu\nu}$ replaced by the total energy-momentum-stress tensor $T_{\mu\nu}^{(tot)}$. For recent mathematical and numerical work on the Einstein-Maxwell-Vlasov system, see for instance Refs.~\cite{pNnNaR04,pN05,nNmT09,hAmEgR09,mT20,hBdF20}.

\section{The collision term and the relativistic Boltzmann equation}
\label{Sec:Collision}

So far, we have only considered the situation in which collisions between gas particles can be neglected. In this section, we consider a gas configuration in which the individual gas particles are subject to collisions between themselves. For simplicity, we only discuss the case of a simple gas consisting of identical, spinless, massive (charged or neutral) particles of mass $m > 0$. Moreover, we assume that the collisions are described by interactions which are short-ranged, such that each collision is regarded as pointlike, taking place at a fixed event $x\in M$ in spacetime. Furthermore, we assume that the gas is sufficiently dilute such that only binary collisions, of the form
\begin{equation}
p_1 + p_2 \mapsto  p_1^* + p_2^*,
\label{Eq:BinaryCollision}
\end{equation}
are relevant. Here and in the following, $p_1,p_2\in P_x^+(m)$ refer to the physical momenta of the incoming particles and $p_1^*,p_2^*\in P_x^+(m)$ to those of the outgoing ones. Next, we assume that each such collision is elastic, that is, it preserves the total energy-momentum, such that
\begin{equation}
p_1 + p_2 = p_1^* + p_2^*.
\label{Eq:ElasticCollision}
\end{equation}
Finally, we assume that the gas is dilute enough such that the important \emph{molecular chaos hypothesis} is satisfied, which assumes that just before collisions, the particles are uncorrelated. This supposition allows one to specify a probabilistic description for the collisions and to derive a closed equation for the one-particle distribution function. As we will see, this leads to the relativistic Boltzmann equation, which is of the form
\begin{equation}
L_F[f] = C_W[f,f],
\label{Eq:BoltzmannEq}
\end{equation}
where $L_F$ is the Liouville vector field (see Eq.~(\ref{Eq:LiouvilleVFCharged}) or Eq.~(\ref{Eq:LiouvilleVF}) for the uncharged case) and where $C_W[f,f]$ is the collision term, which depends quadratically on the one-particle distribution function $f$ and the transition probability density $W_x(p_1 + p_2\mapsto p_1^* + p_2^*)$ for the binary collision.

The goal of this section is to provide a detailed derivation of the collision term $C_W[f,f]$. This will be done in various steps. In a first step, we review some preliminary results regarding binary elastic collisions. In a next step, we introduce the collision manifold $C_x$ and the associated bundle $T^*C$,  and we discuss different ways of parametrizing it. Next, we compute the natural metric and volume form on $C_x$ which are induced from the metric $g_x^{-1}$ on the cotangent space $T^*_x M$. Finally, we introduce the transition probability map $W: T^* C\to \Real$, $(x,p_1,p_2,p_1^*,p_2^*) \mapsto W_x(p_1 + p_2\mapsto p_1^* + p_2^*)$ and derive the integral form of the relativistic Boltzmann equation, from which Eq.~(\ref{Eq:BoltzmannEq}) is shown to follow. For alternative derivations of the collision term in the relativistic case, see for instance Refs.~\cite{wI63,CercignaniKremer-Book,Groot-Book,rS11}.

Before we initiate our derivation, we point out that we work under the assumption that the gas particles can be treated as classical point particles of mass $m$, such that quantum effects can be neglected. This assumption is justified if the temperature $T$ of the gas is high enough such that the particle wave packets are strongly localized compared to the mean inter-particle distance $l_{\text{mpd}} := n^{-1/3}$ with $n$ the particle density. In other words, we assume that the thermal wavelength,
\begin{equation}
\lambda_T := \frac{h}{\sqrt{2\pi m k_B T}}
\end{equation}
is much smaller than $l_{\text{mpd}}$, i.e.  $\lambda_T \ll l_{\text{mpd}}$. For more details regarding the significance of this assumption see Refs.~\cite{Huang-Book} and~\cite{Schwabl-Book}.

\subsection{Kinematics of collisions}

This subsection discusses some basic facts regarding the kinematics of binary collisions; it is mostly based on Ref.~\cite{wI63}, section 3, and on Ref.~\cite{CercignaniKremer-Book}. Throughout this section, we fix an event $x\in M$ and consider a binary collision~(\ref{Eq:BinaryCollision}) taking place at $x$. It will often be useful to introduce an orthonormal basis vectors $\left\lbrace e_{\hat{\alpha}} \right\rbrace$ and the associated dual basis of covectors $\left\lbrace \theta^{\hat{\alpha}} \right\rbrace$ at $x$, such that $g_x = \eta_{\hat{\alpha}\hat{\beta}}\theta^{\hat{\alpha}}\theta^{\hat{\beta}}$. As before, we assume that $e_{\hat{0}}$ is future-directed timelike. Any covector $p\in T_x^* M$ can then be expanded as $p = p_{\hat{\alpha}} \theta^{\hat{\alpha}}$ with $p_{\hat{\alpha}} = p(e_{\hat{\alpha}})$, and we recall that $p\in P_x^+(m)$ if and only if $p_{\hat{0}} = -\sqrt{m^2 + \delta^{\hat{a}\hat{b}} p_{\hat{a}} p_{\hat{b}} }$ with $\hat{a},\hat{b}=1,2,\ldots, d$ (see Eq.~(\ref{Eq:p0})), such that $p\in P_x^+(m)$ can be parametrized by $\ve{p} = (p_{\hat{1}},p_{\hat{2}},\ldots,p_{\hat{d}})$. For the following, it will also be useful to introduce the shorthand notation
\begin{equation}
p_1\cdot p_2 := g_x^{-1}(p_1,p_2) = \eta^{\hat{\alpha}\hat{\beta}} p_{\hat{\alpha}} p_{\hat{\beta}},
\end{equation}
for the inner product between two covectors $p_1,p_2\in T_x^* M$. Further, for a spacelike or null covector $q\in T_x^* M$ we define $|q|:=\sqrt{q\cdot q}$.

After these notational remarks, we go back to the binary collision~(\ref{Eq:BinaryCollision}) and introduce the relative velocities before and after the collision, defined as
\begin{equation}
\boxed{ q := \frac{1}{m}(p_2 - p_1),\qquad
q^* := \frac{1}{m}(p_2^* - p_1^*).}
\label{Eq:qq*Def}
\end{equation}

\begin{lemma}
The covectors $q,q^*\in T_x^* M$ defined in Eq.~(\ref{Eq:qq*Def}) satisfy
\begin{enumerate}
\item[(a)] $q$ and $q^*$ are spacelike provided that $p_1\neq p_2$.
\item[(b)] $|q| = |q^*|$.
\end{enumerate}
\end{lemma}

\begin{proof}
\begin{enumerate}
\item[(a)] Introducing the velocities $u_1 := p_1/m$, $u_2 := p_2/m$ we find
$$
q\cdot q = (u_2 - u_1)\cdot (u_2 - u_1) = -2 - 2u_1\cdot u_2.
$$
Next, we orient the orthonormal frame $\theta^{\hat{\alpha}}$ such that $u_1 = -\theta^{\hat{0}}$ and $u_2 = -\gamma(\theta^{\hat{0}} - v_{\hat{a}}\theta^{\hat{a}})$, $\gamma = (1 - |\ve{v}|^2)^{-1/2} > 1$. Then,
$$
q\cdot q = 2(\gamma - 1) > 0,
$$
which proves that $q$ is spacelike. The proof for $q^*$ is similar.
\item[(b)] From (a) we have $q\cdot q = -2 - 2u_1\cdot u_2$ and $q^*\cdot q^* = -2 - 2u_1^*\cdot u_2^*$. Comparing this with the square norm of Eq.~(\ref{Eq:ElasticCollision}), which implies $-2 + 2u_1\cdot u_2 = -2 + 2u_1^*\cdot u_2^*$, we conclude that $|q| = |q^*|$.
\end{enumerate}
\qed
\end{proof}

Next, we introduce the ``center of mass'' momentum of the collision (see Fig.~\ref{Fig:CoordCM})
\begin{equation}
\boxed{ p^{cm} := \frac{p_1 + p_2}{\sqrt{4 + |q|^2}}
 = (p^{cm})^* := \frac{p_1^* + p_2^*}{\sqrt{4 + |q^*|^2}}.}
\label{Eq:pCM}
\end{equation}

\begin{figure}[ht]
\centerline{\includegraphics[width=3cm]{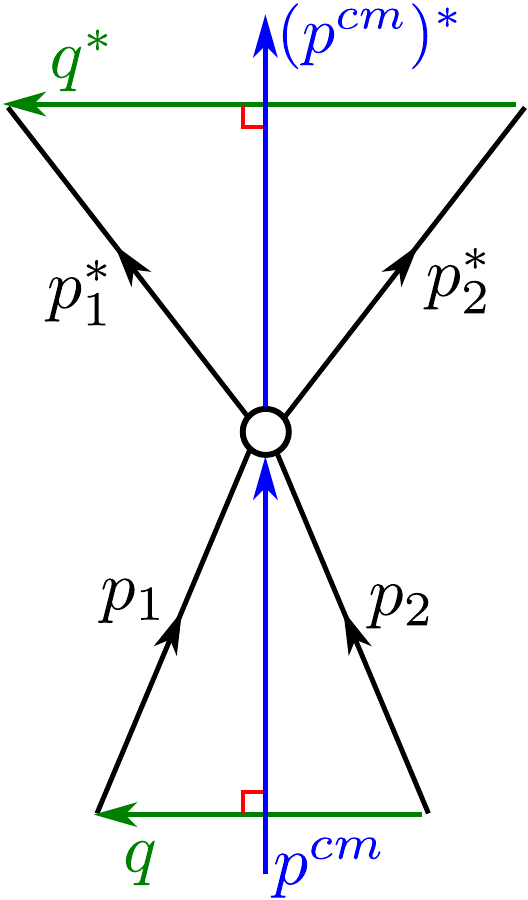}}
\caption{Illustration of the binary elastic collision, its center of mass momentum $p^{cm}$ and the relative velocities $q$ and $q^*$ before and after the collision.}
\label{Fig:CoordCM}
\end{figure}

\begin{lemma}
We have the following properties:
\item[(c)] $p^{cm}\in P_x^+(m)$, that is, $p^{cm}\cdot p^{cm} = -m^2$ and $p^{cm}$ is future-directed.
\item[(d)] $q$ and $q^*$ are orthogonal to $p^{cm}$, that is, $q\cdot p^{cm} = q^*\cdot p^{cm} = 0$.
\end{lemma}

\begin{proof}
\begin{enumerate}
\item[(c)] A straightforward computation yields
$$
p^{cm}\cdot p^{cm} = \frac{m^2}{4 + |q|^2} (u_1 + u_2)\cdot (u_1 + u_2)
 = \frac{m^2}{4 + |q|^2} \left( -2 + 2u_1\cdot u_2 \right)
 = \frac{m^2}{4 + |q|^2} \left( -4 - |q|^2 \right) = -m^2.
$$
Furthermore, $p^{cm}$ is future-directed because it is the sum of two future-directed timelike vector fields.
\item[(d)] Next,
$$
q\cdot p^{cm} = \frac{m}{\sqrt{4 + |q|^2}} (u_2 - u_1)\cdot (u_2 + u_1)
 = \frac{m}{\sqrt{4 + |q|^2}}\left( u_2\cdot u_2 - u_1\cdot u_1 \right) = 0,
$$
and similarly we conclude that $q^*\cdot p^{cm} = q^*\cdot (p^{cm})^* = 0$.
\end{enumerate}
\qed
\end{proof}

Therefore, any binary elastic collision between identical massive particles as in Eq.~(\ref{Eq:ElasticCollision}) gives rise to a momentum $p^{cm}\in P_x^+(m)$ and to two spacelike covectors $q,q^*$ orthogonal to $p^{cm}$ of the same magnitude: $|q| = |q^*|$. Conversely, given $p^{cm}\in P_x^+(m)$ and $q,q^*\in T_x^* M$ such that $q,q^*\perp p^{cm}$ and $|q| = |q^*|$ it follows that
\begin{eqnarray}
p_1 &:=& \sqrt{1 + \frac{1}{4}|q|^2}p^{cm} - \frac{m}{2} q,
\label{Eq:Cp1}\\
p_2 &:=& \sqrt{1 + \frac{1}{4}|q|^2}p^{cm} + \frac{m}{2} q,
\label{Eq:Cp2}\\
p_1^* &:=& \sqrt{1 + \frac{1}{4}|q^*|^2}p^{cm} - \frac{m}{2} q^*,
\label{Eq:Cp1*}\\
p_2^* &:=& \sqrt{1 + \frac{1}{4}|q^*|^2}p^{cm} + \frac{m}{2} q^*,
\label{Eq:Cp2*}
\end{eqnarray}
define covectors $p_1,p_2,p_1^*,p_2^*\in P_x^+(m)$ satisfying Eq.~(\ref{Eq:ElasticCollision}), as one can easily verify. Consequently, any binary elastic collision can be characterized by $p^{cm}\in P_x^+(m)$ and $q,q^*\in T_x^* M$ such that $q,q^*\perp p^{cm}$ and $|q| = |q^*|$.

For later use it is also convenient to introduce the scattering angle $\Theta$, defined by
\begin{equation}
\boxed{ \cos\Theta := \frac{q\cdot q^*}{|q|^2},\qquad
0\leq \Theta\leq \pi,}
\label{Eq:ScatteringAngle}
\end{equation}
and the Mandelstam variables $s$, $t$ and $u$ which are the following collision-invariant quantities:\footnote{The variables $s$ and $t$ were introduced by Stanley Mandelstam in 1958, see Ref.~\cite{sM1958}.} 
\begin{eqnarray}
s &:=& -|p_1 + p_2|^2 = -|p_1^* + p_2^*|^2,
\label{Eq:Mandelstams}\\
t &:=& -|p_1 - p_1^*|^2 = -|p_2 - p_2^*|^2,
\label{Eq:Mandelstamt}\\
u &:=& -|p_1 - p_2^*|^2 = -|p_2 - p_1^*|^2.
\label{Eq:Mandelstamu}
\end{eqnarray}
Using Eqs.~(\ref{Eq:Cp1},\ref{Eq:Cp2},\ref{Eq:Cp1*},\ref{Eq:Cp2*}) we find the following expression in terms of the variables $|q|$ and $\Theta$:
\begin{eqnarray}
s &=& (4 + |q|^2)m^2,
\label{Eq:s}\\
t &=& -\frac{m^2|q|^2}{2}(1 - \cos\Theta),
\label{Eq:t}\\
u &=& -\frac{m^2|q|^2}{2}(1 + \cos\Theta),
\label{Eq:u}
\end{eqnarray}
which implies the relation $s + t + u = 4m^2$. Therefore, only two of the Mandelstam variables are independent.

\subsection{The collision manifold $C_x$}

After discussing the kinematics of binary elastic collisions, in this subsection we introduce the collision manifold $C_x$ which represents the space of all such collisions at $x\in M$. The set $C_x$ is defined as follows:
\begin{equation}
C_x := \{ (p_1,p_2,p_1^*,p_2^*) \in P_x^+(m)\times P_x^+(m)\times P_x^+(m)\times P_x^+(m)
: p_1 + p_2 = p_1^* + p_2^*, p_1\neq p_2 \}.
\label{Eq:Cx}
\end{equation}
This definition essentially coincides with the definition given in Ref.~\cite{jE71} with the exception that here we restrict the incoming momenta $p_1$ and $p_2$ to be different from each other. Physically, this condition makes sense of course, since when $p_1 = p_2$ the two incoming geodesics coincide and thus cannot describe a collision between two different particles. Mathematically, the condition $p_1\neq p_2$ is necessary for $C_x$ to describe a smooth manifold; removing it would result in a set with conical-type singularities at $p_1 = p_2$ (see Appendix~\ref{App:Collisions}).

The goal of this subsection is to find a suitable way to parametrize the set $C_x$ and to prove that it is a manifold. According to the results from the previous subsection, any binary elastic collision is characterized by specifying a covector $p^{cm}\in P_x^+(m)$ and two spatial covectors $q,q^*\in T_x^* M$ of equal magnitude which are both orthogonal to $p^{cm}$, see Eqs.~(\ref{Eq:Cp1},\ref{Eq:Cp2},\ref{Eq:Cp1*},\ref{Eq:Cp2*}). To provide an explicit characterization of the covectors $q$ and $q^*$ one can proceed as follows. Let us assume first that the orthonormal basis $\left\lbrace \theta^{\hat{\alpha}} \right\rbrace$ is oriented along the center of mass frame, such that $p^{cm} = -\theta^{\hat{0}}$. In this case, we obtain from Eqs.~(\ref{Eq:Cp1},\ref{Eq:Cp2})
\begin{eqnarray}
p_1 &=& -\sqrt{1 + \frac{1}{4} g^2} m \theta^{\hat{0}}  
 - \frac{g}{2} m\hat{q}_{\hat{a}} \theta^{\hat{a}},
\label{Eq:Cp1Bis}\\
p_2 &=& -\sqrt{1 + \frac{1}{4} g^2} m \theta^{\hat{0}} 
 + \frac{g}{2} m\hat{q}_{\hat{a}}\theta^{\hat{a}},
\label{Eq:Cp2Bis}
\end{eqnarray}
where we have set $g := |q|$ and where we have expanded $q = g\hat{q}_{\hat{a}}\theta^{\hat{a}}$ with $\ve{\hat q} = (\hat{q}_1,\hat{q}_2,\ldots,\hat{q}_d)\in S^{d-1}$. Similarly,
\begin{eqnarray}
p_1^* &=& -\sqrt{1 + \frac{1}{4} g^2} m \theta^{\hat{0}} 
 - \frac{g}{2} m\hat{q}^*_{\hat{a}} \theta^{\hat{a}},
\label{Eq:Cp1*Bis}\\
p_2^* &=& -\sqrt{1 + \frac{1}{4} g^2} m \theta^{\hat{0}} 
 +  \frac{g}{2} m\hat{q}^*_{\hat{a}} \theta^{\hat{a}},
\label{Eq:Cp2*Bis}
\end{eqnarray}
where $q^* = g(\hat{q}^*)_{\hat{a}} \theta^{\hat{a}}$ with $\ve{\hat q}^* = (\hat{q}^*_1,\hat{q}^*_2,\ldots,\hat{q}^*_d)\in S^{d-1}$. Note that the cosine of the scattering angle is determined by
\begin{equation}
\cos\Theta = \ve{\hat q}\cdot \ve{\hat q}^*,\qquad
0\leq \Theta\leq \pi.
\label{Eq:ScatteringAngleBis}
\end{equation}

The center of mass frame can be used if only one collision needs to be described (or more generally, if several collisions with the same center of mass momentum are described). However, when dealing with arbitrary families of binary collisions one cannot assume a fixed center of mass frame in which $p^{cm}$ is oriented along $-\theta^{\hat{0}}$. In this general situation one needs to consider a Lorentz transformation $\Lambda: T_x^* M\to T_x^* M$ which maps $-\theta^{\hat{0}}$ to $p^{cm}/m$. In other words, $\Lambda$ is a Lorentz transformation which maps the given, fixed orthonormal frame $\left\lbrace \theta^{\hat{\alpha}} \right\rbrace$ to a new orthonormal frame $\left\lbrace \omega^{\hat{\alpha}} \right\rbrace$ such that $-\omega^{\hat{0}} = p^{cm}/m$, i.e. a center of mass frame. A particular transformation of this type is given by the Lorentz boost $\Lambda$ defined by
\begin{eqnarray}
\omega^{\hat{0}} := \Lambda(\theta^{\hat{0}}) 
 &=& \gamma(\theta^{\hat{0}} - \beta_{\hat{a}} \theta^{\hat{a}}),
\label{Eq:LorentzBoost1}\\
\omega^{\hat{b}} := \Lambda(\theta^{\hat{b}}) 
 &=& -\gamma\beta^{\hat{b}} \theta^{\hat{0}} + \left[ \delta_{\hat{a}}{}^{\hat{b}} 
 + (\gamma-1)\hat{\beta}^{\hat{b}}\hat{\beta}_{\hat{a}} \right] \theta^{\hat{a}},\qquad
\hat{\beta}^{\hat{a}} = \beta^{\hat{a}}/\sqrt{\delta_{\hat{c}\hat{d}}\beta^{\hat{c}}\beta^{\hat{d}}},
\label{Eq:LorentzBoost2}
\end{eqnarray}
with $m\gamma = -p^{cm}_{\hat{0}}$ and $m\gamma\beta_{\hat{a}} = p^{cm}_{\hat{a}}$. Under this boost, an arbitrary covector $X = X_{\hat{0}} \theta^{\hat{0}} + X_{\hat{a}} \theta^{\hat{a}} \in T_x^* M$ is mapped to
\begin{equation}
\Lambda(X) = \frac{1}{m}
 \left[ \sqrt{m^2 + |\ve{p}^{cm}|^2} X_{\hat{0}} - \ve{p}^{cm}\cdot\ve{X} \right] \theta^{\hat{0}}
 + \frac{1}{m}
 \left[ -p^{cm}_{\hat{a}} X_{\hat{0}} 
 + \left( \sqrt{m^2 + |\ve{p}^{cm}|^2} - m \right)
  (\ve{\hat{p}}^{cm}\cdot\ve{X})\hat{p}^{cm}_{\hat{a}}  + m X_{\hat{a}} \right] \theta^{\hat{a}},
\label{Eq:LorentzBoost}
\end{equation}
where we have introduced the notation $\ve{p}^{cm} := (p^{cm}_{\hat{1}}, p^{cm}_{\hat{2}}, \ldots, p^{cm}_{\hat{d}})$, $\ve{\hat{p}}^{cm} := \ve{p}^{cm}/|\ve{p}^{cm}|$, and $\ve{X} := (X_{\hat{1}}, X_{\hat{2}}, \ldots, X_{\hat{d}})$, and used the fact that $p^{cm} \in P_x^+(m)$ to conclude that $p^{cm}_{\hat{0}} = -\sqrt{m^2 + |\ve{p}^{cm}|^2}$. The expression for $p_{\hat{1}}$ in terms of the orthonormal frame $\theta^{\hat{0}}, \theta^{\hat{1}}, \ldots, \theta^{\hat{d}}$ can now be obtained from Eq.~(\ref{Eq:Cp1Bis}) by replacing $\theta^{\hat{0}}$ with $\omega^{\hat{0}} = \Lambda(\theta^{\hat{0}})$ and $\theta^{\hat{a}}$ with $\omega^{\hat{a}} = \Lambda(\theta^{\hat{a}})$, $\hat{a}=\hat{1}, \hat{2}, \ldots, \hat{d}$, and using the explicit expression~(\ref{Eq:LorentzBoost}) for $\Lambda$. The result is
\begin{eqnarray}
p_1 &=& -\left[ \sqrt{m^2 + |\ve{p}^{cm}|^2}\sqrt{1 + \frac{1}{4} g^2}
 - \frac{g}{2}\ve{p}^{cm}\cdot\ve{\hat q} \right] \theta^{\hat{0}}
\nonumber\\
 && +\,  \left[ \sqrt{1 + \frac{1}{4} g^2}p^{cm}_{\hat{a}} - \frac{g}{2} m\hat{q}_{\hat{a}} 
 - \frac{g}{2}\left( \sqrt{m^2 + |\ve{p}^{cm}|^2} - m \right)
 (\ve{\hat{p}}^{cm}\cdot\ve{\hat q}) \hat{p}^{cm}_{\hat{a}} \right] \theta^{\hat{a}}.
\label{Eq:Cp1BisBis}
\end{eqnarray}
The expression for $p_2$ can be obtained from this after the change $g\mapsto -g$, and the expressions for $p_1^*$ and $p_2^*$ by changing $\ve{\hat q}\mapsto \ve{\hat q}^*$. This leads to the following transformation that maps $(\ve{p}^{cm},g,\ve{\hat q},\ve{\hat q}^*)$ to $(\ve{p}_1,\ve{p}_2,\ve{p}_1^*,\ve{p}_2^*)$:\footnote{Here, we have also used the identity
$$
\sqrt{m^2 + |\ve{p}^{cm}|^2} - m = 
\frac{|\ve{p}^{cm}|^2}{\sqrt{m^2 + |\ve{p}^{cm}|^2} + m},
$$
in order to transform the last term on the right-hand side of Eq.~(\ref{Eq:Cp1BisBis}) to a form that is manifestly regular at $\ve{p}^{cm} = \ve{0}$.}
\begin{eqnarray}
\ve{p}_1 &=& \sqrt{1 + \frac{1}{4} g^2}\ve{p}^{cm} 
 - \frac{g}{2}\left( m\ve{\hat q} 
 + \frac{\ve{p}^{cm}\cdot\ve{\hat q}}{\sqrt{m^2 + |\ve{p}^{cm}|^2} + m}\ve{p}^{cm}
 \right),
\label{Eq:Cp1Vec}\\
\ve{p}_2 &=& \sqrt{1 + \frac{1}{4} g^2}\ve{p}^{cm}
 + \frac{g}{2}\left( m\ve{\hat q} 
 + \frac{\ve{p}^{cm}\cdot\ve{\hat q}}{\sqrt{m^2 + |\ve{p}^{cm}|^2} + m}\ve{p}^{cm}
 \right),
\label{Eq:Cp2Vec}\\
\ve{p}_1^* &=& \sqrt{1 + \frac{1}{4} g^2}\ve{p}^{cm}
 - \frac{g}{2}\left( m\ve{\hat q}^* 
 + \frac{\ve{p}^{cm}\cdot\ve{\hat q}^*}{\sqrt{m^2 + |\ve{p}^{cm}|^2} + m}\ve{p}^{cm}
 \right),
\label{Eq:Cp1*Vec}\\
\ve{p}_2^* &=& \sqrt{1 + \frac{1}{4} g^2}\ve{p}^{cm}
 + \frac{g}{2}\left( m\ve{\hat q}^* 
 + \frac{\ve{p}^{cm}\cdot\ve{\hat q}^*}{\sqrt{m^2 + |\ve{p}^{cm}|^2} + m}\ve{p}^{cm}
 \right).
\label{Eq:Cp2*Vec}
\end{eqnarray}
In terms of the quantities (note that $\ell_+ = \ell_+^*$)
\begin{equation}
\ell_\pm := \sqrt{m^2 + |\ve{p}_2|^2} \pm \sqrt{m^2 + |\ve{p}_1|^2},\qquad
\ell_\pm^* := \sqrt{m^2 + |\ve{p}_2^*|^2} \pm \sqrt{m^2 + |\ve{p}_1^*|^2},
\end{equation}
the inverse map is given by
\begin{eqnarray}
\ve{p}^{cm} &=& m\frac{\ve{p}_1 + \ve{p}_2}{\sqrt{\ell_+^2 - |\ve{p}_1 + \ve{p}_2|^2}},
\label{Eq:pcm}\\
g &=& \frac{1}{m}\sqrt{|\ve{p}_2 - \ve{p}_1|^2 - \ell_-^2},
\label{Eq:g}\\
\ve{\hat q} &=& \frac{1}{m g}\left[ \ve{p}_2 - \ve{p}_1
  - \frac{\ell_-(\ve{p}_1 + \ve{p}_2)}{\ell_+ + \sqrt{\ell_+^2 - |\ve{p}_1 + \ve{p}_2|^2}} \right],
\\
\ve{\hat q}^* &=& \frac{1}{m g}\left[ \ve{p}_2^* - \ve{p}_1^*
  - \frac{\ell_-^*(\ve{p}_1 + \ve{p}_2)}{\ell_+ + \sqrt{\ell_+^2 - |\ve{p}_1 + \ve{p}_2|^2}} \right].
\end{eqnarray}
By Lorentz-invariance, the scattering angle is still given by the formula~(\ref{Eq:ScatteringAngleBis}).

As a consequence of the above remarks, we see that we can parametrize each point of $C_x$ uniquely by the specification of $\ve{p}^{cm}\in \Real^d$, $g > 0$ and $\ve{\hat q},\ve{\hat q}^*\in S^{d-1}$. Furthermore, the map $(\ve{p}^{cm},g,\ve{\hat q},\ve{\hat q}^*)\mapsto (\ve{p}_1,\ve{p}_2,\ve{p}_1^*,\ve{p}_2^*)$ defined by Eqs.~(\ref{Eq:Cp1Vec}--\ref{Eq:Cp2*Vec}) is smooth. It follows from these considerations:

\begin{lemma}
\label{Lem:Cx}
The collision manifold $C_x$ defined in Eq.~(\ref{Eq:Cx}) is a $(3d-1)$-dimensional submanifold of $[P_x^+(m)]^4\equiv P_x^+(m)\times P_x^+(m)\times P_x^+(m)\times P_x^+(m)$ with topology $\Real^{d+1}\times S^{d-1}\times S^{d-1}$.
\end{lemma}

By choosing local coordinates $({\hat q}^A,({\hat q}^*)^B)$ on the spheres $S^{d-1}$ one obtains corresponding local coordinates $(\ve{p}^{cm},g,{\hat q}^A,({\hat q}^*)^B )$ of $C_x$.

By considering for each $x\in M$ the spaces $C_x$ as a fibre over $x$, we obtain the associated bundle
\begin{equation}
T^* C := \{ (x,p_1,p_2,p_1^*,p_2^*) : x\in M, (p_1,p_2,p_1^*,p_2^*)\in C_x \}.
\label{Eq:TCDef}
\end{equation}

We may construct local coordinates on $T^* C$ as follows: let $(x^\mu)$ be local coordinates defined in a neighborhood $U$ of $M$, and let $\left\lbrace e_{\hat{\alpha}} \right\rbrace$ be an orthonormal frame of vectors on $U$ with the associated dual basis of covectors $\left\lbrace \theta^{\hat{\alpha}} \right\rbrace$, such that $e_{\hat{0}}$ is future-directed timelike. Then we associate to each point $(x,p_1,p_2,p_1^*,p_2^*)\in T^* C$ with $x\in U$ the coordinates
\begin{equation}
(x^\mu,\ve{p}^{cm},g,{\hat q}^A,({\hat q}^*)^B ),
\end{equation}
where $(\ve{p}^{cm},g,{\hat q}^A,({\hat q}^*)^B )$ are the local coordinates of $C_x$ introduced above. This yields local coordinates on $T^*C$. By taking a differentiable atlas of $M$, one obtains a corresponding differential atlas of $T^* C$, and it follows that $T^*C$ is a $4d$-dimensional manifold.

To conclude this subsection, we give the following theorem which will play a key role later, when discussing equilibrium distribution functions.

\begin{theorem}
\label{Thm:CollisionInvariants}
Let $F: P_x^+(m)\to\Real$ be a continuously differentiable function on the future mass hyperboloid. Then, $F$ is a collision invariant, that is, it satisfies
\begin{equation}
F(p_1) + F(p_2) = F(p_1^*) + F(p_2^*)
\label{Eq:FCI}
\end{equation}
for all $p_1,p_2,p_1^*,p_2^*\in P_x^+(m)$ fulfilling Eq.~(\ref{Eq:ElasticCollision}), if and only if
$F$ is of the form
\begin{equation}
\boxed{ F(p) = \alpha + p(\beta),\qquad p\in P_x^+(m) }
\label{Eq:FForm}
\end{equation}
for some $\alpha\in\Real$ and $\beta\in T_x M$.
\end{theorem}

\proof The ``if" part of the theorem is obvious. To prove the "only if" part we use an argument by Ehlers, see section 4.14 in~\cite{jE71}. For a short article containing alternative proofs, original references and a generalization to distributions which allows one to relax the differentiability assumption on $F$, see Ref.~\cite{cCgK99}.

Therefore, suppose $F\in C^1(P_x^+(m))$ is a collision invariant. Let $N_x := [P_x^+(m)]^4$ and introduce the function $H: N_x\to \Real$ defined by
\begin{equation}
H(r) := F(p_1) + F(p_2) - F(p_1^*) - F(p_2^*),\qquad
r := (p_1,p_2,p_1^*,p_2^*)\in N_x.
\end{equation}
According to the hypothesis, $H(r) = 0$ for all $r = (p_1,p_2,p_1^*,p_2^*)\in C_x$ on the collision manifold. Since $C_x\subset N_x$ is a smooth submanifold characterized by $p_1 + p_2 - p_1^* - p_2^* = 0$ and $p_1\neq p_2$, its normal covectors are given by
\begin{equation}
d(p_1)_{\hat{\alpha}} + d(p_2)_{\hat{\alpha}} - d(p_1^*)_{\hat{\alpha}} - d(p_2^*)_{\hat{\alpha}}, \qquad
\hat{\alpha} = \hat{0},\hat{1},\ldots,\hat{d},
\end{equation}
where as before, the index $\hat{\alpha}$ refers to the components with respect to a fixed orthonormal frame. Since $H$ is constant on $C_x$, it follows that
\begin{equation}
dH_r = \beta^{\hat{\alpha}}(r) \left[ d(p_1)_{\hat{\alpha}} + d(p_2)_{\hat{\alpha}}
  - d(p_1^*)_{\hat{\alpha}} - d(p_2^*)_{\hat{\alpha}} \right]
\end{equation}
for some continuous functions $\beta^{\hat{\alpha}}: C_x\to \Real$, $\hat{\alpha} = \hat{0},\hat{1},\ldots,\hat{d}$, on $C_x$. Given the definition of $H$ and the independence of $p_1,p_2,p_1^*,p_2^*$ on $N_x$, this implies that
\begin{equation}
dF_p = \beta^{\hat{\alpha}}(r) dp_{\hat{\alpha}},
\label{Eq:dFCondition}
\end{equation}
where $r = (p_1,p_2,p_1^*,p_2^*)$ and $p$ can be taken to be any of the four momenta in $r$.

We now claim that the functions $\beta^{\hat{\alpha}}$ are constant. This can be shown as follows: since we can parametrize any $p\in P_x^+(m)$ by its spatial orthonormal components $p_{\hat{a}}$ and $p_{\hat{0}} = -\sqrt{m^2 + \delta^{\hat{a}\hat{b}} p_{\hat{a}} p_{\hat{b}}}$, we can also write Eq.~(\ref{Eq:dFCondition}) in the explicit form
\begin{equation}
\frac{\partial F}{\partial p_{\hat{a}}}(p) 
 = \beta^{\hat{a}}(r) - \beta^{\hat{0}}(r)\frac{p^{\hat{a}}}{p^{\hat{0}}},
\label{Eq:dFConditionBis}
\end{equation}
where again $p = p_1,p_2,p_1^*$ or $p_2^*$. Taking the difference between two pairs, for example
\begin{equation}
\frac{\partial F}{\partial p_{\hat{a}}}(p_1) - \frac{\partial F}{\partial p_{\hat{a}}}(p_2)
 = -\beta^{\hat{0}}(r)
\left( \frac{p_1^{\hat{a}}}{p_1^{\hat{0}}} - \frac{p_2^{\hat{a}}}{p_2^{\hat{0}}} \right),
\end{equation}
shows that $\beta^{\hat{0}}(r)$ can only depend on $(p_1,p_2)$, i.e.
\begin{equation}
\beta^{\hat{0}}(p_1,p_2,p_1^*,p_2^*) 
 = \beta^{\hat{0}}(p_1,p_2,\overline{p}_1^*,\overline{p}_2^*)
\end{equation}
for all $(p_1,p_2,p_1^*,p_2^*), (p_1,p_2,\overline{p}_1^*,\overline{p}_2^*)\in C_x$. Since the same is true for any other pairs (for example $(p_1^*,p_2^*)$), it follows that $\beta^{\hat{0}}$ is constant on $C_x$, since it is possible to arrive from any $r\in C_x$ to any other point $s\in C_x$ by a successive change of pairs. It then follows from Eq.~(\ref{Eq:dFConditionBis}) that $\beta^{\hat{a}}(r) = \beta^{\hat{a}}(s)$ if $r$ and $s$ share one common momentum, and thus also the functions $\beta^{\hat{a}}$, $\hat{a} = 1,2,\ldots,d$, are constant.

Eq.~(\ref{Eq:dFConditionBis}) can now be integrated and yields
\begin{equation}
F(p) = \alpha + \beta^{\hat{\alpha}} p_{\hat{\alpha}}
\end{equation}
for some constant $\alpha\in \Real$, which proves the theorem.
\qed

\subsection{Volume form on $C_x$}

After introducing the collision manifold $C_x$, in this subsection we briefly discuss the volume form on it, leaving the details of its derivation (which are rather technical) to Appendix~\ref{App:CollisionManifold}. A natural volume form is obtained as follows: first, we recall that the inverse metric $g_x^{-1}$ at the point $x\in M$ defines a natural metric on $T_x^* M$, which in turn defines a metric $H_x$ on the future mass hyperboloid $P_x^+(m)$, see Eq.~(\ref{Eq:Induced_metric_gx-1}). Taking the fourfold sum $H_x\oplus H_x\oplus H_x\oplus H_x$ of this metric, one obtains a metric on $N_x = [P_x^+(m)]^4$. Since $C_x$ is a submanifold of $N_x$, this product induces a metric on $C_x$, from which the desired volume form $\eta_{C_x}$ is obtained.

As mentioned previously, the calculations are long, so we perform the explicit derivation in  Appendix~\ref{App:CollisionManifold} and directly give the result:
\begin{equation}
\eta_{C_x} = \sqrt{ \left( 2 + \frac{g^2}{4} \right)^2 - \left( \frac{g^2}{4}\right)^2\cos^2\Theta }
\left( 1 + \frac{g^2}{4} \right)^{\frac{d-3}{2}}
\frac{m^{2d}}{\sqrt{m^2 + |\ve{p}^{cm}|^2}}
dp^{cm}_{\hat{1}}\wedge dp^{cm}_{\hat{2}}\wedge\cdots\wedge dp^{cm}_{\hat{d}}\wedge 
g^{2(d-1)} dg\wedge d\Omega(\hat{q})\wedge d\Omega(\hat{q}^*),
\label{Eq:mux}
\end{equation}
where $d\Omega(\hat{q})$ denotes the solid angle associated with $\hat{q}\in S^d$, and similarly for $d\Omega(\hat{q}^*)$. The following comments and remarks will turn out to be useful:

\begin{enumerate}
\item In the non-relativistic limit $g\ll 1$ and $|\ve{p}^{cm}| \ll m$, the expression for the volume form simplifies considerably:
\begin{equation}
\eta_{C_x} \approx 2 m^{2d-1}
dp^{cm}_{\hat{1}}\wedge dp^{cm}_{\hat{2}}\wedge\cdots\wedge dp^{cm}_{\hat{d}}\wedge 
g^{2(d-1)} dg\wedge d\Omega(\hat{q})\wedge d\Omega(\hat{q}^*).
\end{equation}

\item Instead of parametrizing the collision through the ``center of mass'' variables $(\ve{p}^{cm},g,\ve{\hat q},\ve{\hat q}^*)$ we may also parametrize it in terms of the variables $(\ve{p}_1,\ve{p}_2,\ve{\hat q}^*)$, where the center of mass momentum $\ve{p}^{cm}$ and the magnitude $g$ of the relative speed is determined by the momenta of the incoming particles according to Eqs.~(\ref{Eq:pcm},\ref{Eq:g}). Taking into account the equivalent representations in Eqs.~(\ref{Eq:etaPx+2},\ref{Eq:pix2Bis}) for $\eta_{I_x}$, the volume form $\eta_{C_x}$ can also be written as:
\begin{equation}
\eta_{C_x} = \sqrt{\frac{(s+t)(s+u)}{s}} 
\frac{m^d}{\sqrt{m^2 + |\ve{p}_1|^2}\sqrt{m^2 + |\ve{p}_2|^2}}
d(p_1)_{\hat{1}}\wedge\cdots\wedge d(p_1)_{\hat{d}}\wedge
d(p_2)_{\hat{1}}\wedge\cdots\wedge d(p_2)_{\hat{d}}\wedge g^{d-1} d\Omega(\hat{q}^*),
\label{Eq:muxBis}
\end{equation}
where we have reexpressed the factor involving $g$ and $\Theta$ in terms of the Mandelstam variables $t,s,u$. Taking into account the definition~(\ref{Eq:dvolx}) for the volume element $\dvol_x(p)$ on the mass hyperboloid we can also write this in the more compact form
\begin{equation}
\boxed{
\eta_{C_x} = m^d\sqrt{\frac{(s+t)(s+u)}{s}} \dvol_x(p_1)\wedge \dvol_x(p_2)
\wedge g^{d-1} d\Omega(\hat{q}^*).
}
\label{Eq:muxBisVol}
\end{equation}

\item Of course, there are many different possibilities for parametrizing the collision manifold $C_x$. For a discussion and a comparison between different coordinates on $C_x$, see Ref.~\cite{rS11}. Our representation~(\ref{Eq:muxBisVol}) is related to the ``center of momentum reduction" in~\cite{rS11}.

\item The collision manifold $C_x$ possesses the following symmetries which will play an important role later:
\begin{eqnarray}
j &:& C_x\to C_x, (p_1,p_2,p_1^*,p_2^*)\mapsto (p_2,p_1,p_1^*,p_2^*),\qquad
\hbox{(interchange of incoming particles)},
\label{Eq:DefIsoj}\\
j^* &:& C_x\to C_x, (p_1,p_2,p_1^*,p_2^*)\mapsto (p_1,p_2,p_2^*,p_1^*),\qquad
\hbox{(interchange of outgoing particles)},
\label{Eq:DefIsoj*}\\
k &:& C_x\to C_x, (p_1,p_2,p_1^*,p_2^*)\mapsto (p_1^*,p_2^*,p_1,p_2),\qquad
\hbox{(interchange of incoming and outgoing particles)}.
\quad\label{Eq:DefIsok}
\end{eqnarray}
Clearly, these are isometries of the larger space $([P_x^+(m)]^4, H_x\oplus H_x\oplus H_x\oplus H_x)$, and since the restriction $p_1 + p_2 = p_1^* + p_2^*$ is invariant with respect to $j$, $j^*$ and $k$, it follows that $j$, $j^*$ and $k$ are also isometries of the submanifold $C_x$.
\end{enumerate}

\subsection{Transition probability density}
\label{SubSec:TransitionProbability}

Now that we have computed the volume element on the collision manifold $C_x$ we are ready to give the physical interpretation of the transition probability density $W$.\footnote{Sometimes, this function is also called transition rate in the literature, see for instance Refs.~\cite{Groot-Book,Vereshchagin-Book}.} At the mathematical level, this function is a real-valued, nonnegative smooth function
\begin{equation}
W : T^*C\to \Real,\quad (x,p_1,p_2,p_1^*,p_2^*)\mapsto W_x(p_1+p_2\mapsto p_1^*+p_2^*)
\end{equation}
on the collision bundle $T^*C$ (see Eq.~(\ref{Eq:TCDef})) satisfying the following symmetries for all $(x,p_1,p_2,p_1^*,p_2^*)\in T^*C$:
\begin{equation}
W_x(p_1+p_2\mapsto p_1^*+p_2^*) = W_x(p_2+p_1\mapsto p_1^*+p_2^*)
 = W_x(p_1+p_2\mapsto p_2^*+p_1^*).
\label{Eq:WxIdSym}
\end{equation}
These symmetries mean that for each $x\in M$ the function $W_x : C_x\to \Real$ is invariant with respect to the particle-exchange symmetries $j$ and $j^*$, and they incorporate our assumption of the particles involved in the collision being \emph{identical}. The physical interpretation of the transition probability density $W$ is the following. Let $f: \Gamma_m^+\to \Real$ be the one-particle distribution function and let $U\subset T^*C$ be an open neighborhood of $T^*C$ which is invariant under the particle-exchange symmetries $j$ and $j^*$, that is,
$$
(x,p_1,p_2,p_1^*,p_2^*)\in U \hbox{ implies }
(x,p_2,p_1,p_1^*,p_2^*),(x,p_1,p_2,p_2^*,p_1^*)\in U.
$$
Then, the integral
\begin{equation}
\mathcal{N}_c[U]
  := \frac{1}{4}\int\limits_M\left( \int\limits_{C_x} W_x(p_1+p_2\mapsto p_1^*+p_2^*)
f(x,p_1) f(x,p_2) \chi_U(x,p_1,p_2,p_1^*,p_2^*)\eta_{C_x} \right) \eta_M
\label{Eq:NcDef}
\end{equation}
provides the averaged number of collisions satisfying
\begin{equation}
(x,p_1,p_2,p_1^*,p_2^*)\in U.
\end{equation}
Here, $\chi_U: T^*C\to \Real$ denotes the characteristic function of the subset $U\subset T^*C$, defined as
\begin{equation}
\chi_U(\xi) := \left\{ \begin{array}{rl} 
1, & \xi\in U, \\
0, & \xi\in T^*C\setminus U,
\end{array} \right.
\end{equation}
$\eta_{C_x}$ is the volume element defined in Eq.~(\ref{Eq:muxBisVol}), and $\eta_M$ is the natural volume element on spacetime. The factor $1/4$ on the right-hand side of Eq.~(\ref{Eq:NcDef}) arises because the four collisions $p_1+p_2\mapsto p_1^*+p_2^*$, $p_2+p_1\mapsto p_1^*+p_2^*$, $p_1+p_2\mapsto p_2^*+p_1^*$, and $p_2+p_1\mapsto p_2^*+p_1^*$ are identical.

It is important to mention that the interpretation of Eq.~(\ref{Eq:NcDef}) is based on Boltzmann's \emph{molecular chaos hypothesis}, according to which the gas is dilute enough such that the gas particles are \emph{uncorrelated just before the collision}. Without this hypothesis, one should replace the product $f(x,p_1) f(x,p_2)$ in the right-hand side of Eq.~(\ref{Eq:NcDef}) with the two-particle distribution function $f^{(2)}(x_1,p_1,x_2,p_2)$ at the point $x_1 = x_2 = x$. In the following, we assume the validity of the molecular chaos hypothesis, such that $f^{(2)}(x,p_1,x,p_2) = f(x,p_1) f(x,p_2)$ is valid just prior to the collision.

In addition to the symmetries~(\ref{Eq:WxIdSym}) originating from the requirement that the particles are identical, we shall also require the condition of \emph{microscopic reversibility} which means that, in addition to the symmetries~(\ref{Eq:WxIdSym}), one requires the transition probability density to be invariant with respect to the isometry $k$ of $C_x$:
\begin{equation}
W_x(p_1^*+p_2^*\mapsto p_1+p_2) = W_x(p_1+p_2\mapsto p_1^*+p_2^*)
\label{Eq:WxIdSym2}
\end{equation}
for all $(x,p_1,p_2,p_1^*,p_2^*)\in T^*C$.

We conclude this subsection by showing that the symmetry~(\ref{Eq:WxIdSym2}) is a consequence of a stronger requirement on the transition probability density.\footnote{We are thankful to Thomas Zannias for a discussion of this point.} Namely, let us assume that $W_x(p_1 + p_2\mapsto p_1^* + p_2^*)$ is a Lorentz scalar (in the cotangent space $T_x^* M$ with metric $g_x^{-1}$) depending only on the momenta $(p_1,p_2,p_1^*,p_2^*)$. This is certainly a reasonable requirement as long as the interaction does not involve other fields, like a strong magnetic field, for instance (in which case $W$ could also depend on the Faraday tensor $F$). Under this assumption, it is clear that $W_x(p_1^*+p_2^*\mapsto p_1+p_2)$ can only depend on the six inner products $p_1\cdot p_2$, $p_1\cdot p_1^*$, $p_1\cdot p_2^*$, $p_2\cdot p_1^*$, $p_2\cdot p_2^*$, $p_1^*\cdot p_2^*$ between the momenta involved in the collisions (the remaining four inner products $p_i\cdot p_i = p_i^*\cdot p_i^* = -m^2$ being fixed).\footnote{In four spacetime dimensions one could in principle construct a further scalar by taking the dual of $p_1\wedge p_2\wedge p_1^*\wedge p_2^*$; however it is zero due to energy-momentum conservation.} These six Lorentz-invariants are not all independent due to  momentum conservation~(\ref{Eq:ElasticCollision}), and in fact, they can be represented by the Mandelstam variables $s,t$ and $u$, see Eqs.~(\ref{Eq:Mandelstams},\ref{Eq:Mandelstamt},\ref{Eq:Mandelstamu}). These in turn, can be represented solely in  terms of the magnitude of the relative momentum $g := |q|$ and the scattering angle $\Theta$, see Eqs.~(\ref{Eq:s},\ref{Eq:t},\ref{Eq:u}).

Therefore, we arrive at the conclusion that a transition probability density $W_x$ which is a Lorentz scalar depending only on the momenta $(p_1,p_2,p_1^*,p_2^*)$ is a function of only two variables, namely $g$ and $\Theta$. Since these variables are invariant with respect to  $k$, microscopic reversibility is automatically satisfied.\footnote{An independent way of proving that microscopic reversibility is a consequence of Lorentz invariance is to consider the reflection (cf.~\cite{Huang-Book}, pages 59 and 60)
$$
\Lambda(p) := p - 2(n\cdot p) n,\qquad n := \frac{q - q^*}{|q - q^*|},
$$
with respect to the plane orthogonal to the unit spacelike vector $n$. This reflection leaves $p^{cm}$ invariant and interchanges $q$ and $q^*$; hence it interchanges $(p_1,p_2)$ and $(p_1^*,p_2^*)$, see Eqs.~(\ref{Eq:Cp1},\ref{Eq:Cp2},\ref{Eq:Cp1*},\ref{Eq:Cp2*}).}

On the other hand, note that the symmetries $j$ and $j^*$ which exchange $p_1$ and $p_2$ or $p_1^*$ and $p_2^*$ induce the map $(g,\Theta)\mapsto (g,\pi - \Theta)$. Therefore, in order for the transition probability density $W_x$ to be invariant with respect to $j$ and $j^*$, it should also be invariant with respect to $\Theta\mapsto \pi - \Theta$. In subsection~\ref{SubSec:CrossSection} a relation between $W_x$ and the differential cross section will be provided.

\subsection{General relativistic Boltzmann equation}

Now that we have introduced the transition probability function and given its physical interpretation, we are ready to derive the relativistic Boltzmann equation in integral form, and from this its standard integro-differential form.

In order to derive the Boltzmann equation in integral form, we consider a subset $V\subset \Gamma_m^+$ of the form given in Eq.~(\ref{Eq:VDef}) and recall the interpretation of the quantity $\mathcal{N}[\Sigma_1] - \mathcal{N}[\Sigma_0]$ in Eq.~(\ref{Eq:BalanceLaw}) as the ensemble average of the net change in number of occupied trajectories between $\Sigma_0$ and $\Sigma_1$ due to collisions. In case the particles are charged, one needs to replace $\mathcal{N}[\Sigma]$ with $\mathcal{N}_F[\Sigma]$ and $L$ with $L_F$ in this identity. The net change in the number of occupied trajectories is a consequence of the following types of collisions (see Fig.~\ref{Fig:Collision} for an illustration):
\begin{itemize}
\item Binary collisions for which both $(x,p_1)$ and $(x,p_2)$ belong to the set $V$, however $(x,p_1^*)$ and $(x,p_2^*)$ lie outside $V$, see Fig.~\ref{subFig:BC_U20}. These collisions correspond to points in the following subset of $T^* C$:
\begin{equation}
U_{2,0} 
 := \{ (x,p_1,p_2,p_1^*,p_2^*)\in T^*C : (x,p_1), (x,p_2)\in V, (x,p_1^*),(x,p_2^*)\notin V \},
\end{equation}
and they lead to a \emph{depletion} of the averaged number of occupied trajectories in $V$ by the value $2\mathcal{N}_c[U_{2,0}]$.

\item Binary collisions for which both $(x,p_1)$ and $(x,p_2)$ belong to the set $V$ (as in the previous case); however now $(x,p_1^*)\in V$ and $(x,p_2^*)\notin V$ or vice versa: $(x,p_1^*)\notin V$ and $(x,p_2^*)\in V$, see Fig.~\ref{subFig:BC_U21}. Denoting the corresponding set by $U_{2,1}$, these collisions lead to a depletion of the averaged number of occupied trajectories in $V$ by the value $\mathcal{N}_c[U_{2,1}]$.

\item Binary collisions for which $(x,p_1)\in V$ and $(x,p_2)\notin V$ or vice versa: $(x,p_1)\notin V$ and $(x,p_2)\in V$ and $(x,p_1^*),(x,p_2^*)\notin V$, see Fig.~\ref{subFig:BC_U10}. Denoting the corresponding set by $U_{1,0}$, these collisions lead to a depletion of the averaged number of occupied trajectories in $V$ by the value $\mathcal{N}_c[U_{1,0}]$.

\item Likewise, the collisions $U_{0,2}$, $U_{1,2}$ and $U_{0,1}$ (the first index indicating the number of incoming particles belonging to $V$ and the second one the number of outgoing particles belonging to $V$) lead to an \emph{increase} of the averaged number of occupied trajectories in $V$ by the values $2\mathcal{N}_c[U_{0,2}]$, $\mathcal{N}_c[U_{1,2}]$ and $\mathcal{N}_c[U_{0,1}]$, respectively.
\end{itemize}

\begin{figure}[ht]
\vspace{1cm}
  \subfloat[Two internal (i.e. lying inside $V$) incoming particles resulting in two external (i.e. lying outside $V$) outgoing particles.]{
	\begin{minipage}[c][1\width]{
	   0.3\textwidth}
	   \centering
	   \includegraphics[width=1\textwidth]{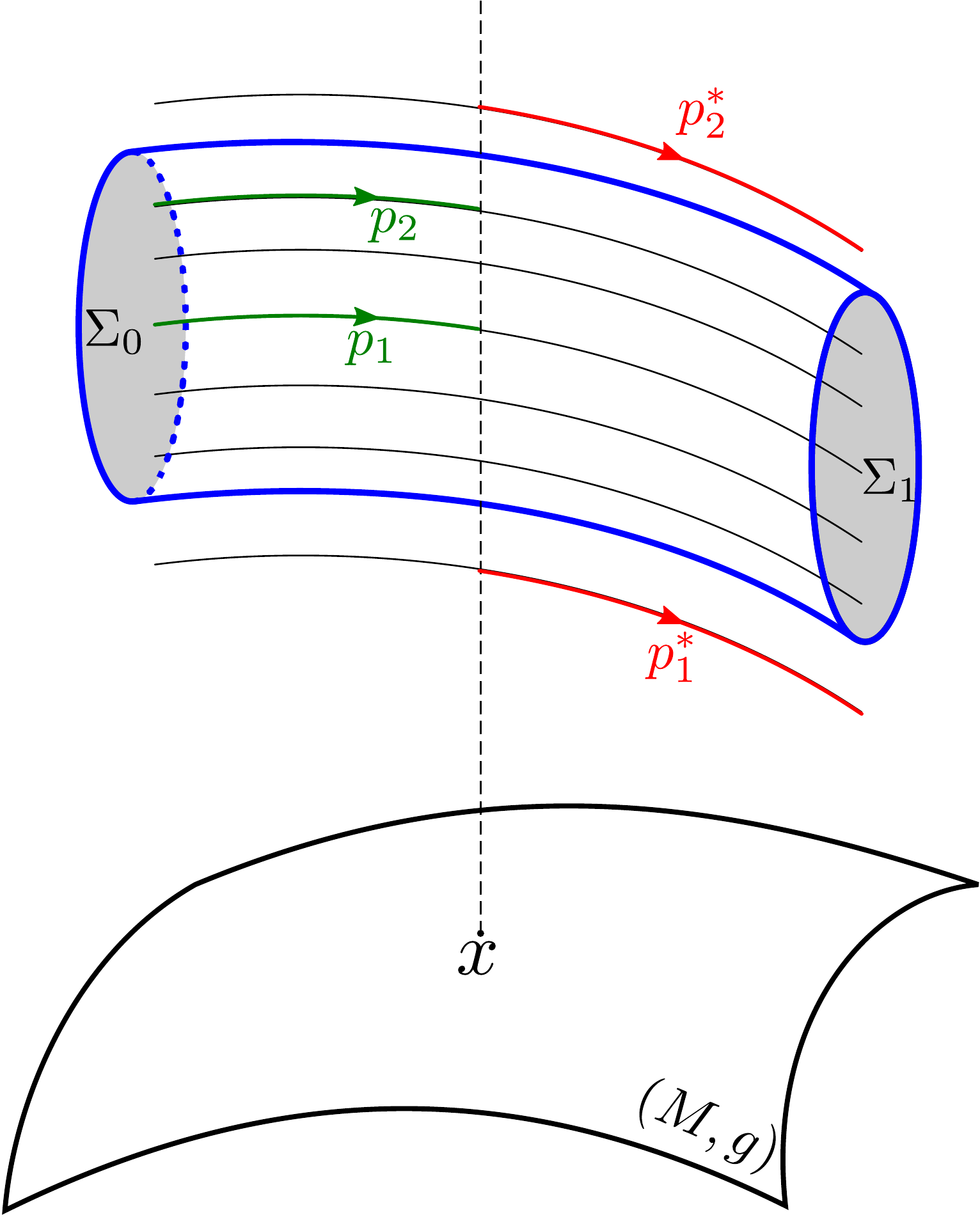}
	   \vspace{0.7cm}
	   \label{subFig:BC_U20}
	\end{minipage}}
  \hspace{2.5cm}
  \subfloat[Two internal incoming particles resulting in one internal and one external outgoing particle.]{
	\begin{minipage}[c][1\width]{
	   0.3\textwidth}
	   \centering
	   \includegraphics[width=1\textwidth]{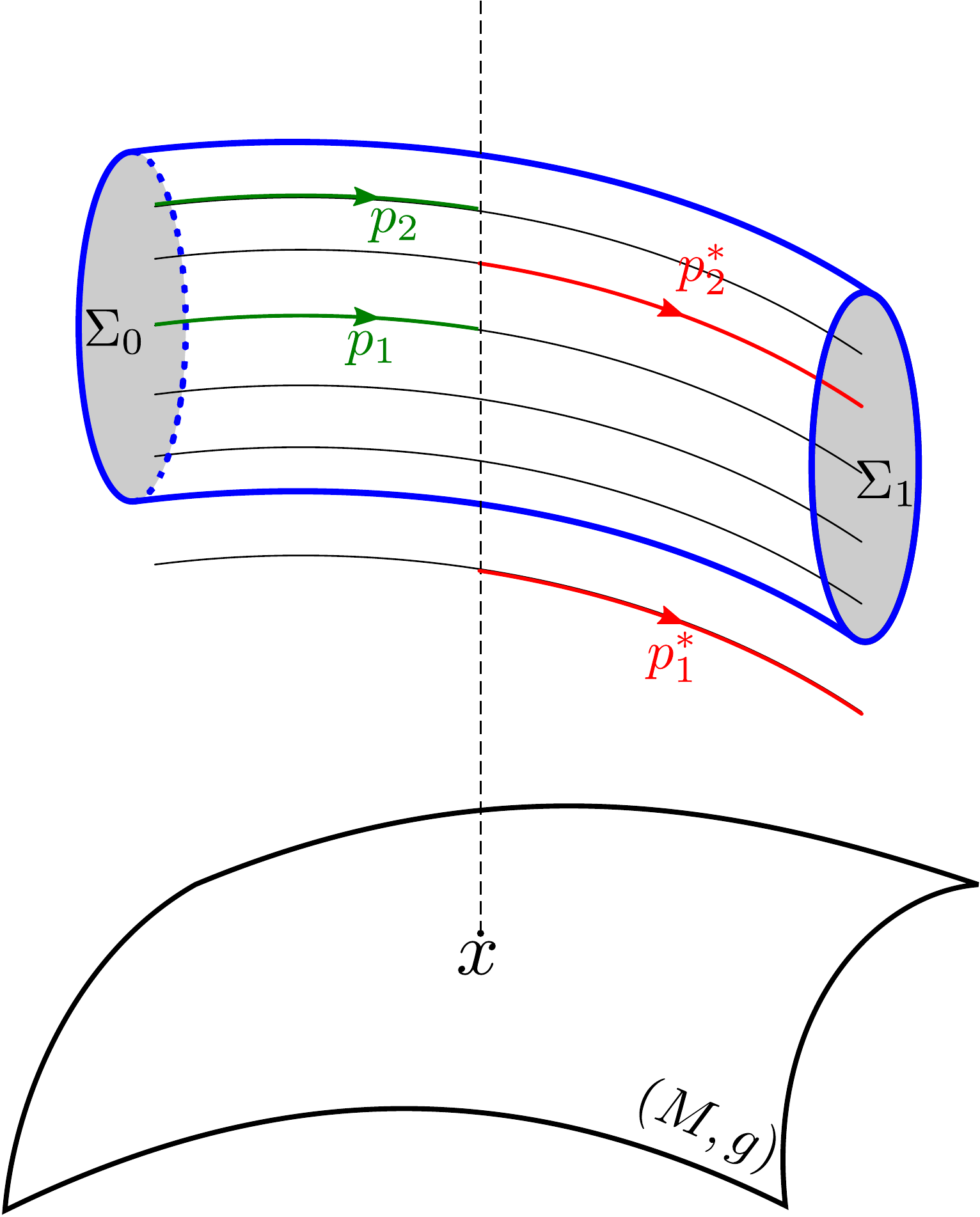}
	   \vspace{1.5cm}
	   \label{subFig:BC_U21}
	\end{minipage}}
  \vspace{2cm}
	\subfloat[One external and one internal incoming particle resulting in two external outgoing particles.]{
	\begin{minipage}[c][1\width]{
	   0.3\textwidth}
	   \centering
	   \includegraphics[width=1\textwidth]{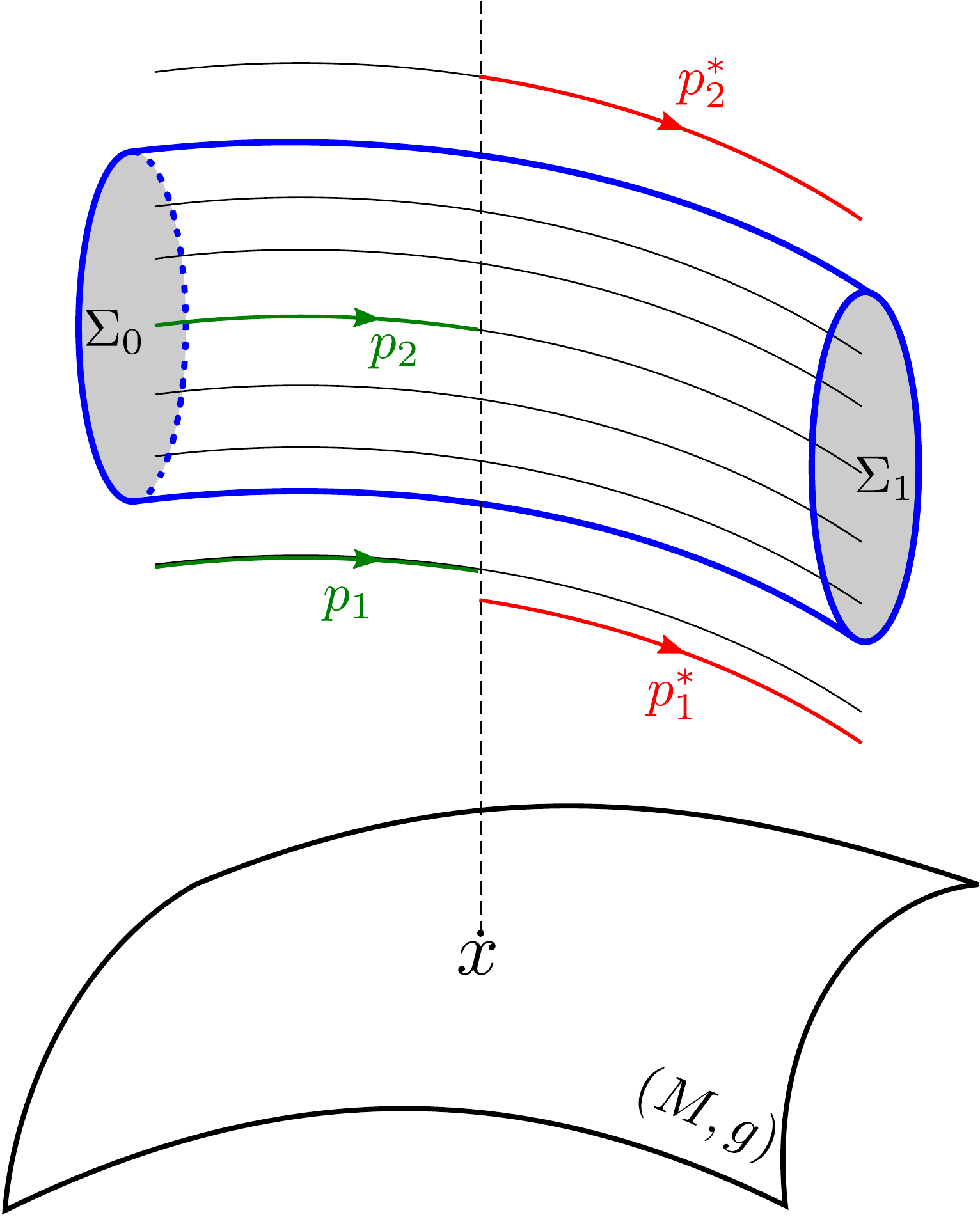}
	   \vspace{1.5cm}
	   \label{subFig:BC_U10}
	\end{minipage}}
  \hspace{2.5cm}
  \subfloat[Two external incoming particles resulting in one internal and one external outgoing particle.]{
	\begin{minipage}[c][1\width]{
	   0.3\textwidth}
	   \centering
	   \includegraphics[width=1\textwidth]{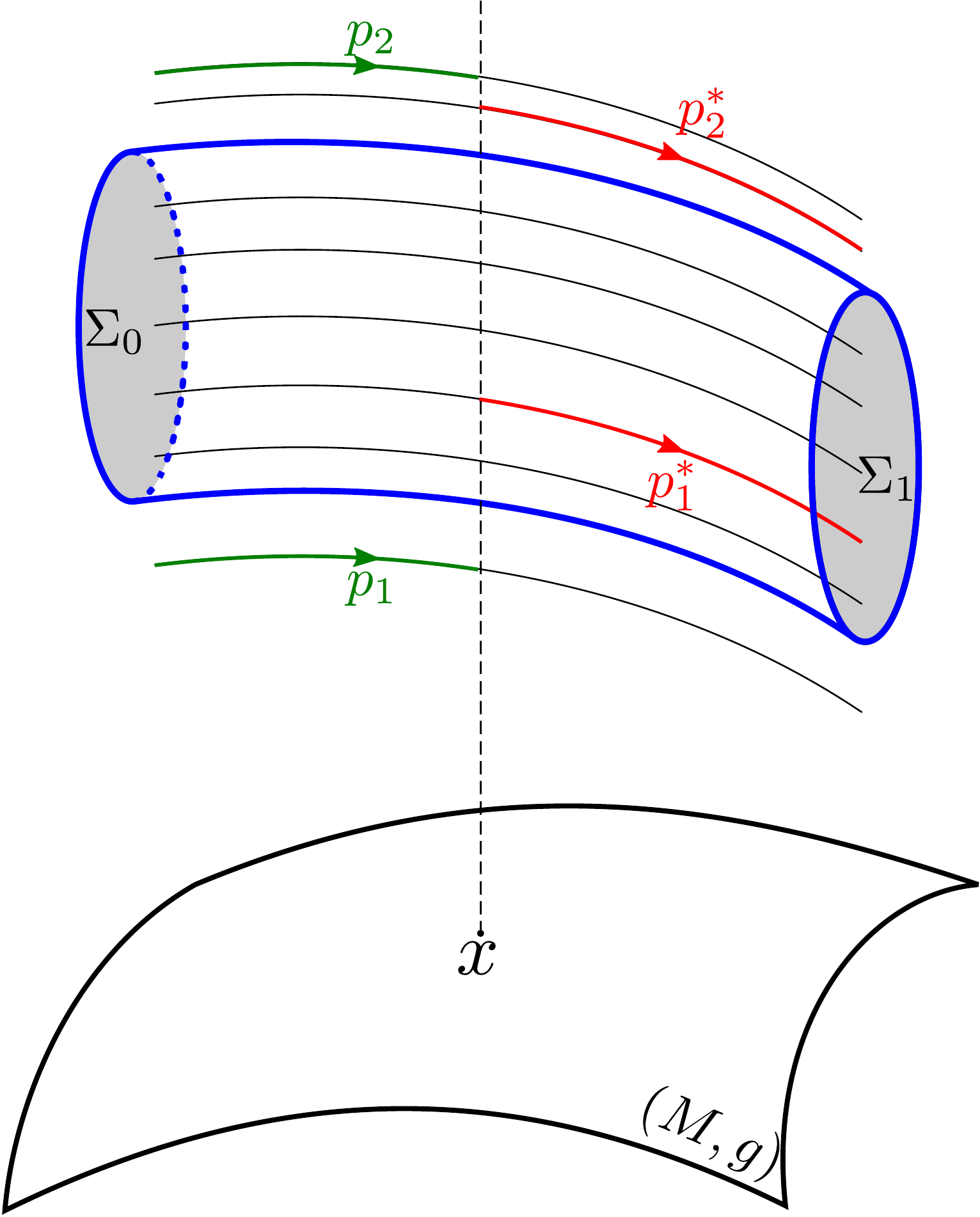}
	   \vspace{1.5cm}
	   \label{subFig:BC_U01}
	\end{minipage}}
	
\caption{A set of binary collisions that takes place at the event $x\in M$ with incoming and outgoing momenta in or outside the subset $V$ (in blue) of relativistic phase space $\Gamma_m^+$.}
\label{Fig:Collision}
\end{figure}

Taking into account these observations, and based on our assumption that only binary collisions take place, we obtain Boltzmann's equation in integral form:
\begin{equation}
\mathcal{N}_F[\Sigma_1] - \mathcal{N}_F[\Sigma_0] 
 = 2\mathcal{N}_c[U_{0,2}] + \mathcal{N}_c[U_{1,2}] + \mathcal{N}_c[U_{0,1}]
 - 2\mathcal{N}_c[U_{2,0}] - \mathcal{N}_c[U_{2,1}] - \mathcal{N}_c[U_{1,0}],
\label{Eq:BoltzmannIntegralForm}
\end{equation}
where $\mathcal{N}_F[\Sigma]$ and $\mathcal{N}_c[U]$ are defined in Eqs.~(\ref{Eq:NFluxF}) and (\ref{Eq:NcDef}), respectively. Using Eq.~(\ref{Eq:NcDef}) and the observation that a collision in $U_{r,l}$ changes the number of occupied trajectories by $l-r$, the right-hand side can be written as follows:
\begin{equation}
\frac{1}{4}\int\limits_M\left( \int\limits_{C_x} W_x(p_1+p_2\mapsto p_1^*+p_2^*)
f(x,p_1) f(x,p_2)\left[ \chi_V(x,p_1^*) +  \chi_V(x,p_2^*) - \chi_V(x,p_1) - \chi_V(x,p_2) \right]
\eta_{C_x} \right) \eta_M.
\end{equation}
Using the isometries $j$ and $j^*$ defined in Eqs.~(\ref{Eq:DefIsoj},\ref{Eq:DefIsoj*}) and the hypothesis that $W_x(p_1+p_2\mapsto p_1^*+p_2^*)$ is invariant with respect to them, we can rewrite this term as
\begin{equation}
\frac{1}{2}\int\limits_M\left( \int\limits_{C_x} W_x(p_1+p_2\mapsto p_1^*+p_2^*)
f(x,p_1) f(x,p_2)\left[ \chi_V(x,p_1^*) - \chi_V(x,p_1)  \right]\eta_{C_x} \right) \eta_M.
\end{equation}
Next, using the isometry $k$ defined in Eq.~(\ref{Eq:DefIsok}) we obtain, interchanging the integration variables $(p_1,p_2)$ with $(p_1^*,p_2^*)$ in the first term:
\begin{eqnarray}
\mathcal{N}_F[\Sigma_1] - \mathcal{N}_F[\Sigma_0]
  &=& \frac{1}{2}\int\limits_M\Biggl\{ \int\limits_{C_x} \left[ 
 W_x(p_1^*+p_2^*\mapsto p_1+p_2) f(x,p_1^*) f(x,p_2^*) 
 - W_x(p_1+p_2\mapsto p_1^*+p_2^*) f(x,p_1) f(x,p_2) \right]
\nonumber\\
&&  \qquad\qquad\qquad  \times  \chi_V(x,p_1) \eta_{C_x} \Biggr\} \eta_M.
\end{eqnarray}
In a next step, we use the representation of the volume element $\eta_{C_x}$ given in Eq.~(\ref{Eq:muxBisVol}) to rewrite the right-hand side as follows:
\begin{eqnarray}
&& \mathcal{N}_F[\Sigma_1] - \mathcal{N}_F[\Sigma_0]
  = \frac{m^d}{2}\int\limits_M\int\limits_{P_x^+(m)}  \chi_V(x,p_1) \left\{
  \int\limits_{P_x^+(m)}\int\limits_{S^{d-1}} 
  \left[  W_x(p_1^*+p_2^*\mapsto p_1+p_2) f(x,p_1^*) f(x,p_2^*)  \right. \right.
\nonumber\\
 && \qquad \left. \left. - W_x(p_1+p_2\mapsto p_1^*+p_2^*) f(x,p_1) f(x,p_2) \right] 
\sqrt{\frac{(s+t)(s+u)}{s}} g^{d-1} d\Omega(\hat{q}^*) \dvol_x(p_2) \right\} \dvol_x(p_1) \eta_M.
\label{Eq:BStep1}
\end{eqnarray}
Finally, we compare this result with the identity
\begin{equation}
\mathcal{N}_F[\Sigma_1] - \mathcal{N}_F[\Sigma_0] 
 = \int\limits_M \int\limits_{P_x^+(m)}\chi_V(x,p_1) L_F[f](x,p_1)\dvol_x(p_1) \eta_M,
\label{Eq:BStep2}
\end{equation}
which follows from the charged generalization of Eq.~(\ref{Eq:BalanceLaw}) and the Fubini-type formula~(\ref{Eq:Fubini}). Since Eqs.~(\ref{Eq:BStep1},\ref{Eq:BStep2}) hold for an arbitrary subset $V\subset \Gamma_m^+$ of the required form, we obtain the relativistic Boltzmann equation
\begin{equation}
\boxed{ L_F[f](x,p) = C_W[f,f](x,p), }
\label{Eq:Boltzmann}
\end{equation}
where the collision term $C_W[f,f]$ is defined by the following quadratic form of the one-particle distribution function:
\begin{eqnarray}
C_W[f,h](x,p_1) &:=& \frac{m^d}{2}\int\limits_{P_x^+(m)}\int\limits_{S^{d-1}} 
\left[  W_x(p_1^*+p_2^*\mapsto p_1+p_2) f(x,p_1^*) h(x,p_2^*) 
 - W_x(p_1+p_2\mapsto p_1^*+p_2^*) f(x,p_1) h(x,p_2) \right]
\nonumber\\
&&\qquad\qquad\qquad \times\sqrt{\frac{(s+t)(s+u)}{s}} g^{d-1}d\Omega(\hat{q}^*)\dvol_x(p_2),
\label{Eq:CollisionTerm1}
\end{eqnarray}
for $f,h: \Gamma_m^+\to \Real$ smooth functions with compact support and $(x,p_1)\in \Gamma_m^+$. Here, the Mandelstam variables $s,t,u$ and the variable $g=|q|$ are computed in terms of $p_1$ and $p_2$ according to the definitions in Eqs.~(\ref{Eq:qq*Def},\ref{Eq:Mandelstams},\ref{Eq:Mandelstamt},\ref{Eq:Mandelstamu}), and the outgoing momenta $p_1^*$ and $p_2^*$ are computed according to
\begin{equation}
p_1^* = \frac{p_1+p_2}{2} - \frac{m}{2} q^*,\qquad
p_2^* = \frac{p_1+p_2}{2} + \frac{m}{2} q^*,
\end{equation}
with
\begin{equation}
m q^* = g\left\{ -\ve{p}^{cm}\cdot\ve{\hat q}^* \theta^{\hat{0}} 
 + \left[ m(\hat{q}^*)_{\hat{a}} + \left( \sqrt{m^2 + |\ve{p}^{cm}|^2} - m \right)
 (\ve{\hat{p}}^{cm}\cdot\ve{\hat q}^*) \hat{p}^{cm}_{\hat{a}} \right] \theta^{\hat{a}} \right\},
\end{equation}
and $\ve{p}^{cm}$ and $g$ can be obtained from the inverse formulae~(\ref{Eq:pcm},\ref{Eq:g}).

We conclude this subsection with the following remarks:
\begin{enumerate}
\item Under the additional assumption of \emph{microscopic reversibility}, see Eq.~(\ref{Eq:WxIdSym2}), the collision term can be simplified to
\begin{eqnarray}
C_W[f,h](x,p_1) &:=& \frac{m^d}{2}\int\limits_{P_x^+(m)}\int\limits_{S^{d-1}} 
W_x(p_1^*+p_2^*\mapsto p_1+p_2)\nonumber\\
&& \qquad\times
\left[  f(x,p_1^*) h(x,p_2^*)- f(x,p_1) h(x,p_2) \right]\sqrt{\frac{(s+t)(s+u)}{s}} g^{d-1}d\Omega(\hat{q}^*)\dvol_x(p_2).
\label{Eq:CollisionTerm2}
\end{eqnarray}

\item Multiplying Boltzmann's equation~(\ref{Eq:Boltzmann}) by an arbitrary smooth function $\Psi:\Gamma_m^+\to\Real$ and integrating the result over $P_x^+(m)$ one obtains the equation
\begin{eqnarray}
&& \int\limits_{P_x^+(m)} \Psi(x,p_1) L_F[f](x,p_1) \dvol_x(p_1)
\nonumber\\
 &=& \frac{1}{2}\int\limits_{C_x}
\left[  W_x(p_1^*+p_2^*\mapsto p_1+p_2) f(x,p_1^*) f(x,p_2^*) 
 - W_x(p_1+p_2\mapsto p_1^*+p_2^*) f(x,p_1) f(x,p_2) \right]\Psi(x,p_1) \eta_{C_x}
\nonumber\\
 &=& -\frac{1}{4}\int\limits_{C_x}
W_x(p_1+p_2\mapsto p_1^*+p_2^*) f(x,p_1) f(x,p_2)
\left[ \Psi(x,p_1) + \Psi(x,p_2) - \Psi(x,p_1^*) - \Psi(x,p_2^*) \right]\eta_{C_x},
\label{Eq:IntBoltzmannIdentity}
\end{eqnarray}
where we have again used the definition~(\ref{Eq:muxBisVol}) of the volume element on the collision manifold $C_x$ and its isometries $j$, $j^*$ and $k$. Eq.~(\ref{Eq:IntBoltzmannIdentity}), which is a consequence of the Boltzmann equation~(\ref{Eq:Boltzmann}) and holds for arbitrary $x\in M$, has a number of important consequences which we discuss briefly.

First, we note that the right-hand side of Eq.~(\ref{Eq:IntBoltzmannIdentity}) vanishes if $\Psi(x,\cdot)$ is a collision invariant on $P_x^+(m)$. (In fact, provided that $f(x,\cdot) > 0$ and $W_x > 0$ are positive, the right-hand side vanishes if and only if $\Psi(x,\cdot)$ is a collision invariant.) According to Theorem~\ref{Thm:CollisionInvariants}, a collision invariant is of the form
\begin{equation}
\Psi(x,p) = \alpha(x) + \beta_x(p)
\label{Eq:Psi}
\end{equation}
for some function $\alpha$ and a one-form $\beta$ on $M$. Choosing $\beta = 0$ leads to
\begin{equation}
\int\limits_{P_x^+(m)} L_F[f](x,p) \dvol_x(p) = 0,
\end{equation}
which, in view of Theorem~\ref{Thm:sMomentsF} with $s=1$, leads to the continuity equation for the current particle density:
\begin{equation}
\boxed{ \nabla^\mu J_\mu = \divrg J = 0.}
\end{equation}
Therefore, as in the collisionless case, the continuity equation expressing particle number conservation is a consequence of the relativistic Boltzmann equation~(\ref{Eq:Boltzmann}).

Choosing instead $\alpha = 0$ in Eq.~(\ref{Eq:Psi}) and expanding the one-forms $\beta_x(p) = \beta^\mu(x) p_\mu$ leads to the identity
\begin{equation}
\int\limits_{P_x^+(m)} p_\mu L_F[f](x,p) \dvol_x(p) = 0.
\end{equation}
In view of Theorem~\ref{Thm:sMomentsF} with $s=2$, this implies that the energy-momentum-stress tensor $T$ satisfies the identity
\begin{equation}
\boxed{ \nabla^\mu T_{\mu\nu} = (\divrg T)_\nu = q F_\nu{}^\beta J_\beta.}
\label{Eq:DivTBoltzmann}
\end{equation}
Therefore, the relativistic Boltzmann equation~(\ref{Eq:Boltzmann}) implies the same equations~(\ref{Eq:JTConservation}) for the divergence of $J$ and $T$ as in the collisionless case.

Other important consequences of the identity~(\ref{Eq:IntBoltzmannIdentity}) will be discussed in the next section.
\end{enumerate}

\subsection{Relation between transition probability density and differential cross-section}
\label{SubSec:CrossSection}

We conclude this section by providing a relation between the transition probability density $W_x$ and the differential cross-section associated with the binary collision~(\ref{Eq:BinaryCollision}). Physically, the differential cross-section is defined as
\begin{equation}
\frac{d\sigma}{d\Omega} := \frac{\hbox{number of particles scattered in the solid angle $d\Omega$ per unit time}}{\hbox{incident flux}},
\label{Eq:DiffCrossSection}
\end{equation}
and has units of area, that is (length)$^{d-1}$. To provide the precise relation between $d\sigma/d\Omega$ and $W_x$, one needs to determine in which reference frame the solid angle and the incident flux should be measured, or --even better-- to provide frame-independent definitions for these quantities.

To gain some insight regarding the definition of the scattered number of particles and incident flux, it is useful to start with a dimensional analysis. Denoting by $[X]$ the units of a quantity $X$, and recalling that we are using geometrized units, we find:
\begin{align*}
\left[ p \right] &= \ell,                           && \mbox{(the particle momentum)},\\
\left[ \dvol_x(p) \right] &= \ell^{d-1},   && \mbox{(see Eq.~(\ref{Eq:dvolx}))},\\
\left[ \eta_M \right] &= \ell^{d+1},        && \mbox{(the volume element on $M$)},\\
\left[ \eta_{C_x} \right] &= \ell^{3d-1}, && \mbox{(see Eq.~(\ref{Eq:muxBisVol}))},\\ 
\left[ f \right] &= \ell^{-2d},                   && \mbox{(the one-particle distribution function)},\\
\left[ L_F \right] &= \ell^0,                   && \mbox{(the Liouville vector field, see Eq.~(\ref{Eq:LiouvilleVFCharged}))},
\end{align*}
where $\ell$ stands for a length unit. From the interpretation of the transition probability density in Eq.~(\ref{Eq:NcDef}) it then follows that
\begin{equation}
\left[ W_x(p_1+p_2\mapsto p_1^*+p_2^*) \right] = \ell^0,
\end{equation}
that is, $W_x$ is dimensionless. To make further progress, we fix an event $x\in M$ and choose local inertial coordinates $x^0,x^1,\ldots,x^d$ in the vicinity of $x$ with corresponding orthonormal frame $\displaystyle \left. e_{\hat{\alpha}}^\mu \right|_x = \delta^\mu{}_{\hat{\alpha}}$ at $x$. In the following, we assume that for each collision this frame is aligned with its center of mass frame. It follows from Eqs.~(\ref{Eq:NcDef}) and~(\ref{Eq:muxBisVol}) that the dimensionless quantity
\begin{equation}
\frac{1}{4} W_x(p_1+p_2\mapsto p_1^*+p_2^*) f(x,p_1) f(x,p_2)
\sqrt{\frac{(s+t)(s+u)}{s}} m^d \dvol_x(p_1)\dvol_x(p_2) g^{d-1} d\Omega(\hat{q}^*)\eta_M
\label{Eq:NumScatteredParticles}
\end{equation}
represents the number of binary collisions taking place in the (spacetime) volume element $\eta_M$ centered at $x$, for which the incoming particles have momenta lying in $\dvol_x(p_1)$ and $\dvol_x(p_2)$ and one particle is deflected through the angle $\Theta$ and scattered into the solid angle $d\Omega(\hat{q}^*)$ (as measured in the center of mass frame of the collision). Dividing the quantity (\ref{Eq:NumScatteredParticles}) by $d\Omega(\hat{q}^*) dx^0$ provides the numerator in Eq.~(\ref{Eq:DiffCrossSection})

To define the incident flux, we use Eq.~(\ref{Eq:ParticleDensityCoord}) from which it follows that the quantity
\begin{equation}
N^{\hat a} p_{\hat a} f(x,p) \dvol_x(p)
\label{Eq:NormalFlux}
\end{equation}
(whose dimensions are $\ell^{-d}$) represents the particle flux in direction $N^{\hat a}$ for particles with momenta $p$ in $\dvol_x(p)$. Further, it follows from Eq.~(\ref{Eq:NSigmaFromS}) that the dimensionless quantity
\begin{equation}
f(x,p)|p_{\hat 0}| \dvol_x(p) \eta_S,\qquad
\eta_S := dx^1\wedge dx^2\wedge\cdots\wedge dx^d
\label{Eq:Number}
\end{equation}
is the number of particles contained in the space volume $\eta_S$ with momenta in $\dvol_x(p)$. Choosing $N^{\hat a}$ in the same direction as $p$, it follows from Eqs.~(\ref{Eq:NormalFlux},\ref{Eq:Number}) that the incident flux of particles with momentum $p_1$ which are involved in a possible collision with a second particle with momentum $p_2$ at $x$, such that $p_1$ and $p_2$ lie in $\dvol_x(p_1)$ and $\dvol_x(p_2)$, respectively, is equal to
\begin{equation}
|\ve{p}_1|\sqrt{m^2 + |\ve{p}_2|^2} f(x,p_1) f(x,p_2)\dvol_x(p_1)\dvol_x(p_2)\eta_S,
\label{Eq:IncidentFlux}
\end{equation}
where $\ve{p}_1$ and $\ve{p}_2$ refer to the spatial orthonormal components of $p_1$ and $p_2$. In the center of mass frame one obtains from Eqs.~(\ref{Eq:Cp1Bis},\ref{Eq:Cp2Bis}),
\begin{equation}
|\ve{p}_1| = \frac{m g}{2},\qquad
\sqrt{m^2 + |\ve{p}_2|^2} = m\sqrt{1 + \frac{g^2}{4}} = \frac{\sqrt{s}}{2}.
\end{equation}
Gathering the results and noting that $\eta_M = dx^0\wedge \eta_S$, the following important relation between the transition probability density and the differential cross section is obtained:
\begin{equation}
\boxed{ \frac{1}{2} W_x(p_1+p_2\mapsto p_1^*+p_2^*) \sqrt{\frac{(s+t)(s+u)}{s}} m^d g^{d-1} 
 = m^2 g\sqrt{1 + \frac{g^2}{4}}\frac{d\sigma}{d\Omega},}
\label{Eq:RelationWDiffCross}
\end{equation}
where we recall that Lorentz invariance and the symmetries $j$ and $j^*$ require $W_x$ (and hence also $d\sigma/d\Omega$) to be a function of $g$ and $\Theta$ only which is invariant with respect to $\Theta\mapsto \pi - \Theta$. With these requirements, which also imply the validity of microscopic reversibility, the collision term~(\ref{Eq:CollisionTerm2}) can be rewritten as\footnote{Sometimes, the collision term~(\ref{Eq:CollisionTerm3}) is written with a factor $1/2$ in front of it (see for instance Eq.~(39) on page 22 in Ref.~\cite{Groot-Book}), which is due to a different definition of the  differential cross section. For further details on this difference, see the discussion on page~37 in~\cite{CercignaniKremer-Book}.}
\begin{equation}
\boxed{ C_W[f,f](x,p_1) = m^2\int\limits_{P_x^+(m)}\int\limits_{S^{d-1}}
 g\sqrt{1 + \frac{g^2}{4}}\frac{d\sigma}{d\Omega}(g,\Theta)
  \left[  f(x,p_1^*) f(x,p_2^*) -  f(x,p_1) f(x,p_2) \right] d\Omega(\hat{q}^*)\dvol_x(p_2). }
\label{Eq:CollisionTerm3}
\end{equation}
This is the form of the collision term for a single non-degenerate gas\footnote{For a generalized version of this collision term which takes into account the quantum statistics of the particles, see for example Chapter~2.2 in Ref.~\cite{CercignaniKremer-Book} and references therein.} that can be found in most textbooks, see for example Eq.~(12.23) in Ref.~\cite{CercignaniKremer-Book}. The quantity
\begin{equation}
\mathcal{F} := m^2 g\sqrt{1 + \frac{g^2}{4}} = \sqrt{(p_1\cdot p_2)^2 - m^4}
\end{equation}
is called the invariant flux. It is sometimes written in the form $\mathcal{F} = p_1^{\hat 0} p_2^{\hat 0} g_{\o}$ with M{\o}ller's relative speed
\begin{equation}
g_{\o} := \sqrt{|\ve{v}_2 - \ve{v}_1|^2 - |\ve{v}_1\wedge \ve{v}_2|^2},\qquad
\ve{v}_1 := \frac{1}{p_1^{\hat 0}}\ve{p}_1,\quad
\ve{v}_2 := \frac{1}{p_2^{\hat 0}}\ve{p}_2.
\end{equation}
However, note that while $\mathcal{F}$ and $g$ are Lorentz scalars, $g_{\o}$ is not. For further details on the definition of the invariant flux, a Lorentz-invariant definition of the relative velocity and their role for the definition of the relativistic cross-section, we refer the reader to~\cite{mC17}.

We conclude this section with a few remarks regarding the form of the differential cross section (or scattering kernel) that is used in the literature in the $(d=3)$-dimensional case. As follows from the discussion in Section~\ref{SubSec:TransitionProbability}, the differential cross section should be a non-negative function depending only on the relative velocity $g$ and the scattering angle $\Theta$. Furthermore, it is convenient to require the following hypothesis (see~\cite{jJyS19}):
\begin{equation}
0\leq \frac{d\sigma}{d\Omega}(g,\Theta)  \leq K g^a \Gamma(\Theta), \qquad g > 0,\quad
0\leq \Theta\leq \pi,
\label{Eq:HypotesisDiffCrossSection}
\end{equation}
with $K > 0$ a  positive constant and $\Gamma$ a non-negative, bounded function depending only on $\Theta$. The constant $a > -3$ depends on the nature of the interaction between the particles and divides the class of collisional kernels in hard- ($a\geq 0$) and soft-potential ($a < 0$) kernels. Representative  examples in each class are:
\begin{enumerate}
\item The ``hard-sphere" model, which is characterized by a constant differential cross section, that is,
\begin{equation}
\frac{d\sigma}{d\Omega}(g,\Theta) = \textrm{const.} = \frac{\sigma_T}{4\pi},
\end{equation}
with $\sigma_T$ the total cross section. Apart from its mathematical simplicity, it might also serve as an approximate model for describing strongly interacting particles at high energy, see the discussion in the introductory paragraphs of chapter~XI in~\cite{Groot-Book}. In the non-relativistic limit, this model reduces to the cross section corresponding to hard spheres of diameter $D$, in which case $\sigma_T = \pi D^2$~\cite{CercignaniKremer-Book}.

\item Israel particles~\cite{wI63,CercignaniKremer-Book} are characterized by a differential cross section of the form
\begin{equation}
\frac{d\sigma}{d\Omega}(g,\Theta) = \frac{m^2}{g\left( 1 + \frac{g^2}{4} \right)}\Gamma(\Theta),
\end{equation}
with an arbitrary, non-negative and bounded function $\Gamma(\Theta)$ of the scattering angle $\Theta$. Note that this model satisfies the bound~(\ref{Eq:HypotesisDiffCrossSection}) with $a = -1$.
\end{enumerate}

There are also interesting proposals that have been considered in the literature which interpolate between the two cases, for example (see section~6.6 in Ref.~\cite{CercignaniKremer-Book} and references therein)
\begin{equation}
\frac{d\sigma}{d\Omega}(g,\Theta) = \frac{m^2}{g}\sqrt{1 + \frac{g^2}{4}} \Gamma(\Theta),
\end{equation}
where in the classical limit $g\to 0$ one obtains the differential cross section associated with Maxwell particles, whereas in the ultra-relativistic limit $g\to \infty$ one obtains a constant total cross section, like in the ``hard-sphere" model.

For recent studies regarding the relativistic Boltzmann or Einstein-Boltzmann systems in a cosmological setting, see~\cite{kT03,nNeT06,hLaR13,hL13,dBgDuHmMjN16,dBgDuHmMjN16b,hLeN17,hLeN18,hLeNpT19}.

\section{H-theorem and equilibrium configurations}
\label{Sec:Equilibrium}

After having derived the relativistic Boltzmann equation, in this section we discuss one of its  most important consequence, namely Boltzmann's celebrated H-theorem which implies that any solution of the Boltzmann equation has a non-decreasing entropy.

\subsection{H-theorem}

In order to formulate the H-theorem, we introduce the \emph{entropy flux covector field} $S\in {\cal X}^*(M)$, defined as
\begin{equation}
\boxed{S_x(X) := -k_B\int\limits_{P_x^+(m)} f(x,p)\log(A f(x,p)) p(X) \dvol_x(p),
\qquad X\in T_x M,}
\end{equation}
where here and in the following, $k_B$ denotes Boltzmann's constant and $A$ an arbitrary constant with units $\ell^{2d}$, such that the argument of the logarithm is dimensionless.\footnote{Note that a rescaling $A\mapsto \lambda A$ of $A$ by a positive constant $\lambda$ induces the transformation $S_\mu\mapsto S_\mu - \log(\lambda) k_B J_\mu$ with $J_\mu$ the particle current density, see Eq.~(\ref{Eq:ParticleDensityCoord}). However, note that $\nabla^\mu S_\mu$ is invariant with respect to this rescaling since $J_\mu$ is divergence-free.} In terms of adapted local coordinates, this can also be written as
\begin{equation}
\boxed{S_\mu(x) := -k_B\int\limits_{P_x^+(m)} f(x,p)\log(A f(x,p)) p_\mu \dvol_x(p). }
\label{Eq:EntropyFlux}
\end{equation}
Next, we compute the divergence of this entropy flux. Using Theorem~\ref{Thm:sMomentsF} with $s=1$ and $-k_B f\log(A f)$ instead of $f$, we obtain
\begin{equation}
\nabla^\mu S_\mu(x) = -k_B\int\limits_{P_x^+(m)} L_F[ f\log(A f) ] \dvol_x(p)
 = -k_B\int\limits_{P_x^+(m)} \left[ 1 + \log(A f) \right] L_F[f] \dvol_x(p).
\end{equation}
The right-hand side can be rewritten by applying the identity~(\ref{Eq:IntBoltzmannIdentity}) to the function $\Psi(x,p) = 1 + \log(A f(x,p) )$, giving
\begin{equation}
\nabla^\mu S_\mu(x) =  \frac{k_B}{4}\int\limits_{C_x}
W_x(p_1+p_2\mapsto p_1^*+p_2^*) f(x,p_1) f(x,p_2)
\left[ \log(A^2 f(x,p_1)f(x,p_2)) - \log(A^2 f(x,p_1^*)f(x,p_2^*)) \right]\eta_{C_x}.
\end{equation}
Assuming the validity of microscopic reversibility (see Eq.~(\ref{Eq:WxIdSym2})), this equation can be rewritten in the following form:
\begin{equation}
\nabla^\mu S_\mu(x) = \frac{k_B}{8}\int\limits_{C_x}
W_x(p_1+p_2\mapsto p_1^*+p_2^*) \left[ f_1 f_2 - f_1^* f_2^* \right]
\left[ \log(A^2 f_1 f_2) - \log(A^2 f_1^* f_2^*) \right]\eta_{C_x},
\label{Eq:DivS}
\end{equation}
where for convenience we have abbreviated $f_1:=f(x,p_1)$, $f_2:=f(x,p_2)$, $f_1^*:=f(x,p_1^*)$, $f_2^*:=f(x,p_2^*)$. The product of the last two factors in the integrand on the right-hand side has the form
\begin{equation}
G(x,y) := (x - y)(\log x - \log y),\qquad x,y > 0,
\end{equation}
with $x = A^2 f_1 f_2$ and $y = A^2 f_1^* f_2^*$. Since $\log$ is a strictly increasing function, it follows that $G(x,y) > 0$ for all $x\neq y$ and that $G(x,y) = 0$ if and only if $x = y$. From this observation, we obtain the H-theorem:
\begin{equation}
\boxed{\nabla^\mu S_\mu(x) \geq 0,}
\label{Eq:HTheorem}
\end{equation}
with the equality if and only if $W_x = 0$ or $W_x > 0$ and $f_1 f_2 = f_1^* f_2^*$ at the point $x$. As an application of this important inequality, we consider an asymptotically flat, globally hyperbolic spacetime $(M,g)$ with two Cauchy surfaces $C_2$ and $C_1$ whose future-directed normal vector field is denoted by $\nu$. Assuming that $f$ decays sufficiently rapidly at infinity and that $C_2$ lies to the future of $C_1$, the inequality~(\ref{Eq:HTheorem}) and Gauss' theorem imply that
\begin{equation}
\boxed{ S[C_2] \geq S[C_1],}
\label{Eq:SecondLaw}
\end{equation}
with
\begin{equation}
S[C_i] := -\int\limits_{C_i} S_\mu \nu^\mu \eta_{C_i},\qquad i=1,2,
\label{Eq:EntropyCS}
\end{equation}
the entropy contained in the Cauchy surface $C_i$. Eq.~(\ref{Eq:SecondLaw}) describes the second law of thermodynamics, according to which the total entropy is a non-decreasing function of time. Therefore, a solution of the relativistic Boltzmann equation~(\ref{Eq:Boltzmann}) on a globally hyperbolic spacetime $(M,g)$ foliated by Cauchy surfaces $C_t$ has an associated entropy $S[C_t]$ which grows in time $t$ until an equilibrium configuration is reached (if such a state exists at all!). In the following, we discuss and characterize such equilibrium configurations.

\subsection{Global equilibrium configurations}
\label{SubSec:GlobalEquilibrium}

Next, we discuss the equilibrium configurations, that is, the one-particle distribution functions $f(x,p)$ for which the equality holds for all $x\in M$ in Eq.~(\ref{Eq:HTheorem}). For the following, we assume that $W_x > 0$ is strictly positive for all $x\in M$. According to Eq.~(\ref{Eq:DivS}), the divergence of the entropy flux is zero if and only if $f_1 f_2 = f_1^* f_2^*$ for all binary elastic collisions, which is equivalent to the condition that $\log(A f(x,p))$ is a collision invariant. In view of Theorem~\ref{Thm:CollisionInvariants}, $f$ must have the form
\begin{equation}
f(x,p) = \alpha(x)\exp\left[ p(\beta_x) \right] = \alpha(x)\exp\left[ \beta^\mu(x) p_\mu \right],
\label{Eq:LocalEquilibrium}
\end{equation}
with $\alpha\in {\cal F}(M)$ a positive function and $\beta\in {\cal X}(M)$ a vector field on $M$. This is precisely of the form of the example studied in subsection~\ref{SubSec:Juttner}, where it was noted that the corresponding spacetime observables describe a perfect fluid characterized by a $(d+1)$-velocity proportional to $\beta^\mu$. Note that in order for the observables to be well-defined, the vector field $\beta$ is required to be everywhere future-directed timelike.

Clearly, any distribution function of the form~(\ref{Eq:LocalEquilibrium}) cancels the collision term $C_W[f,f]$ on the right-hand side of the Boltzmann equation~(\ref{Eq:Boltzmann}). Hence, in order to be a solution of the Boltzmann equation, $f$ must also satisfy the condition $L_F[f] = 0$. Using Eq.~(\ref{Eq:LFpbeta}) this yields
\begin{equation}
0 = L_F[\log(A f)] = p_\mu p_\nu\nabla^\mu \beta^\nu(x)
 + p_\nu\left[ \nabla^\nu\log(A\alpha(x)) + q\beta^\mu(x) F_\mu{}^\nu(x) \right]
\label{Eq:TransportZero}
\end{equation}
for all $(x,p)\in \Gamma_m^+$. This implies the following two conditions for $\alpha$ and $\beta$:
\begin{equation}
\nabla^{(\mu} \beta^{\nu)} = 0,\qquad
\nabla^\nu\log(A\alpha(x)) + q\beta^\mu(x) F_\mu{}^\nu(x) = 0.
\label{Eq:TransportConditions}
\end{equation}
The first condition means that $\pounds_\beta g = 0$, that is the timelike vector field $\beta\in {\cal X}(M)$ must be a Killing vector field. Using the notation of differential forms the second condition can be rewritten as
\begin{equation}
d\log(A\alpha) + q i_\beta F = 0.
\label{Eq:EquilibriumFirstCondition}
\end{equation}
Applying the exterior derivative $d$ on both sides of this equation, using Cartan's formula and the fact that $dF = 0$ one obtains
\begin{equation}
q\pounds_\beta F = 0.
\label{Eq:FStationary}
\end{equation}
Therefore, we arrive at the important conclusion that a global equilibrium configuration (that is, a distribution function satisfying the Boltzmann equation for which the entropy is conserved) exists if and only the fields $g$ and $q F$ are stationary, that is, they admit the same, globally-defined, timelike Killing vector field $\beta\in {\cal X}(M)$. In particular, this implies that on a generic globally hyperbolic spacetime, the entropy function~(\ref{Eq:EntropyCS}) increases strictly in time along any foliation by Cauchy surfaces.

If the conditions for stationarity on $g$ and $F$ are satisfied, we may decompose the timelike Killing vector field $\beta$ in its norm and its normalized part, according to Eq.~(\ref{Eq:IdealGas}):
\begin{equation}
\boxed{ \beta^\mu = \frac{1}{k_B T} u^\mu,\qquad
k_B T := [-g(\beta,\beta)]^{-1/2} = (-\beta^\mu\beta_\mu)^{-1/2}, }
\label{Eq:Temperature}
\end{equation}
with the vector field $u$ satisfying $g(u,u) = u^\mu u_\mu = -1$ and representing the $(d+1)$-velocity of the associated perfect fluid flow (see Eqs.~(\ref{Eq:PerfectFluid}) and (\ref{Eq:PerfectFluidnu})). Denoting by $E := -i_u F$ the electric field associated with observers which are comoving with the fluid flow, the condition~(\ref{Eq:FStationary}) implies $d(i_\beta F) = 0$ which in turn implies the local existence of an ``electric'' potential $\Phi$, such that $E/(k_B T) = -i_\beta F = -d\Phi$ and $\pounds_\beta\Phi = 0$. Provided this potential exists globally, Eq.~(\ref{Eq:EquilibriumFirstCondition}) can readily be integrated and one obtains\footnote{An alternative expression for the equilibrium distribution function can be given in terms of the canonical momentum $\Pi := p + q A$ (see Eq.~(\ref{Eq:Pi})), provided the electromagnetic field admits a global potential $A$ satisfying $\pounds_\beta A = 0$:
$$
f(x,p) = \alpha_0 e^{\beta^\mu\Pi_\mu},
$$
see for instance Ref.~\cite{Groot-Book}.}
\begin{equation}
\boxed{ f(x,p) = \alpha_0\exp \left[\frac{u^\mu(x) p_\mu}{k_B T(x)} - q\Phi(x) \right], }
\label{Eq:DistributionFunctionOfChargedGas}
\end{equation}
with $\alpha_0\geq 0$ a constant with units $\ell^{-2d}$. As follows from the spacetime observables computed in subsection~\ref{SubSec:Juttner} (see Eqs.~(\ref{Eq:PerfectFluid},\ref{Eq:PerfectFluidnu},\ref{Eq:PerfectFluidhp})), this distribution function describes a perfect fluid whose $(d+1)$-velocity is given by $u^\mu$ and whose pressure $P$ and particle density $n$ satisfy the ideal gas equation $P = n k_B T$.

We conclude this subsection with two remarks. The first remark is related to the definition of the temperature $T(x)$ in Eq.~(\ref{Eq:Temperature}), which has the nice property of yielding the ideal gas equation, as we have just discussed. However, there is a much deeper reason for accepting the definition~(\ref{Eq:Temperature}) to describe the temperature of the gas which is independent of the assumption that the gas is an ideal one. To discuss this point, we follow~\cite{wI63} and compute the entropy flux (see Eq.~(\ref{Eq:EntropyFlux})) for an arbitrary distribution function which is of the form~(\ref{Eq:LocalEquilibrium}). This yields
\begin{eqnarray}
S_\mu(x) &=& -k_B \int\limits_{P_x^+(m)} f(x,p) \log\left[A\alpha(x) e^{\beta(x)^\nu p_\nu} \right] p_\mu \dvol_x(p) \nonumber\\
&=& -k_B \log[A\alpha(x)] \int\limits_{P_x^+(m)} f(x,p) p_\mu \dvol_x(p)
 - k_B \beta(x)^\nu \int\limits_{P_x^+(m)} f(x,p) p_\nu p_\mu \dvol_x(p)
\nonumber\\
&=& -k_B \log[A\alpha(x)] J_\mu(x) - k_B \beta(x)^\nu T_{\mu\nu}(x),
\end{eqnarray}
where in the last step we have used Eqs.~(\ref{Eq:ParticleDensityCoord}) and (\ref{Eq:TsCoord}) with $s=2$. On the other hand, we have already computed the particle current density and the energy-momentum-stress tensor for a distribution function of the form~(\ref{Eq:LocalEquilibrium}) in subsection~\ref{SubSec:Juttner}, see Eq.~(\ref{Eq:PerfectFluid}). Using these expressions and the decomposition~(\ref{Eq:Temperature}) one obtains $S^\mu(x) = \mathcal{S}(x) u^\mu(x)$, with
\begin{equation}
\mathcal{S}(x) = -k_B\log[A\alpha(x)] n(x) + \frac{n(x) h(x) - P(x)}{T(x)}
\end{equation}
the entropy density measured by an observer who is comoving with the fluid flow, where here $n(x)$, $h(x)$ and $P(x)$ refer to the particle density, enthalpy per particle and pressure of the fluid, see Eqs.~(\ref{Eq:PerfectFluidnu},\ref{Eq:PerfectFluidhp}). The corresponding entropy per particle is
\begin{equation}
s(x) := \frac{\mathcal{S}(x)}{n(x)} = -k_B\left[ 1 + \log(A\alpha(x)) \right] + \frac{h(x)}{T(x)},
\label{Eq:EntropyPerParticle}
\end{equation}
where we have used the ideal gas equation $P = n k_B T$. By varying the parameters $\alpha$ and $T$, one finds
\begin{equation}
T ds = -k_B T \frac{d\alpha}{\alpha} + dh - h\frac{dT}{T}.
\end{equation}
Here, we can eliminate $\alpha$ by varying the expression~(\ref{Eq:PerfectFluidhp}) for the pressure, which yields
\begin{equation}
\frac{d\alpha}{\alpha} = \frac{dP}{P} - \frac{h}{k_B T} \frac{dT}{T},
\label{Eq:DifferentialTwo}
\end{equation}
and hence,
\begin{equation}
T ds = dh - k_B T\frac{dP}{P} = dh - \frac{dP}{n},
\label{Eq:FirstLawOfThermodynamics}
\end{equation}
where we have used again the ideal gas equation in the last step. We recognize in Eq.~(\ref{Eq:FirstLawOfThermodynamics}) the first law of thermodynamics, which reinforces the interpretation for $T$ describing the correct temperature of the system. We stress that this result holds for any distribution function of the form~(\ref{Eq:LocalEquilibrium}). In particular, it holds for the configurations in global equilibrium of the form~(\ref{Eq:DistributionFunctionOfChargedGas}). In this case, the entropy per particle is
\begin{equation}
\boxed{ s(x) = -k_B\left[ 1 + \log(A\alpha_0) - q\Phi(x) \right] + \frac{h(x)}{T(x)}. }
\end{equation}
Finally, we note that the temperature $T(x)$ of a gas described by a global equilibrium configuration on a stationary curved spacetime $(M,g)$ is not necessarily constant. If $\beta\in {\cal X}(M)$ is the timelike Killing vector field which is proportional to the flow's $(d+1)$-velocity $u$, it follows from Eq.~(\ref{Eq:Temperature}) that
\begin{equation}
\boxed{ \sqrt{-g(\beta,\beta)} T(x) = \textrm{const.} }
\end{equation}
This is known as the Tolman-Ehrenfest theorem, see~\cite{rT1930,rTeP1930}, and is a purely relativistic effect.

\subsection{Local equilibrium configurations}

In the last subsection we have discussed the global equilibrium configurations of the relativistic Boltzmann equation, and have seen that the existence of such configurations is very restrictive: it requires spacetime (and the electromagnetic field, if present) to be globally stationary. Therefore, a natural question is what occurs in spacetimes which are not stationary. Is it still possible to talk about configurations which are locally in equilibrium in some sense? One possibility is to define a local equilibrium configuration as one for which the distribution function has the form of Eq.~(\ref{Eq:LocalEquilibrium}) for an arbitrary timelike vector field $\beta\in {\cal X}(M)$ and function $\alpha\in {\cal F}(M)$. This is equivalent to the form~\cite{Groot-Book,Vereshchagin-Book,CercignaniKremer-Book}
\begin{equation}
\boxed{ f(x,p) = \frac{1}{A}\exp \left[\frac{u^\alpha(x) p_\alpha + \mu(x)}{k_B T(x)} \right], }
\label{Eq:LocalEquilibriumBis}
\end{equation}
where $A$ is a positive constant with units $\ell^{2d}$, $u\in {\cal X}(M)$ is a timelike vector field normalized such that $g(u,u) = -1$ and $T,\mu\in {\cal F}(M)$ are functions representing, respectively, the temperature $T > 0$ and Gibbs potential $\mu = h - T s$ per particle. By construction, this distribution function has associated to it a vanishing collision term. However, recall from the previous section that the relativistic Boltzmann equation~(\ref{Eq:Boltzmann}) is only satisfied if the conditions~(\ref{Eq:TransportConditions}) hold, which is never the case if spacetime fails to be stationary. Hence, can one still envisage a situation in which the distribution function~(\ref{Eq:LocalEquilibriumBis}) makes sense, at least in an approximate way?

To analyze this question we must first realize that a solution of the Boltzmann equation has two length scales associated to it. First, there is a macroscopic characteristic length $l_{\text{ms}}$, defined as a typical length scale over which the spacetime observables (particle current density $J_\mu$, energy-momentum-stress tensor $T_{\mu\nu}$ etc.) vary.\footnote{For example, one could define $l_{\text{ms}}^{-2}$ as the maximum value of $g^{-1}(D n,D n)/n^2$ with $n := \sqrt{-J^\mu J_\mu}$ the invariant particle density and $D n := dn + (\pounds_u n)u$ the spatial gradient of $n$ orthogonal to the velocity vector $u := J/n$. Instead of $n$ one could also consider other scalars constructed from $J_\mu$ or $T_{\mu\nu}$ in the previous definition.} The second important length scale is of microscopic nature and consists of the \emph{mean free path} $\ell_{\text{mfp}}$, that is the average distance travelled by a particle between successive collisions. Locally, it is defined through the relation
\begin{equation}
\boxed{\sigma_T\ell_{\text{mfp}} n = 1, }
\label{Eq:ImportantRelation}
\end{equation}
with
\begin{equation}
\sigma_T := \int\limits_{S^{d-1}} \frac{d\sigma}{d\Omega} d\Omega
\label{Eq:TotalCrossSection}
\end{equation}
the total cross section and $n$ the particle density. (Recall that the differential cross section $d\sigma/d\Omega$ and hence also $\sigma_T$ have units of $\ell^{d-1}$ while $n$ has units of $\ell^{-d}$.)

Based on the two length scales $l_{\text{ms}}$ and $\ell_{\text{mfp}}$, we can cast the relativistic Boltzmann equation~(\ref{Eq:Boltzmann}) in terms of dimensionless variables by writing\footnote{We are thankful to Ana Laura Garc\'ia Perciante for explaining this point to us.}
\begin{equation}
x^\mu = l_{\text{ms}}\overline{x}^\mu,\qquad p_\mu = m\overline{p}_\mu,\qquad
q F_{\mu\nu} = \frac{m}{l_{\text{ms}}} \overline{q}\overline{F}_{\mu\nu},\qquad
f = \frac{1}{l_{\text{ms}}^d m^d} \overline{f},\qquad
\frac{d\sigma}{d\Omega} 
 = \frac{l_{\text{ms}}^d}{\ell_{\text{mfp}}}\frac{d\overline{\sigma}}{d\Omega},
\end{equation}
where all the quantities with a bar are dimensionless. In terms of these quantities, one obtains the dimensionless (or rescaled) relativistic Boltzmann equation
\begin{equation}
\boxed{ \overline{L_F}[\overline{f}] =
 \frac{1}{\text{Kn}}\overline{C}_W[\overline{f},\overline{f}], \qquad
 \text{Kn} := \frac{\ell_{\text{mfp}}}{l_{\text{ms}}}.}
\label{Eq:DimensionlessBoltzmann}
\end{equation}
The quantity $\text{Kn}$ representing the ratio between the mean free path and the macroscopic scale is called the Knudsen number, see Refs.~\cite{mK1909a} and~\cite{wS1986} for the original reference and a historical account.

After these remarks, we return to the question regarding the interpretation of the distribution function~(\ref{Eq:LocalEquilibriumBis}) as describing local equilibrium. Indeed, if $\text{Kn} \ll 1$ is small, we see from Eq.~(\ref{Eq:DimensionlessBoltzmann}) that the Boltzmann equation is dominated by the collision term which vanishes precisely for the distribution functions of the form~(\ref{Eq:LocalEquilibriumBis}). In this limit, which is also called the \emph{hydrodynamic limit} one can try to formally expand the distribution function according to
\begin{equation}
\overline{f} = \overline{f}^{(0)} + \text{Kn}\overline{f}^{(1)} + \text{Kn}^2\overline{f}^{(2)} 
+ \cdots,
\end{equation}
with $\overline{f}^{(0)}$ of the form~(\ref{Eq:LocalEquilibriumBis}) and $\overline{f}^{(1)}$, $\overline{f}^{(2)}$ correction terms, which are determined by the integral equations
\begin{eqnarray}
\overline{C}_W[\overline{f}^{(0)},\overline{f}^{(1)}] &=& \frac{1}{2}\overline{L_F}[\overline{f}^{(0)}],\\ 
\overline{C}_W[\overline{f}^{(0)},\overline{f}^{(2)}] &=&
 -\frac{1}{2} \overline{C}_W[\overline{f}^{(1)},\overline{f}^{(1)}] + \frac{1}{2} \overline{L_F}[\overline{f}^{(1)}].
\end{eqnarray}
This approach lies at the base of the Hilbert expansion and Chapman-Enskog methods, see Refs.~\cite{wI63,CercignaniKremer-Book,Groot-Book} for more details. In particular, see Refs.~\cite{aGaSlG08,aSaGlG09,aGmRo20} for recent applications of the Chapman-Enskog method and Ref.~\cite{yGqX21} for recent rigorous results on the validity of the Hilbert expansion for the relativistic Vlasov-Maxwell-Boltzmann system.

We end this section by remarking that another possible scenario, relevant in plasma physics, occurs when the Knudsen number is large, in which case the Boltzmann equation is dominated by the transport part. In this limit, called the Vlasov limit, one can formally expand
\begin{equation}
\overline{f} = \overline{f}^{(0)}+\frac{1}{\text{Kn}}\overline{f}^{(1)}+\frac{1}{\text{Kn}^2}\overline{f}^{(2)}+\cdots,
\end{equation}
with $\overline{f}^{(0)}$ satisfying the collisionless Boltzmann equation $\overline{L_F}[\overline{f}^{(0)}] = 0$ and where the correction terms $\overline{f}^{(1)}$, $\overline{f}^{(2)}$ are determined by the differential equations
\begin{eqnarray}
\overline{L}_F[\overline{f}^{(1)}] &=& \overline{C}_W[\overline{f}^{(0)},\overline{f}^{(0)}],\\ 
\overline{L}_F[\overline{f}^{(2)}] &=& 2\overline{C}_W[\overline{f}^{(0)},\overline{f}^{(1)}].
\end{eqnarray}
For a review on mathematical results regarding the nonrelativistic Boltzmann equation, see~\cite{cV2002}.

\section{The method of moments}
\label{Sec:MomentMethod}

In this section we provide a brief outline regarding the method of moments. The main idea is to convert the Boltzmann equation~(\ref{Eq:BoltzmannEq}), which is an integro-differential equation for the one-particle distribution function $f$, into a system of partial differential equations for certain moment fields defined on the spacetime manifold $(M,g)$. For more details and applications we refer the reader to Refs.~\cite{Stewart-Book,CercignaniKremer-Book}, see also~\cite{Cercignani-Book69} for the Newtonian case.

Suppose $\Psi^{(A)}$, $A=0,1,2,\ldots$, is a family of (sufficiently well-behaved) functions on the future mass shell $\Gamma_m^+$. It is customary to choose polynomials in $p$, such that $\Psi^{(0)}(x,p) = 1$, $\Psi^{(1)}(x,p) = p_0$, $\Psi^{(2)}(x,p) = p_1$ etc., although for what follows below we shall not necessarily assume this particular form. We define the associated currents by
\begin{equation}
J^{(A)}_\mu(x) := \int\limits_{P_x^+(m)} \Psi^{(A)}(x,p) f(x,p) p_\mu \dvol_x(p),\qquad A = 0,1,2,\ldots
\label{Eq:MomentCurrents}
\end{equation}
Note that for the polynomial choice, $J^{(0)}_\mu = J_\mu$ agrees with the particle particle current density, see Eq.~(\ref{Eq:ParticleDensityCoord}), the $J^{(A)}_\mu$'s with $A = 1,2,\ldots,d$ yield the components of the energy-momentum-stress tensor $T_{\mu\nu}$ and for $A > d$ the $J^{(A)}_\mu$'s yield  the components of the symmetric tensor fields $T^{(s)}$ defined in Eq.~(\ref{Eq:Ts}) with $s > 2$.

Applying the identity~(\ref{Eq:DivergenceIdentityF}) to the case $s=1$ with $f$ replaced by $\Psi^{(A)} f$, one obtains
\begin{equation}
\nabla^\mu J^{(A)}_\mu(x) = \int\limits_{P_x^+(m)} L_F[\Psi^{(A)}](x,p) f(x,p) \dvol_x(p)
 + \int\limits_{P_x^+(m)} \Psi^{(A)}(x,p) L_F[f](x,p) \dvol_x(p).
\label{Eq:DivergenceLaws}
\end{equation}
The second term on the right-hand side can be re-expressed in terms of the collision integral using the identity~(\ref{Eq:IntBoltzmannIdentity}). Using Eq.~(\ref{Eq:RelationWDiffCross}) in order to write the result in terms of the differential cross section one obtains the family of divergence laws
\begin{equation}
\boxed{
\nabla^\mu J^{(A)}_\mu(x) = C_F^{(A)}(x) + C_{\text{gain}}^{(A)}(x) - C_{\text{loss}}^{(A)}(x),
\qquad A = 0,1,2,\ldots,
}
\label{Eq:BoltzmannMomentEq}
\end{equation}
where the source terms on the right-hand side consist of
\begin{equation}
C_F^{(A)}(x) := \int\limits_{P_x^+(m)} L_F[\Psi^{(A)}](x,p) f(x,p) \dvol_x(p)
\end{equation}
and the gain and loss collision terms
\begin{eqnarray}
C_{\text{gain}}^{(A)}(x) &=& \int\limits_{P_x^+(m)} \int\limits_{P_x^+(m)} \int\limits_{S^{d-1}}
m^2 g\sqrt{1 + \frac{g^2}{4}}\frac{d\sigma}{d\Omega}(g,\Theta) f(x,p_1) f(x,p_2) 
\Psi^{(A)}(x,p_1^*) d\Omega(\hat{q}^*) \dvol_x(p_1)\dvol_x(p_2),
\label{Eq:Cgain}\\
C_{\text{loss}}^{(A)}(x) &=& \int\limits_{P_x^+(m)} \int\limits_{P_x^+(m)} \int\limits_{S^{d-1}}
m^2 g\sqrt{1 + \frac{g^2}{4}}\frac{d\sigma}{d\Omega}(g,\Theta) f(x,p_1) f(x,p_2)
 \Psi^{(A)}(x,p_1) d\Omega(\hat{q}^*) \dvol_x(p_1)\dvol_x(p_2).
\label{Eq:Closs}
\end{eqnarray}
For the aforementioned polynomial choice, the first $n(n+3)/2$ of these equations are equivalent to
\begin{eqnarray}
\nabla^\mu J_\mu &=& 0,
\label{Eq:DivergenceEq0}\\
\nabla^\mu T_{\mu\nu} &=& q F_\nu{}^\mu J_\mu,
\label{Eq:DivergenceEq1}\\
\nabla^\mu S_{\mu\alpha\beta} &=& 2q F_{(\alpha}{}^\mu T_{\beta)\mu}
 + C^{\text{gain}}_{\alpha\beta} - C^{\text{loss}}_{\alpha\beta},
\label{Eq:DivergenceEq2}
\end{eqnarray}
where here $S_{\mu\alpha\beta} = T^{(3)}_{\mu\alpha\beta}$ and $C^{\text{gain,loss}}_{\alpha\beta}$ are given by the same expressions as in Eqs.~(\ref{Eq:Cgain},\ref{Eq:Closs}) with $\Psi^{(A)}(x,p)$ replaced by $p_\alpha p_\beta$. Note that both sides of Eq.~(\ref{Eq:DivergenceEq2}) are symmetric and trace-free in $\alpha\beta$ (by virtue of Eq.~(\ref{Eq:DivergenceEq0})), which implies that the latter gives rise to $n(n+1)/2 - 1$ independent equations.

To convert Eq.~(\ref{Eq:BoltzmannMomentEq}) into a closed system of equations one may proceed as follows: let $u\in {\cal X}(M)$ be a future-directed unit timelike vector field on $M$ and let $f^{(0)}\in {\cal F}(\Gamma_m^+)$ be a given, strictly positive (reference) distribution function. Consider for each $x\in M$ the Hilbert space ${\cal H}_x := L^2(P_x^+(m), f^{(0)}[-p(u)]\dvol_x(p) )$ of square-integrable functions on the future mass shell $P_x^+(m)$ with measure $f^{(0)}[-p(u)]\dvol_x(p)$ and associated scalar product
\begin{equation}
(f_1,f_2)_x := \int\limits_{P_x^+(m)} f_1(p) f_2(p) f^{(0)}(x,p) [-p(u)] \dvol_x(p),\qquad
f_1,f_2\in {\cal H}_x.
\end{equation}
Assume that for each $x\in M$, the functions $\Psi^{(A)}(x,\cdot)$, $A=0,1,2,\ldots$, form a linearly independent and complete set in ${\cal H}_x$, and consider the associated orthonormal basis $\hat{\Psi}^{(A)}(x,\cdot)$, $A=0,1,2,\ldots$, which is constructed from this family using the Gram-Schmidt procedure, such that
\begin{equation}
\hat{\Psi}^{(A)}(x,p) = \sum\limits_{B=0}^A c_{AB}(x)\Psi^{(B)}(x,p),\qquad A = 0,1,2,\ldots,
\end{equation}
for some coefficients $c_{AB}(x)$, which are uniquely determined by the orthonormality condition $(\hat{\Psi}^{(A)}(x,\cdot), \hat{\Psi}^{(B)}(x,\cdot))_x = \delta^{AB}$. It follows that any function $h\in {\cal H}_x$ can be expanded according to
\begin{equation}
h(p) = \sum\limits_{A=0}^\infty h_A\hat{\Psi}^{(A)}(x,p),
\end{equation}
with coefficients
\begin{equation}
h_A = (\hat{\Psi}^{(A)}(x,\cdot), h)_x = \sum\limits_{B=0}^A c_{AB}(x)
\int\limits_{P_x^+(m)} \Psi^{(B)}(x,p) h(p) f^{(0)}(x,p) [-p(u)] \dvol_x(p). 
\end{equation}
In particular for $h(p) = f(x,p)/f^{(0)}(x,p)$ one obtains
\begin{equation}
\boxed{
f(x,p) = f^{(0)}(x,p)\sum\limits_{A=0}^\infty h_A(x)\hat{\Psi}^{(A)}(x,p),\qquad
h_A(x) = \sum\limits_{B=0}^A c_{AB}(x)\Pi^{(B)}(x),
}
\label{Eq:fMomentExpansion}
\end{equation}
with the moments
\begin{equation}
\boxed{
\Pi^{(A)}(x) := -u^\mu J_\mu^{(A)} = \int\limits_{P_x^+(m)} \Psi^{(A)}(x,p) f(x,p) [-p(u)] \dvol_x(p),
\qquad A = 0,1,2,\ldots}
\label{Eq:DefMoments}
\end{equation}
Eqs.~(\ref{Eq:fMomentExpansion},\ref{Eq:DefMoments}) allow one to formally express all the quantities appearing in Eq.~(\ref{Eq:BoltzmannMomentEq}) in terms of the moments $\Pi^{(A)}(x)$, which yields an infinite set of divergence laws for these moments. By suitably truncating the series in Eq.~(\ref{Eq:fMomentExpansion}) and restricting oneself to a finite number $A = 0,1,2,...,N-1$ of divergence laws~(\ref{Eq:BoltzmannMomentEq}), one obtains a system of $N$ partial differential equations for a finite number of moments which can be analyzed and (hopefully!) approximates the infinite system in an appropriate sense. Of course, apart from the choice for the functions $\Psi^{(A)}$, the question is how to choose the future-directed unit timelike vector field $u$ and the reference distribution function $f^{(0)}$. When studying near-equilibrium gas configurations it makes sense to choose $f^{(0)}$ and $u$ as in Eq.~(\ref{Eq:LocalEquilibriumBis}) to represent a local equilibrium distribution function. A simple example is discussed in the next subsection and a few comments on more sophisticated examples are made in the subsequent one. An application of the method of moments to a cosmology setting is discussed in the next section.

\subsection{Local equilibrium fluids}

The simplest truncation corresponds to the system formed by the first $n+1$ divergence laws~(\ref{Eq:DivergenceEq0}) and (\ref{Eq:DivergenceEq1}), where one sets $f = f^{(0)}$ and $u$ as in the local equilibrium distribution function in~(\ref{Eq:LocalEquilibriumBis}). This function depends on the $n+1$ unknowns $T$, $\mu$ and $u$, and the associated current density and energy-momentum-stress tensor are given by the perfect fluid expressions~(\ref{Eq:PerfectFluid}) in which $n$, $h$ and $P$ are the functions of $(T,\mu)$ obtained from Eqs.~(\ref{Eq:PerfectFluidnu},\ref{Eq:PerfectFluidhp}) by substituting $\alpha = A^{-1}\exp(\mu/k_B T)$. According to Eq.~(\ref{Eq:EntropyPerParticle}), the entropy per particle is
\begin{equation}
s = -k_B + \frac{h - \mu}{T},
\end{equation}
and satisfies the first law of thermodynamics~(\ref{Eq:FirstLawOfThermodynamics}). Eqs.~(\ref{Eq:DivergenceEq0},\ref{Eq:DivergenceEq1}) yield the following system
\begin{eqnarray}
&& \nabla_u s = 0,
\label{Eq:EntropyConservation}\\
&& \nabla_u n = -\theta n,
\label{Eq:Continuity}\\
&& h\nabla_u u_\nu = -D_\nu h + T D_\nu s + q F_\nu{}^\beta u_\beta,
\label{Eq:Euler}
\end{eqnarray}
where here $\theta := \nabla_\mu u^\mu$ refers to the expansion of the fluid flow and $D_\mu := \nabla_\mu + u_\mu u^\nu\nabla_\nu$ is the gradient operator projected onto the space orthogonal to $u$. Together with Maxwell's equations $\nabla_{[\mu} F_{\alpha\beta]} = 0$ and $\nabla_\mu F^{\mu\nu} = -q n u^\nu$ this yields an evolution system for the fields $(s,n,u_\nu,F_{\alpha\beta})$ in which the function $h$ and the temperature $T$ can be obtained by inverting the map $(T,\mu)\mapsto (n,s)$ in order to express $T$ in terms of the particle density $n$ and entropy per particle $s$. Note that Eq.~(\ref{Eq:EntropyConservation}) together with the continuity equation~(\ref{Eq:Continuity}) imply entropy conservation $\nabla_\mu S^\mu = \nabla_\mu(s n u^\mu) = 0$, while Eq.~(\ref{Eq:Euler}) describes the relativistic Euler equation, the last term on the right-hand side expressing the Lorentz acceleration exerted by the electromagnetic field $F$ on the fluid. This system describes the equations of motion for an ideal charged relativistic fluid (without conductivity) in local thermodynamic equilibrium. Recall that (unless a global timelike Killing vector field exists) the solutions of these equations do not correspond to an exact solution of the Maxwell-Vlasov equation since the underlying distribution function is only a local equilibrium function.

\subsection{A few comments on non-equilibrium fluids}

A method for describing near-equilibrium configurations is based on the following representation:
\begin{equation}
f(x,p) = f^{(0)}(x,p)\left[ 1 + a^\alpha(x) p_\alpha + a^{\alpha\beta}(x) p_\alpha p_\beta \right],
\label{Eq:GradDistributionFunction}
\end{equation}
which expands the distribution function $f$ as a second-order polynomial in the moment covector $p$ about the local equilibrium distribution function $f^{(0)}$ of Eq.~(\ref{Eq:LocalEquilibriumBis}). Here, the coefficients $a^\alpha(x)$ and $a^{\alpha\beta}(x) = a^{\beta\alpha}(x)$ can be determined from the moments $J_\mu$ and $T_{\mu\nu}$ using the method described above or similar methods. The functions $T$, $\mu$ and $u$ parametrizing $f^{(0)}$ need to be determined by appropriate matching conditions which relate them to the corresponding quantities computed from $f$. One obtains the equations of motion from Eqs.~(\ref{Eq:DivergenceEq0},\ref{Eq:DivergenceEq1},\ref{Eq:DivergenceEq2}), neglecting the higher-order moment equations. In particular, Grad's method is based on a choice of the form~(\ref{Eq:GradDistributionFunction}) for the distribution function. For more details, see chapter~VII in~\cite{Groot-Book}, chapter~6 in~\cite{CercignaniKremer-Book} and Ref.~\cite{Stewart-Book}.

The method of moment plays an important role for the motivation of extended thermodynamics~\cite{MullerRuggeri-Book}, divergence-type fluids~\cite{iLiMtR86,sP87,rGlL90,gNoR95,oRgN97} and other related fluid theories~\cite{rGlL91}. In particular, divergence-type fluids (in the absence of the electromagnetic field) consider the system of the form~(\ref{Eq:DivergenceEq0},\ref{Eq:DivergenceEq1},\ref{Eq:DivergenceEq2}) augmented by the entropy law $\nabla_\mu S^\mu = \sigma\geq 0$, and assume that all the involved quantities depend only algebraically on $J_\mu$ and $T_{\mu\nu}$. One can show that such theories are characterized by a single scalar generating function depending on $n(n+3)/2$ variables, from which the moments $J_\mu$, $T_{\mu\nu}$ and $S_{\mu\alpha\beta}$, the entropy current $S_\mu$ and its divergence $\sigma$ can be constructed. By analyzing the properties of this generating function one can obtain theories for which the propagation is hyperbolic and causal. For a recent application to conformal dissipative fluids, see~\cite{lLoRmR18}.

\section{Application to homogeneous and isotropic universes}
\label{Sec:Application}

In this section we apply the method of moments to the Einstein-Boltzmann system of equations
\begin{eqnarray}
L[f] &=& C_W[f,f],
\label{Eq:BoltzmannFLRW}\\
G_{\mu\nu} + \Lambda g_{\mu\nu} &=& \kappa T_{\mu\nu},
\label{Eq:EinsteinFLRW}
\end{eqnarray}
where $\Lambda$ denotes the cosmological constant and $\kappa$ is the gravitational coupling constant. To simplify the problem, we restrict ourselves to a simple, uncharged gas configuration which is homogeneous and isotropic, such that the metric describes a Friedmann-Lema\^itre-Robertson-Walker (FLRW) spacetime and the electromagnetic field vanishes. Since we only consider one species of particles, the system we consider does not provide a realistic model for the large-scale description of our universe and only constitutes a toy model. However, as we will see, this simplified model already exhibits interesting phenomena. For recent studies regarding the properties of the solutions of this system, see also~\cite{kT03,nNeT06,hLaR13,hL13,dBgDuHmMjN16,dBgDuHmMjN16b,hLeN18,hLeNpT19}.

A FLRW spacetime manifold is of the form $M = (0,\infty)\times S_k$ with metric
\begin{equation}
g = -dt^2 + a(t)^2\frac{\delta_{ij} dx^i dx^j}{\left( 1 + \frac{k}{4}|\ve{x}|^2 \right)^2},
\qquad
|\ve{x}| := \sqrt{\delta_{ij} x^i x^j},
\label{Eq:FLRWMetric}
\end{equation}
where $a(t)$ is the scale factor and $S_k$ is either hyperbolic space $H^3$, Euclidean space $E^3$ or spherical space $S^3$ depending on the (normalized) value of the constant spatial curvature $k=-1$, $0$ or $1$. The metric~(\ref{Eq:FLRWMetric}) is spatially homogeneous and isotropic, and a natural orthonormal basis of vector fields is given by
\begin{equation}
e_{\hat{0}} := \frac{\partial}{\partial t},\quad
e_{\hat{1}} := \frac{1 + \frac{k}{4}|\ve{x}|^2}{a(t)}\frac{\partial}{\partial x^1},\quad
e_{\hat{2}} := \frac{1 + \frac{k}{4}|\ve{x}|^2}{a(t)}\frac{\partial}{\partial x^2},\quad
e_{\hat{3}} := \frac{1 + \frac{k}{4}|\ve{x}|^2}{a(t)}\frac{\partial}{\partial x^3}.
\end{equation}
Using the tools described in appendix~\ref{App:Symmetries} a systematic analysis~\cite{fAoStZ14} reveals that the most general spatially homogeneous and isotropic distribution function $f$ on a FLRW spacetime must be of the form
\begin{equation}
f(x,p) = F(t,\mathcal{C}),
\label{Eq:FLRWDF}
\end{equation}
for some sufficiently smooth function $F$ of the two variables $t$ and $\mathcal{C}$, where $\mathcal{C}$ is the following integral of motion:
\begin{equation}
\mathcal{C}(x,p) := a(t)|\ve{p}|,\qquad
|\ve{p}| := \sqrt{p_{\hat{1}}^2 + p_{\hat{2}}^2 + p_{\hat{3}}^2}.
\label{Eq:CDef}
\end{equation}
By noting that the the free-particle Hamiltonian defined in Eq.~(\ref{Eq:FreeParticleH}) can be written as
\begin{equation}
\mathcal{H}(x,p) = \frac{1}{2} \left(-p_t^2 + \frac{\mathcal{C}^2}{a(t)^2}\right),
\end{equation}
a straightforward calculation reveals that $\{ \mathcal{H}, \mathcal{C} \} = 0$ which implies that $L[\mathcal{C}] = 0$. As a consequence, the Boltzmann equation~(\ref{Eq:BoltzmannFLRW}) simplifies to\footnote{It follows from the results in appendix~\ref{App:Symmetries} that the homogeneity and isotropy conditions can be consistently imposed on the distribution function, provided the transition probability density $W$ is a function of the Mandelstam variables only. This will in fact be assumed further below.}
\begin{equation}
p^t\frac{\partial F}{\partial t} = C_W[f,f].
\label{Eq:LfIso}
\end{equation}
Einstein's field equations~(\ref{Eq:EinsteinFLRW}) for the FLRW model reduce to the well-known Friedmann equations (see, for instance~\cite{Carroll-Book})
\begin{eqnarray}
\frac{3}{a^2}\left( \frac{da}{dt} \right)^2 &=& -\frac{3k}{a^2} + \Lambda + 8\pi\varepsilon,
\label{Eq:Friedmann1}\\
\frac{3}{a}\frac{d^2 a}{dt^2} &=& \Lambda - 4\pi(\varepsilon + 3P),
\label{Eq:Friedmann2}
\end{eqnarray}
where the energy density $\varepsilon$ and pressure $P$ are given by
\begin{eqnarray}
\varepsilon(t) &=& \int\limits_{P_x^+(m)} f(x,p) [-p(u)]^2\dvol_x(p)
 = 4\pi\int\limits_0^\infty F(t,a(t) p) \sqrt{m^2 + p^2} p^2 dp,
\label{Eq:FLRWepsilon}\\
P(t) &=& \frac{1}{3}\int\limits_{P_x^+(m)} f(x,p) [g^{-1} + u\otimes u](p,p) \dvol_x(p)
 = \frac{4\pi}{3}\int\limits_0^\infty F(t,a(t) p) \frac{p^4 dp}{\sqrt{m^2 + p^2}},
\label{Eq:FLRWpressure}
\end{eqnarray}
with $u = e_{\hat{0}}$ the four-velocity of the isotropic observers. It follows from Boltzmann's equation~(\ref{Eq:BoltzmannFLRW}) and the vanishing of the electromagnetic field tensor that $T_{\mu\nu}$ is divergence-free (see Eq.~(\ref{Eq:DivTBoltzmann})), which yields the following equation:
\begin{equation}
\frac{d\varepsilon}{dt} + 3H (\varepsilon + P) = 0,
\label{Eq:FLRWEuler}
\end{equation}
with the Hubble parameter $H(t) := a(t)^{-1} da(t)/dt$. Hence, the spatially homogeneous and isotropic Einstein-Boltzmann system consists of the evolution equation~(\ref{Eq:LfIso}) for the function $F(t,\mathcal{C})$ which is coupled to the evolution equation~(\ref{Eq:Friedmann2}) for the scale factor $a(t)$, where $\varepsilon$ and $P$ are given by~(\ref{Eq:FLRWepsilon},\ref{Eq:FLRWpressure}) and Eq.~(\ref{Eq:Friedmann1}) acts as a constraint. The collisionless case, which corresponds to $C_W = 0$, has been analyzed in Ref.~\cite{fAoStZ14} and references therein. In the absence of collisions the qualitative features of the evolution of the universe can be easily understood. As discussed in~\cite{fAoStZ14}, in the early epoch ($a(t)\to 0$) the kinetic gas behaves as a gas of massless particles with equation of state $\displaystyle P(t) \approx \frac{1}{3}\varepsilon(t)$, while at the late epoch ($a(t)\to \infty$) the energy density dominates the pressure and the gas behaves more and more like dust.

Below, we shall use the method of moments (without truncation) described in the previous section in order to reformulate Eq.~(\ref{Eq:LfIso}) in terms of an evolution system for suitable moments and analyze the behavior of the solutions in the early and late epochs. One interesting question we would like to address in this section is whether or not the gas settles down to a local equilibrium configuration, see Eq.~(\ref{Eq:LocalEquilibriumBis}). Since the gas is isotropic and spatially homogeneous, its four-velocity $u$ coincides with the four-velocity of the isotropic observers, and the local equilibrium distribution function must have the form
\begin{equation}
f^{(0)}(x,p) = \frac{1}{A}\exp\left[ \frac{p_t + \mu(t)}{k_B T(t)} \right]
 = \frac{1}{A}\exp\left[ -\frac{\sqrt{m^2 + |\ve{p}|^2} - \mu(t)}{k_B T(t)} \right],
\label{Eq:FLRWLocalEquilibrium}
\end{equation}
depending on the functions $T > 0$ and $\mu$ of $t$. As discussed in section~\ref{SubSec:GlobalEquilibrium}, there is no \emph{global} equilibrium, since the FLRW metric~(\ref{Eq:FLRWMetric}) does not possess any timelike Killing vector fields unless the scale factor $a(t)$ is constant. However, one might still ask whether or not the distribution approaches a \emph{local} equilibrium function of the form~(\ref{Eq:FLRWLocalEquilibrium}). This question will be analyzed towards the end of this section in the limits $a(t)\to 0$ and $a(t)\to \infty$.

\subsection{Recasting the problem in terms of the moments}

As mentioned above, we use the method of moments outlined in section~\ref{Sec:MomentMethod} in order to rewrite the Boltzmann equation~(\ref{Eq:LfIso}) as an infinite system of evolution equation for the moments
\begin{equation}
\Pi_s(t) = \int\limits_{P_x^+(m)} \Psi_s(x,p) f(x,p) [-p(u)] \dvol_x(p)
 = 4\pi\int\limits_0^\infty \psi_s(t,a(t) p) F(t,a(t) p) p^2 dp,
\label{Eq:FLRWDefMoments}
\end{equation}
see Eq.~(\ref{Eq:DefMoments}), where for convenience we shall use the index $s$ instead of the superscript $(A)$ in order to label the moments. We find it convenient to choose the following basis functions
\begin{equation}
\Psi_s(x,p)  = \psi_s(t,\mathcal{C}) 
 = \left. \left( \sqrt{m^2 + |\ve{p}|^2} - m \right)^s \right|_{|\ve{p}| = \frac{\mathcal{C}}{a(t)}},
 \qquad s = 0,1,2,\ldots
\end{equation}
Note that $\psi_s(t,\mathcal{C}) \to |\ve{p}|^s$ in the ultrarelativistic limit, while $\psi_s(t,\mathcal{C}) \to \left( \frac{|\ve{p}|^2}{2m} \right)^s$ in the nonrelativistic limit. Further, note that $\Pi_0 = n$ is the mean particle density and $\Pi_1 = \varepsilon - m n$ the internal energy density of the gas configuration. Using the fact that
\begin{equation}
L[\Psi_s] = p^t\frac{\partial}{\partial t}\psi_s 
 = -s H\left( \sqrt{m^2 + |\ve{p}|^2} + m \right)\psi_s,
\end{equation}
and that (by isotropy) the currents defined in Eq.~(\ref{Eq:MomentCurrents}) are equal to $u_\mu\Pi_s(t)$, the divergence laws~(\ref{Eq:DivergenceLaws}) yield
\begin{equation}
\frac{d}{dt} \Pi_s(t) +  (3 + s)H(t)\Pi_s(t) + 4\pi s H(t)\int\limits_0^\infty \psi_s(t,a(t)p) F(t,a(t)p) 
\frac{m p^2 dp}{\sqrt{m^2 + p^2}}
 = C_{\text{gain}}^{(s)}(t) - C_{\text{loss}}^{(s)}(t),\qquad s = 0,1,2,\ldots,
\label{Eq:MomentsEvolution}
\end{equation}
where the gain and loss terms can be written as
\begin{equation}
C_{\text{gain,loss}}^{(s)}(t) = 16\pi^2
\int\limits_0^\infty \int\limits_0^\infty K_{\text{gain,loss}}^{(s)}(p_1,p_2)
F(t,a(t)p_1) F(t,a(t) p_2) p_1^2 p_2^2 dp_1 dp_2,
\label{Eq:Cgainloss}
\end{equation}
with the kernels
\begin{eqnarray}
K_{\text{gain}}^{(s)}(p_1,p_2) &:=&
\frac{1}{2\sqrt{m^2 + p_1^2}\sqrt{m^2 + p_2^2}}
\int\limits_0^\pi d\vartheta\sin\vartheta m^2 g\sqrt{1 + \frac{g^2}{4}}
\int\limits_{S^2} d\Omega(\ve{\hat q}^*) \frac{d\sigma}{d\Omega}(g,\Theta)
\psi_s(t,a(t)|\ve{p}_1^*|),
\label{Eq:GainKernel}\\
K_{\text{loss}}^{(s)}(p_1,p_2) &:=&
\frac{1}{2\sqrt{m^2 + p_1^2}\sqrt{m^2 + p_2^2}}
\int\limits_0^\pi d\vartheta\sin\vartheta m^2 g\sqrt{1 + \frac{g^2}{4}}
\int\limits_{S^2} d\Omega(\ve{\hat q}^*) \frac{d\sigma}{d\Omega}(g,\Theta) \psi_s(t,a(t)p_1).
\label{Eq:LossKernel}
\end{eqnarray}
Herein, $\vartheta$ denotes the angle between $\ve{p}_1$ and $\ve{p}_2$ and $g$ is the function of $p_1$, $p_2$ and $\vartheta$ determined by Eq.~(\ref{Eq:g}) and similarly, $|\ve{p}_1^*|$ is the function of $p_1$, $p_2$, $\vartheta$ and $\ve{\hat q}^*$ which can be determined by Eq.~(\ref{Eq:Cp1*Vec}) or the zero component of Eq.~(\ref{Eq:Cp1BisBis}) and the definition of $p^{cm}$ in Eq.~(\ref{Eq:pCM}). It is worthwhile observing that for $s = 0$ and $s = 1$ the gain and loss terms cancel each other since in these cases $\psi_s$ are collision invariants. For $s = 0$, Eq.~(\ref{Eq:MomentsEvolution}) reduces to the continuity equation
\begin{equation}
\frac{dn}{dt} + 3H n = 0,
\label{Eq:FLRWContinuity}
\end{equation}
while for $s = 1$, Eqs.~(\ref{Eq:MomentsEvolution},\ref{Eq:FLRWContinuity}) yield Eq.~(\ref{Eq:FLRWEuler}). In order to obtain a closed system, one needs to express the function $F$ appearing in Eqs.~(\ref{Eq:MomentsEvolution},\ref{Eq:Cgainloss}) and the expression for the pressure in Eq.~(\ref{Eq:FLRWpressure}) in terms of the moments $\Pi_s$. This can be achieved, in principle, using the method described in the previous section, where in the present case it is natural to choose $f^{(0)}$ to be a local equilibrium distribution function of the form~(\ref{Eq:FLRWLocalEquilibrium}), that is $f^{(0)}(x,p) = F^{(0)}(t,\mathcal{C})$ with
\begin{equation}
F^{(0)}(t,\mathcal{C}) = \lambda(t) e^{-\frac{\psi_1(t,\mathcal{C})}{k_B T(t)}},\qquad
\lambda(t) = \frac{1}{A} e^{\frac{\mu(t) - m}{k_B T(t)}}.
\label{Eq:F0Def}
\end{equation}
The associated moments can be computed from (cf. section~\ref{SubSec:Juttner})
\begin{equation}
\Pi_s^{(0)}(t) = \lambda(t)
\left. \left( -\frac{d}{d\beta} \right)^s \zeta(\beta) \right|_{\beta = \frac{1}{k_B T(t)}},
\label{Eq:FLRWMoments0}
\end{equation}
with the generating function
\begin{equation}
\zeta(\beta) := \int\limits_{P_x^+(m)} e^{-\beta\psi_1(t,\mathcal{C})} [-p(u)] \dvol_x(p)
 = \frac{4\pi m^2}{\beta} e^{m\beta} K_2(m\beta).
\label{Eq:FLRWGenerator}
\end{equation}
The functions $\lambda(t)$ and $T(t)$ in Eq.~(\ref{Eq:F0Def}) can be fixed by imposing the matching conditions $\Pi_s^{(0)} = \Pi_s$ for the first two moments $s=0$ and $s=1$, which is equivalent to the requirement that $n^{(0)} = n$ and $\varepsilon^{(0)} = \varepsilon$.

In the following subsections, after deriving more explicit expressions for the collision kernels~(\ref{Eq:GainKernel},\ref{Eq:LossKernel}) in the case of a constant cross section, we shall analyze the moment equations in the early and late epochs, assuming a differential cross section with a ``hard-sphere" behavior at high energies while describing Maxwell particles at low energies, such that
\begin{equation}
\lim\limits_{g\to\infty} \frac{d\sigma}{d\Omega}(g,\Theta) = \frac{\sigma_T}{4\pi},\qquad
\lim\limits_{g\to 0} g\frac{d\sigma}{d\Omega}(g,\Theta) 
 = \frac{\sigma_0}{4\pi}\Gamma(\Theta),
\label{Eq:FLRWDiffCrossSec}
\end{equation}
with constants $\sigma_T,\sigma_0 > 0$ and $\Gamma(\Theta)$ a positive function of the scattering angle $\Theta$ satisfying $\Gamma(\pi - \Theta) = \Gamma(\Theta)$ and the normalization condition
\begin{equation}
\int\limits_0^\pi \Gamma(\Theta) \sin\Theta d\Theta = 2.
\label{Eq:GammaNormalization}
\end{equation}

\subsection{More explicit expressions for the collision kernels for a constant cross section}

For a constant cross section, such that $\displaystyle \frac{d\sigma}{d\Omega} = \frac{\sigma_T}{4\pi}$ with constant total cross section $\sigma_T > 0$, the collision kernels~(\ref{Eq:GainKernel},\ref{Eq:LossKernel}) can be simplified. We start with the computation of the loss term which is simpler since the integral over $\ve{\hat q}^*$ simply gives a factor $4\pi$. Using Eq.~(\ref{Eq:g}) one finds
\begin{equation}
g = \frac{\sqrt{2}}{m}\sqrt{\ell - m^2},\qquad
\sqrt{1 + \frac{g^2}{4}} = \frac{1}{\sqrt{2} m}\sqrt{ \ell + m^2},\qquad
 \ell :=\sqrt{m^2 + p_1^2}\sqrt{m^2 + p_2^2} - p_1 p_2\cos\vartheta,
\label{Eq:gsg}
\end{equation}
such that
\begin{eqnarray}
K_{\text{loss}}^{(s)}(p_1,p_2) 
 &=& \frac{\sigma_T}{2}\frac{(E_1 - m)^s}{E_1 E_2}
\int\limits_0^\pi d\vartheta\sin\vartheta \sqrt{\ell^2 - m^4}
\nonumber\\
 &=& \frac{\sigma_T (E_1 - m)^s}{4p_1 p_2 E_1 E_2}
 \left[ L_+\sqrt{L_+^2 - m^4} - L_-\sqrt{L_-^2 - m^4}
 - m^4\log\left( \frac{L_+ + \sqrt{L_+^2 - m^4}}{L_- + \sqrt{L_-^2 - m^4}} \right) \right],
\label{Eq:Kloss}
\end{eqnarray}
where we have set $L_\pm := E_1 E_2 \pm p_1 p_2$ and abbreviated $E_j := \sqrt{m^2 + p_j^2}$ for $j=1,2$. When $m > 0$ it is convenient to reparametrize $p_1 = m\sinh\chi_1$, $p_2 = m\sinh\chi_2$ in terms of the hyperbolic angles $\chi_1$ and $\chi_2$, such that $L_\pm = m^2\cosh(\chi_1\pm \chi_2)$, which yields
\begin{equation}
K_{\text{loss}}^{(s)}(p_1,p_2) = \frac{\sigma_T \left[ 2m\sinh\left( \frac{\chi_1}{2} \right) \right]^s}{2\sinh(2\chi_1)\sinh(2\chi_2)}
\left[ \sinh(2\chi_1+2\chi_2) - \sinh|2\chi_1-2\chi_2| - (2\chi_1+2\chi_2) + |2\chi_1-2\chi_2| \right].
\end{equation}

The gain term is more difficult to deal with, since the integral over $\ve{\hat q}^*$ involves the quantity $|\ve{p}_1^*|$ which depends on $\ve{\hat q}^*$. From the zero component of Eq.~(\ref{Eq:Cp1BisBis}) one finds
\begin{equation}
\sqrt{m^2 + |\ve{p}_1^*|^2} = \sqrt{m^2 + |\ve{p}^{cm}|^2}\sqrt{1 + \frac{1}{4} g^2}
 - \frac{g}{2}\ve{p}^{cm}\cdot\ve{\hat q}^* = A - B\cos\phi,
\label{Eq:sqrtmmp1*}
\end{equation}
where $\phi$ denotes the angle between $\ve{p}^{cm}$ and $\ve{\hat q}^*$ and where
\begin{equation}
A := \frac{E_1 + E_2}{2},\qquad
B := \frac{1}{2}\sqrt{\frac{\ell-m^2}{\ell+m^2}} \sqrt{p_1^2 + p_2^2 + 2p_1 p_2\cos\vartheta}.
\end{equation}
In deriving these expressions, Eqs.~(\ref{Eq:pCM}) and (\ref{Eq:gsg}) have been used. By means of the variable substitutions $y := \sqrt{p_1^2 + p_2^2 + 2p_1 p_2\cos\vartheta}/(2A)$ and $x := 1 - (A - m)^{-1} B\cos\phi$, one finds
\begin{equation}
K_{\text{gain}}^{(s)}(p_1,p_2) 
 = \frac{2\sigma_T A^3 (A - m)^{s+1}}{p_1 p_2 E_1 E_2}
\int\limits_{y_-}^{y_+} dy(1-y^2)\int\limits_{1 - z(y)}^{1 + z(y)} dx x^s,
\end{equation}
with $y_- := |p_2 - p_1|/(2A)$, $y_+ := (p_1 + p_2)/(2A)$ and
\begin{equation}
z(y) := \frac{A y}{A-m} \sqrt{1 - \frac{m^2}{A^2}\frac{1}{1-y^2}}.
\end{equation}
The inner integral can be computed explicitly, and one obtains
\begin{equation}
K_{\text{gain}}^{(s)}(p_1,p_2) 
 = \frac{2\sigma_T A^3 (A - m)^{s+1}}{p_1 p_2 E_1 E_2}
\int\limits_{y_-}^{y_+} (1-y^2) \frac{[1 + z(y)]^{s+1} - [1 - z(y)]^{s+1}}{s+1} dy,\qquad
s = 0,1,2,\ldots
\label{Eq:Kgain}
\end{equation}
The remaining integral can be computed explicitly in the limit $m = 0$ as we show next.

\subsection{The massless case}

When $m = 0$, Eq.~(\ref{Eq:Kloss}) reduces to
\begin{equation}
K_{\text{loss}}^{(s)}(p_1,p_2) = \sigma_T p_1^s.
\end{equation}
Furthermore, taking into account that $A = (p_1 + p_2)/2$, Eq.~(\ref{Eq:Kgain}) simplifies to
\begin{equation}
K_{\text{gain}}^{(s)}(p_1,p_2) 
 = \frac{2\sigma_T A^{s+4}}{p_1^2 p_2^2}
\int\limits_{\kappa}^1 dy (1-y^2)\frac{(1+y)^{s+1} - (1-y)^{s+1} }{s+1},
\end{equation}
with $\kappa := |p_2 - p_1|/(p_1 + p_2)$. The integral is elemental and yields
\begin{equation}
\int\limits_{\kappa}^1 dy (1-y^2)\frac{(1+y)^{s+1} - (1-y)^{s+1} }{(s+1)}
 = \frac{2^{s+4} - 2(s+4)\left[ (1+\kappa)^{s+3} + (1-\kappa)^{s+3} \right]
  + (s+3)\left[ (1+\kappa)^{s+4} + (1-\kappa)^{s+4}\right]}{(s+1)(s+3)(s+4)}.
\end{equation}
Combining this with the observation that $1 + \kappa = 2p_2/(p_1 + p_2)$ and $1 - \kappa = 2p_1/(p_1 + p_2)$ for $p_1\leq p_2$ and vice-versa for $p_1\geq p_2$, and using $A = (p_1+p_2)/2$ one finds
\begin{equation}
K_{\text{gain}}^{(s)}(p_1,p_2) 
 = \frac{2\sigma_T}{(s+1)(s+3)(s+4)}\frac{ (p_1 + p_2)^{s+4} - p_1^{s+4} - p_2^{s+4}
  - (s+4)\left[ p_1^{s+3} p_2 - p_2^{s+3} p_1 \right] }{p_1^2 p_2^2}.
\end{equation}
With the help of the binomial expansion this can be rewritten as
\begin{equation}
K_{\text{gain}}^{(s)}(p_1,p_2) 
 = 2\sigma_T\sum\limits_{r=0}^s \frac{s! (s+2)}{(s+2-r)! (r+2)!} p_1^{s-r} p_2^r.
\end{equation}

Therefore, in the massless case, Eq.~(\ref{Eq:MomentsEvolution}) yields
\begin{equation}
\frac{d}{dt} \Pi_s(t) + (s+3) H(t) \Pi_s(t)
 = 2\sigma_T\sum\limits_{r=0}^s \frac{s! (s+2)}{(s+2-r)! (r+2)!} \Pi_{s-r}(t) \Pi_{r}(t)
  - \sigma_T \Pi_{s}(t) \Pi_0(t),
\label{Eq:MomentsEvolutionMassless}
\end{equation}
for $s = 0,1,2,\ldots$, which agrees precisely with Eq.~(23) in Ref.~\cite{dBgDuHmMjN16b}. Note that the right-hand side is zero for $s=0$ and $s=1$, which reflects the fact that these moments correspond to collision invariants, as discussed previously. In Refs.~\cite{dBgDuHmMjN16,dBgDuHmMjN16b} these results were applied to the analysis for the propagation of a photon gas in an expanding FLRW universe. In the next subsection, we show that Eq.~(\ref{Eq:MomentsEvolutionMassless}) also holds approximately for a gas of massive particles in the early universe.

\subsection{The early epoch}
\label{Subsec:Early_epoch}

As we have already mentioned previously, in the limit $a(t)\to 0$ the mass term appearing on the right-hand sides Eqs.~(\ref{Eq:FLRWepsilon},\ref{Eq:FLRWpressure}) can be neglected and the gas behaves as a gas of massless particles. This means that most particles have a large momentum $|\ve{p}| = \mathcal{C}/a(t) \gg m$ (i.e. $\Pi_2(t) \gg m^2\Pi_0(t)$), and hence one can replace the collision kernels $K_{\text{gain,loss}}^{(j)}(p_1,p_2)$ with their corresponding expressions in the massless case computed in the previous subsection. Hence, in this limit one obtains again the evolution equation~(\ref{Eq:MomentsEvolutionMassless}) for the moments.

For the following, it is convenient to compare the moments' time evolution with those associated with the local equilibrium distribution function $F^{(0)}$ defined in Eq.~(\ref{Eq:F0Def}). From Eqs.~(\ref{Eq:FLRWMoments0},\ref{Eq:FLRWGenerator}) one obtains $\zeta(\beta) = 8\pi\beta^{-3}$ and hence
\begin{equation}
\Pi_s^{(0)}(t) = 4\pi(s+2)! \lambda(t) \left[ k_B T(t) \right]^{s+3},
\end{equation}
such that
\begin{equation}
n^{(0)}(t) = \Pi_0^{(0)}(t) = 8\pi\lambda(t)[ k_B T(t) ]^3,\qquad
\varepsilon^{(0)}(t) = \Pi_1^{(0)}(t) = 24\pi\lambda(t)[ k_B T(t) ]^4.
\label{Eq:ntepsilont0}
\end{equation}
On the other hand, it follows from the continuity equation~(\ref{Eq:FLRWContinuity}) and the Euler equation~(\ref{Eq:FLRWEuler}) with $P = \varepsilon/3$ that
\begin{equation}
n(t) = n_i\frac{a_i^3}{a(t)^3},\qquad
\varepsilon(t) = \varepsilon_i\frac{a_i^4}{a(t)^4},
\label{Eq:ntepsilont}
\end{equation}
where the index $i$ refers to the quantity evaluated at a given initial time $t = t_i > 0$. Imposing the matching conditions $n^{(0)}(t) = n(t)$ and $\varepsilon^{(0)}(t) = \varepsilon(t)$ for all times $t\geq t_i$ implies that $\lambda(t) = \lambda_0$ and that the temperature has the form $T(t) = T_0/a(t)$ with positive constants $\lambda_0,T_0 > 0$ satisfying
\begin{equation}
n_i = 8\pi\lambda_0 (k_B T_i)^3,\qquad
\varepsilon_i = 24\pi\lambda_0 (k_B T_i)^4 = 3n_i k_B T_i,
\label{Eq:niepsiloni}
\end{equation}
with the initial temperature $T_i := T_0/a_i$.

At this point it is worthwhile emphasizing the following important fact. The specific forms $\lambda(t) = \lambda_0$ and $T(t) = T_0/a(t)$ obtained from the matching conditions imply that the distribution function $F^{(0)}$ defined in Eq.~(\ref{Eq:F0Def}) describes a state in \emph{global} thermodynamic equilibrium. This can be seen by rewriting
\begin{equation}
F^{(0)}(t,\mathcal{C}) = \lambda_0 e^{-\frac{\mathcal{C}}{k_B T_0}}
 = \lambda_0 e^{\beta^\mu(t) p_\mu},\qquad
\beta^\mu(t)\frac{\partial}{\partial x^\mu} = \frac{1}{k_B T(t)}\frac{\partial}{\partial t}.
\label{Eq:F0DefMassless}
\end{equation}
If $\partial/\partial t$ was a Killing vector field and $T(t)$ was constant in time, it would follow from the considerations in subsection~\ref{SubSec:GlobalEquilibrium} that $F^{(0)}$ describes global equilibrium. When $da/dt \neq 0$ this is no longer true in the massive case; however, for massless particles one can show that $F^{(0)}$ does satisfy the full Boltzmann equation. This is a consequence of the fact that the vector field $\beta^\mu(t)$ defined in Eq.~(\ref{Eq:F0DefMassless}) is a \emph{conformal Killing vector field} of the FLRW spacetime, which means that $\nabla^{(\mu} \beta^{\nu)}$ is proportional to $g^{\mu\nu}$ and implies that Eq.~(\ref{Eq:TransportZero}) is still satisfied in the massless case where the momentum $p$ of the particle is a null covector.

For the following we introduce the normalized moments
\begin{equation}
M_s(t) := \frac{\Pi_s(t)}{\Pi_s^{(0)}(t)},\qquad s = 0,1,2,\ldots
\label{Eq:NormalizedMoments}
\end{equation}
The matching conditions~(\ref{Eq:niepsiloni}) imply that the first two normalized moments satisfy $M_0(t) = M_1(t) = 1$ for all $t > t_i$, while Eq.~(\ref{Eq:MomentsEvolutionMassless}) yields the following evolution equation for the remaining $M_s$:
\begin{equation}
l_{\text{mfp}}(t) \frac{d}{dt} M_s(t) + M_s(t) 
 = \frac{1}{s+1}\sum\limits_{r=0}^s M_{s-r}(t) M_r(t),\qquad
 s = 2,3,4,\ldots,
 \label{Eq:MomentsEvolutionMasslessBis}
\end{equation}
with $l_{\text{mfp}}(t) := [\sigma_T n(t)]^{-1}$ the mean free path at time $t$ (see Eq.~(\ref{Eq:ImportantRelation})). Introducing the new time coordinate
\begin{equation}
\tau(t) := \int\limits_{t_i}^t \frac{dt'}{l_{\text{mfp}}(t')},
\label{Eq:ReescalingTime}
\end{equation}
and the corresponding moments $\hat{M}_s(\tau) := M_s(t)$, Eq.~(\ref{Eq:MomentsEvolutionMasslessBis}) can be rewritten as~\cite{dBgDuHmMjN16,dBgDuHmMjN16b}
\begin{equation}
\frac{d}{d\tau} \hat{M}_s(\tau) + \omega_s\hat{M}_s(\tau)
 = \frac{1}{s+1}\sum\limits_{r=1}^{s-1} \hat{M}_{s-r}(\tau)\hat{M}_r(\tau),\qquad
\omega_s := \frac{s-1}{s+1},\qquad
s = 2,3,4,\ldots,
\label{Eq:MomentsEvolutionMasslessInTau}
\end{equation}
where the two terms corresponding to $r=0$ and $r=s$ in the sum on the right-hand side of Eq.~(\ref{Eq:MomentsEvolutionMasslessBis}) have been moved to the left-hand side of the equation. Remarkably, the scale factor $a(t)$ does not appear anymore in these equations, and hence the evolution equations for the moments completely decouple from the dynamics describing the expansion of the universe.\footnote{A massless collisionless gas propagating on a FLRW background clearly does not feel the scale factor $a(t)$, since in this case the gas particles follow null geodesics which remain invariant with respect to conformal transformations of the metric. Therefore, the evolution of such a gas is equivalent to its evolution on the time-independent conformal metric
$$
-d\hat{\tau}^2 + \frac{\delta_{ij} dx^i dx^j}{\left( 1 + \frac{k}{4}|\ve{x}|^2 \right)^2},
$$
with the conformal time $d\hat{\tau} = dt/a(t)$. In the collisional case, this would still be true if the mean free path scaled like $a(t)$ (or the total cross section like $a(t)^2$). However, since we have assumed $\sigma_T$ to be constant, $l_{\text{mfp}}(t)$ is proportional to $a(t)^3$ and the rescaled time $\tau$ defined in Eq.~(\ref{Eq:ReescalingTime}) differs from the conformal time $\hat{\tau}$.
}
For each $s\geq 2$, Eq.~(\ref{Eq:MomentsEvolutionMasslessInTau}) provides a linear ordinary differential equation for the normalized moment $\hat{M}_s$ where the source term on the right-hand side only depends on the moments $\hat{M}_1 = 1, \hat{M}_2,\ldots, \hat{M}_{s-1}$ with order smaller than $s$. As a consequence, the equations can be solved successively for $s=2,3,4,\ldots$, which yields~\cite{dBgDuHmMjN16b}
\begin{eqnarray}
\hat{M}_2(\tau) &=& \left[ \hat{M}_2(0) - 1 \right] e^{-\omega_2\tau} + 1,\\
&\vdots&\nonumber\\
\hat{M}_s(\tau) &=& \hat{M}_s(0) e^{-\omega_s\tau}
 + \frac{1}{s+1}\sum\limits_{r=1}^{s-1}
\int\limits_0^\tau \hat{M}_{s-r}(\tau')\hat{M}_r(\tau') e^{-\omega_s(\tau-\tau')} d\tau'.
\end{eqnarray}
It follows from these equations that for all $s\in \Natural_0$, $\hat{M}_s(\tau)\to 1$ as $\tau\to \infty$ which seems to indicate that the gas converges to the global equilibrium configuration $F^{(0)}$. However, this is not necessarily the case as we discuss now. For this, we need to recall that the rescaled time $\tau$ is related to the physical (cosmic) time $t$ through the relation~(\ref{Eq:ReescalingTime}) which depends on the inverse mean free path $l_{\text{mfp}}(t)^{-1} = \sigma_T n(t) = \sigma_T n_i a_i^3/a(t)^3$ at time $t$ which, in turn, depends on the scale factor $a(t)$. To determine $a(t)$, we solve the Friedmann equation~(\ref{Eq:Friedmann1}), where for simplicity we assume a spatially flat universe $(k=0)$ and neglect the cosmological constant ($\Lambda = 0$).\footnote{In fact, since $\varepsilon$ is proportional to $1/a(t)^4$, one can neglect the terms proportional to $k$ and $\Lambda$ in Eq.~(\ref{Eq:Friedmann1}) at sufficiently early times. For analytic solutions with $k\neq 0$ and $\Lambda = 0$ see, for instance, Section~2.4.3 in Ref.~\cite{Piatella-Book}.
}
Using Eq.~(\ref{Eq:ntepsilont}) one finds, assuming that $a(t) = 0$ at $t=0$,
\begin{equation}
a(t) = a_i\sqrt{\frac{t}{t_i}},\qquad
t_i = \sqrt{\frac{3}{32\pi\varepsilon_i}}.
\label{Eq:atRadiation}
\end{equation}
Using Eqs.~(\ref{Eq:ntepsilont},\ref{Eq:ReescalingTime}) one obtains from this
\begin{equation}
\tau(t) = \sigma_T\int\limits_{t_i}^t n(t')^3 dt' = \tau_\infty\left[ 1 - \sqrt{\frac{t_i}{t}} \right],\qquad
\tau_\infty := 2\sigma_T n_i t_i.
\label{Eq:tauRadiation}
\end{equation}
Therefore, $\tau$ cannot become arbitrarily large even as $t\to \infty$, and it is bounded from above by $\tau_\infty$. Consequently, the gas does not converge to an equilibrium configuration. However, if $\tau_\infty \gg 1$, the quantities $|\hat{M}_s(\tau) - 1|$ reduce to a tiny fraction of their initial values for large $t$, and in this sense the gas does reach a configuration close to equilibrium. Combining the condition $\tau_\infty \gg 1$ with Eqs.~(\ref{Eq:niepsiloni},\ref{Eq:atRadiation}) yields
\begin{equation}
8\pi k_B T_i \ll n_i\sigma_T^2.
\end{equation}
We note in passing that the Knudsen number defined in Eq.~(\ref{Eq:DimensionlessBoltzmann}), where we take $l_{\text{ms}} := 1/H = 2t$ to be the Hubble length, is equal to
\begin{equation}
\text{Kn}(t) = \frac{1}{2\sigma_T n(t) t} = \frac{1}{\tau_\infty}\frac{T_i}{T(t)}.
\end{equation}
Therefore, under the condition $\tau_\infty \gg 1$ the problem can in principle be treated using the Hilbert expansion, as long as $T_i/T(t)$ does not become too large.

Finally, we need to check the range of validity of the approximation made here, in which it was assumed that most particles have large momenta $|\ve{p}| \gg m$. In terms of the moments this condition can be characterized by $\Pi_2(t)\gg m^2\Pi_0(t)$ or $\Pi_1(t)\gg m\Pi_0(t)$. Using the fact that $\Pi_1(t) = \varepsilon(t)$ and $\Pi_0(t) = n(t)$ and Eqs.~(\ref{Eq:ntepsilont0},\ref{Eq:ntepsilont},\ref{Eq:atRadiation}) one obtains
\begin{equation}
\frac{m\Pi_0(t)}{\Pi_1(t)} = \frac{m}{3k_B T(t)} = \frac{m}{3k_B T_i}\sqrt{\frac{t}{t_i}} \ll 1.
\end{equation}
In particular for $t = t_i$ this requires $k_B T_i \gg m$. Using the Cauchy-Schwarz inequality it follows that $\Pi_1(t)^2\leq \Pi_0(t)\Pi_2(t)$ which shows that $\Pi_1(t)\gg m\Pi_0(t)$ automatically implies $\Pi_2(t)\gg m^2\Pi_0(t)$. Therefore, we conclude that for initial conditions satisfying
\begin{equation}
m c^2\ll k_B T_i \ll n_i\frac{c^4\sigma_T^2}{8\pi G_N},
\label{Eq:ConditionEarlyTime}
\end{equation}
where we have reintroduced the speed of light $c$ and Newton's constant $G_N$, the gas behaves as a photon gas for some time during which it reaches a near-equilibrium state. Since $T(t)$ scales like $1/a(t)$ and $n(t)$ like $1/a(t)^3$, the condition~(\ref{Eq:ConditionEarlyTime}) is clearly satisfied if $t_i$ is small enough.

It is also interesting to note that for a contracting universe, for which the scale factor $a(t)$ is given by Eq.~(\ref{Eq:atRadiation}) with the factor $t$ replaced by $-t$ (and $t$ running from $-t_i$ to $0$), one obtains $\tau(t) = \tau_\infty(\sqrt{-t_i/t} - 1)$ for $-t_i < t < 0$ instead of Eq.~(\ref{Eq:tauRadiation}), and $\tau(t)\to \infty$ as $t\to 0$, which means that the gas configuration does have ``enough time" to reach thermal equilibrium before the universe collapses.

\subsection{The late epoch}

In this subsection we analyze the asymptotic limit $a(t) \to \infty$ corresponding to the late epoch. In this limit, most of the particles are expected to have small momentum $\displaystyle |\ve{p}|=\mathcal{C}/a(t) \ll m$, which implies that it is sufficient to compute the collision kernels $K^{(s)}_{\text{gain},\text{loss}}(p_1,p_2)$ in Eqs.~(\ref{Eq:GainKernel}) and~(\ref{Eq:LossKernel}) using the nonrelativistic limit $p_1,p_2\ll m$.

Since
\begin{equation}
\ell = m^2 +\frac{1}{2} (p_1^2 + p_2^2 -2 p_1 p_2 \cos\vartheta) + \mathcal{O}(p_i^4/m^2),
\end{equation}
one finds the non-relativistic expressions
\begin{equation}
m^2 g\sqrt{1 + \frac{g^2}{4}} = \sqrt{\ell^2 - m^4} = m |\ve{p}_2 - \ve{p}_1|
\end{equation}
and
\begin{equation}
2\ve{p}^{cm} = \ve{p}_1 + \ve{p}_2 = \ve{p}_1^* + \ve{p}_2^*,\qquad
m\ve{q} = \ve{p}_2 - \ve{p}_1,\qquad
m\ve{q}^* = \ve{p}_2^* - \ve{p}_1^*,
\end{equation}
such that the gain term reduces to
\begin{equation}
K^{(s)}_{\text{gain}}(p_1,p_2) = \frac{1}{2m}\int\limits_0^\pi d\vartheta \sin\vartheta
|\ve{p}_2 - \ve{p}_1 | \int\limits_{S^2} d\Omega(\ve{\hat{q}}^*) \frac{d\sigma}{d\Omega} 
\left( \frac{|\ve{p}_1^*|^2}{2m} \right)^s.
\end{equation}

For the following, we compute this kernel for the case of Maxwell particles (cf. Eq.~(\ref{Eq:FLRWDiffCrossSec})), for which
\begin{equation}
| \ve{p}_2 - \ve{p}_1| \frac{d\sigma}{d\Omega} 
 = \frac{m\sigma_0}{4\pi}\Gamma(\Theta),
\end{equation}
with a positive constant $\sigma_0$ and an arbitrary function of the scattering angle $\Gamma(\Theta)$. To this purpose we follow~\cite{mKtW76,mKtW77} and introduce a Cartesian coordinate system such that
\begin{equation}
\ve{q} = g(0,0,1),\qquad
\ve{p}^{cm} = \left( p_{\hat{1}}^{cm},0,p_{\hat{3}}^{cm} \right),\qquad p_{\hat{1}}^{cm}\geq 0.
\end{equation}
In terms of the usual polar coordinates $(\Theta,\Phi)$ with respect to this system, one can write
\begin{equation}
\ve{q}^* = g\left( \sin\Theta\cos\Phi,\sin\Theta\sin\Phi,\cos\Theta \right),
\end{equation}
where $\Theta$ is the scattering angle, see Eq.~(\ref{Eq:ScatteringAngleBis}). One easily finds
\begin{equation}
g p_{\hat{1}}^{cm}
 = |\ve{p}^{cm}\wedge\ve{q}| = \frac{1}{m}| \ve{p}_1\wedge \ve{p}_2 |,\qquad
g p_{\hat{3}}^{cm}
 = \ve{p}^{cm}\cdot\ve{q} = \frac{1}{2m}\left( |\ve{p}_2|^2 - |\ve{p}_1|^2 \right), 
\end{equation}
which yields
\begin{equation}
|\ve{p}_1^*|^2 = \cos^2\left( \frac{\Theta}{2} \right) |\ve{p}_1|^2 
 + \sin^2\left( \frac{\Theta}{2} \right) |\ve{p}_2|^2 - | \ve{p}_1\wedge \ve{p}_2 |\sin\Theta\cos\Phi.
\label{Eq:p1*2}
\end{equation}
Using this, the gain term can be written as
\begin{equation}
K^{(s)}_{\text{gain}}(p_1,p_2) = \frac{\sigma_0}{8\pi(2m)^s}\int\limits_0^\pi d\Theta \sin\Theta \Gamma(\Theta) \left. \left( \frac{d}{d\lambda} \right)^s \right|_{\lambda=0} R(p_1,p_2,\Theta;\lambda),
\end{equation}
with the generating function
\begin{equation}
R(p_1,p_2,\Theta;\lambda) := \int\limits_0^\pi d\vartheta \sin\vartheta\int\limits_0^{2\pi} d\Phi e^{\lambda |\ve{p}_1^*|^2}.
\end{equation}
Using Eq.~(\ref{Eq:p1*2}) and the series representation of the exponential function, one finds after some calculations:
\begin{equation}
R(p_1,p_2,\Theta;\lambda) = \frac{2\pi}{p_1 p_2\sin\Theta}\sum\limits_{n=0}^\infty
\frac{\lambda^n}{(n+1)!}\left\{ \left[ p_1\cos\left( \frac{\Theta}{2} \right) + p_2\sin\left( \frac{\Theta}{2} \right) \right]^{2n+2} -  \left[ p_1\cos\left( \frac{\Theta}{2} \right) - p_2\sin\left( \frac{\Theta}{2} \right) \right]^{2n+2} \right\},
\end{equation}
such that
\begin{equation}
\left. \left( \frac{d}{d\lambda} \right)^s \right|_{\lambda=0} R(p_1,p_2,\Theta;\lambda)
 = \frac{4\pi}{(s+1)\sin\Theta}\sum\limits_{j=0}^{s} {2s+2\choose 2j +1} p_1^{2j} p_2^{2s-2j}
\cos^{2j+1}\left( \frac{\Theta}{2} \right)\sin^{2s+1-2j}\left( \frac{\Theta}{2} \right),
\end{equation}
and thus
\begin{equation}
K^{(s)}_{\text{gain}}(p_1,p_2) = \frac{\sigma_0}{2(s+1)}
\sum\limits_{j=0}^{s} {2s+2\choose 2j +1}\frac{j! (s-j)!}{(s+1)!}\Gamma_{sj}
 \left( \frac{p_1^2}{2m} \right)^j\left( \frac{p_2^2}{2m} \right)^{s-j},
\label{Eq:KsGainNonRel}
\end{equation}
with coefficients
\begin{equation}
\Gamma_{sj} := (s+1){s\choose j}\int\limits_0^\pi \Gamma(\Theta)\cos^{2j+1}\left( \frac{\Theta}{2} \right)\sin^{2s+1-2j}\left( \frac{\Theta}{2} \right) d\Theta,
\end{equation}
which are normalized such that $\Gamma_{sj} = 1$ when $\Gamma(\Theta) = 1$. In general, if $\Gamma(\pi-\Theta) = \Gamma(\Theta)$, these coefficients satisfy $\Gamma_{sj} = \Gamma_{s,s-j}$ and due to the normalization condition~(\ref{Eq:GammaNormalization}) it follows that $\Gamma_{10} = \Gamma_{11} = \Gamma_{00} = 1$.

In the non-relativistic limit the loss term reduces to
\begin{equation}
K^{(s)}_{\text{loss}}(p_1,p_2) = \sigma_0\left( \frac{p_1^2}{2m} \right)^s.
\label{Eq:KsLossNonRel}
\end{equation}
Eqs.~(\ref{Eq:KsGainNonRel},\ref{Eq:KsLossNonRel}) lead to the moment equations
\begin{equation}
\frac{d}{dt} \Pi_s(t) + (2s+3) H(t) \Pi_s(t)
 = \frac{\sigma_0}{2(s+1)}\sum\limits_{j=0}^{s} 
 {2s+2\choose 2j +1}\frac{j! (s-j)!}{(s+1)!}\Gamma_{sj}\Pi_{j}(t)\Pi_{s-j}(t)
  - \sigma_0\Pi_{s}(t) \Pi_0(t),\qquad s = 0,1,2,\ldots
\label{Eq:MomentsEvolutionMassive}
\end{equation}
The right-hand side vanishes for $s=0$ and $s=1$, as expected.

As in the early epoch calculation, these equations can be simplified by replacing the moments $\Pi_s$ with their normalized moments defined as in Eq.~\eqref{Eq:NormalizedMoments}, where the moments $\Pi_s^{(0)}$ can be computed according to Eq.~\eqref{Eq:FLRWMoments0} using the generating function~\eqref{Eq:FLRWGenerator} and taking into account the approximation for $|\ve{p}| = \mathcal{C}/a(t) \ll m$. This yields $\displaystyle\zeta(\beta) = (2\pi m/\beta)^{3/2}$, from which
\begin{equation}
\Pi^{(0)}_s(t) = (2\pi m)^{3/2} \lambda(t)\frac{(2s+1)!!}{2^s}[k_B T(t)]^{s+3/2}.
\end{equation}
As in section~\ref{Subsec:Early_epoch} we impose the matching conditions $\Pi_s= \Pi^{(0)}_s$ for $s=0,1$. From Eqs.~(\ref{Eq:FLRWEuler},\ref{Eq:FLRWDefMoments},\ref{Eq:FLRWContinuity}) one obtains
\begin{eqnarray}
\Pi_0(t) & = & n(t) = n_i\frac{a_i^3}{a(t)^3},
\label{Eq:ntepsilontMassive1}\\
\Pi_1(t) & = & \varepsilon(t) - m n(t) = (\varepsilon_i - m n_i)\frac{a_i^5}{a(t)^5}.
\label{Eq:ntepsilontMassive2}
\end{eqnarray}
From this, one concludes that $\lambda(t) = \lambda_0$ is constant and that $T(t) = T_i a_i^2/a(t)^2$ with
\begin{equation}
n_i = \lambda_0 (2\pi m k_B T_i)^{3/2},\qquad
\varepsilon_i = n_i\left( m + \frac{3}{2} k_B T_i \right).
\label{Eq:niepsiloniMass}
\end{equation}
In terms of the normalized moments defined in Eq.~(\ref{Eq:NormalizedMoments}) the  evolution equations for the moments~(\ref{Eq:MomentsEvolutionMassive}) read
\begin{equation}
l^*_{\text{mfp}}(t) \frac{d}{dt} M_s(t) + M_s(t) 
 = \frac{1}{s+1}\sum\limits_{j=0}^s \Gamma_{sj} M_{s-j}(t) M_j(t), \qquad
 s = 2,3,4,\ldots,
 \label{Eq:MomentsEvolutionMass}
\end{equation}
with $l^*_{\text{mfp}}(t) := [\sigma_0 n(t)]^{-1}$ (note that this quantity is not the mean free path, since $\sigma_0 = g\sigma_T$ differs from the total cross section $\sigma_T$ by a factor of $g$, see Eq.~(\ref{Eq:FLRWDiffCrossSec})), and we recover the evolution equations for the moments derived in~\cite{dBgDuHmMjN16,dBgDuHmMjN16b} when $\Gamma(\Theta)=1$. Introducing the new time coordinate
\begin{equation}
\tau^*(t) := \int\limits_{t_i}^t \frac{dt'}{l^*_{\text{mfp}}(t')},
\label{Eq:ReescalingTimeBis}
\end{equation}
and the corresponding moments $\hat{M}_s(\tau^*) := M_s(t)$, Eq.~(\ref{Eq:MomentsEvolutionMass}) can be rewritten as~\cite{dBgDuHmMjN16,dBgDuHmMjN16b}
\begin{equation}
\frac{d}{d\tau^*} \hat{M}_s(\tau^*) + \omega_s\hat{M}_s(\tau^*)
 = \frac{1}{s+1}\sum\limits_{j=1}^{s-1} \Gamma_{sj} \hat{M}_{s-j}(\tau^*) \hat{M}_j(\tau^*), \qquad s = 2,3,4,\ldots
\label{Eq:MomentsEvolutionMassInTau}
\end{equation}
Proceeding as in section~\ref{Subsec:Early_epoch}, we determine the scale factor $a(t)$ by solving the Friedmann equation~(\ref{Eq:Friedmann1}). Assuming a spatially flat universe $(k=0)$ and that the cosmological constant $\Lambda \gg 8\pi\varepsilon_i$ dominates\footnote{In the case when $\Lambda=k=0$ the dominant term in the Friedmann equation is the energy density $\varepsilon(t) \sim 1/a(t)^3$, which gives rise to the scale factor $\displaystyle a(t) = a_i\left[ 1 + \sqrt{6\pi m n_i}(t-t_i) \right]^{2/3}$ with a power-law behavior in $t$ instead of the exponential form. The condition for reaching a near-equilibrium state turns out to be
\begin{equation}
\frac{6\pi G_N m}{c^2} \gg n_i \sigma_0^2,
\nonumber
\end{equation}
and the initial temperature should be chosen such that $3k_B T_i \ll 2mc^2$ in order to guarantee the validity of the late epoch.
}
the energy density at initial time $t_i$, it follows that for $t\geq t_i$,
\begin{equation}
a(t) = a_i e^{\frac{t-t_i}{\ell}},\qquad
\ell := \sqrt{\frac{3}{\Lambda}}.
\label{Eq:atMass}
\end{equation}
Using Eqs.~(\ref{Eq:ntepsilontMassive1},\ref{Eq:ReescalingTimeBis},\ref{Eq:atMass}) one obtains from this
\begin{equation}
\tau^*(t) = \sigma_0 \int\limits_{t_i}^t n(t') dt' = \tau^*_\infty\left[ 1 - e^{-3\frac{t-t_i}{\ell}} \right], \qquad
\tau^*_\infty := \frac{\ell}{3} \sigma_0 n_i.
\end{equation}
Similarly to the early epoch calculation, $\tau^*(t)$ cannot become arbitrarily large as $t \to \infty$ and is bounded from below by $\tau^*_\infty$. Hence, the gas does not converge to an equilibrium configuration. However, near-equilibrium is reached provided that $\tau^*_\infty \gg 1$, which is equivalent to $\sigma_0^2 n_i^2 \gg 3\Lambda$. Combining this with $\Lambda\gg 8\pi\varepsilon_i \geq 8\pi m n_i$ and using Eqs.~(\ref{Eq:niepsiloniMass},\ref{Eq:atMass}), the  conditions for reaching a near-equilibrium state can be summarized as follows:
\begin{equation}
\frac{\sqrt{3\Lambda}}{\sigma_0}\ll n_i \ll \frac{\Lambda c^2}{8\pi m G_N},
\label{Eq:niBounds}
\end{equation}
where we have reintroduced the speed of light $c$ and Newton's constant $G_N$. Taking a gas consisting of protons and introducing the value of the cosmological constant $\Lambda \approx 1.09\times 10^{-56} \hbox{cm}^{-2}$ inferred from recent observations from the Planck Collaboration~\cite{Planck2018,Zyla:2020zbs}, one obtains
\begin{equation}
\frac{\Lambda c^2}{8\pi m G_N} \approx 3.49\times 10^6 \hbox{cm}^{-3},\qquad
\sigma_0 \gg \frac{8\pi m G_N}{c^2} \sqrt{\frac{3}{\Lambda}}\approx 0.517b,
\end{equation}
where $b$ refers to the unit of barns.\footnote{$1b=10^{-24}\hbox{cm}^2$.} The upper bound in Eq.~(\ref{Eq:niBounds}) is well satisfied in our current universe, for which the baryon density is $n_b\approx 2.5\times 10^{-7} \hbox{cm}^{-3}$~\cite{Zyla:2020zbs}.

The Knudsen number in this case is (taking again $l_{\text{ms}} := 1/H = \ell$ to be the Hubble length)
\begin{equation}
\text{Kn}(t) = \frac{1}{\ell\sigma_T(t) n(t)} 
 = \frac{\bar{g}(t)}{3\tau^*_\infty}\left( \frac{T_i}{T(t)} \right)^{3/2},
\end{equation}
with $\bar{g}$ the mean relative speed between the particles. The condition $\tau^*_\infty \gg 1$ implies that $\text{Kn}(t_i) \ll 1$ at the initial time; however assuming $\bar{g}(t)\sim\sqrt{T(t)/T_i}$, it follows that $\text{Kn}(t)$ diverges as $t\to \infty$, meaning that the hydrodynamic limit eventually breaks down.

Finally, we need to check the range of validity of the approximation made here, which assumed that most particles have small momenta $|\ve{p}| \ll m$. In terms of the moments this condition can be characterized by $m\Pi_0(t) \gg \Pi_1(t)$. Using the fact that $\Pi_1(t) = \varepsilon(t) - mn(t)$ and $\Pi_0(t) = n(t)$ and Eqs.~(\ref{Eq:ntepsilontMassive1},\ref{Eq:ntepsilontMassive2},\ref{Eq:niepsiloniMass}) one obtains
\begin{equation}
\frac{m\Pi_0(t)}{\Pi_1(t)} = \frac{2}{3}\frac{m}{k_B T(t)} = \frac{2}{3}\frac{m}{k_B T_i}
e^{2\frac{t-t_i}{\ell}} \gg 1,
\end{equation}
which requires the initial temperature to be small enough such that $k_B T_i \ll 2mc^2/3$.

\section{Conclusions}
\label{Sec:Conclusions}

We have provided a self-contained, pedagogical review of the relativistic kinetic theory of  dilute gases propagating on a curved spacetime manifold $(M,g)$ of arbitrary dimension $n$. In contrast to most previous work in the literature, which is based on the tangent bundle formulation, we formulated the theory on the cotangent bundle $T^* M$ associated with $M$, i.e. the set of pairs $(x,p)$ in which $x$ is a spacetime event and $p$ a momentum covector at $x$. Although the tangent and cotangent formulations are equivalent, the latter has the advantage of being more naturally adapted to the Hamiltonian framework on which statistical physics considerations are usually based on. As we have discussed, the cotangent bundle $T^* M$ possesses rich geometric structures which form the backbone of the theory and lead naturally to a manifestly covariant (coordinate-independent) formulation of relativistic kinetic theory. As a first example, we provided such a formulation for the description of a collisionless gas consisting of identical, massive and uncharged particles and then showed how to generalize the description to more general situations. In particular, we included a brief discussion of a kinetic gas consisting of different species of charged particles, a novel derivation of the collision term based on the structure of the collision manifold and the transition probability density, the formulation of the relativistic H-theorem and the associated equilibrium states, the difference between global and local equilibrium, a brief discussion of the method of moments and its application to the propagation of an homogeneous, isotropic kinetic gas in a FLRW universe.

Let us describe some of the highlights of our review in more detail. As mentioned above, we made emphasis on the geometric structures of the cotangent bundle $T^* M$ associated with the spacetime manifold $(M,g)$. One of these structures is the naturally-defined symplectic form $\Omega_s = d\Theta$ on $T^*M$, which arises as the exterior differential of the Poincar\'e one-form $\Theta$ on $T^* M$ and is a standard construction in classical mechanics. The symplectic form allows one to reformulated the geodesic motion on $(M,g)$ as the Hamiltonian flow on $T^* M$ associated with the free one-particle Hamiltonian $\mathcal{H}$ on $T^* M$. The generator of this flow (i.e. the Hamiltonian vector field corresponding to $\mathcal{H}$) is the Liouville vector field, which is the operator that appears in the transport part of the relativistic Boltzmann equation. A further structure on $T^*M$ arises from the spacetime metric $g$ and the associated Levi-Civita connection $\nabla$, which induce a natural metric $\hat{g}$ on $T^* M$, called the Sasaki metric. As reviewed in our article, this metric satisfies several important properties, including the fact that the
Liouville vector field is geodesic with respect to $(T^*M,\hat{g})$, and it yields a natural volume form on $T^* M$ which coincides (up to a constant numerical factor) with the $n$-fold wedge product of the symplectic form $\Omega_s$. This volume form allows one to integrate functions on the cotangent bundle, and ultimately it provides the means for the physical interpretation of the one-particle distribution function. For the particular case of a simple gas of massive particles, one considers instead of $T^*M$ the future mass shell $\Gamma_m^+$, a Lorentzian submanifold of $(T^* M,\hat{g})$ of codimension one on which the one-particle distribution function $f$ is defined. As we have shown, the Liouville vector field is tangent to $\Gamma_m^+$ and satisfies Liouville's theorem, that is, it generates an incompressible flow on $\Gamma_m^+$.

As we have discussed, the one-particle distribution function provides a density function on $\Gamma_m^+$ which (together with the Liouville vector field) allows one to count the averaged number of occupied trajectories crossing a given spatial hypersurface in $\Gamma_m^+$ through a flux integral, analogous to the way the particle number density on $M$, together with the $n$-velocity, determines the number of particles contained in a given spatial volume in $M$ through a flux integral over the particle current density vector field on $M$. Further, the one-particle distribution function defines the macroscopic observables on the spacetime manifold through suitable fibre integrals over the $p$-space. The most important observables for the case of general relativity are the particle current density vector $J$ and the energy-momentum-stress tensor $T$ which appear as source terms in the Maxwell and Einstein equations and which contain all the information regarding the particle and energy densities, the mean particle velocity, the heat flow and the pressure tensor of the gas. However, higher moments are also relevant, as we have encountered in the method of moments, for instance.

The free one-particle Hamiltonian and symplectic form $\Omega_s$ describe the motion of a kinetic gas which, in the absence of collisions, is ``freely falling" in the gravitational field $g$, that is, the individual gas particles follow timelike geodesics in $(M,g)$. As we have shown, our formulation can easily be generalized to the case of charged particles in an external electromagnetic field $F$ by keeping the same Hamiltonian $\mathcal{H}$ as in the free case and modifying the symplectic form $\Omega_s$ on $T^*M$ by adding to it a term involving the pull-back of $F$ with respect to the natural projection $\pi: T^*M\to M$. This allows one to provide a gauge-invariant formulation of a charged gas which avoids the need of introducing an (in general only locally defined) electromagnetic potential $A$ such that $F = dA$ and keeps the interpretation of $p$ as the physical (as opposed to the gauge-dependent canonical) momentum of the particle. Instead of being freely falling and in the absence of collisions, charged particles follow the Hamiltonian flow corresponding to $\mathcal{H}$ which arises from this new symplectic structure depending on $F$, and this flow, when projected onto the spacetime manifold, describes the usual motion of a charged test particle in external electromagnetic and gravitational fields. By promoting the metric and electromagnetic fields from being mere external fields to dynamical fields satisfying the Einstein and Maxwell equations with the appropriate source terms, one can take into account the self-gravity and the electromagnetic interactions between the gas particles in a self-consistent way. In the absence of other interactions between the gas particles, this leads to the Einstein-Vlasov-Maxwell system of equations.

When (in addition to the gravitational and electromagnetic interactions) the gas particles are subject to short-ranged interactions, a given particle follows the Hamiltonian flow only as long as it lies sufficiently far away from other gas particles. Assuming that the gas is sufficiently dilute, this leads to the model in which the gas particles' trajectories consist of broken segments of Hamiltonian orbits, where the particle's momentum $p$ is abruptly modified each time a binary collision takes place. Such binary collisions are described in a statistical way through the collision term in the Boltzmann equation. In this article, we have provided a formal derivation of the collision term for a simple gas, based on a systematic study of the kinematics of binary elastic collisions and the resulting collision manifold $C_x$, its volume form and the transition probability density which is related to the differential cross-section and which contains all the information on the short-ranged interaction relevant for the statistical description.

One of the most important consequences of Boltzmann's equations is the famous H-theorem which leads to a dynamical description for the second law of thermodynamics and, in some simple situations, can be used to show the approach of an initially off-equilibrium configuration to an equilibrium one. However, as we have emphasized in our review, the existence of such global equilibrium configurations is only possible under rather restrictive conditions on the spacetime manifold $(M,g)$ and electromagnetic field $F$. In particular, $(M,g)$ is required to possess a global timelike Killing vector field for such a configuration to exist. This means that many interesting spacetimes, including ones involving isolated systems containing black holes and cosmological models describing an expanding universe do not possess global equilibrium configurations. This makes the problem of understanding the approach to equilibrium in general relativity even more interesting than in flat Minkowski spacetime or the Newtonian case.

In the hydrodynamic limit of small values of the Knudsen parameter, in which the Boltzmann equation is dominated by the collision term, one can assume that the gas is (in first approximation) in a local equilibrium state. This is a far less restrictive condition than the existence of global equilibrium and can be achieved on generic curved spacetimes $(M,g)$ by modeling the distribution function (in first approximation) by a Maxwell-J\"uttner distribution function, i.e. the relativistic generalization of the local Maxwell-Boltzmann distribution function, which is characterized by a temperature, $n$-velocity and amplitude, all of which are allowed to be functions on $(M,g)$. This assumption offers the possibility to study near-equilibrium configurations through well-known methods, like the Chapman-Enskog method or the methods of moments we have made brief comments on. These methods are paramount for the description of non-equilibrium relativistic fluids which, despite much progress and many proposals, is still an open field of investigation.

Our review culminates in the discussion of the Einstein-Boltzmann system for a Friedmann-Lema\^itre-Robertson-Walker (FLRW) universe filled with an homogeneous and isotropic kinetic gas, consisting of identical massive particles subject to binary elastic collisions. Due to the presence of the time-dependent scale factor $a(t)$ in the metric, a global equilibrium configuration does not exist, and hence studying the dynamics of the gas is a very interesting and relevant problem, even under the restricted symmetry assumptions of homogeneity and isotropy! While we are not aware of a complete study of this interesting problem, in the present work we have reformulated it by replacing the distribution function with appropriate moments thereof. This leads to a nonlinear infinity system of ordinary differential equations for the scale factor and these moments which generalizes previous work which was either restricted to the case of massless particles or to a Newtonian gas. Unfortunately, the solution of this system does not seem to be possible using purely analytic methods. However, in the limit of either the early or the late epochs of the universe the system does simplify sufficiently such that it can be formally solved and such that qualitative statements can be made without necessarily assuming that the state lies close to local equilibrium.

In the early epoch, the energies of the particles are dominated by their kinetic part, and hence the gas can be considered to be ultra-relativistic which is equivalent to taking the limit of zero  particle mass. In this limit, it turns out there \emph{does} exist a global equilibrium distribution function which is due to the fact that the FLRW spacetime admits a timelike \emph{conformal} Killing vector field. The amplitud and temperature characterizing the equilibrium function can be determined by matching its first two moments to those of the full (non-equilibrium) distribution function. Assuming that the differential cross section is constant in the ultra-relativistic limit, one can show that the infinite family of moment equations can be decoupled from the Friedmann equations by absorbing the scale factor into a suitable redefinition of the time coordinate and normalizing the moments by the ones associated with the equilibrium distribution function. Additionally, it turns out the moment equations form a hierarchical structure which allows one to solve the system explicitly moment-by-moment. On the other hand, in the ultra-relativistic limit the Euler equations imply that the energy density scales as $1/a(t)^4$ which dominates the spatial curvature and cosmological constant terms in the Friedmann equation at early times and leads to an explicit approximate solution for the scale factor $a(t)$. The solution of the moment equations reveals that the approach to equilibrium is characterized by a timescale which, in general, is much larger than the timescale corresponding to the validity of the early epoch during which the gas is ultra-relativistic. In other words, within the early epoch and for generic initial data for the gas configuration, there is, in general, not enough time for the gas to relax to an equilibrium state. However, we have also specified conditions on the initial data which do guarantee that the gas reaches a state lying close to equilibrium within the early epoch, in the sense that the normalized moments (defined as the ratio between the moments associated with the distribution function and the equilibrium function) approach a value close to one.

In the late epoch, the energy of the particles is dominated by their rest mass energy, and thus the gas is non-relativistic. Assuming that in the non-relativistic limit the particles behave like Maxwell particles, for which the differential cross section is inversely proportional to the relative speed between the particles, one can show that once again, the family of moment equations can be decoupled from the Friedmann equation and has the same hierarchical structure as in the ultrarelativistic limit. (Interestingly, for the special case in which the cross section is independent of the scattering angle, the resulting equations for the normalized moments are \emph{exactly} the same as those obtained in the ultrarelativistic limit.) In the late epoch, the internal energy density of the gas scales like $1/a(t)^5$ and in this case it is the cosmological term that dominates in the Friedmann equations and yields the typical exponential growing form of the scale factor. This growth is so fast that the limit $t\to \infty$ corresponds to a finite time scale for the moment equations. Hence, similar to what happens within the early epoch, the gas does  not necessarily reach an equilibrium configuration within the late epoch and generally  ``freezes" to a non-equilibrium configuration for large cosmological times. In particular, there is a large class of initial data for which the gas does not even approach a local equilibrium state. However, we have also identified sufficient conditions on the initial data for the distribution function which guarantee that the end state is a near-equilibrium configuration, the normalized moments approaching a value close to one.

It is clear from these results that it should be rather interesting to understand the dynamics of the full problem (spanning from the big bang to the late universe) and that this problem should serve as a nice example for the understanding of non-equilibrium phenomena in curved spacetimes in which a global equilibrium configuration is nonexistent. More generally, we hope the material discussed in this review will serve as a solid introduction to the relativistic Boltzmann equation and, although we have just briefly mentioned many interesting applications, it will serve as a starting point for future research in the field.

\acknowledgments
It is a pleasure to thank H\r{a}kan Andr\'easson, Ana Laura Garc\'ia Perciante, J\'er\'emie Joudioux, Paola Rioseco, Elmar Wagner, and Thomas Zannias for fruitful and stimulating discussions throughout the elaboration of this work. We also thank Thomas Zannias for comments on an earlier version of this manuscript. R.A. and C.G. were supported by a PhD CONACyT fellowship. O.S. was partially supported by a CIC Grant to Universidad Michoacana. We also acknowledge support from the CONACyT Network Project No. 376127 ``Sombras, lentes y ondas gravitatorias generadas por objetos compactos astrof\'isicos".

\appendix
\section{Manifold structure of the cotangent bundle}
\label{App:Cotangent}

In this appendix we provide the details for the proof of Lemma~\ref{Lem:CotangentSpace2n} and show that a differentiable atlas $(U_\alpha,\phi_\alpha)$ of $M$ induces a differentiable atlas $(V_\alpha,\psi_\alpha)$ of $T^*M$, where $V_\alpha := \pi^{-1}(U_\alpha) = \{ (x,p)\in U_\alpha: x\in U_\alpha, p\in T_x^*M \}$ and
\begin{eqnarray}
\psi_\alpha : V_\alpha &\longrightarrow& \phi_\alpha(U_\alpha)\times \Real^n
\subset\Real^{2n},\nonumber\\
(x,p) &\longmapsto& \left(x^{\mu},p_{\mu}\right),
\end{eqnarray}
with $\displaystyle x^{\mu} := \phi_\alpha\left(x\right)^{\mu}$ and $\displaystyle p_{\mu} := p\left(\left.\frac{\partial}{\partial x^{\mu}}\right|_{x}\right)$. Notice that, due to the properties of $(U_\alpha,\phi_\alpha)$, the sets $\psi_\alpha(V_\alpha) = \phi_\alpha(U_\alpha)\times \Real^n$ are open subsets of $\Real^{2n}$ and that the maps $\psi_\alpha$ are invertible with inverse $\psi_\alpha^{-1}: \phi_\alpha(U_\alpha)\times \Real^n\to V_\alpha$ given by
\begin{equation}
\psi_\alpha^{-1}\left(x^\mu, p_\mu \right) 
 = \left(\phi_\alpha^{-1}(x^{\mu}),p_\mu dx^{\mu}|_x\right),\qquad
 (x^\mu,p_\mu)\in \phi_\alpha(U_\alpha)\times \Real^n.
\end{equation}
Clearly, the local charts $(V_\alpha,\psi_\alpha)$ cover $T^* M$ since the charts $(U_\alpha,\phi_\alpha)$ cover $M$. To show that $(V_\alpha,\psi_\alpha)$ defines a differentiable atlas in $T^* M$ it remains to show that the transition maps are $C^\infty$-differentiable. For this, take two overlapping charts $\left(V_1,\psi_1\right)$ and $\left(V_2,\psi_2\right)$, say, such that $V_1\cap V_2\neq \emptyset$, i.e. $U_1\cap U_2\neq \emptyset$, see Fig~\ref{Fig:InducedMap_psi}.
\begin{figure}[h]
\begin{centering}
\includegraphics[scale=0.40]{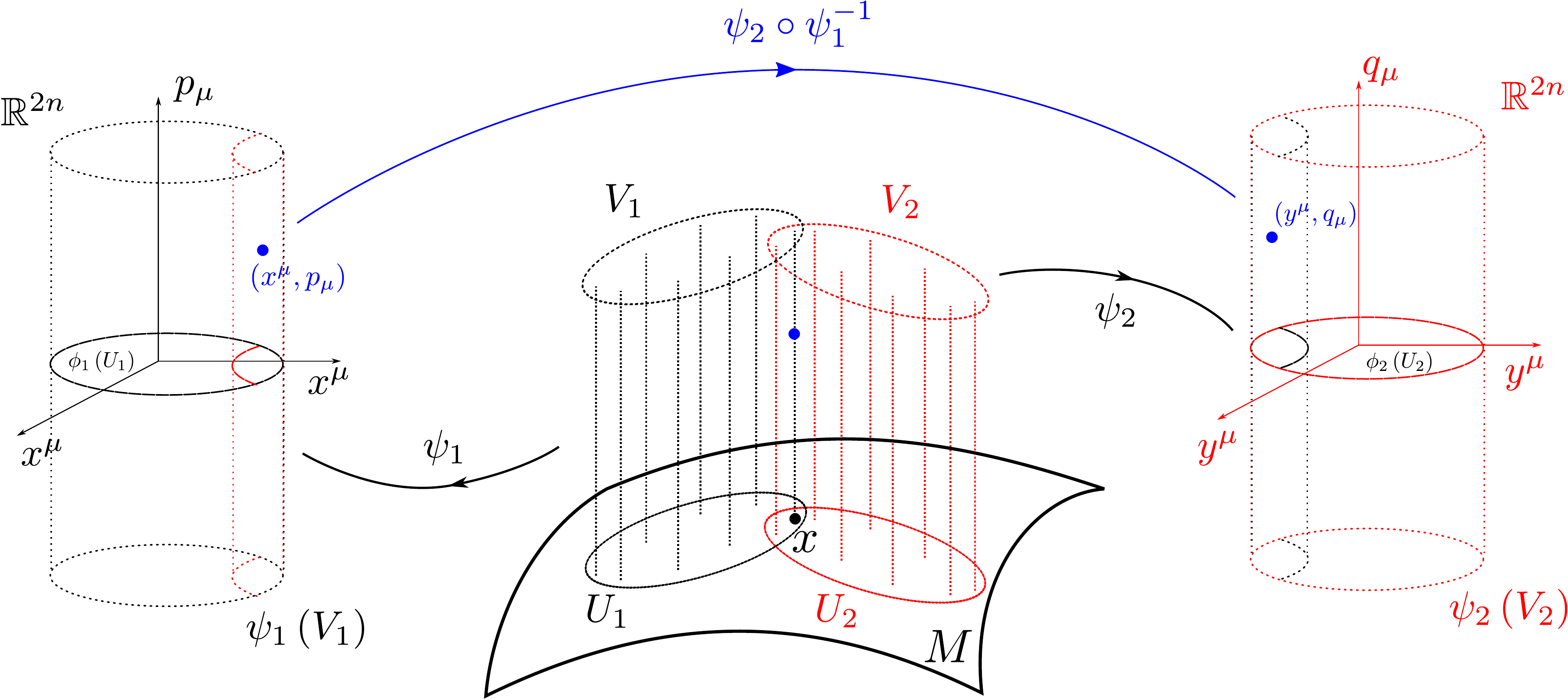}
\par\end{centering}
\caption{Illustration of the transition map $\psi_{12} = \psi_2\circ\psi_1^{-1}$ between the two overlapping local charts $\left(V_1,\psi_1\right)$ and $\left(V_2,\psi_2\right)$ of the cotangent bundle.}
\label{Fig:InducedMap_psi}
\end{figure} 
Let us call the corresponding local coordinates $\psi_1(x,p) = (x^\mu,p_\mu)$ and $\psi_2(x,p) = (y^\mu,q_\mu)$, respectively. The transition map $\psi_{12} := \psi_2 \circ \psi_1^{-1} : \psi_1(V_1\cap V_2)\to \psi_2(V_1\cap V_2)$ is invertible, because $\psi_1$ and $\psi_2$ are, and it is explicitly given by
\begin{equation}
 (y^\mu,q_\mu) = \psi_{12}(x^\mu,p_\mu) = \psi_2\left(\psi_1^{-1}(x^\mu,p_\mu)\right)
 = \psi_2 \left(\phi_1^{-1}(x^\mu),p_\mu dx^\mu_x \right)
 =  \left(\phi_{12}(x^\mu),p_\alpha \left.\frac{\partial x^\alpha}{\partial y^\mu}\right|_x \right),
 \label{Eq:TransitionMap}
\end{equation}
with $\phi_{12} := \phi_2\circ\phi_1^{-1}: \phi_1(U_1\cap U_2) \to \phi_2(U_1\cap U_2)$ the corresponding transition map on $M$. Since $\phi_{12}$ is $C^\infty$-differentiable, the same is true for its inverse Jacobi matrix $(\frac{\partial x^\alpha}{\partial y^\mu})$, and it follows that $\psi_{12}$ is $C^\infty$-differentiable as well. This proves that $T^* M$ is a $2n$-dimensional differential manifold.

To complete the proof of the lemma, it remains to show that $T^*M$ is orientable, that is, that $T^* M$ possesses an atlas with the property that all its transition maps preserve orientation (i.e. their Jacobi matrix have positive determinant). Let $M^\mu{}_\alpha{}(x) := \left.\frac{\partial y^\mu}{\partial x^\alpha}\right|_x$ be the Jacobi matrix of $\phi_{12}$ at $x\in U_1\cap U_2$. Then, according to Eq.~(\ref{Eq:TransitionMap}),  $q_\mu = \left(M^{-1}(x)\right)^\alpha{}_\mu p_\alpha$, and thus the Jacobi matrix of $\psi_{12}$ is
\begin{equation}
D\psi_{12}(x^{\mu},p_{\mu}) = \left(\begin{array}{cc}
\frac{\text{\ensuremath{\partial y^{\mu}}}}{\partial x^{\alpha}} & \frac{\text{\ensuremath{\partial y^{\mu}}}}{\partial p_{\alpha}}\\
\frac{\text{\ensuremath{\partial q_{\mu}}}}{\partial x^{\alpha}} & \frac{\text{\ensuremath{\partial q_{\mu}}}}{\partial p_{\alpha}}
\end{array}\right) = \left(\begin{array}{cc}
M^\mu{}_{\alpha}(x) & 0\\
\frac{\partial}{\partial x^\alpha}\left[\left(M^{-1}(x)\right)^{\beta}{}_{\mu}\right]p_{\beta} & \left(M^{-1}(x)\right)^{\alpha}{}_{\mu}
\end{array}\right),
\end{equation}
and it follows that $\det{D\psi_{12}(x^{\mu},p_{\mu})} = 1$. This shows that the differentiable atlas $(\psi_\alpha,V_\alpha)$ induced by the differentiable atlas $(U_\alpha,\phi_\alpha)$ is not only oriented, but also volume-preserving. In fact, this property follows directly from the fact  that the adapted local coordinates $(x^\mu,p_\mu)$ are symplectic coordinates on $T^* M$ (see subsection~\ref{SubSec:Ham}).

\section{Further details on the local coordinates $(x^\mu,p_{\hat{\alpha}})$ on $T^*M$}
\label{App:ON_basis}

Towards the end of subsection~\ref{SubSec:Liouville} we introduced the new local coordinates $(x^\mu,p_{\hat{\alpha}})$ in which the momentum covector $p = p_{\hat{\alpha}}\theta^{\hat{\alpha}}$ is expanded in terms of a local orthonormal frame of covector fields $\theta^{\hat{\alpha}}$ on $M$ instead of the expansion $p_\mu dx^\mu$ in terms of the coordinate basis $dx^\mu$. The relation between the two local coordinate systems $(x^\mu,p_\mu)$ and $(x^\mu,p_{\hat{\alpha}})$ is given by the transformation
\begin{equation}
p_{\mu} = p_{\hat{\alpha}}\theta^{\hat{\alpha}}{}_\mu,\qquad
\theta^{\hat{\alpha}}{}_\mu := \theta^{\hat{\alpha}}\left(\frac{\partial}{\partial x^{\mu}}\right) ,
\end{equation}
where $\theta^{\hat{\alpha}}{}_\mu$ are the components of the one-forms $\theta^{\hat{\alpha}}$ with respect to the coordinate basis $dx^\mu$, i.e. $\theta^{\hat{\alpha}} =
\theta^{\hat{\alpha}}{}_\mu dx^\mu$. For the purpose of this appendix, the orthonormality condition is unimportant; hence we shall consider an arbitrary frame $\theta^{\hat{\alpha}}$ of covectors on $M$.

In this appendix we show that the family of one-forms $Dp_\mu$ defined in Eq.~\eqref{Eq:dualbasis} can be interpreted as the absolute exterior derivative of $p_{\hat{\alpha}}$ (see, for instance, Section~15.8.1 in~\cite{Straumann-Book}), which is defined (in terms of an arbitrary frame) as
\begin{equation}
Dp_{\hat{\alpha}} := dp_{\hat{\alpha}} - \omega^{\hat{\beta}}{}_{\hat{\alpha}} p_{\hat{\beta}},
\end{equation}
with the connection one-form
\begin{equation}
\omega^{\hat{\beta}}{}_{\hat{\alpha}}(X) 
:= \theta^{\hat{\beta}}\left( \nabla_X e_{\hat{\alpha}} \right)
 = \hat{\Gamma}^{\hat{\beta}}{}_{\mu\hat{\alpha}} X^\mu,
\qquad X\in {\cal X}(M).
\end{equation}
Note that for a coordinate frame $\theta^{\hat{\alpha}} = dx^{\hat{\alpha}}$, the connection coefficients $\hat{\Gamma}^{\hat{\beta}}{}_{\mu\hat{\alpha}}$ coincide with the usual Christoffel symbols and in this case $Dp_\mu$ indeed reduces to the expression defined in Eq.~\eqref{Eq:dualbasis}. On the other hand, if $\theta^{\hat{\alpha}}$ is orthonormal, then $\hat{\Gamma}_{\hat{\beta}\mu\hat{\alpha}} := \eta_{\hat{\beta}\hat{\gamma}}\hat{\Gamma}^{\hat{\gamma}}{_{\mu\hat{\alpha}}}$ is \emph{antisymmetric} in the indices $\hat{\alpha}\hat{\beta}$ and $Dp_{\hat{\alpha}}$ reduces to the expression in Eq.~\eqref{Eq:Dphat}. With respect to a change of basis,
\begin{equation}
\bar{\theta}^{\hat{\alpha}} = A^{\hat{\alpha}}{}_{\hat{\beta}}\theta^{\hat{\beta}}
\end{equation}
one finds
\begin{equation}
\bar{p}_{\hat{\alpha}} = p_{\hat{\beta}} (A^{-1})^{\hat{\beta}}{}_{\hat{\alpha}},\qquad
\bar{\omega}^{\hat{\alpha}}{}_{\hat{\beta}} = A^{\hat{\alpha}}{}_{\hat{\gamma}}\omega^{\hat{\gamma}}{}_{\hat{\delta}} (A^{-1})^{\hat{\delta}}{}_{\hat{\beta}} 
 + A^{\hat{\alpha}}{}_{\hat{\gamma}} d(A^{-1})^{\hat{\gamma}}{}_{\hat{\beta}},
\end{equation}
which implies
\begin{equation}
D\bar{p}_{\hat{\alpha}} = Dp_{\hat{\beta}} (A^{-1})^{\hat{\beta}}{}_{\hat{\alpha}},
\end{equation}
that is, $Dp_{\hat{\alpha}}$ transforms like $p_{\hat{\alpha}}$. In particular, it follows that
\begin{equation}
Dp_{\mu} = (Dp_{\hat{\alpha}})\theta^{\hat{\alpha}}{}_\mu.
\end{equation}

\section{Volume forms on the mass hyperboloid and mass shell}
\label{App:VolumeFormOnTheHyperboloidOfMass}

In this appendix we provide an alternative derivation for the volume forms on the future mass hyperboloid $P_x^+(m)$ and the future mass shell $\Gamma_m^+$. As in the previous appendix, we work in local coordinates $(x^\mu,p_{\hat{\alpha}})$ in which the momentum $p = p_{\hat{\alpha}} \theta^{\hat{\alpha}}$ is expanded in terms of an arbitrary basis of covectors $\left\{ \theta^{\hat{\alpha}}\right\}$ on $M$.

We start with the volume form on the cotangent space $T_x^* M$, given by
\begin{equation}
\eta_{T_x^*M} = -\sqrt{-\det(g^{\hat{\alpha}\hat{\beta}})} 
dp_{\hat{0}} \wedge dp_{\hat{1}} \wedge\cdots \wedge dp_{\hat{d}}.
\end{equation}
The induced volume form on the future mass hyperboloid $P_x^+(m)$ can be computed by taking the interior derivative of $\eta_{T_x^*M}$ with respect to the normal vector $N$ given in Eq.~(\ref{Eq:NormalVector}). This yields
\begin{eqnarray*}
i_N\eta_{T_x^*M} &=& -\sqrt{-\det(g^{\hat{\alpha}\hat{\beta}})} \left[ 
(i_{N}dp_{\hat{0}}) \wedge dp_{\hat{1}} \wedge \cdots \wedge dp_{\hat{d}} 
- dp_{\hat{0}} \wedge (i_{N}dp_{\hat{1}}) \wedge \cdots \wedge dp_{\hat{d}} + \cdots 
+ (-1)^d dp_{\hat{0}} \wedge \cdots \wedge (i_{N} dp_{\hat{d}})
\right]\\
 &=& -\frac{1}{m}\sqrt{-\det(g^{\hat{\alpha}\hat{\beta}})} \left[ 
 p_{\hat{0}} dp_{\hat{1}} \wedge \cdots \wedge dp_{\hat{d}} 
 - p_{\hat{1}} dp_{\hat{0}} \wedge dp_{\hat{2}} \wedge \cdots \wedge dp_{\hat{d}} + \cdots 
  - (-1)^d p_{\hat{d}} dp_{\hat{0}} \wedge \cdots \wedge dp_{\widehat{d-1}} 
\right],
\end{eqnarray*}
where we have used the fact that $\displaystyle i_N dp_{\hat{\alpha}} = dp_{\hat{\alpha}}(N)=p_{\hat{\alpha}}/m$ in the second step. Next, we use the fact that $p^{\hat{\alpha}} dp_{\hat{\alpha}} = 0$ on the mass hyperboloid, such that
\begin{equation}
dp_{\hat{0}} = -\frac{p^{\hat{a}}}{p^{\hat{0}}} dp_{\hat{a}}.
\end{equation}
This yields
\begin{equation}
\eta_{P_x^+(m)} = \frac{-\sqrt{-\det(g^{\hat{\alpha}\hat{\beta}})}}{m p^{\hat{0}}} \left[ 
 p^{\hat{0}} p_{\hat{0}} dp_{\hat{1}} \wedge \cdots \wedge dp_{\hat{d}} 
 + p^{\hat{a}} p_{\hat{a}} dp_{\hat{1}} \wedge \cdots \wedge dp_{\hat{d}}\right]
  = \frac{m}{p^{\hat{0}}}\sqrt{-\det(g^{\hat{\alpha}\hat{\beta}})} 
  dp_{\hat{1}} \wedge  dp_{\hat{2}} \wedge\cdots \wedge dp_{\hat{d}}.
\end{equation}
For the particular case in which the basis $\left\{ \theta^{\hat{\alpha}}\right\}$ is \emph{orthonormal} one has $\det(g^{\hat{\alpha}\hat{\beta}}) = -1$ and this reduces to
\begin{equation}
\eta_{P_x^+(m)} = \frac{m}{p^{\hat{0}}} 
dp_{\hat{1}} \wedge  dp_{\hat{2}} \wedge\cdots \wedge dp_{\hat{d}},
\end{equation}
which agrees with Eq.~(\ref{Eq:etaPxm}).

We may compute the volume form $\eta_{\Gamma_m^+}$ on the future mass shell in a similar way. We start from the volume form on $T^*M$, see Eq.~(\ref{Eq:EtaOnT*M}),
\begin{equation}
\eta_{T^*M} = -\sqrt{-\det(g_{\mu\nu})} \sqrt{-\det(g^{\hat{\alpha}\hat{\beta}})} 
dp_{\hat{0}}\wedge dp_{1}\wedge\cdots \wedge dp_{\hat{d}}\wedge
dx^0\wedge dx^1\wedge\cdots\wedge dx^d
 =\eta_{T_x^* M} \wedge \eta_M.
\end{equation}
Applying Definition~\ref{Def:InducedVolumeForm}, using the fact that $i_N\eta_M = 0$ and the previous result for $i_N\eta_{T_x^* M}$ one finds
\begin{equation}
\eta_{\Gamma_m^+} = \eta_{P_x^+(m)} \wedge \eta_M = \eta_M\wedge \eta_{P_x^+(m)},
\end{equation}
which coincides with Eq.~(\ref{Eq:etaGammam+Short}).

\section{Further details on the collision manifold}
\label{App:Collisions}

In this appendix we make some remarks on the set
\begin{equation}
\tilde{C}_x := \{ (p,\pi,p^*,\pi^*) \in [P_x^+(m)]^4 : p + \pi  = p^* + \pi^* \},
\end{equation}
which contains the collision manifold $C_x$ defined in Eq.~(\ref{Eq:Cx}) as a subset. The difference between $\tilde{C}_x$ and $C_x$ is the condition $p\neq \pi$ which is relaxed in the definition of $\tilde{C}_x$. We show in this appendix that this leads to conical-type singularities at the points $(p,\pi,p^*,\pi^*)\in\tilde{C}_x$ for which $p = \pi$.

In order to prove this, consider for each fixed $\ve{p}^{cm}\in\Real^d$, $\ve{e},\ve{e}^*\in S^{d-1}$ the curve $\gamma(t)$ in $\tilde{C}_x\subset [P_x^+(m)]^4$ through the point $p_1 = p_2 = p^{cm}$, which is parametrized by setting
\begin{equation}
g = t,\qquad \ve{\hat q} = \ve{e},\qquad \ve{\hat q}^* = \ve{e}^*
\end{equation}
into Eqs.~(\ref{Eq:Cp1Vec},\ref{Eq:Cp2Vec},\ref{Eq:Cp1*Vec},\ref{Eq:Cp2*Vec}), where in this appendix we use the notation $(\ve{p},\ve{\pi})$ and $(\ve{p}^*,\ve{\pi}^*)$ instead of $(\ve{p}_1,\ve{p}_2)$ and $(\ve{p}_1^*,\ve{p}_2^*)$. Since $[P_x^+(m)]^4$ is a (product) manifold, we can compute the tangent vector of $\gamma(t)$ in $[P_x^+(m)]^4$ at the point $\gamma(0)$, which yields
\begin{equation}
\dot{\gamma}(0) = \frac{1}{2}\left( m e_{\hat{a}} 
 + \frac{\ve{p}^{cm}\cdot\ve{e}}{\sqrt{m^2 + |\ve{p}^{cm}|^2} + m} p^{cm}_{\hat{a}} \right)
\left( \frac{\partial}{\partial \pi_{\hat{a}}} - \frac{\partial}{\partial p_{\hat{a}}} \right)
 + \frac{1}{2}\left( m e^*_{\hat{a}}
 + \frac{\ve{p}^{cm}\cdot\ve{e}^*}{\sqrt{m^2 + |\ve{p}^{cm}|^2} + m} p^{cm}_{\hat{a}} \right)
\left( \frac{\partial}{\partial \pi_{\hat{a}}^*} - \frac{\partial}{\partial p_{\hat{a}}^*} \right).
\end{equation}
By taking linear combinations with different choices of $\ve{e}$ and $\ve{e}^*$, one obtains the $2d$ linearly independent tangent vectors
\begin{equation}
\frac{\partial}{\partial \pi_{\hat{a}}} - \frac{\partial}{\partial p_{\hat{a}}},\qquad
\frac{\partial}{\partial \pi_{\hat{a}}^*} - \frac{\partial}{\partial p_{\hat{a}}^*},\qquad
\hat{a} = 1,2,\ldots,d,
\end{equation}
at $T_{\gamma(0)} [P_x^+(m)]^4$. Furthermore, by considering the curve $\mu(t)$ in $\tilde{C}_x\subset [P_x^+(m)]^4$ defined by the substitutions
\begin{equation}
\ve{p}^{cm}\mapsto \ve{p}^{cm} + t\ve{e},\qquad g\mapsto 0,
\end{equation}
in Eqs.~(\ref{Eq:Cp1Vec},\ref{Eq:Cp2Vec},\ref{Eq:Cp1*Vec},\ref{Eq:Cp2*Vec}), one obtains the additional $d$ tangent vectors
\begin{equation}
\frac{\partial}{\partial p_{\hat{a}}} + \frac{\partial}{\partial \pi_{\hat{a}}}
 + \frac{\partial}{\partial p_{\hat{a}}^*} + \frac{\partial}{\partial \pi_{\hat{a}}^*},
 \qquad \hat{a} = 1,2,\ldots,d,
\end{equation}
at $T_{\gamma(0)} [P_x^+(m)]^4$. If $\tilde{C}_x$ were smooth at $\gamma(0)$, it would follow that the dimension of its tangent space at $\gamma(0)$ is at least $3d$-dimensional. However, this contradicts the result from Lemma~\ref{Lem:Cx} which implies that each point of $\tilde{C}_x$ for which $p\neq \pi$ possesses a neighborhood which can be smoothly parametrized by only $3d-1$ coordinates.

\section{Derivation of metric and volume form on the collision manifold}
\label{App:CollisionManifold}

This appendix is devoted to a detailed derivation of the induced metric and volume form on the collision manifold $C_x$. Since the calculation is long, we divide it into several steps. In a first step, we recall the induced metric $H_x$ on the future mass hyperboloid $P_x^+(m)$ and the induced volume form. In a second step, we introduce the ``momentum space of incoming particles" $I_x$, which consists of two copies of $P_x^+(m)$. The bulk of the calculation consists in representing the induced metric in terms of the coordinates $(\ve{p}^{cm},\ve{q})$, representing the center of mass momentum and the relative velocity. Likewise, one considers the ``momentum space of outgoing particles" $O_x$ with coordinates $((\ve{p}^*)^{cm},\ve{q}^*)$ which is just another copy of $I_x$. The collision manifold is equal to the submanifold of the product $I_x\times O_x$ subject to the restrictions $(\ve{p}^*)^{cm} = \ve{p}^{cm}$ and $|\ve{q}| = |\ve{q}^*|$, and in the final step we use these restrictions to compute the induced metric and volume form on $C_x$.

{\bf Step 1}: Recall that the metric $H_x$ on $P_x^+(m)$ which is induced from  $g_x^{-1}$ is given by Eq.~(\ref{Eq:Induced_metric_gx-1}):
\begin{equation}
H_x = \left(\delta^{\hat{a}\hat{b}}-\frac{p^{\hat{a}}p^{\hat{b}}}{m^{2}+|\ve{p}|^{2}}\right) dp_{\hat{a}}\otimes dp_{\hat{b}}.
\label{Eq:Hx}
\end{equation}
An alternative form for writing this metric is $H_x = \delta^{\hat{a}\hat{b}}d'p_{\hat{a}}\otimes d'p_{\hat{b}}$, with
\begin{equation}
d'p_{\hat{a}} := dp_{\hat{a}} - \left(1-\frac{m}{\sqrt{m^{2}+|\ve{p}|^{2}}}\right)\hat{p}_{\hat{a}}\hat{p}^{\hat{b}}dp_{\hat{b}}
 = dp_{\hat{a}}-\frac{p_{\hat{a}}}{\sqrt{m^{2}+|\ve{p}|^{2}}}\frac{p^{\hat{b}}dp_{\hat{b}}}{\sqrt{m^{2}+|\ve{p}|^{2}}+m},
\label{Eq:Hxprime}
\end{equation}
where $\hat{p}_{\hat{a}} := p_{\hat{a}}/|\ve{p}|$ and the second representation shows that $d'p_{\hat{a}}$ is well-defined everywhere including at $\ve{p} = \ve{0}$. The associated volume form is (see Eq.~(\ref{Eq:etaPxm}))
\begin{equation}
\eta_{P_x^+(m)} = \frac{m}{\sqrt{m^{2}+|\ve{p}|^{2}}}dp_{\hat{1}}\wedge dp_{\hat{2}}\wedge\cdots\wedge dp_{\hat{d}},
\end{equation}
which can be obtained either from Eq.~(\ref{Eq:Hx}) by using a particular frame in which $\ve{p} = (p_{\hat{1}},0,\ldots,0)$, or directly from Eq.~(\ref{Eq:Hxprime}).

{\bf Step 2}:
 In a second step we compute the induced metric and volume form on the ``momentum space of incoming particles" (where here we use the notation $(p,\pi)$ instead of $(p_1,p_2)$ in order to simplify the notation)
$$
I_x := \{ (p,\pi)\in P_x^+(m)\times P_x^+(m) : p\neq \pi \}.
$$
In terms of the global coordinates $(\ve{p},\ve{\pi})\in \Real^{2d}$ of $I_x$ the induced metric and associated volume form are simply:
\begin{equation}
H_x^{(2)} = 
\left(\delta^{\hat{a}\hat{b}}-\frac{p^{\hat{a}}p^{\hat{b}}}{m^{2}+|\ve{p}|^{2}}\right)dp{}_{\hat{a}}\otimes dp{}_{\hat{b}}
 + \left(\delta^{\hat{a}\hat{b}}-\frac{\pi^{\hat{a}}\pi^{\hat{b}}}{m^{2}+|\ve{\pi}|^{2}}\right)d\pi{}_{\hat{a}}\otimes d\pi{}_{\hat{b}}
\label{Eq:Hx2}
\end{equation}
and
\begin{equation}
\eta_{I_x} = \frac{m^{2}}{\sqrt{m^{2}+|\ve{p}|^{2}}\sqrt{m^{2}+|\ve{\pi}|^{2}}}
dp_{\hat{1}}\wedge dp{}_{\hat{2}}\wedge\cdots\wedge dp{}_{\hat{d}}\wedge
d\pi{}_{\hat{1}}\wedge d\pi{}_{\hat{2}}\wedge\cdots\wedge d\pi{}_{\hat{d}}.
\label{Eq:etaPx+2}
\end{equation}
However, it is also possible to parametrize the space $I_x$ using the coordinates $(\ve{p}^{cm},\ve{q})\in \Real^{2d}$, where $p^{cm}=-\sqrt{m^{2}+|\ve{p}^{cm}|^{2}}\theta^{\hat{0}}+ p^{cm}_{\hat{a}} \theta^{\hat{a}}$ is the center of mass momentum defined in Eq.~(\ref{Eq:pCM}) and $\ve{q}$ parametrizes the relative velocity of the two particles. According to Eqs.~(\ref{Eq:Cp1Vec},\ref{Eq:Cp2Vec}) the transformation between $(\ve{p},\ve{\pi})$ and $(\ve{p}^{cm},\ve{q})$ is given by
\begin{eqnarray}
\ve{p} &=& \sqrt{1 + \frac{1}{4} |\ve{q}|^2}\ve{p}^{cm}
 - \frac{1}{2}\left( m\ve{q}
 + \frac{\ve{p}^{cm}\cdot\ve{q}}{\sqrt{m^2 + |\ve{p}^{cm}|^2} + m}\ve{p}^{cm}
 \right),\\
\ve{\pi} &=& \sqrt{1 + \frac{1}{4} |\ve{q}|^2}\ve{p}^{cm}
 + \frac{1}{2}\left( m\ve{q}
 + \frac{\ve{p}^{cm}\cdot\ve{q}}{\sqrt{m^2 + |\ve{p}^{cm}|^2} + m}\ve{p}^{cm}
 \right).
\end{eqnarray}
In principle we could compute the Jacobian of this transformation and obtain the metric and volume form in terms of the coordinates $(\ve{p}^{cm},\ve{q})$ from Eqs.~(\ref{Eq:Hx2},\ref{Eq:etaPx+2}). This is clearly a tough calculation, so we proceed in a slightly different manner. Namely, we recall that $p$ and $\pi$ can be expressed in terms of Lorentz transformations,
$$
p = \Lambda(\pz),\qquad
\pi = \Lambda(\piz),
$$
where $\pz,\piz$ are the momenta of the incoming particles in the center of mass frame, that is, according to Eqs.~(\ref{Eq:Cp1Bis},\ref{Eq:Cp2Bis}),
\begin{eqnarray}
\pz = -\sqrt{1+\frac{1}{4}|\ve{q}|^{2}}m\theta^{\hat{0}}-\frac{m}{2}q_{\hat{a}}\theta^{\hat{a}},
\label{Eq:Cpz}\\
\piz = -\sqrt{1+\frac{1}{4}|\ve{q}|^{2}}m\theta^{\hat{0}}+\frac{m}{2}q_{\hat{a}}\theta^{\hat{a}},
\label{Eq:Cpiz}
\end{eqnarray}
and $\Lambda$ is the Lorentz boost defined in Eq.~(\ref{Eq:LorentzBoost}) which maps $-m \theta^{\hat{0}}$ to $p^{cm}$. Hence,
$$
dp_{\hat{\alpha}} = \Lambda_{\hat{\alpha}}{}^{\hat{\beta}} d\pz_{\hat{\beta}} 
+ d\Lambda_{\hat{\alpha}}{}^{\hat{\beta}}\pz_{\hat{\beta}}
 = \Lambda_{\hat{\alpha}}{}^{\hat{\beta}} \left( d\pz_{\hat{\beta}}
   + \Omega_{\hat{\beta}}{}^{\hat{\lambda}}\pz_{\hat{\lambda}} \right),
$$
with $\Omega_{\hat{\lambda}}{^{\hat{\beta}}}$ the one-form defined as $\Omega_{\hat{\alpha}}{^{\hat{\beta}}}:=(\Lambda^{-1})_{\hat{\alpha}}{^{\hat{\lambda}}}d\Lambda_{\hat{\lambda}}{^{\hat{\beta}}}$. Since $\Lambda$ is a Lorentz transformation, it follows that $\Omega^{\hat{\alpha}\hat{\beta}} + \Omega^{\hat{\beta}\hat{\alpha}} = 0$, and one obtains
$$
\eta^{\hat{\alpha}\hat{\beta}} dp_{\hat{\alpha}}\otimes dp_{\hat{\beta}} 
 = \eta^{\hat{\alpha}\hat{\beta}} d\pz_{\hat{\alpha}}\otimes d\pz_{\hat{\beta}}
 + d\pz_{\hat{\lambda}}\otimes\Omega^{\hat{\lambda}\hat{\mu}}\pz_{\hat{\mu}}
 + \Omega^{\hat{\mu}\hat{\lambda}}\pz_{\hat{\lambda}}\otimes d\pz_{\hat{\mu}}
 + \eta_{\hat{\alpha}\hat{\mu}} \Omega^{\hat{\alpha}\hat{\lambda}}\pz_{\hat{\lambda}}\otimes\Omega^{\hat{\mu}\hat{\nu}}\pz_{\hat{\nu}},
$$
and similarly for $\eta^{\hat{\alpha}\hat{\beta}} d\pi_{\hat{\alpha}}\otimes d\pi_{\hat{\beta}}$. To make further progress we note Eq.~(\ref{Eq:Cpz}) which yields immediately
$$
d\pz_{\hat{0}} = -\frac{1}{4}\frac{m}{\sqrt{1+\frac{1}{4}|\ve{q}|^{2}}}q^{\hat{a}}dq_{\hat{a}},\quad
d\pz_{\hat{a}} = -\frac{m}{2}dq_{\hat{a}}.
$$
Next, one notes that for any Lorentz boost as defined in Eqs.~(\ref{Eq:LorentzBoost1},\ref{Eq:LorentzBoost2}) one finds
\begin{eqnarray}
\Omega^{\hat{0}\hat{b}} &=&
  \gamma\left[\delta^{\hat{b}\hat{c}}+(\gamma-1)\hat{\beta}^{\hat{b}}\hat{\beta}^{\hat{c}}\right]d\beta_{\hat{c}},\\
\Omega^{\hat{a}\hat{b}} &=& -\frac{\gamma^{2}}{\gamma+1}\left(\beta^{\hat{a}}d\beta^{\hat{b}}-\beta^{\hat{b}}d\beta^{\hat{a}}\right),
\end{eqnarray}
which yields
$$
\Omega^{\hat{0}\hat{\beta}}\pz_{\hat{\beta}}=-\frac{1}{2}q^{\hat{b}}d'p^{cm}_{\hat{b}},\qquad
 \Omega^{\hat{a}\hat{\beta}}\pz_{\hat{\beta}}
 = \eta^{\hat{a}\hat{b}}\sqrt{1+\frac{1}{4}|\ve{q}|^{2}}d'p^{cm}_{\hat{b}}+\frac{1}{2}\vartheta^{\hat{a}\hat{b}}q_{\hat{b}},
$$
with $d'p^{cm}_{\hat{a}}$ defined in Eq.~(\ref{Eq:Hxprime}) above and
$$
\vartheta_{\hat{a}\hat{b}}:=\frac{p^{cm}_{\hat{a}}dp^{cm}_{\hat{b}} - p^{cm}_{\hat{b}}dp^{cm}_{\hat{a}}}
 {\sqrt{m^2 + |\ve{p}^{cm}|^2} + m}
 = \frac{p^{cm}_{\hat{a}}d'p^{cm}_{\hat{b}}-p^{cm}_{\hat{b}}d'p^{cm}_{\hat{a}}}
 {\sqrt{m^2 + |\ve{p}^{cm}|^2} + m}
 = -\vartheta_{\hat{b}\hat{a}}.
$$
The corresponding expressions for $d\piz_{\hat{a}}$ and $\Omega^{\hat{a}\hat{b}}\piz_{\hat{b}}$ are obtained from this by substituting $\ve{q}\mapsto -\ve{q}$. The induced metric $H_x^{(2)}$ in terms of the coordinates $(\ve{p}^{cm},\ve{q})$ is obtained by adding together the two contributions $\eta^{\hat{\alpha}\hat{\beta}} dp_{\hat{\alpha}}\otimes dp_{\hat{\beta}}$ and $\eta^{\hat{\alpha}\hat{\beta}} d\pi_{\hat{\alpha}}\otimes d\pi_{\hat{\beta}} $. This yields, finally,
\begin{equation}
H_x^{(2)} = \left( \delta^{\hat{a}\hat{b}} - \frac{1}{4}\frac{q^{\hat{a}} q^{\hat{b}}}{1 + \frac{1}{4}|\ve{q}|^2} \right)
\left[ 2\left( 1 + \frac{1}{4}|\ve{q}|^2 \right) d'p^{cm}_{\hat{a}}\otimes d'p^{cm}_{\hat{b}} 
 + \frac{m^2}{2} Dq_{\hat{a}}\otimes Dq_{\hat{b}} \right],
\end{equation}
where we recall the definition
$$
d'p^{cm}_{\hat{a}} = dp^{cm}_{\hat{a}} - \frac{p^{cm}_{\hat{a}}}{\sqrt{m^2 + |\ve{p}^{cm}|^2}}\frac{(p^{cm})^{\hat{b}} dp^{cm}_{\hat{b}}}{\sqrt{m^2 + |\ve{p}^{cm}|^2} + m},
$$
and where we have also introduced the notation
$$
Dq_{\hat{a}} := dq_{\hat{a}} - \frac{1}{m}\vartheta_{\hat{a}\hat{b}} q^{\hat{b}}
  = dq_{\hat{a}} - \frac{q^{\hat{b}}}{m\left(\sqrt{m^2 + |\ve{p}^{cm}|^2} + m\right)}\left[ p^{cm}_{\hat{a}} dp^{cm}_{\hat{b}}
 - p^{cm}_{\hat{b}} dp^{cm}_{\hat{a}} \right].
$$
(Note that $q^{\hat{a}} Dq_{\hat{a}} = q^{\hat{a}} dq_{\hat{a}}$ since $\vartheta_{\hat{a}\hat{b}} = -\vartheta_{\hat{a}\hat{b}}$ is antisymmetric.) The associated volume form is
\begin{equation}
\eta_{I_x} = m^d\left( 1 + \frac{1}{4}|\ve{q}|^2 \right)^{\frac{d}{2}-1}
\frac{m}{\sqrt{m^2 + |\ve{p}^{cm}|^2}} dp^{cm}_{\hat{1}}\wedge dp^{cm}_{\hat{2}}\wedge\cdots
\wedge dp^{cm}_{\hat{d}}\wedge dq_{\hat{1}}\wedge dq_{\hat{2}}\wedge\cdots\wedge dq_{\hat{d}}.
\end{equation}

If we split $q_{\hat{a}} = g\hat{q}_{\hat{a}}$, $g:=|\ve{q}|$, $\ve{\hat q}\in S^{d-1}$, then a short calculation reveals that
\begin{equation}
H_x^{(2)} = 2\left( 1 + \frac{g^2}{4} \right)
\left( \delta^{\hat{a}\hat{b}} - \frac{1}{4}\frac{g^2}{1 + \frac{1}{4}g^2}\hat{q}^{\hat{a}} \hat{q}^{\hat{b}} \right)
d'p^{cm}_{\hat{a}}\otimes d'p^{cm}_{\hat{b}} 
 + \frac{m^2}{2}\left( \frac{dg\otimes dg}{1 + \frac{1}{4}g^2} + g^2 D\Omega(\hat{q})^2 \right),
\end{equation}
with
\begin{equation}
D\Omega(\hat{q})^2 = \delta^{\hat{a}\hat{b}} (D\hat{q}_{\hat{a}})\otimes (D\hat{q}_{\hat{b}}),
\end{equation}
and
\begin{equation}
\eta_{I_x} = m^d\left( 1 + \frac{g^2}{4} \right)^{\frac{d}{2}-1}
\frac{m}{\sqrt{m^2 + |\ve{p}^{cm}|^2}} dp^{cm}_{\hat{1}}\wedge dp^{cm}_{\hat{2}}\wedge\cdots
\wedge dp^{cm}_{\hat{d}}\wedge g^{d-1} dg\wedge d\Omega(\hat{q}),
\label{Eq:pix2Bis}
\end{equation}
with $d\Omega(\hat{q})$ the solid angle belonging to the unit vector field $\ve{\hat q}\in S^{d-1}$. For $m=1$ and $d=3$ this agrees precisely with Eq.~(3.6) in Ref.~\cite{wI63}.

{\bf Step 3}: Similarly, we may consider the ``momentum space of outgoing particles":
$$
O_x := \{ (p^*,\pi^*)\in P_x^+(m)\times P_x^+ (m): p^*\neq \pi^* \}.
$$
The metric and volume form are the same as in the previous case, we just need to replace $(\ve{p},\ve{\pi})$ with $(\ve{p}^*,\ve{\pi}^*)$ and $(\ve{p}^{cm},g,\ve{\hat q})$ with
$((\ve{p}^*)^{cm},g^*,\ve{\hat q}^*)$.

{\bf Step 4}: After the above remarks we are ready to compute the induced metric and volume form on the collision manifold $C_x$. For this, we observe that $C_x$ is just the submanifold of
$$
I_x\times O_x
$$
for which $p^{cm} = (p^*)^{cm}$ and $g = g^*$. Therefore, the induced metric on $C_x$ can be obtained by adding together the metrics $H_x^{(2)}$ on $I_x$ and $O_x$ and setting $(\ve{p}^*)^{cm} = \ve{p}^{cm}$ and $g^* = g$ in the final result. This yields the following expression for $\gamma_x$ in terms of the coordinates $(\ve{p}^{cm},g,\ve{\hat q},\ve{\hat q}^*)$ on $C_x$:
\begin{equation}
\gamma_x = 4\left( 1 + \frac{g^2}{4} \right)
\left( \delta^{\hat{a}\hat{b}} - \frac{1}{4}\frac{g^2}{1 
 + \frac{1}{4}g^2}\frac{\hat{q}^{\hat{a}}\hat{q}^{\hat{b}} + (\hat{q}^*)^{\hat{a}}(\hat{q}^*)^{\hat{b}}}{2} \right)
d'p^{cm}_{\hat{a}}\otimes d'p^{cm}_{\hat{b}} 
 + m^2\left( \frac{dg\otimes dg}{1 + \frac{1}{4}g^2} 
 + g^2\frac{D\Omega(\hat{q})^2 +  D\Omega(\hat{q}^*)^2}{2} \right).
\end{equation}
To computation of the associated volume form involves the determinant of the $d\times d$ matrix whose coefficients are
$$
\delta^{\hat{a}\hat{b}} - \frac{1}{4}\frac{g^2}{1 
 + \frac{1}{4}g^2}\frac{\hat{q}^{\hat{a}}\hat{q}^{\hat{b}} + (\hat{q}^*)^{\hat{a}}(\hat{q}^*)^{\hat{b}}}{2}.
$$
In order to calculate this determinant, we use a basis for which
$$
\ve{\hat q} = (1,0,0,\ldots,0),\qquad
\ve{\hat q}^* = (\cos\Theta,\sin\Theta,0,\ldots,0),
$$
with $\Theta$ the scattering angle defined in Eq.~(\ref{Eq:ScatteringAngle}). In this basis, the matrix has the form
$$
\left( \begin{array}{ccccc}
 1 - \frac{1}{8}\frac{g^2}{1 + \frac{1}{4} g^2}(1 + \cos^2\Theta)
& -\frac{1}{8}\frac{g^2} {1 + \frac{1}{4} g^2}\cos\Theta\sin\Theta & 0 & \ldots & 0 \\
 -\frac{1}{8}\frac{g^2} {1 + \frac{1}{4} g^2}\cos\Theta\sin\Theta
& 1 - \frac{1}{8}\frac{g^2}{1 + \frac{1}{4} g^2}\sin^2\Theta & 0 & \ldots & 0 \\
0 & 0 & 1 & \ldots & 0 \\
\vdots & \vdots & \vdots & & \vdots \\
0 & 0 & 0 & \ldots & 1 \end{array} \right),
$$
and the determinant is
$$
\frac{1}{\left( 1 + \frac{1}{4}g^2 \right)^2}\left[
 \left( 1 + \frac{g^2}{8} \right)^2 - \left( \frac{g^2}{8} \right)^2\cos^2\Theta \right]. 
$$
Using these observations, one finds the expression for the volume form announced in Eq.~(\ref{Eq:mux}).

\section{Symmetries}
\label{App:Symmetries}

In many situations one is interested in restricting the gas configuration to satisfy certain symmetries (i.e. stationarity, spherical symmetry, axisymmetry etc.), which leads to the need of understanding how such symmetries should be imposed on the distribution function. In this appendix the necessary tools for performing this task are discussed. We restrict ourselves to the case of continuous (as opposed to discrete) symmetry groups.

Before talking about a specific symmetry that should be imposed on the gas configuration, it is important to realize that this only makes sense in general if the underlying spacetime itself possesses the required symmetry.\footnote{There also exists the possibility of encountering situations in which ``hidden symmetries" arise which do not originate from a symmetry of spacetime itself but rather from a symmetry of the associated cotangent bundle. Such symmetries are described by symplectic vector fields which commute with the Liouville vector field. A prominent example is the geodesic flow in the Kerr spacetime~\cite{MTW-Book}.} For example, the notation of a steady-state axisymmetric configuration only makes sense if spacetime itself is both stationary and axisymmetric. For this reason, we will start with the assumption that spacetime $(M,g)$ possesses a one-parameter group of isometries generated by a Killing vector field $\xi\in {\cal X}(M)$, or more generally, that spacetime possesses a Lie group $G$ of isometries generated by a finite number of Killing vector fields $\xi_1,\xi_2,\ldots,\xi_r$ satisfying commutation relations $[\xi_a,\xi_b] = C^d{}_{ab} \xi_d$ with associated structure constants $C^d{}_{ab}$. As we will see shortly, there is a natural way of lifting each Killing vector field $\xi$ to a vector field $\hat{\xi}\in {\cal X}(T^*M)$ on the cotangent bundle. This lifted vector field, called the \emph{complete lift} of $\xi$, satisfies many nice properties which we shall summarize in this appendix (see~\cite{pRoS17} for further details and~\cite{oStZ14b} and references therein for the corresponding formulation on the tangent bundle). In particular, $\hat{\xi}$ commutes with the Liouville vector field $L$ and generates a one-parameter group of isometries of the cotangent bundle $(T^*M,\hat{g})$ endowed with the Sasaki metric. In this way, a Lie group $G$ of isometries on $(M,g)$ naturally lifts to a Lie group of isometries on $(T^*M,\hat{g})$ where the lifted generators $\hat{\xi}_1,\hat{\xi}_2,\ldots,\hat{\xi}_r$ can be shown to satisfy the commutation relation $[\hat{\xi}_a,\hat{\xi}_b] = C^d{}_{ab}\hat{\xi}_d$ with the same structure constants $C^d{}_{ab}$. A distribution function $f$ is then called $G$-symmetric if it is invariant with respect to the flow generated by each $\hat{\xi}_a$, i.e. if
\begin{equation}
\pounds_{\hat{\xi}_a}[f] = 0,\qquad a=1,2,\ldots,r.
\label{Eq:Gsymmetric}
\end{equation}
One question that arises now is whether or not the conditions~(\ref{Eq:Gsymmetric}) are compatible with the relativistic Boltzmann equation. Suppose that $f$ is a solution of Eq.~(\ref{Eq:Boltzmann}):
\begin{equation}
\pounds_{L_F} [f] = C_W[f,f],
\label{Eq:BoltzmannBis}
\end{equation}
with $L_F$ the Liouville vector field (see Eq.~(\ref{Eq:LiouvilleVFCharged})) and $C_W[f,f]$ the collision term (see Eq.~(\ref{Eq:CollisionTerm1}) or (\ref{Eq:CollisionTerm3})). Applying the vector fields $\hat{\xi}_a$ on both sides of Eq.~(\ref{Eq:BoltzmannBis}) yields
\begin{equation}
\pounds_{L_F}\pounds_{\hat{\xi}_a} f = \pounds_{ [L_F,\hat{\xi}_a] } f 
+ \pounds_{\hat{\xi}_a} C_W[f,f].
\label{Eq:IntegrabilityCondition}
\end{equation}
In the absence of collisions and in the uncharged case when $L_F = L$, the right-hand side vanishes automatically and it follows that the conditions~(\ref{Eq:Gsymmetric}) are compatible with Eq.~(\ref{Eq:BoltzmannBis}). In order for this to be the case for the full Boltzmann equation one needs to guarantee that $[L_F,\hat{\xi}_a] $ and $\pounds_{\hat{\xi}_a} C_W[f,f]$ vanish if $f$ satisfies~(\ref{Eq:Gsymmetric}), and these conditions will form part of the discussion of this appendix.

\subsection{The complete lift on $T^*M$}

Let $\varphi^\lambda: M\to M$ be a one-parameter group of diffeomorphisms on the base manifold $M$ with corresponding infinitesimal generator $\xi\in {\cal X}(M)$. Irrespectively of whether or not $\xi$ is a Killing vector field, we can naturally lift $\varphi^\lambda$ to the cotangent bundle $T^*M$ by defining~\cite{pRoS17}
\begin{equation}
\hat{\varphi}^\lambda: T^*M\to T^*M,\quad
(x,p)\mapsto \hat{\varphi}^\lambda(x,p) := (\varphi^\lambda(x), [(d\varphi^\lambda_x)^*]^{-1}(p) ),
\label{Eq:LiftedFlow}
\end{equation}
where $d\varphi^\lambda_x : T_x M\to T_{\varphi^\lambda(x)} M$ denotes the differential (or push-forward) of the map $\varphi^\lambda$ at the point $x\in M$ and $(d\varphi^\lambda_x)^* : T_{\varphi^\lambda(x)}^* M\to T_x^* M$ its adjoint (or pull-back). It is a simple matter to verify that $\hat{\varphi}^\lambda$ defines a one-parameter group of diffeomorphisms on $T^*M$.

\begin{definition}
\label{Def:CompleteLift}
Let $\xi\in {\cal X}(M)$, and let $\varphi^\lambda: M\to M$ be the (local) one-parameter group of diffeomorphisms generated by $\xi$.  Consider the corresponding lifted group $\hat{\varphi}^\lambda: T^*M\to T^*M$ defined by Eq.~(\ref{Eq:LiftedFlow}). Its generator
\begin{equation}
\hat{\xi}_{(x,p)} := \left. \frac{d}{d\lambda} \right|_{\lambda = 0} \hat{\varphi}^\lambda(x,p)
\label{Eq:CompleteLift}
\end{equation}
is called the complete lift of $\xi$.
\end{definition}

In adapted local coordinates $(x^\mu,p_\mu)$ one finds the following expression:
\begin{equation}
\boxed{
\hat{\xi}_{(x,p)} = \xi^\mu(x)\left. \frac{\partial}{\partial x^\mu} \right|_{(x,p)}
 - p_\alpha\frac{\partial\xi^\alpha}{\partial x^\mu}(x)
 \left. \frac{\partial}{\partial p_\mu}  \right|_{(x,p)},\qquad
\xi_x = \xi^\mu(x)\left. \frac{\partial}{\partial x^\mu} \right|_x\, .}
\label{Eq:xihat}
\end{equation}

The most important properties of the complete lift are summarized in the next proposition (cf. Proposition~1 in~\cite{pRoS17}, Proposition~4 in~\cite{oStZ14b} and references therein) whose proof is included for completeness. For this, we recall the definitions in Eqs.~(\ref{Eq:Omega},\ref{Eq:Sasaki}) of the symplectic form $\Omega_s$ and Sasaki metric $\hat{g}$ on $T^* M$.

\begin{proposition}
\label{Prop:CompleteLift}
Let $\xi,\eta\in {\cal X}(M)$. Then one has:
\begin{enumerate}
\item[(i)] $[\hat{\xi},\hat{\eta}] = \widehat{[\xi,\eta]}$ for all $\xi,\eta\in {\cal X}(M)$, i.e. the complete lift preserves the Lie-brackets.
\item[(ii)] $\pounds_{\hat{\xi}}\Omega_s = 0$, i.e. $\hat{\xi}$ generates a symplectic flow on $(T^*M,\Omega_s)$.
\item[(iii)] $\hat{\xi} = X_{\mathcal{F}}$ with $X_{\mathcal{F}}$ the Hamiltonian vector field generated by the function $\mathcal{F} = p(\xi)$.
\item[(iv)] Let $m > 0$. $\hat{\xi}$ is tangent to the future mass shell $\Gamma_m^+$ at each point  $(x,p)\in \Gamma_m^+$ if and only if $\xi$ is a Killing vector field of $(M,g)$.
\item[(v)] $[L,\hat{\xi}] = 0$ if and only if $\xi$ is a Killing vector field of $(M,g)$.
\item[(vi)] $\xi$ is a Killing vector field of $(M,g)$ if and only $\hat{\xi}$ is a Killing vector field of $(T^*M,\hat{g})$.
\end{enumerate}
\end{proposition}

\proof (i) can be verified directly using Eq.~(\ref{Eq:xihat}). As for (ii) and (iii) we use Eqs.~(\ref{Eq:Omega},\ref{Eq:xihat}) and compute
\begin{equation}
i_{\hat{\xi}}\Omega_s = i_{\hat{\xi}}(dp_\mu\wedge dx^\mu)
 = dp_\mu(\hat{\xi}) dx^\mu - dx^\mu(\hat{\xi}) dp_\mu
 = -p_\alpha\frac{\partial \xi^\alpha}{\partial x^\mu} dx^\mu - \xi^\mu dp_\mu
 = -d\mathcal{F},
\label{Eq:ihatxiOmega}
\end{equation}
with $\mathcal{F} = p(\xi) = p_\mu\xi^\mu$. It follows that $\hat{\xi} = X_{\mathcal{F}}$ is the Hamiltonian vector field associated with $\mathcal{F}$ (see Definition~\ref{Def:HamVF}), which implies (ii) and (iii).

As for (iv) and (v), we remark first that $\hat{\xi}$ is tangent to $\Gamma_m^+$ if and only if $\pounds_{\hat{\xi}}\mathcal{H} = d\mathcal{H}(\hat{\xi}) = \{ \mathcal{F},\mathcal{H} \} = 0$ at each $(x,p)\in \Gamma_m^+$, since $\Gamma_m^+$ is a level surface of the free-particle Hamiltonian $\mathcal{H}$ (see Eq.~(\ref{Eq:FreeParticleH})). On the other hand, combining  the definition of $\mathcal{F}$ with Eqs.~(\ref{Eq:FreeParticleH},\ref{Eq:PoissonBracketCoordinate}) one finds for all $(x,p)\in T^* M$,
\begin{equation}
\{ \mathcal{F},\mathcal{H} \} = \frac{1}{2} (\pounds_\xi g)^{\mu\nu} p_\mu p_\nu
 = \frac{1}{2}(\pounds_\xi g^{-1})(p,p),
\label{Eq:PoissonFH}
\end{equation}
which implies (iv). To prove (v) we use the identity\footnote{This identity follows from the identity $[\pounds_X,i_Y] = i_{[X,Y]}$ by noticing that $i_{[ X_{\mathcal{F}},X_{\mathcal{H}}]}\Omega_s = \pounds_{X_{\mathcal{F}}} i_{\mathcal{H}}\Omega_s = -\pounds_{X_{\mathcal{F}}} d\mathcal{H} = -d\pounds_{X_{\mathcal{F}}}\mathcal{H} = -d\{ \mathcal{F},\mathcal{H} \}$.}
\begin{equation}
[ X_{\mathcal{F}},X_{\mathcal{H}}] = X_{ \{ \mathcal{F},\mathcal{H} \} },
\end{equation}
which shows that $\hat{\xi} = X_{\mathcal{F}}$ and $L = X_{\mathcal{H}}$ commute with each other if and only if $\{\mathcal{F},\mathcal{H} \}$ is constant. By taking $p=0$ in Eq.~(\ref{Eq:PoissonFH}) one concludes that this constant must be zero, and (v) follows.

Finally, to prove (vi) we evaluate the right-hand side of the identity
\begin{equation}
(\pounds_{\hat{\xi}}\hat{g})(X,Y) 
 = \hat{\xi}[ \hat{g}(X,Y) ] - \hat{g}(\pounds_{\hat{\xi}} X,Y) - \hat{g}(X,\pounds_{\hat{\xi}} Y)
\label{Eq:Lieghat}
\end{equation}
for the particular basis of vector fields $\displaystyle X,Y =  \frac{D}{dx^{\mu}}, \frac{\partial}{\partial p_{\mu}}$ (see Eq.~(\ref{Eq:Ddx})). In order to compute the Lie derivatives on the right-hand side of Eq.~(\ref{Eq:Lieghat}) it is useful to rewrite Eq.~(\ref{Eq:xihat}) in terms of covariant derivatives, such that
\begin{equation}
\boxed{ \hat{\xi} = \xi^\mu \frac{D}{dx^\mu} 
- p_\alpha(\nabla_\mu\xi^\alpha) \frac{\partial}{\partial p_\mu},}
\end{equation}
and use the commutation relations~(\ref{Eq:Commutators}) to find:
\begin{eqnarray}
\pounds_{\hat{\xi}} \frac{D}{dx^{\mu}} &=& \left[ \hat{\xi}, \frac{D}{dx^{\mu}} \right] 
 = -\frac{\partial\xi^\alpha}{\partial x^\mu}\frac{D}{dx^\alpha}  
+  \left( \nabla_\mu\nabla_\beta\xi^\alpha - R^\alpha{}_{\beta\mu\sigma}\xi^\sigma \right)
 p_\alpha\frac{\partial}{\partial p_\beta}, \\
\pounds_{\hat{\xi}} \frac{\partial}{\partial p_\mu} &=& \left[ \hat{\xi}, \frac{\partial}{\partial p_\mu} \right] = \frac{\partial\xi^\mu}{\partial x^\alpha}\frac{\partial}{\partial_\alpha}.
\end{eqnarray}
Taking into account the expressions~(\ref{Eq:ghatComponents}) for the components of the Sasaki metric it follows that
\begin{eqnarray}
(\pounds_{\hat{\xi}}\hat{g})\left( \frac{D}{dx^\mu},\frac{D}{dx^\nu} \right) 
 &=& \pounds_\xi g_{\mu\nu},\\
(\pounds_{\hat{\xi}}\hat{g})\left( \frac{\partial}{\partial p_\mu},\frac{\partial}{\partial p_\nu} \right) &=& \pounds_\xi g^{\mu\nu},\\
(\pounds_{\hat{\xi}}\hat{g})\left( \frac{D}{dx^\mu},\frac{\partial}{\partial p_\nu} \right) 
 &=& -\left( \nabla_\mu\nabla_\nu \xi^\alpha - R^\alpha{}_{\nu\mu\sigma}\xi^\sigma \right) p_\alpha.
\end{eqnarray}
Since any Killing vector field $\xi$ on $(M,g)$ satisfies the equation $\nabla_\mu\nabla_\nu \xi^\alpha = R^\alpha{}_{\nu\mu\sigma}\xi^\sigma$, the statement (v) follows immediately from these identities.
\qed

It follows from the previous proposition that a Lie group $G$ of isometries on $(M,g)$ naturally lifts to a group $G$ of symplectic isometries on $(T^*M,\Omega_s,\hat{g})$. Furthermore, the action of this group on $T^* M$ leaves the future mass shell $\Gamma_m^+$ invariant and commutes with the Liouville vector field $L$. In particular, it follows that this action is an isometry of $(\Gamma_m^+,\hat{h})$ with $\hat{h}$ the induced metric on the future mass shell. Note also that the current density vector field ${\cal J} = f L/m$ (see Eq.~(\ref{Eq:JDef})) associated with a $G$-symmetric distribution function $f$ is also invariant with respect to $G$. 

The next result implies that a $G$-symmetric distribution function has associated to it $G$-symmetric observables.

\begin{theorem}
\label{Thm:TsLieIdentity}
Let $f\in C_0^\infty(\Gamma_m^+)$ and $s\in \Natural_0$. Then, for any Killing vector field $\xi\in {\cal X}(M)$ on $(M,g)$ the $s$-rank symmetric tensor field $T^{(s)}$ defined in Eq.~(\ref{Eq:Ts}) satisfies
\begin{equation}
\boxed{\pounds_\xi T_x^{(s)}(X_1,X_2,\ldots,X_s)
 = \int\limits_{P_x^+(m)} (\pounds_{\hat{\xi}} f)(x,p) p(X_1) p(X_2)\cdots p(X_s)
 \dvol_x(p),}
\label{Eq:TsLieIdentity}
\end{equation}
for all $x\in M$ and $X_1,X_2,\ldots,X_s\in T_x M$. In particular, $\pounds_{\hat{\xi}} f = 0$ implies $\pounds_\xi T^{(s)} = 0$ if $\xi$ is a Killing vector field.
\end{theorem}

\proof
The proof proceeds in a similar way as the one of Theorem~\ref{Thm:sMoments}. We start with the more elegant version based on Gauss' theorem. Let $K\subset M$ be a compact subset of $M$ with $C^\infty$-boundary $\partial K = S$ and unit outward normal $s$, and let $V = \{ (x,p) : x\in K, p\in P_x^+(m) \}$ be the corresponding subset of the future mass shell $\Gamma_m^+$ with boundary $\partial V$ whose unit normal is $\displaystyle \nu = s^\mu\frac{D}{dx^\mu}$ (see Eq.~(\ref{Eq:nu})). Using the Fubini-type formula~(\ref{Eq:Fubini}), Gauss' theorem and the fact that $\hat{\xi}$ is divergence-free on $(\Gamma_m^+,\hat{h})$ due to its Killing property, yields
\begin{eqnarray*}
\int\limits_K \left( \int\limits_{P_x^+(m)} \pounds_{\hat{\xi}} f(x,p) \dvol_x(p) \right) \eta_M
&=& \frac{1}{m}\int\limits_V \pounds_{\hat{\xi}} f \eta_{\Gamma_m^+}
 = \frac{1}{m}\int\limits_V \divrg_{\hat{h}}( f\hat{\xi} ) \eta_{\Gamma_m^+}
 = \frac{1}{m}\int\limits_{\partial V} \hat{h}(f\hat{\xi},\nu) \eta_{\partial V}\\
&=& \frac{1}{m}\int\limits_{\partial V} g(\xi,s) f \eta_{\partial V}
 = \int\limits_{\partial K} F(x) g(\xi,s) \eta_S
 = \int\limits_K \divrg_g (F\xi) \eta_M,
\end{eqnarray*}
where we have set
\begin{equation}
F(x) := \int\limits_{P_x^+(m)} f(x,p) \dvol_x(p) = T^{(0)}(x).
\end{equation}
Since $\xi$ is divergence-free on $(M,g)$, $\divrg_g(F\xi) = \pounds_\xi F$ and the statement follows for $s = 0$. For $s = 1$ one takes a vector field $X\in {\cal X}(M)$ and replaces $f(x,p)$ with $\hat{f}(x,p) := f(x,p) p(X)$ in the identity~(\ref{Eq:TsLieIdentity}) with $s=0$ that has just been proven. This gives
\begin{equation}
\pounds_\xi[ T^{(1)}(X) ] = \int\limits_{P_x^+(m)} \pounds_{\hat{\xi}} \hat{f}(x,p) \dvol_x(p)
 = \int\limits_{P_x^+(m)} \left[ (\pounds_{\hat{\xi}} f) p(X) + f p(\pounds_\xi X) \right]\dvol_x(p),
\end{equation}
where in the last step we have used the identity $\pounds_{\hat{\xi}}[ p(X) ] = p(\pounds_\xi X)$ which follows directly from Eq.~(\ref{Eq:xihat}). It follows from this that
\begin{equation}
\pounds_\xi T^{(1)}(X) = \pounds_\xi[ T^{(1)}(X) ] - T^{(1)}(\pounds_\xi X)
 = \int\limits_{P_x^+(m)} (\pounds_{\hat{\xi}} f) p(X) \dvol_x(p),
\end{equation}
which proves the identity for $s = 1$. The proof for $s\geq 2$ is analogous.

A simple alternative proof is based on the use of the mixed local coordinates $(x^\mu,p_{\hat{\alpha}})$ introduced towards the end of subsection~\ref{SubSec:Liouville}, where $p_{\hat{\alpha}}$ refer to the components of $p$ with respect to an orthonormal basis $\{ e_{\hat{\alpha}} \}$ of vector fields. In terms of these coordinates, the complete lift of $\xi$ has the following representation:
\begin{equation}
\hat{\xi} = \xi^\mu\frac{\partial}{\partial x^\mu} + (\pounds_\xi e_{\hat{\alpha}})^\mu p_\mu\frac{\partial}{\partial p_{\hat{\alpha}}}.
\end{equation}
Since $\xi$ is a Killing vector field one can choose $\pounds_\xi e_{\hat{\alpha}} = 0$, and then the theorem with $s=0$ follows by applying the operator $\xi^\mu\partial_\mu$ on both sides of Eq.~(\ref{Eq:TsCoord}) with $\dvol_x(p)$ expressed in terms of orthonormal components of $p$.
\qed

\subsection{The complete lift on $T^*C$}

Next, we lift the (local) one-parameter group of diffeomorphisms $\varphi^\lambda: M\to M$ associated with a vector field $\xi\in {\cal X}(M)$ on the collision bundle $T^*C$ defined in Eq.~(\ref{Eq:TCDef}). For this, we note first that $T^* C$ is a submanifold of the bundle
\begin{equation}
{\cal B}_4 := \{ (x,p_1,p_2,p_3,p_4) : x\in M, (p_1,p_2,p_3,p_4)\in T_x^* M \},
\end{equation}
on which we can define the lifted flow in an analogous way as in Eq.~(\ref{Eq:LiftedFlow}), that is,\footnote{With a slight abuse of notation we shall denote this lift again by $\hat{\varphi}^\lambda$.}
\begin{equation}
\hat{\varphi}^\lambda: {\cal B}_4\to {\cal B}_4,\quad
(x,p_1,p_2,p_3,p_4)\mapsto 
 (\varphi^\lambda(x), [(d\varphi^\lambda_x)^*]^{-1}(p_1),[(d\varphi^\lambda_x)^*]^{-1}(p_2),[(d\varphi^\lambda_x)^*]^{-1}(p_3),[(d\varphi^\lambda_x)^*]^{-1}(p_4) ).
\label{Eq:LiftetFlowB4}
\end{equation}
It is straightforward to verify that $\hat{\varphi}^\lambda$ defines a (local) one-parameter group of diffeomorphisms on ${\cal B}_4$. Denote by $\hat{\xi}$ the associated infinitesimal generator. In adapted local coordinates it can be written as
\begin{equation}
\hat{\xi} = \xi^\mu\frac{\partial}{\partial x^\mu} 
 - \frac{\partial\xi^\alpha}{\partial x^\mu}\left[
 (p_1)_\alpha\frac{\partial}{\partial (p_1)_\mu} 
 + ( p_2)_\alpha\frac{\partial}{\partial (p_2)_\mu} 
+ (p_3)_\alpha\frac{\partial}{\partial (p_3)_\mu} 
+ (p_4)_\alpha\frac{\partial}{\partial (p_4)_\mu} 
\right],\qquad
\xi = \xi^\mu\frac{\partial}{\partial x^\mu},
\label{Eq:xihatB4}
\end{equation}
which generalizes the expression~(\ref{Eq:xihat}) for ${\cal B}_4$.

Next, we claim that the flow $\hat{\varphi}^\lambda$ leaves the collision bundle $T^*C$ invariant, provided that $\xi$ is a Killing vector field on $(M,g)$. For this, we notice that for each $i,j\in \{ 1,2,3,4 \}$ the functions $\mathcal{F}_{ij}: {\cal B}_4\to \Real$ defined by
\begin{equation}
\mathcal{F}_{ij}(x,p_1,p_2,p_3,p_4) := g_x^{-1}(p_i,p_j),
\label{Eq:FijDef}
\end{equation}
satisfy
\begin{equation}
(\pounds_{\hat{\xi}}\mathcal{F}_{ij})(x,p_1,p_2,p_3,p_4) = (\pounds_\xi g^{-1})(p_i,p_j),
\end{equation}
such that $\hat{\xi}$ leaves $\mathcal{F}_{ij}$ invariant if $\xi$ is a Killing vector field. Furthermore, by the linearity and invertibility of the operators $(d\varphi^\lambda_x)^*$, it is also clear that $\hat{\varphi}^\lambda$ leaves the equation $p_1 + p_2 - p_3 - p_4 = 0$ and the inequality $p_1\neq p_2$ invariant. These observations imply that the flow $\hat{\varphi}^\lambda$ leaves $T^* C$ invariant and that $\hat{\xi}$ is tangent to $T^* C$ at each of its point provided that $\xi$ is a Killing vector field on $(M,g)$.

In analogy to Theorem~\ref{Thm:TsLieIdentity} one has:

\begin{theorem}
\label{Thm:TsLieIdentityCollision}
Suppose $\xi\in {\cal X}(M)$ is a Killing vector field on $(M,g)$ and let $F\in C^\infty_0(T^*C)$. Then, it follows that
\begin{equation}
\pounds_\xi \int\limits_{C_x} F(x,p_1,p_2,p_3,p_4)\eta_{C_x} 
= \int\limits_{C_x} \pounds_{\hat{\xi}} F(x,p_1,p_2,p_3,p_4)\eta_{C_x},
\label{Eq:TsLieIdentityCollision}
\end{equation}
where $\eta_{C_x}$ is the induced volume element on $C_x$ defined in Eq.~(\ref{Eq:mux}).
\end{theorem}

\proof
As in the proof of the previous theorem we use the mixed local coordinates $(x^\mu,p_{\hat{\alpha}})$ where $p_{\hat{\alpha}}$ refer to the components of $p$ with respect to an orthonormal basis $\{ e_{\hat{\alpha}} \}$ of vector fields. In terms of these coordinates, the complete lift of $\xi$ on ${\cal B}_4$ has the following representation:
\begin{equation}
\hat{\xi} = \xi^\mu\frac{\partial}{\partial x^\mu} 
 + (\pounds_\xi e_{\hat{\alpha}})^\mu\sum\limits_{j=1}^4( p_j)_\mu\frac{\partial}{\partial (p_j)_{\hat{\alpha}}}.
\end{equation}
Since $\xi$ is a Killing vector field one can choose $\pounds_\xi e_{\hat{\alpha}} = 0$, and then the theorem follows by expressing the volume form $\eta_{C_x}$ in terms of orthonormal components, as in Eq.~(\ref{Eq:mux}).
\qed

With the help of the previous theorem one can show:

\begin{proposition}
\label{Prop:SymCollisionTerm}
Let $\xi\in {\cal X}(M)$ be a Killing vector field on $(M,g)$, and let $W\in {\cal F}(T^*C)$ be a smooth function on the collision bundle. Then, the collision term $C_W[f,h]$ defined in Eq.~(\ref{Eq:CollisionTerm1}) satisfies
\begin{equation}
\pounds_{\hat{\xi}} C_W[f,h] 
 = C_{\pounds_{\hat{\xi}} W} [f,h] + C_W[\pounds_{\hat{\xi}} f,h] + C_W[f,\pounds_{\hat{\xi}} h]
\end{equation}
for all $f,h\in C_0^\infty(\Gamma_m^+)$
\end{proposition}

\proof
To prove the statement, we multiply $C_W[f,h]$ with a smooth function $\Psi\in C_0^\infty(\Gamma_m^+)$ and integrate the result over the mass hyperboloid $P_x^+(m)$:
\begin{eqnarray}
&& \int\limits_{P_x^+(m)} C_W[f,h]\Psi(x,p) \dvol_x(p)
\nonumber\\
 &=& \frac{1}{2}\int\limits_{C_x}
\left[  W_x(p_1^*+p_2^*\mapsto p_1+p_2) f(x,p_1^*) h(x,p_2^*) 
 - W_x(p_1+p_2\mapsto p_1^*+p_2^*) f(x,p_1) h(x,p_2) \right]\Psi(x,p_1) \eta_{C_x},
\label{Eq:IntegralIdentity}
\end{eqnarray}
where we have used the expression~(\ref{Eq:muxBisVol}) for $\eta_{C_x}$. Applying the operator $\pounds_\xi$ on both sides of this equation, using Eqs.~(\ref{Eq:TsLieIdentity},\ref{Eq:TsLieIdentityCollision}) and the integral identity~(\ref{Eq:IntegralIdentity}) again yields
\begin{equation}
\int\limits_{P_x^+(m)} \left( \pounds_{\hat{\xi}}  C_W[f,h] \right)\Psi(x,p) \dvol_x(p)
 = \int\limits_{P_x^+(m)} \left(
 C_{\pounds_{\hat{\xi}} W} [f,h] + C_W[\pounds_{\hat{\xi}} f,h] + C_W[f,\pounds_{\hat{\xi}} h]
\right) \Psi(x,p) \dvol_x(p).
\end{equation}
Since $\Psi\in C_0^\infty(\Gamma_m^+)$ is arbitrary, the statement of the proposition follows.
\qed

\subsection{Consequences for the relativistic Boltzmann equation}

With the help of the previous results we can answer the question posed below Eq.~(\ref{Eq:Gsymmetric}): suppose that $\xi\in {\cal X}(M)$ is a Killing vector field on $(M,g)$, what are the conditions that guarantee that $[L_F,\hat{\xi}_a]$ and $\pounds_{\hat{\xi}_a} C_W[f,f]$ vanish if $f$ satisfies~(\ref{Eq:Gsymmetric}), such that no further integrability conditions arise from solving the relativistic Boltzmann equation?

Regarding the vanishing of the commutator $[L_F,\hat{\xi}_a]$, we already know from Proposition~\ref{Prop:CompleteLift}(v) that $[L,\hat{\xi}] = 0$. Hence, it remains to analyze the commutator $[V,\hat{\xi}_a]$ involving the vertical part of $L_F$. Using Eqs.~(\ref{Eq:LiouvilleVFChargedBis},\ref{Eq:xihat}) a short calculation reveals that
\begin{equation}
[V,\hat{\xi}_a] = -q(\pounds_{\xi_a} F)_\alpha{}^\beta p_\beta\frac{\partial}{\partial p_\alpha},
\end{equation}
which shows that $[L_F,\hat{\xi}_a] = 0$ if and only if $\pounds_{\xi_a} F = 0$ for all $a=1,2,\ldots,r$, that is, if and only if the electromagnetic field tensor is $G$-invariant.

Considering the second condition on the collision term, it follows from Proposition~\ref{Prop:SymCollisionTerm} that $\pounds_{\hat{\xi}_a} C_W[f,f] = 0$ for all $f$ satisfying Eq.~(\ref{Eq:Gsymmetric}) if and only if the transition probability density $W$ is $G$-invariant in the sense that
\begin{equation}
\pounds_{\hat{\xi}_a} W = 0,\qquad a = 1,2,\ldots,r.
\end{equation}
If $W$ is a function depending only on the Mandelstam variables $s,t,u$ (see Eqs.~(\ref{Eq:Mandelstams},\ref{Eq:Mandelstamt},\ref{Eq:Mandelstamu})) these conditions are automatically satisfied since the functions $\mathcal{F}_{ij}$ defined in Eq.~(\ref{Eq:FijDef}) are invariant with respect to the lifted Killing flow. However, more general conditions are also allowed. For example, if $t$ is a $G$-invariant contravariant tensor field on $M$ and $W$ is an algebraic function depending only on expressions of the form
\begin{equation}
t(p_1,\ldots,p_1,p_2,\ldots,p_2,p_3,\ldots,p_3,p_4,\ldots,p_4),
\end{equation}
then $W$ is automatically $G$-invariant. Such a tensor field $t$ can be constructed from the metric field (in which case one recovers the Mandelstam variables), the electromagnetic field tensor $F$, or any other $G$-invariant field.

We summarize the most relevant findings of this appendix. A Lie-group $G$ of isometries on $(M,g)$ can be naturally lifted to the cotangent bundle $T^* M$ and the collision bundle $T^* C$. The distribution function $f$ is called $G$-symmetric if the condition~(\ref{Eq:Gsymmetric}) is satisfied, and this condition is compatible with the relativistic Boltzmann equation, provided the electromagnetic field $F$ and the transition probability density $W$ are $G$-invariant. The latter condition is automatically satisfied if $W$ is a function of the relative velocity $g$ and the scattering angle $\Theta$ only, although more general forms for $W$ are also possible.

\bibliographystyle{unsrt}
\bibliography{refs_kinetic}

\end{document}